\documentclass[authorversion=true,nonacm,acmsmall]{acmart}
\setcopyright{none}
\startPage{1}

\usepackage{adjustbox,rotating,booktabs,tcolorbox} 

\usepackage{quiver}
\usepackage{etoolbox}
\usepackage{fullshort}
\setboolean{fullpaper}{true}

\usepackage[toc,page]{appendix}
\usepackage{savesym}
\savesymbol{widering}
\usepackage{xspace,mathpartir,xparse,url,enumitem,multirow,wrapfig}
\usepackage{thmtools,thm-restate}
\usepackage{ebmath,ebproof,ebutf8}
\usepackage{listings}

\newcommand{\isempty}[1]%
{
  \ifthenelse{\equal{#1}{}}%
    {EMPTY}
    {FULL, it contains the string '#1'}
}

\usepackage{hyperref}
\newlist{enumerati}{enumerate}{10}
\setlist[enumerati]{label=\emph{(\roman*)}, ref=\emph{(\roman*)}}
\newlist{enumeratiwide}{enumerate}{10}
\setlist[enumeratiwide]{label=\emph{(\roman*)}, ref=\emph{(\roman*)}, wide, labelindent=0pt}
\newlist{enumerata}{enumerate}{10}
\setlist[enumerata]{label=\emph{(\alph*)}, ref=\emph{(\alph*)}}
\newlist{enumeratawide}{enumerate}{10}
\setlist[enumeratawide]{label=\emph{(\alph*)}, ref=\emph{(\alph*)}, wide, labelindent=0pt}
\newlist{enumeratewide}{enumerate}{10}
\setlist[enumeratewide]{label=\emph{\arabic*.}, ref=\emph{(\arabic*)}, wide, labelindent=0pt, before = \leavevmode}
\newlist{enumeratp}{enumerate}{10}
\setlist[enumeratp]{label=\emph{(\arabic*)}, ref=\emph{(\arabic*)}}
\usepackage{environ}

\newcommand{\repeattheorem}[1]{%
  \begingroup
  \renewcommand{\thetheorem}{\ref{#1}}%
  \expandafter\expandafter\expandafter\theorem
  \csname reptheorem@#1\endcsname
  \endtheorem
  \endgroup
}

\NewEnviron{reptheorem}[1]{%
  \global\expandafter\xdef\csname reptheorem@#1\endcsname{%
    \unexpanded\expandafter{\BODY}%
  }%
  \expandafter\theorem\BODY\unskip\label{#1}\endtheorem
}

\newcommand{\repeatcorollary}[1]{%
  \begingroup
  \renewcommand{\thecorollary}{\ref{#1}}%
  \expandafter\expandafter\expandafter\corollary
  \csname repcorollary@#1\endcsname
  \endcorollary
  \endgroup
}

\NewEnviron{repcorollary}[1]{%
  \global\expandafter\xdef\csname repcorollary@#1\endcsname{%
    \unexpanded\expandafter{\BODY}%
  }%
  \expandafter\corollary\BODY\unskip\label{#1}\endcorollary
}

\newcommand{\repeatlemma}[1]{%
  \begingroup
  \renewcommand{\thelemma}{\ref{#1}}%
  \expandafter\expandafter\expandafter\lemma
  \csname replemma@#1\endcsname
  \endlemma
  \endgroup
}

\NewEnviron{replemma}[1]{%
  \global\expandafter\xdef\csname replemma@#1\endcsname{%
    \unexpanded\expandafter{\BODY}%
  }%
  \expandafter\lemma\BODY\unskip\label{#1}\endlemma
}

\newcommand{\repeatproposition}[1]{%
  \begingroup
  \renewcommand{\theproposition}{\ref{#1}}%
  \expandafter\expandafter\expandafter\proposition
  \csname repproposition@#1\endcsname
  \endproposition
  \endgroup
}

\NewEnviron{repproposition}[1]{%
  \global\expandafter\xdef\csname repproposition@#1\endcsname{%
    \unexpanded\expandafter{\BODY}%
  }%
  \expandafter\proposition\BODY\unskip\label{#1}\endproposition
}

\usepackage{tikz-cats}
\usetikzlibrary{matrix,cd}
\tikzcdset{arrow style=tikz}
\tikzcdset{every arrow/.append style = -{Computer Modern Rightarrow[]}}
\tikzset{every arrow/.append style = -{Computer Modern Rightarrow[]}}
\tikzset{
  labelslt/.style n args={2}{%
    labelat={[left]{$\scriptstyle #1$}}{.45},%
    labelat={[right]{$\scriptstyle #2$}}{.55}%
    } %
  }
\tikzset{
  labelsltat/.style n args={4}{%
    labelat={[left]{$\scriptstyle #1$}}{#3},%
    labelat={[right]{$\scriptstyle #2$}}{#4}%
    } %
  }
\tikzset{twoof/.style={twocenter={#1}}}

\newcounter{trou}[figure]
\renewcommand{\diaggrandhauteur}{.6}
\renewcommand{\diaggrandlargeur}{1}

\newcommand{\xto}[1]{\xrightarrow{#1}}

\newcommand{\xot}[1]{\xleftarrow{#1}}

\newcommand{\ens}[1]{\{ #1 \}}

\newcommand{\op}[1]{{#1}^{\mathit{op}}}
\newcommand{\myop}{\mathit{op}}

\DeclareMathOperator{\alg}{-\,\mathbf{alg}}
\DeclareMathOperator{\Alg}{-\,\mathbf{Alg}}

\newcommand{\pt}{∘}

\DeclareMathOperator{\Mon}{-\,\mathbf{Mon}}

\DeclareMathOperator{\id}{id}

\DeclareMathOperator{\bool}{bool}

\DeclareMathOperator{\tl}{tl}

\DeclareTextMath\tmlambda[lambda]{\lambda}

\newcommand\Tstrut{\rule{0pt}{3ex}}

\newcommand\Bstrutbas{\rule[-1.2ex]{0pt}{0pt}}

\newcommand{\subst}{\mathit{subst}}

\newcommand{\add}{\mathit{add}}
\newcommand{\mul}{\mathit{mul}}
\newcommand{\app}{\mathit{app}}
\newcommand{\lapp}{\mathit{lapp}}
\newcommand{\abs}{\mathit{abs}}
\newcommand{\diff}{\mathit{diff}}
\newcommand{\plus}{\mathit{plus}}

\usepackage{bbding}
\newcommand{\fleur}{*}
\newcommand{\iso}{≅}
\newcommand{\red}[1]{#1}

\newcommand{\liftleft}[1]{{{}^{#1}{↑}}}

\newcommand{\ajustedroit}[2][1]{\adjustbox{max width=#1\columnwidth,max height=.95\textheight}{#2}}%
\newcommand{\ajustetourne}[1]{\adjustbox{max width=\columnwidth,max height=.95\textheight}{
    \begin{turn}{90}
      #1
    \end{turn}
  }}%

\newcommand{\alert}[1]{\textbf{#1}}

\NewDocumentCommand{\GammaInv}{}{\mathpalette\doGammaInv\relax}
\makeatletter
\NewDocumentCommand{\doGammaInv}{mm}{%
  \reflectbox{$\m@th#1\Gamma$}%
}
\makeatother
\renewcommand{\daleth}{\GammaInv}

\newcommand{\montitre}{A unified treatment of structural definitions
on syntax \\ for capture-avoiding substitution, context application,
named substitution, partial differentiation, and so on}
\newcommand{\monsoustitre}{} \newcommand{\montitrecourt}{A unified
treatment of structural definitions on syntax}

\newcommand{\monabstract}{ We introduce a category-theoretic
abstraction of a syntax with auxiliary functions, called an admissible
monad morphism. Relying on an abstract form of structural recursion,
we then design generic tools to construct admissible monad morphisms
from basic data.  These tools automate ubiquitous standard patterns
like (1)~defining auxiliary functions in successive, potentially
dependent layers, and (2)~proving properties of auxiliary functions by
induction on syntax.  We cover significant examples from the
literature, including the standard lambda-calculus with
capture-avoiding substitution, a lambda-calculus with binding
evaluation contexts, the lambda-mu-calculus with named substitution, and
the differential lambda-calculus.  }

\newcommand{\shift}{\mathrm{shift}}

\newtheorem{notation}{Notation}
\newtheorem{remark}{Remark}

\definecolor{bois}{rgb}{.5,0,0}
\definecolor{vert}{rgb}{0,0.6,0}
\definecolor{rouge}{rgb}{0.8,0,0}
\definecolor{violet}{rgb}{0.8,0,0.4}
\definecolor{marron}{rgb}{0.4,0.4,0}
\definecolor{bleu}{rgb}{0,0,0.6}
\AtBeginDocument
{
\newcommand{\colorie}[2]{{\color{#1} #2}}
\renewcommand{\vert}[1]{\colorie{vert}{#1}}
\newcommand{\rouge}[1]{\colorie{red}{#1}}
\newcommand{\orange}[1]{\colorie{orange}{#1}}

}

\author[T.\ Hirschowitz]{Tom Hirschowitz} %
\affiliation{ \institution{Univ. Grenoble Alpes, Univ. Savoie Mont
    Blanc, CNRS, LAMA, 73000}
  \city{Chambéry}
  \country{France}
}

\author[A.\ Lafont]{Ambroise Lafont}
\affiliation{
  \institution{University of Cambridge}
  \city{Cambridge}
  \country{United Kingdom}
}

\title[\montitrecourt]{\montitre}
\subtitle{\monsoustitre}



\begin{document}
\pagestyle{fancy} 
\fancyfoot[C]{\thepage}

\begin{abstract}
  \monabstract
\end{abstract}
\maketitle              

\section{Introduction}\label{s:intro}
  \paragraph{Motivation}
  The literature offers several initial-algebra semantics frameworks
  for generating and reasoning about syntax with variable binding
  (e.g.,
  \cite{fiore:presheaf,PittsAM:newaas,hofmann:presheaf,FioreHurEquational}).
  Still, even state-of-the-art frameworks lack some expressiveness to
  suit the working operational semanticist's needs.  One typical
  limitation, which is the topic of active
  research~\cite{GratzerSterling,CoragliaDiLiberti}, concerns some
  complex typing features like dependent types.

  In this paper, we are concerned with a different limitation:
  although existing frameworks do explain capture-avoiding
  substitution satisfactorily, they largely ignore the numerous
  similar auxiliary operations on syntax, like context application
  (e.g., in ML~\cite[Chapter~10]{ATAPL}), named substitution in the
  $λμ$-calculus~\cite{Parigot} (a.k.a.\ named
  application~\cite{Iouri}), partial differentiation in the
  differential $λ$-calculus~\cite{LambdaDiff}, and so on.

  One way of dealing with such auxiliary operations is to build them
  into the syntax, and view their defining equations as an equational
  theory, by which to quotient the initial model (see, e.g.,
  \cite{FioreHurEquational,GratzerSterling}). The problem with this
  approach is that it lacks the basic, inductive construction of the
  initial model. Or in other words, it misses the fact that auxiliary
  operations are… auxiliary, i.e., that they are admissible ($=$
  encodable) in the initial model.

  \paragraph{Contributions}
  In this paper, we propose a categorical foundation for such
  admissible operations, called \alert{admissible monad morphisms},
  together with a general framework for constructing them and
  reasoning about them. In particular, the framework offers tools to
  \begin{enumerata}
  \item \label{item:layers} define successive layers of auxiliary
    operations, each layer potentially depending on previous ones (as,
    e.g., in the differential $λ$-calculus~\cite{LambdaDiff}),
  \item \label{item:benign} automatically derive \alert{benign}
    equations, i.e., equations involving auxiliary operations that are
    satisfied by the syntax (such as, e.g., associativity of
    substitution $e[σ][θ] = e[σ[θ]]$).
  \end{enumerata}
  \begin{remark}
    The framework does not yet allow the user to automatically prove
    that auxiliary operations are compatible with \alert{severe}
    equations, i.e., equations not involving auxiliary operations (as,
    again, in the differential $λ$-calculus), whose initial model thus
    may be a proper quotient of the syntax.
  \end{remark}

  We illustrate the expressiveness of the framework by reconstructing
  (1)~Fiore et al.'s~[\citeyear{fiore:presheaf}] approach as a special
  case, and (2)~a few concrete examples from the literature, notably a
  $λ$-calculus with explicit
  substitution~\cite{DBLP:conf/rta/Accattoli19} (which involves
  variable capture by evaluation contexts) and the differential
  $λ$-calculus~\cite{LambdaDiff} (without quotienting by structural
  equations). We also demonstrate that our framework is not tied to
  any particular setting for modelling variable binding, by
  reconstructing pure $λ$-calculus with capture-avoiding substitution
  in De Bruijn representation (§\ref{ss:db}).



  

  The framework consists of the following components.
  \begin{enumeratewide}
  \item First, we offer a fundamental, very easy construction of
    admissible monad morphisms from suitable monad distributive laws
    in the sense of~\citet{BeckDistlaws}.  A monad
    distributive law is a natural transformation $TS→ST$, where we
    think of $S$ as the monad of basic operations, hence call it the
    \alert{basic} monad, and of $T$ as the one of auxiliary
    operations, hence call it the \alert{auxiliary} monad.
    We show that, whenever $T$ preserves the initial object,
    any monad distributive law induces an admissible
    monad morphism $S → ST$ to the composite monad.
  \item Then, we design a toolbox for constructing such
    distributive laws from more basic data:
    \begin{enumerati}
    \item We first introduce a simple notion of signature for
      distributive laws. A signature consists of an endofunctor $Σ$,
      equipped with an abstract analogue of a structurally recursive
      definition, called a \alert{simple structural law}.  From such a
      signature, we generate a distributive law of the desired form,
      whose basic monad is the free monad $Σ^*$ on $Σ$.  We also
      establish a simple characterisation of \alert{augmented
        algebras}, i.e., $Σ$-algebras featuring the specified
      auxiliary operations in a compatible way.

      Simple structural laws are expressive enough to equip
      (potentially binding) syntax with capture-avoiding substitution,
      but not to tackle Features~\ref{item:layers}
      and~\ref{item:benign} above.
    \item We then introduce a second notion of signature, which is an
      incremental form of the first.  A signature consists of a
      distributive law $δ∶ TS → ST$ whose basic monad $S$ is free on
      some endofunctor $Σ$, equipped with a so-called
      \alert{incremental structural law} over it.  The idea is that we
      have already specified a few auxiliary operations over $S$,
      bundled as the monad $T$, and an incremental structural law is
      like a structurally recursive definition by pattern-matching on
      basic operations, whose body may use all auxiliary
      operations. From this, we construct a new distributive law with
      the same basic monad, but whose auxiliary monad contains new
      operations. We again present a simple explicit description of
      augmented algebras.

      Incremental structural laws are expressive enough to equip
      syntax with incremental layers of auxiliary operations, hence
      implement Feature~\ref{item:layers} above.
    \item We finally address the benign equations
      feature~\ref{item:benign}.  For this, we again start from a
      distributive law with a free basic monad, and introduce
      \alert{structural equational systems}, which, roughly, consist
      of:
      \begin{itemize}
      \item (an abstract form of) equations on terms, potentially
        using auxiliary operations, and
      \item data that ensures that the equations will be satisfied by
        the syntax.
      \end{itemize}
      From this, we construct a distributive law with the same basic
      monad.
    \end{enumerati}
  \end{enumeratewide}

  \paragraph{Plan}
  We start in~§\ref{s:admissible} by introducing admissible monad
  morphisms and showing how they may be derived from suitable monad
  distributive laws.  We then study simple structural laws
  in~§\ref{s:ssl}, incremental structural laws in~§\ref{s:isl}, and
  benign equations in~§\ref{s:benign}, giving applications along the
  way. Finally, we conclude and give some perspectives
  in~§\ref{s:conclu}.

  \paragraph{Related work}
  Our framework abstracts over the idea of a definition by structural
  recursion, to categories other than sets.  The abstraction
  emphasises the idea that the syntax equipped with auxiliary
  functions is initial in a category of augmented algebras.  Our
  framework is directly inspired by (and abstracts over) Fiore et
  al.'s~\cite{fiore:presheaf,DBLP:conf/lics/Fiore08}.  A notable
  difference is that Fiore et al.\ rely on a technical notion called
  \alert{pointed strong} endofunctors, which intuitively amounts to
  assuming that the argument of the endofunctor already has part of
  the desired structure -- in this case, variables.  The corresponding
  instance of our abstract framework avoids this astute trick by
  taking the free structure. Technically, this is visible in the
  presence of $X⊗S(Y)$ in the codomain of our simple structural law
  in~§\ref{ss:lambda}. Indeed, taking $S(Y)$ instead of merely $Y$ is
  what allows us to use $σ^↑$ in the definition.

  Our work is related to the idea of
  traversals~\cite{DBLP:journals/pacmpl/AllaisA0MM18}.  The latter are
  however limited to categories of families of sets indexed by some
  fixed set of types, and tailored for capture-avoiding substitution.

  \paragraph{Notation and prerequisites}
  We often conflate natural numbers with the corresponding finite
  ordinals, or some choice of equipotent set, hopefully made clear by
  the context.

  We assume basic knowledge of category theory~\cite{MacLane:cwm} and
  locally finitely presentable categories~\cite{Adamek} (although the
  latter may be ignored on a superficial reading).  By default, our
  categories are locally small, although we occasionally repeat it for
  emphasis.  We let $𝐂𝐀𝐓$ denote the category of (locally small)
  categories.


  We generally denote initial objects by $∅$, relying on context to
  infer the corresponding category.  For any category $𝐂$ and object
  $c ∈ 𝐂$, we denote by $𝐂/c$ the slice category over $c$.

  On any category $𝐂$ with binary coproducts, for any object $E ∈ 𝐂$,
  we denote the corresponding \alert{option} functor by $O_E$, i.e.,
  $O_E(C) = C+E$, for some choice of coproducts.

  For any endofunctor $F$, we denote its category of algebras by
  $F\alg$.  For any monad $T$, we denote its category of monad
  algebras by $T\Alg$; it is then a full subcategory of $T\alg$.  We
  denote the coproduct of two monads $T$ and $T'$ by $T⊕T'$, to
  distinguish it from the coproduct $T+T'$ of underlying endofunctors.

  Any finitary endofunctor $F$ on a locally finitely presentable
  category $𝐂$ generates a free monad, that we denote by $F^*$.  In
  particular, the initial algebra is denoted by $F^*∅$. As is
  well-known, this may often be computed as a directed colimit which
  we denote by $μE.F(E)$.  We furthermore denote the unit $F → F^*$ by
  $η_F$.  Please also note that $F^*\Alg$ is isomorphic to $F\alg$:
  the isomorphism maps an algebra $F^*c → c$ to $Fc → F^*c → c$.
  
  We write $Δ$ for the diagonal functor $𝐂 → 𝐂×𝐂$ mapping an object
  $c$ of a category $𝐂$ to $(c,c)$.  We use the letters $Γ$ and $Θ$ to
  denote bifunctors $𝐂×𝐂 → 𝐂$. We sometimes write $Γ$ instead of
  $ΓΔ$. E.g., we talk about $Γ$-algebras instead of $ΓΔ$-algebras, and
  denote $(ΓΔ)^*$ by $Γ^*$. We furthermore sometimes denote $Γ(X,Y)$
  by $Γ_YX$. If $F$ is a functor to $𝐂$, the bifunctor mapping $(X,Y)$
  to $Γ(X,FY)$ is denoted by $Γ_F$.

  Moreover, for denoting the components of a natural transformation
  $α∶ F → G$, we freely switch between $α_C$ and $α C$, and we denote
  horizontal composition in $𝐂𝐀𝐓$ by mere juxtaposition. E.g.,.  given
  any $f∶ C → D$ in the domain category, we also write $α_f$, or
  $α f$, for either composite $α_D ∘ F(f) = G(f) ∘ α_C$.

  The middle dot $·$ is overloaded: it may have the followin three
  meanings, depending on context:
  \begin{itemize}
  \item for any presheaf $F∶ \op{𝐂} → 𝐒𝐞𝐭$, morphism $f∶ c → d$ in
    $𝐂$, and element $x ∈ F(d)$, we define $x · f ≔ F(f)(x)$;
  \item for any set $X$ and object $c$ of any category with enough
    coproducts, we denote by $X · c$ the $X$-fold coproduct
    $∑_{x ∈ X} c$; and finally,
  \item the middle dot is used in the syntax of differential
    $λ$-calculus.    
  \end{itemize}
  \begin{definition}
    A \alert{monad morphism} $S → T$ between monads $S$ and $T$ on a
    given category $𝐂$ is a natural transformation $α∶ S → T$
    commuting with multiplication and unit, i.e., making the following
    diagrams commute.
    \begin{center}
      \diag{%
        SS \& T T \\
        S \& T
      }{%
        (m-1-1) edge[labela={α α}] (m-1-2) %
        edge[labell={μ^S}] (m-2-1) %
        (m-2-1) edge[labelb={α}] (m-2-2) %
        (m-1-2) edge[labelr={μᵀ}] (m-2-2) %
      }
      \hfil
      \diag{%
        \& \id_𝐂 \\
        S \& \& T
      }{%
        (m-1-2) edge[labelal={η^S}] (m-2-1) %
        (m-2-1) edge[labelb={α}] (m-2-3) %
        (m-1-2) edge[labelar={ηᵀ}] (m-2-3) %
      }
    \end{center}
    We let $𝐌𝐧𝐝(𝐂)$ denote the category of monads in $𝐂$ and monad
    morphisms between them, and $𝐌𝐧𝐝_f(𝐂)$ denote its full subcategory
    spanned by finitary monads.
  \end{definition}

  \section{Admissible monad morphisms and distributive laws}\label{s:admissible}
  In this section, we introduce admissible monad morphisms,
  and show how to construct them from monad distributive laws.
  
  \subsection{Admissible monad morphisms}
  \begin{definition}
    Let $𝐂$ have an initial object.  A morphism $α∶ R → S$ of monads
    on $𝐂$ is \alert{admissible} iff its component $α_{∅}∶ R∅ → S∅$ at
    the initial object is an isomorphism.
  \end{definition}
  \begin{remark}
  Intuitively, thinking of $R$ and $S$ as algebraic structures, or
  families of operations, $R(∅)$ and $S(∅)$ are the initial algebras,
  i.e., morally the syntaxes of both languages. The monad morphism
  translates each operation from $R$ to some (derived) operation from
  $S$.  Admissibility then amounts to the translation being an
  isomorphism.
\end{remark}

We will give the paradigmatic example of admissible morphism just
below, but before that, let us mention a different point of view on
admissible morphisms.

\begin{definition}[{\cite[§VI.3]{MacLane:cwm}}]
We call \alert{monadic} (over $𝐂$)
any functor isomorphic in $𝐂𝐀𝐓/𝐂$ to some forgetful functor
$T\Alg → 𝐂$, and let $𝐌𝐨𝐧𝐚𝐝𝐢𝐜/𝐂$ denote the full subcategory of
$𝐂𝐀𝐓/𝐂$ spanned by monadic functors.
\end{definition}

Given a category $𝐂$, the assignment mapping any monad $T$ on $𝐂$ to
the forgetful functor $Uᵀ∶ T\Alg → 𝐂$ extends to a functor
$𝐬𝐞𝐦∶ \op{𝐌𝐧𝐝(𝐂)} → 𝐌𝐨𝐧𝐚𝐝𝐢𝐜/𝐂$, and we have:
\begin{lemma}
  The functor $𝐬𝐞𝐦$ is an equivalence of categories.
\end{lemma}
\begin{proof}
  The functor $𝐬𝐞𝐦$ is essentially surjective by definition of monadic
  functors. It is also full and faithful by
  \cite[Proposition~5.3]{Barr1970CoequalizersAF}.
\end{proof}
\begin{proposition}\label{prop:admissible:street}
  Given a category $𝐂$ with initial object, monads $S$ and $T$ on $𝐂$,
  and a monad morphism $α∶ S → T$, the following are equivalent:
  \begin{enumeratiwide} 
  \item \label{item:admissible} $α$ is admissible;
    \item \label{item:preserves} $𝐬𝐞𝐦(α)∶ T\Alg → S\Alg$ preserves the initial object;
    \item \label{item:creates} $𝐬𝐞𝐦(α)∶ T\Alg → S\Alg$ creates the
      initial object in the sense of~\cite[§V.1]{MacLane:cwm}, which
      in this case means that the initial $S$-algebra $S∅$ possesses a
      unique $T$-algebra structure $a∶ TS∅ → S∅$
      making the triangle
      \begin{center}
        \diag{%
          SS∅ \& \& TS∅ \\
          \& S∅
        }{%
          (m-1-1) edge[labela={α_{S∅}}] (m-1-3) %
          edge[labelbl={μ^S_∅}] (m-2-2) %
          (m-1-3) edge[labelbr={a}] (m-2-2) %
        }
      \end{center}
      commute, which furthermore makes it initial in $T\Alg$.
    \end{enumeratiwide}
\end{proposition}
\begin{proof}
  See~§\ref{s:admissible:street}.
\end{proof}

\begin{example}\label{ex:fpt}
  Our motivating example is in fact a class of examples.  For any
  finitary, \alert{pointed strong} endofunctor $Σ$ on any nice monoidal
  category $(𝐂,⊗,I,α,λ,ρ)$ (the impatient reader may consult
  Definitions~\ref{def:pts} and~\ref{def:lfpmon} below for details,
  but these are not yet needed), \citet{fiore:presheaf} introduce a category $Σ\Mon$ of
  \alert{$Σ$-monoids}, which are objects $X ∈ 𝐂$ with both $Σ$-algebra
  structure $ΣX → X$ and monoid structure $I → X ← X⊗ X$, satisfying a
  standard coherence condition (Definition~\ref{def:SMon} below).

  They then show that the initial $(I+Σ)$-algebra, or equivalently the
  free $Σ$-algebra on $I$, admits a unique $Σ$-monoid structure, which
  makes it initial in $Σ\Mon$. Furthermore, both forgetful functors
  from $(I + Σ)\alg$ and $Σ\Mon$ to $𝐂$ are monadic, and there is an
  obvious forgetful functor $Σ\Mon → (I + Σ)\alg$.

  Thus, denoting by $(I + Σ)^*$ and $Σ^⊛$ the corresponding monads,
  their result may be read as proving that the induced monad morphism
  $$(I + Σ)^* → Σ^⊛$$
  is admissible.
\end{example}

\subsection{From distributive laws}
Let us now show how to construct admissible monad morphisms from
distributive laws. We first recall from~\cite{BeckDistlaws} that
  \begin{enumeratiwide}
  \item a \alert{monad distributive law} of $S$ over $T$ is a natural
    transformation $δ∶ TS → ST$, commuting with unit and
    multiplication of both monads $S$ and $T$, and that
  \item any such distributive law equips the composite functor $ST$
    with monad structure, which in particular makes the natural
    transformation $Sη^T∶ S → ST$ into a monad morphism.
\end{enumeratiwide}
The idea is to start from a distributive law with a suitable
constraint on $T$, namely that $T(∅) ≅ ∅$. Intuitively, $T$ may have
many operations, but no constants to feed them with. We name such
monads accordingly:
\begin{definition}
  A monad $T$ on a category with an initial object $∅$ is
  \alert{constant-free} iff its unit at $∅$, $ηᵀ_{∅}∶ ∅ → T(∅)$, is an
  isomorphism.
\end{definition}

\begin{proposition}
  For any distributive law $δ∶ TS → ST$ with $T$ constant-free, the
  monad morphism $Sηᵀ∶ S → ST$ is admissible.
\end{proposition}
\begin{proof}
  Immediate.
\end{proof}

This result might seem purely academic, but in fact all of our
applications arise in this way. The technical core of the paper thus
consists in designing tools to construct such distributive laws over
constant-free monads.

\section{Simple structural laws}\label{s:ssl}
In this section, we introduce our first construction of distributive
laws with constant-free auxiliary monad, from what we call simple
structural laws.  Let us start by working on a simple (perhaps
surprising) example, and abstract over it in the following subsection.
\subsection{On a simple example}\label{ss:addition}
We consider the unary (Peano) natural numbers, viewed as the initial
algebra of the endofunctor $Σ$ on sets defined by $Σ(X) = 1 + X$.  We
use standard syntax, i.e., $0$ for the constant and $s$ for successor.
Let us now consider the category, say $Σ\alg^{\add}$, of $Σ$-algebras $X$
equipped with a binary operation, denoted by $(x,y) ↦ x+y$, satisfying
the following equations.
  \begin{mathpar}
    s(x)+y = s(x+y) \and 0+y = y
  \end{mathpar}
  We call such an operation an addition.

  \begin{proposition}\label{prop:triv}
    The
    initial algebra $Σ^*(∅)$ is equipped with a unique addition, which
    makes it initial in the category of $Σ$-algebras with addition.
  \end{proposition}
  \begin{proof}
    An easy induction.
  \end{proof}
  Furthermore, the forgetful functor $U^{\add}∶ Σ\alg^{\add} → 𝐒𝐞𝐭$ is
  monadic, and the corresponding monad, say $S^{\add}$ maps any set
  $X$ of variables to terms generated from $0$ and all $x ∈ X$ by $s$
  and $+$, modulo the above equations.

  Now, for each set $X$, there is an obvious inclusion $SX ↪ S^{\add}X$.
  This family of inclusions induces a monad morphism $S → S^{\add}$,
  and it is admissible by Propositions~\ref{prop:triv}
  and~\ref{prop:admissible:street}.

  Let us now describe the monad $S^{\add}$ and its relation to $S$
  more carefully, which will lead us to distributive laws.

  Orienting equations, terms generated by $0$, $s$, and $+$ from a
  given set of variables, quotiented by the equations, have a normal
  form in which the first argument of an addition is either a
  variable, or a further addition.  Otherwise said, normal forms with
  variables in $X$ are generated by the following grammar:
  \begin{equation}
    \label{eq:grammar-add}
  \begin{array}{rcll}
    e & ::= & 0 ｜ s(e) ｜ a \\
    a & ::= & x ｜ a+e.
  \end{array}
  \end{equation}
  Clearly, if no variable is available, only the first two cases can
  occur in a term $e$, which together correspond precisely to the
  syntax generated by $S$.  We thus recover the fact that
  $S∅→ S^{\add}∅$ is an isomorphism.

  Let us replay this reasoning, at a slightly more abstract level.
  The starting point is to refine the arity of the auxiliary function
  (addition), here the endofunctor $X ↦ X²$, into a bifunctor
  $$\begin{array}{rcl}
    Γ∶ 𝐒𝐞𝐭² & → & 𝐒𝐞𝐭 \\
    (X,Y) & ↦ & X×Y\rlap{.}
    \end{array}$$
    Intuitively, this allows us to distinguish the ``decreasing''
    occurrence of the argument.

    By mimicking the recursive definition of addition on $S∅$, we then
    define the following natural transformation
    \begin{equation}
      \begin{array}{rcl}
        d_{X,Y}∶ Γ(Σ(X),Y) & → & S(\rouge{Γ(X,Y)}+\vert{Y}) \\
        (0,y) & ↦ & η^S(\vert{in₂}(y)) \\
        (s(x),y) & ↦ & s (η^S (\rouge{in₁}(x,y)))\rlap{,}
      \end{array}
      \label{eq:d}
    \end{equation}
    where we write elements of $Σ(Z)$ as terms of depth $1$ in $S(Z)$,
    for any set $Z$.  (We will use a slightly more general codomain in the
    abstract case.)  Elements of the domain are thought of as patterns
    in the first argument, and the natural transformation maps them to
    ``definition bodies'', which are basic terms generated from the
    auxiliary arguments in $Y$, and potentially a ``recursive call''
    in $Γ(X,Y)$ --- hence with ``strictly smaller'' main argument.

    We furthermore define a monad $T ≔ Γ_S^\fleur$, which, we recall
    from~§\ref{s:intro}, denotes $(Γ_SΔ)^*$.  Concretely,
    $T(X) ≔ μA.(X + A×S(A))$ corresponds to the syntactic category $a$
    above: terms consist of additions in normal form, i.e., additions
    whose first argument may be a variable or some further addition,
    inductively, and whose second argument is an arbitrary
    expression. We thus have $S^{\add} ≅ S ∘ T$.

    Finally, we construct a distributive law $TS → ST$, by induction
    from the natural transformation~\eqref{eq:d}, which repeatedly
    applies oriented equations until some normal form is reached.  We
    will be done if we prove $T(∅) ≅ ∅$, for then
    $$S^{\add}(∅) ≅ ST(∅) ≅ S(∅)\rlap{,}$$
    as desired.  But $T(∅) ≅ ∅$ follows from the next lemma with
    $Θ = Γ_S$.
    \begin{lemma}
      For any category $𝐂$ and functor $Θ∶ 𝐂² → 𝐂$, if $Θ$ is
      cocontinuous in its first argument and finitary in its second
      argument, then the initial object possesses a unique
      $ΘΔ$-algebra structure, which furthermore makes it initial in
      $ΘΔ\alg$.
    \end{lemma}
    \begin{proof}
      By cocontinuity, $Θ_∅∅$ is initial, hence there is a unique
      morphism $Θ_∅∅ → ∅$. The rest follows easily.
    \end{proof}

    It may not be entirely obvious that $Γ_S$ is cocontinuous in its
    first argument, but $X × Y$ is isomorphic to the coproduct
    $∑_{y ∈ Y} X$, which is cocontinuous by interchange of colimits.
    Most of our examples below follow the same pattern.

\subsection{The abstract case}

In this subsection, we introduce the general notion of simple
structural law, and construct, from any such law, a monad distributive
law with constant-free auxiliary monad.
  \begin{definition}
    A \alert{simple structural law} on a given locally finitely
    presentable category $𝐂$ consists of
    \begin{itemize}
    \item a \alert{basic} finitary endofunctor $Σ∶ 𝐂 → 𝐂$,
    \item an \alert{auxiliary} functor $Γ∶ 𝐂² → 𝐂$ which is
      cocontinuous in its first argument and finitary in its second
      argument, and
    \item a natural transformation
      $$d_{X,Y}∶ Γ_Y(Σ(X)) → S(Γ_{S(Y)}(X) + X + Y)\rlap{,}$$
      where $S ≔ Σ^*$ denotes the free monad generated by $Σ$.
  \end{itemize}
  \end{definition}
  \begin{remark}
    Comparing with the example~\eqref{eq:d} of the previous section,
    we have added a new base case $X$, and replaced $Y$ with the more
    general $S(Y)$ in the recursive call.
  \end{remark}

  Let us now state the construction result, relying on the following
  lemma.
\begin{lemma}\label{lem:constantless}
    For any finitary bifunctor $B∶ 𝐂² → 𝐂$ which is cocontinuous in
    its first argument, $B^\fleur$ is constant-free.
\end{lemma}
\begin{proof}
  More generally, if an endofunctor $F$ preserves the initial object $∅$,
  then $∅$ equipped with the isomorphism $F(∅) → ∅$ is easily seen to be the initial
  $F$-algebra, hence is isomorphic to $F^*∅$.
\end{proof}

\begin{theorem}\label{thm:simple} Any simple structural law
    $$d_{X,Y}∶ Γ_Y(Σ(X)) → S(Γ_{S(Y)}(X) + X + Y)$$
    (with again $S = Σ^*$) induces a monad distributive law
    $$TS → ST\rlap{,}$$
    where $T ≔ Γ_S^\fleur$, making the following diagram commute.
    \begin{equation}
    \begin{tikzcd}[ampersand replacement=\&]
      {Γ_X ΣX} \&\& {S(Γ_{SX}X + X + X)} \\
      {Γ_{SX}SX} \\
      {Γ_{SSX}SX} \\
      TSX \&\& STX
      \arrow["{S[η_{Γ_S,X},ηᵀ_X,ηᵀ_X]}", from=1-3, to=4-3]
      \arrow["{d_{X,X}}", from=1-1, to=1-3]
      \arrow["{Γ_{η^S_X} η_{Σ,X}}"', from=1-1, to=2-1]
      \arrow["{Γ_{η^S_{SX}}SX}"', from=2-1, to=3-1]
      \arrow["{η_{Γ_S,SX}}"', from=3-1, to=4-1]
      \arrow["{δ_X}"', from=4-1, to=4-3]
    \end{tikzcd}
    \label{eq:simple:d:delta}
    \end{equation}

    Furthermore, $T$ being constant-free by
    Lemma~\ref{lem:constantless}, the monad morphism $S→ ST$ is
    admissible.
\end{theorem}    
\begin{proof}
  This is a special case of Theorem~\ref{thm:incremental} below, instantiating
  $T$ with the identity monad.
\end{proof}
\begin{remark}
  By the usual characterisation of free algebras for finitary
  endofunctors~\cite{Reiterman}, we have
$$T(X) = μA.(X + Γ_{SA}(A)).$$
\end{remark}

  \subsection{Augmented algebras}
  As a bonus, we may characterise algebras for the generated monad
  $ST$, which we call augmented algebras.

  \begin{definition}\ \hfill
  \begin{itemize}
  \item 
    An \alert{algebra} for a simple structural law $L=(Σ,Γ,d)$ on $𝐂$ consists of an object $X ∈ 𝐂$, equipped with
    \begin{itemize}
    \item $Σ$-algebra structure $𝐚∶ Σ(X) → X$ and
    \item $ΓΔ$-algebra structure $𝐛∶ Γ_X(X) → X$,
    \end{itemize}
    making the following diagram commute, 
    \begin{center}
      \diag{%
        Γ_X(Σ(X)) \& \& S(Γ_{SX}(X) + X + X) \\
        \& \& S(Γ_X(X) + X + X) \\
        Γ_X (X) \& \& S X \\
        \& X
      }{%
        (m-1-1) edge[labela={d_{X,X}}] (m-1-3) %
       edge[labell={Γ_X (𝐚)}] (m-3-1) %
       (m-1-3) edge[labelr={S(Γ_{\bar{𝐚}}(X) + X + X)}] (m-2-3) %
       (m-2-3) edge[labelr={S[𝐛,X,X]}] (m-3-3) %
       (m-3-1) edge[labelbl={𝐛}] (m-4-2) %
       (m-3-3) edge[labelbr={\bar{𝐚}}] (m-4-2) %
      }
    \end{center}
    where $\bar{𝐚}∶ S(X) → X$ is freely induced by $𝐚$.
  \item A morphism $X → Y$ of algebras for $L = (Σ,Γ,d)$ is a morphism
    between underlying objects which is both a morphism of $Σ$- and
    $ΓΔ$-algebras.
  \item Let $L\alg$ denote the category of algebras for $L$, or
    $L$-algebras.
\end{itemize}
\end{definition}

\begin{proposition}
  Let $L=(Σ,Γ,d)$ be any simple structural law on $𝐂$, and let
  $T = Γ_S^\fleur$.  Then we have
  $$L\alg ≅ ST\Alg$$ over
  $𝐂$, where we recall from the basic notations of~§\ref{s:intro} that
  capital $𝐀𝐥𝐠$ denotes monad algebras.
\end{proposition}
\begin{proof}
  This is a special case of Theorem~\ref{thm:models} below, instantiating
  $T$ with the identity monad.
\end{proof}

  \subsection{Application: evaluation contexts}\label{ss:evconts}
  Let us present a first application of Theorem~\ref{thm:simple}, to
  generate a simple language with context application.
  The language is generated by the following grammar
  $$
  \begin{array}{rcl}
    e, f & ::= & x ｜ e\ f ｜ λx.e \\
    E & ::= & ▫ ｜ E\ e 
  \end{array}$$
  and context application is defined inductively by
  $$
  \begin{array}{rcl}
    ▫[e] & = & e \\
    (E\ e)[f] & = & E[f]\ e.
  \end{array}$$

  We now define a simple structural law, whose basic monad is the term
  monad, including contexts, and whose associated auxiliary monad will
  account for context application.
  \begin{itemize}
  \item Following~\citet{fiore:presheaf}, at least in spirit, we first
    choose as ambient category the category $[𝐒𝐞𝐭,𝐒𝐞𝐭²]_f$ of finitary
    functors $𝐒𝐞𝐭 → 𝐒𝐞𝐭²$, or equivalently $[𝔽,𝐒𝐞𝐭²]$, where $𝔽$
    denotes the category of finite ordinals and all maps between them.

    We call $𝐩$ and $𝐜$ the two elements of $2$, respectively for
    ``program'' and ``context''.
    
    For any object $X ∈ [𝔽,𝐒𝐞𝐭²]$ and $n ∈ 𝔽$, $X(n)$ is a pair of sets, which we
    denote by $(X(n)_𝐩,X(n)_𝐜)$.  We think of
    \begin{itemize}
    \item $X(n)_𝐩$ as a set of \alert{programs} with $n$ free
      variables, and of
    \item $X(n)_𝐜$ as a set of \alert{contexts} with $n$ free
      variables.
    \end{itemize}
  \item The basic endofunctor is defined by: $$
    \begin{array}{rcl}
      Σ(X)(n)_𝐩 & = & n + X(n)_𝐩² + X(n+1)_𝐩 \\
      (e, f & ::= & x ｜ \ \ e\ f\ \  ｜ λx.e) \\
      Σ(X)(n)_𝐜 & = & 1 + X(n)_𝐜×X(n)_𝐩 \\
      (E & ::= & ▫ ｜ E\ e).
    \end{array}$$
  \item The auxiliary bifunctor is
    $$\begin{array}{rcl}
      Γ(X,Y)(n)_𝐩 & = & X(n)_𝐜 × Y(n)_𝐩 \\
      Γ(X,Y)(n)_𝐜 & = & ∅.
    \end{array}$$
  \item For our simple structural law,
    which is trivial by construction at $𝐜$,
    we take at $𝐩$:
    $$
    \begin{array}{rcl}
      Σ(X)(n)_𝐜 × Y(n)_𝐩 & → & S (\rouge{Γ(X,S(Y))} + \orange{X}
                               + \vert{Y})(n)_𝐩 \\
      (▫,e) & ↦ & η^S(\vert{in₃}(e)) \\
      (E\ f, e) & ↦ & η^S(\rouge{in₁}(E,η^S(e)))\ η^S(\orange{in₂}(f)).
    \end{array}$$

  \end{itemize}

  \subsection{Application: capture-avoiding substitution}\label{ss:lambda}
  In this section, we present a third application of
  Theorem~\ref{thm:simple}, to pure $λ$-calculus with capture-avoiding
  substitution. This will be subsumed by~§\ref{ss:ptstrengths}, but we
  find it instructive to unfold the development on a concrete example.

  \begin{itemize}
  \item We take as ambient category the category $[𝐒𝐞𝐭,𝐒𝐞𝐭]_f$ of
    finitary endofunctors of sets, or equivalently
    $[𝔽,𝐒𝐞𝐭]$.
  \item The basic endofunctor on $[𝔽,𝐒𝐞𝐭]$ is defined by
    $$\begin{array}{rcl}
      Σ(X)(n) & = & n + X(n)² + X(n+1) \\
      (e, f & ::= & x ｜ \ \ e\ f\ \  ｜ λe).
    \end{array}$$
  \item The auxiliary bifunctor is~\cite{fiore:presheaf}'s
    substitution tensor product
    $$(X⊗Y)(n) = ∫ᵖX(p)×Y(n)ᵖ\rlap{,}$$
    or, expressed in $[𝐒𝐞𝐭,𝐒𝐞𝐭]_f$, $(X⊗Y)(n) = X (Y (n))$.
  \item For our simple structural law,
    we take (sometimes omitting $η^S$ for readability)
    $$
    \begin{array}{rcl}
      ∫ᵖΣ(X)(p) × Y(n)ᵖ & → & S (\rouge{X⊗S(Y)}) + \orange{X}
                               + \vert{Y})(n) \\
      (x,σ) & ↦ & \vert{in₃}(σ(x)) \\
      (e\ f, σ) & ↦ & \rouge{in₁}(e, η^S∘σ)\ \rouge{in₁}(f,η^S∘σ) \\
      (λe,σ) & ↦ & λ (\rouge{in₁} (e, σ^↑))\rlap{,}
    \end{array}$$
    where $σ^↑∶ p+1 → S(Y)(n+1)$ denotes the copairing of the following
    two maps.
    \begin{mathpar}
      p \xto{σ} Y(n) \xto{Y(in₁)} Y(n+1) \xto{η^S_{Y,n+1}}
      S(Y)(n+1)
      \and
      1 \xto{in₂} n+1 \xto{in₁} Σ(Y)(n+1) ↪ S(Y)(n+1)
    \end{mathpar}
\end{itemize}

\begin{remark}
  In the final case, we crucially rely on the recursive call being in
  $Γ(X,SY)$, rather than just $Γ(X,Y)$.  In~\citet{fiore:presheaf}, this is done by taking $Y$
  ``pointed'', in the sense of being equipped with a natural
  transformation $n → Y(n)$.
\end{remark}
\begin{remark}
 We obtain the expected syntax and substitution,
    but not yet the following standard equations,
        \begin{mathpar}
          e[σ][σ'] = e[σ[σ']] \and e[\id] = e
        \end{mathpar}
        where  $\id∶ n → S(X)(n)$ picks the variables.
        We will complete the picture in~§\ref{s:benign}.
\end{remark}

\subsection{Application: binding contexts}
For a slightly more involved example, we consider in this subsection
the \alert{sharing}
$λ$-calculus~\cite[§4.1]{DBLP:conf/rta/Accattoli19} (but see
also~\cite{DBLP:journals/lisp/HirschowitzLW09,DBLP:journals/jfp/SewellSHBW08}).

Following Accattoli, the syntax is given by
$$
\begin{array}{rcll}
  e, f & ::= & x ｜ e\ f ｜ λx.e ｜ e⟨x↦f⟩ & \\
  E & ::= & ▫ ｜ E⟨x↦f⟩.
\end{array}$$
Context application is then defined inductively by
$$
\begin{array}{rcl}
  ▫[e] & = & e \\
  (E⟨x↦f⟩)[e] & = & E[e]⟨x ↦ f⟩.
\end{array}$$

\begin{remark}
  Context application may give rise to variable capture. E.g.,
  $(▫⟨x↦f⟩)[x] = x⟨x↦f⟩$.
\end{remark}

In order to model this, we extend the setting of the previous section
as follows.  In~§\ref{ss:evconts}, for any functor $X∶ 𝔽 → 𝐒𝐞𝐭²$, we
thought of $X(n)_𝐩$ and $X(n)_𝐜$ as sets of programs, resp.\ contexts
with $n$ free variables. We now need to refine this point of view, and
index contexts over the number of \alert{capturing} variables, i.e.,
variables bound above the context hole $▫$.  Instead of functors
$𝔽 → 𝐒𝐞𝐭²$, we thus consider functors $𝔽 → 𝐒𝐞𝐭^{1+ℕ}$. Of course, we
have $ℕ≅1+ℕ$, but we write $1+ℕ$ to emphasise the fact that $in₁(⋆)$,
the unique element of the left-hand summand, models terms, while each
$in₂(n)$ models contexts with $n$ capturing variables.

\begin{notation}
  We abbreviate $in₁(⋆)$ to $𝐩$ and $in₂(n)$ to $𝐜ₙ$, so that, e.g.,
  $X(n)_{𝐜ₘ}$ is thought of as a set of contexts with $n$ free
  variables and $m$ capturing variables.
\end{notation}

Accordingly, we specify the syntax by the endofunctor

\noindent $
\begin{array}{r@{\ =\ }l}
  Σ(X)(n)_𝐩 &  n + X(n)_𝐩² + X(n+1)_𝐩 + X(n+1)_𝐩×X(n)_𝐩 \\
  Σ(X)(n)_{𝐜_{m+1}} &  X(n+1)_{𝐜ₘ} × X(n)_𝐩 \\
  Σ(X)(n)_{𝐜₀} &  1.
\end{array}$

\begin{remark}
  On the second line, the expression $X(n+1)_{𝐜ₘ} × X(n)_𝐩$ reflects
  the fact that in the above grammar, in $E⟨x↦f⟩$, $E$ may use the
  bound variable $x$, hence the use of $n+1$, and has one less
  capturing variable than $E⟨x↦f⟩$, hence the passing from $m+1$ to
  $m$. Otherwise said, $x$ is free in $E$, but capturing in
  $E⟨x↦f⟩$.
\end{remark}
The generated monad $S = Σ^*$ may be presented syntactically
by the following rules
\begin{mathpar}
  \inferrule{k ∈ X(n)_𝐩}{n ⊢ k : 𝐩} \and
  \inferrule{i ∈ n}{n ⊢ xᵢ : 𝐩} \and
  \inferrule{n ⊢ e : 𝐩 \\ n ⊢ f : 𝐩}{n ⊢ e\ f : 𝐩} \and
  \inferrule{n+1 ⊢ e : 𝐩}{n ⊢ λ(e) : 𝐩} \and
  \inferrule{n+1 ⊢ e : 𝐩 \\ n ⊢ f : 𝐩}{n ⊢ e⟨f⟩ : 𝐩} \and
  \inferrule{K ∈ X(n)_{𝐜ₘ}}{n ; m ⊢ K : 𝐜} \and
  \inferrule{ }{n ; 0 ⊢ ▫ : 𝐜} \and
  \inferrule{n+1;m ⊢ E : 𝐜 \\ n ⊢ f : 𝐩}{n;m+1 ⊢ E⟨f⟩ : 𝐜}\ ,
\end{mathpar}
with $S(n)_𝐩 = \{ e ｜ n ⊢ e : 𝐩 \}$ and
$S(n)_{𝐜ₘ} = \{ E ｜ n ; m ⊢ E : 𝐜 \}$.  One then straightforwardly defines the
functorial action in $n$: for any renaming $ρ∶ n → n'$,
one replaces all $xᵢ$ with $x_{ρ(i)}$, for $i ∈ n$.

Let us now model context application by a simple structural law:
\begin{itemize}
\item the arity functor is defined by
  $$
  \begin{array}{rcl}
    Γ(X,Y)(n)_{𝐜ₘ} & = & ∅ \\
    Γ(X,Y)(n)_𝐩 & = & ∑_{m ∈ ℕ} X(n)_{𝐜ₘ} × Y(n+m)_𝐩\rlap{,}
  \end{array}$$
  reflecting the fact that
  in a context application $E[e]$, $e$ has as free variables the disjoint union
  of the free and capturing variables of $E$;
\item the simple structural law is defined componentwise by
      $$
    \begin{array}{rcll}
      Σ(X)(n)_{𝐜ₘ} × Y(n+m)_𝐩 & → & S (\rouge{Γ(X,S(Y))} + \orange{X}
                               + \vert{Y})(n)_𝐩 \\
      (▫,e) & ↦ & \vert{in₃}(e) & \mbox{(if $m=0$)}\\
      (E⟨f⟩, e) & ↦ & \rouge{in₁}(E,e) ⟨\orange{in₂}(f)⟩ & \mbox{(if $m=m'+1$)}\rlap{,}
    \end{array}$$
    for all $n,m ∈ ℕ$ (omitting $η^S$ again for readability).
\end{itemize}

\begin{remark}
  The reason $e$ may be applied to both $E$ and $E⟨f⟩$ is that both
  contexts have the same number of free or capturing variables
  ($(n+1)+m'$ and $n + (m'+1)$, respectively).
\end{remark}

\subsection{Application: named substitution}
Let us now consider a last illustration of simple structural laws:
named substitution in \linebreak $λμ$-calculus~\cite{Parigot,Iouri}.
Its usual form is as follows.
The syntax of $λμ$-calculus is:
$$
\begin{array}{rcl}
  e,f,g & ::= & x ｜ e\ f ｜ λx.e ｜ μα.c \\
  c,d & ::= & [α]e.
\end{array}$$
There are two syntactic categories: \alert{programs}, ranged over by $e,f,g,…$, and
\alert{continuations}, ranged over by $c,d$.
Accordingly, there are two kinds of variables: \alert{program variables}, ranged
over by $x,y,z,…$, and \alert{continuation variables}, ranged over
by $α,β,…$

\begin{remark}
  In this subsection, $μ$ always denote the syntactic operation, as
  opposed to any monad multiplication, or least fixed-point operator.
\end{remark}

Intuitively, named substitution $(t)_α\ g$ (notation
from~\cite[Definition~6.5]{Iouri}) takes as argument any term $t$
(program or continuation) with a distinguished continuation variable
$α$, together with a program $g$, and replaces all subterms of the
form $[α]e$ with $[α](e\ g)$ in $t$. This may be defined recursively
by:
$$
\begin{array}{rcll}
  (x)_α\ g & = & x \\
  (λx.e)_α\ g & = & λx.((e)_α\ g) & \mbox{($x ∉ g$)} \\
  (e\ f)_α\ g & = & ((e)_α\ g)\ ((f)_α\ g) \\
  (μβ.c)_α\ g & = & μβ.((c)_α\ g) & \mbox{($β ∉ α,g$)} \\
  ([β]e)_α\ g & = & \left \{
                    \begin{array}{ll}
                      [α] (((e)_α\ g)\ g) & \mbox{(if $α = β$)} \\
                      {}[β] ((e)_α\ g) & \mbox{(otherwise).}
                    \end{array}
                    \right .
\end{array}$$

In order to model this categorically, since we have two syntactic
categories, each with its own set of variables, we work with finitary
functors $𝐒𝐞𝐭² → 𝐒𝐞𝐭²$, or equivalently the category $[𝔽²,𝐒𝐞𝐭²]$.  We
write $𝐩$ and $𝐜$ for the elements of $2$ (respectively standing for
``program'' and ``continuation''). Thus, for any $X ∈ [𝔽²,𝐒𝐞𝐭²]$, we
think of
\begin{itemize}
\item $X(m,n)_𝐩$ as a set of programs with $m$ free program
  variables and $n$ continuation variables, and of 
\item $X(m,n)_𝐜$ as a set of continuations with $m$ free program
  variables and $n$ continuation variables.
\end{itemize}
The basic syntax is specified by the endofunctor $Σ$ defined by
$$
\begin{array}{rcl}
  Σ(X)(m,n)_𝐩 & = & m + X(m,n)_𝐩² + X(m+1,n)_𝐩 + X(m,n+1)_𝐜 \\
  Σ(X)(m,n)_𝐜 & = & n × X(m,n)_𝐩.
\end{array}$$
For specifying named substitution, we
take as auxiliary bifunctor
$$
\begin{array}{rcl}
  Γ(X)(m,n)_𝐩 & = & X(m,n)_𝐩 × n × Y(m,n)_𝐩 \\
  Γ(X)(m,n)_𝐜 & = & X(m,n)_𝐜 × n × Y(m,n)_𝐩.
\end{array}$$

\begin{notation}
  Slightly generalising previous notation, we denote
  by $xᵢ$ the $i$th program variable, for $i ∈ m$, and
  by
  $αⱼ$ the $j$th continuation variable, for $j ∈ n$.
\end{notation}
Using this notation, we model named substitution by the following
simple structural law (omitting $η^S$ again for readability):
$$
\begin{array}{rcll}
  d_{X,Y,m,n,𝐩}∶ Σ(X)(m,n)_𝐩 × n × Y(m,n)_𝐩 & → & S(Γ(X,SY)+X+Y)(m,n)_𝐩 & \\
  (xᵢ,j,g) & ↦ & in₂(xᵢ) \\
  (λ(e),j,g) & ↦ & λ(in₁(e,j,w^𝐩_{m,n}·g)) \\
  (e\ f, j,g) & ↦ & in₁(e,j,g)\ in₁(f,j,g) \\
  (μ(e),j,g) & ↦ & μ(in₁(e,j,w^𝐜_{m,n}·g)) \\ [1em]
  d_{X,Y,m,n,𝐜}∶ Σ(X)(m,n)_𝐜 × n × Y(m,n)_𝐩 & → & S(Γ(X,SY)+X+Y)(m,n)_𝐜\rlap{,} & \mbox{i.e.,}\\
    d_{X,Y,m,n,𝐜}∶ n × X(m,n)_𝐩 × n × Y(m,n)_𝐩 & → & S(Γ(X,SY)+X+Y)(m,n)_𝐜 \\
  ([αⱼ]e,j,g) & ↦ & [αⱼ] (in₁(e,j,g)\ in₃(g)) \\
  ([α_{j'}]e,j,g) & ↦ & [α_{j'}] in₁(e,j,g) & \mbox{(for $j' ≠ j$),}
\end{array}$$
where
\begin{itemize}
\item $w^𝐩_{m,k} ≔ (in₁,idₖ)∶ (m,k) ↪ (m+1,k)$;
\item $w^𝐜_{m,k} ≔ (idₘ,in₁)∶ (m,k) ↪ (m,k+1)$; and
\item $u·a$ denotes the action of a morphism $u∶ (m,k) → (m',k')$ in
  $𝔽²$ on an element $a$ of some $A(m,k)ₛ$, for some $A∶ 𝔽² → 𝐒𝐞𝐭²$
  and $s ∈ \{𝐩,𝐜\}$.
\end{itemize}

\section{Incremental structural laws}\label{s:isl}
\label{s:incr-struct-laws}
In this section, we introduce the incremental variant of structural
laws.  In~§\ref{ss:mult}, we sketch the idea on a simple example.
In~§\ref{ss:incrlaws}, we introduce \alert{incremental structural
  laws} in full generality, and prove that they induce admissible
morphisms (Theorem~\ref{main:thm:incremental}). In~§\ref{ss:models},
we characterise the composite monad induced by an incremental
structural law, as objects equipped with suitably compatible algebra
structures.  Then, in~§\ref{ss:indep}, we explain how to combine
several, independent incremental structural laws.  Finally, we cover
applications: partial differentiation in the differential $λ$-calculus
(§\ref{ss:diff}), and De Bruijn's presentation of capture-avoiding
substitution in the $λ$-calculus (§\ref{ss:db}).

\subsection{On a simple example}\label{ss:mult}
Recalling the distributive law, say $δ∶ TS → ST$ constructed
in~§\ref{ss:addition} from the simple structural law~\eqref{eq:d}, we
now want to extend the language with a second binary
operation, multiplication, satisfying
    $$
    \begin{array}{rcl}
      {s}(x) × y & = & (x × y)  +  y \\
      0 × y & = & 0. %
    \end{array}$$
    Furthermore, the full quotiented term language yields a monad
    $S^{\add,\mul}$, and a monad morphism $S → S^{\add,\mul}$.
    As in~§\ref{ss:addition}, orienting equations, we obtain
    normal forms, which are generated for any argument $X$  by the
    grammar
  $$
  \begin{array}{rcll}
    e & ::= & 0 ｜ s(e) ｜ b \\
    b & ::= & x ｜ b+e ｜ b×e.
  \end{array}$$
  Of course, by the same reasoning as for mere addition, the initial $S$-algebra is initial in the
  category of models of the whole language (i.e., $S$-algebras with
  addition and multiplication, satisfying all four equations), hence
  the monad morphism $S → S^{\add,\mul}$
  is admissible.

  Categorically, because the new equation uses addition, we cannot
  define a second, independent structural extension of $S$.  However,
  the recursive definition of multiplication yields a natural
  transformation
    \begin{equation}
      \label{eq:di}
      Θ(Σ(X),Y) → ST (Θ(X,Y) + Y)\rlap{,}
    \end{equation}
    where $Θ(X,Y) ≔ X × Y$.  Furthermore, inspecting the above grammar
    for normal forms, we find $S^{\add,\mul} ≅ S ∘ (Γ_S + Θ_S)^*$:
    from the top level, we have a first layer of basic operations,
    until we meet a binary operation; the first argument of the latter
    may then only consist of binary operations until it reaches a
    variable $x ∈ X$, while the second argument is
    arbitrary. Equivalently, letting $T' = Θ_S^\fleur$, we have
    $$(Γ_S + Θ_S)^* ≅ Γ_S^* ⊕ Θ_S^* ≅ T ⊕ T'\rlap{,}$$

    hence $$S^{\add,\mul} ≅ S ∘ (T ⊕ T').$$ By induction, i.e.,
    applying~\eqref{eq:di} repeatedly, we define a distributive law
    $(T ⊕ T') ∘ S → S ∘ (T ⊕ T')$.

  Of course $Θ$ is cocontinuous in its first argument, which easily
  entails that $T⊕T'$ is constant-free, and so
  $S → S^{\add,\mul} ≅ S ∘ (T⊕T')$ is admissible.

\subsection{The abstract case}\label{ss:incrlaws}
\subsubsection{Main result}
Let us now abstract over the previous section.
\begin{definition}
  An \alert{incremental structural law} on a locally finitely
  presentable category $𝐂$ consists of
  \begin{itemize}
  \item a distributive law $δ∶ TS → ST$ with free basic monad
    $S = Σ^*$ and constant-free auxiliary monad $T$, together with
  \item a functor $Γ∶ 𝐂² → 𝐂$ which is cocontinuous in its first
    argument and finitary in its second argument, equipped with
  \item a natural transformation
    $$Γ_Y(Σ(X)) → ST(Γ_{STY}(X) + X + Y)\rlap{.}$$
  \end{itemize}
\end{definition}

\begin{theorem}\label{main:thm:incremental}
  Any incremental structural law
  $$d_{X,Y}∶ Γ_Y(Σ(X)) → ST(Γ_{ST(Y)}(X) + X + Y)$$
  over $δ∶ TS → ST$, with $S ≔ Σ^*$, induces a  
  distributive law
  $$d_{/δ}∶ (T⊕T')S → S(T⊕T')\rlap{,}$$
  where $T' = Γ_S^\fleur$, making the following diagram commute.
  \begin{equation}
\begin{tikzcd}[ampersand replacement=\&]
	{Γ_X ΣX} \&\& {ST(Γ_{STX}X + X + X)} \\
	\&\& {ST(Γ_{STX}TX + X)} \\
	{Γ_{SX}SX} \&\& {ST(T'TX + X)} \\
	{Γ_{SSX}SX} \&\& {ST((T ⊕ T')(T⊕T')X + X)} \\
	{T'SX} \&\& {S(T⊕T')(T⊕T')X} \\
	{(T ⊕ T')SX} \&\& {S(T⊕T')X}
	\arrow["{d_{X,X}}", from=1-1, to=1-3]
	\arrow["{Γ_{η^S_X} η_{Σ,X}}"', from=1-1, to=3-1]
	\arrow["{Γ_{η^S_{SX}}SX}"', from=3-1, to=4-1]
	\arrow["{d_{/δ}X}"', from=6-1, to=6-3]
	\arrow["{Sμ^{T⊕T'}_X}", from=5-3, to=6-3]
	\arrow["{η_{Γ_S,SX}}"', from=4-1, to=5-1]
	\arrow["{in_2 SX}"', from=5-1, to=6-1]
	\arrow["{Sin₁[μ^{T⊕T'}_X, η^{T⊕T'}X]}", from=4-3, to=5-3]
	\arrow["{ST(in₂ in₁ X + X)}", from=3-3, to=4-3]
	\arrow["{ST(Γ_{STX} ηᵀ_X + [X, X])}", from=1-3, to=2-3]
	\arrow["{ST(\eta_{\Gamma_S,TX} + X)}", from=2-3, to=3-3]
      \end{tikzcd}
      \label{eq:d:delta}
  \end{equation}
  Furthermore, $T⊕T'$ is constant-free, hence the monad morphism
  $$S → S (T⊕T')$$
  is admissible.
\end{theorem}

The next subsection is devoted to sketching the proof, and may be
safely ignored.

\subsubsection{Proof sketch}
Our first step will consist in defining an intermediate notion
called incremental lifting, fitting in the following
process:
\begin{center}
  incremental structural law $→$ incremental lifting $→$ distributive
  law $→$ admissible morphism.
\end{center}






To explain the idea, let us start by recalling
from~\citet{BeckDistlaws} that distributive laws are equivalent to
monad liftings.

\begin{definition}
  A \alert{lifting} of a monad $S∶ 𝐂 → 𝐂$ along a functor
  $U∶ 𝐄 → 𝐂$ is a monad $S'∶ 𝐄 → 𝐄$ such that the following square commutes
  \begin{center}
    \diag{%
      𝐄 \& 𝐄 \\
      𝐂 \& 𝐂 %
    }{%
      (m-1-1) edge[labela={S'}] (m-1-2) %
      edge[labell={U}] (m-2-1) %
      (m-2-1) edge[labelb={S}] (m-2-2) %
      (m-1-2) edge[labelr={U}] (m-2-2) %
    }
  \end{center}
  and, furthermore, $U$ preserves multiplication and unit, i.e., for
  all $E ∈ 𝐄$, $U(μ^{S'}_E) = μ^S_{U(E)}$ and
  $U(η^{S'}_E) = η^S_{U(E)}$.

  Given any monads $T$ and $S$ on a category $𝐂$, a \alert{$T$-lifting}
  of $S$ is a lifting of $S$ along the forgetful functor
  $T\Alg → 𝐂$.
\end{definition}
\begin{proposition}
  For all $S$ and $T$, $T$-liftings of $S$ are in one-to-one
  correspondence with endofunctors $S'∶ T\Alg → T\Alg$, such that the
  following square commutes
  \begin{center}
    \diag{%
      T\Alg \& T\Alg \\
      𝐂 \& 𝐂
    }{%
      (m-1-1) edge[labela={S'}] (m-1-2) %
      edge[labell={}] (m-2-1) %
      (m-2-1) edge[labelb={S}] (m-2-2) %
      (m-1-2) edge[labelr={}] (m-2-2) %
    }
  \end{center}
  and, for all $X ∈ T\Alg$, $μ^S_X∶ SSX → SX$ and $η^S_X∶ X → SX$ are
  $T$-algebra morphisms.
\end{proposition}
\begin{proof}
  Straightforward.
\end{proof}

\begin{lemma}[{\citet[§1]{BeckDistlaws}}]\label{lem:dist:lifting}
  For any monads $S,T∶ 𝐂 → 𝐂$, monad distributive laws $TS → ST$ are
  in one-to-one correspondence with $T$-liftings of $S$.
\end{lemma}

We now introduce incremental liftings.
\begin{definition}
  Let $S$, $T$, and $T'$ be monads on $𝐂$, and $δ∶ TS → ST$ be a monad
  distributive law. An \alert{incremental lifting} of $S$ to $T'\Alg$
  along $δ$ is a functorial assignment, to each pair of a $T$-algebra
  structure and a $T'$-algebra structure on $X$, of a $T'$-algebra
  structure on $SX$, such that for each such $X$, the multiplication
  $SSX → SX$ and unit $X → SX$ are $T'$-algebra morphisms.
\end{definition}

This rather technical definition in fact unfolds rather simply, as we
now explain. We first recall the following, well-known
characterisation of $(T⊕T')$-algebras.

\begin{lemma}\label{lem:monad:coprod:alg}
  Let $T$ and $T'$ be monads on a complete (locally small) category $𝐂$
  such that the monad coproduct $T⊕T'$ exists. Then, the category of
  $(T⊕T')$-algebras consists of objects of $𝐂$ equipped with algebra
  structures for both $T$ and $T'$. Otherwise said, the following
  square is a pullback.
  \begin{center}
    \Diag{%
      \stdpbk %
    }{%
      (T⊕T')\Alg \& T\Alg \\
      T'\Alg \& 𝐂
    }{%
      (m-1-1) edge[labela={}] (m-1-2) %
      edge[labell={}] (m-2-1) %
      (m-2-1) edge[labelb={}] (m-2-2) %
      (m-1-2) edge[labelr={}] (m-2-2) %
    }%
    \end{center}
\end{lemma}
\begin{proof}
  See~§\ref{s:monad:colims}.
\end{proof}
\begin{remark}
  In passing, this result yields an easy proof that $T⊕T'$ is
  constant-free.
\end{remark}

This directly entails the following equivalent presentation of
incremental liftings.
\begin{corollary}\label{cor:incr:liftings}
  Incremental liftings of $S$ to $T'\Alg$ along $δ$ on a complete
  category $𝐂$ such that the coproduct $T⊕T'$ exists are in one-to-one
  correspondence with liftings along the projection
  $(T⊕T')\Alg → T\Alg$ of the monad on $T\Alg$, say $Sᵟ$, itself
  obtained by lifting $S$, as in
  \begin{center}
    \diag{%
      (T⊕T')\Alg \&       (T⊕T')\Alg \\
      T\Alg \&       T\Alg \\
      𝐂 \& 𝐂\rlap{.} %
    }{%
      (m-1-1) edge[labela={S^{δ,T'}}] (m-1-2) %
      edge[labell={}] (m-2-1) %
      (m-2-1) edge[labela={Sᵟ}] (m-2-2) %
      (m-1-2) edge[labelr={}] (m-2-2) %
      (m-2-1)
      edge[labell={}] (m-3-1) %
      (m-3-1) edge[labelb={S}] (m-3-2) %
      (m-2-2) edge[labelr={}] (m-3-2) %
    }
  \end{center}
  In such a situation, $S^{δ,T'}$ is furthermore a $(T⊕T')$-lifting of
  $S$.
\end{corollary}
It is now easy to see that incremental liftings give rise to distributive laws.
\begin{corollary}\label{cor:incrlift:distlaw}
  Let $S$, $T$, and $T'$ be monads on a complete category $𝐂$ such
  that the monad coproduct $T ⊕ T'$ exists, and let $δ∶ TS → ST$ be
  any monad distributive law.  Then, any incremental lifting of $S$ to
  $T'\Alg$ along $δ$ gives rise to a distributive law
  $δ'∶ (T⊕T')S → S (T⊕T')$.

  Furthermore, if $T$ and $T'$ are constant-free, so is $T⊕T'$, hence
  the induced monad morphism $S → S (T⊕T')$ is admissible.
\end{corollary}
\begin{proof}
  By Corollary~\ref{cor:incr:liftings} and
  Lemma~\ref{lem:dist:lifting}.  Constant-freeness of $T⊕T'$ follows
  by Lemma~\ref{lem:monad:coprod:alg}.
\end{proof}
\begin{remark}
  \label{rk:admis-STT}
  With the same hypotheses, it is straightforward to deduce that
  $ST → S(T⊕T')$ is admissible, since $S → ST$  and $S → ST → S(T⊕T')$ both are.
\end{remark}

We thus mostly reduce Theorem~\ref{main:thm:incremental} to the following.
\begin{theorem}\label{thm:incremental}
  Any incremental structural law
  $$d_{X,Y}∶ Γ_Y(Σ(X)) → ST(Γ_{ST(Y)}(X) + X + Y)$$
  over $δ∶ TS → ST$, with $S ≔ Σ^*$, induces an
  incremental lifting of $S$ to $T'\Alg$ along $δ$,
  where $T' = Γ_S^\fleur$.
      
\end{theorem}
\begin{remark}
  This result does not directly entail~\eqref{eq:d:delta}, which will
  follow from the construction of the incremental lifting.
\end{remark}


\begin{proof}[Proof sketch, see~§\ref{app:thm-incremental} for a complete proof]

  We first define from $d$ a natural transformation
  $$d^{•ω}_{X,Y}∶ Γ_{SY}(S(X)) → ST (Γ_{STY}(X) + X + Y)\rlap{,}$$
  notably using the fact that
  $Γ_{STY}$, being a cocontinuous endofunctor on a locally finitely
  presentable category, admits a right adjoint.

  We then want to construct a lifting of $S$ to some functor
  $(T⊕T')\Alg → T'\Alg$. But we have $T'\Alg ≅ Γ_S\alg$, so we reduce
  to constructing a functor $(T⊕T')\Alg → Γ_S\Alg$ -- still behaving
  like $S$ on underlying objects.
  
  For any $T$-algebra $𝐛∶ TX → X$ equipped with
  $Γ_SΔ$-algebra structure $𝐜∶ Γ_{SX}X → X$, we 
  construct the desired $Γ_SΔ$-algebra
  structure on $S(X)$ as
  $$Γ_{SSX}SX \xto{Γ_{μ^S_X}SX}
  Γ_{SX} SX \xto{d^{•ω}_{X,X}}
  ST (Γ_{STX}(X) + X + X)
  \xto{ST (𝐛⊳𝐜)}
  ST X
  \xto{S𝐛} SX\rlap{,}$$
  where $(𝐛⊳𝐜)∶ Γ_{STX}X + X + X → X$ is defined as follows.
\begin{definition}\label{def:triangleright}
  For any $T$-algebra $𝐛∶ TX → X$ equipped with $Γ_SΔ$-algebra
  structure $𝐜∶ Γ_{SX}X → X$, let $𝐛⊳𝐜$ denote the following composite.
  $$Γ_{STX}(X)+X+X \xto{Γ_{S𝐛}X+[X,X]} Γ_{SX}X+X \xto{[𝐜,X]} X.$$
\end{definition}
  
We then show that $μ^S$ and $η^S$ are algebra morphisms, as
desired.
\end{proof}

\subsection{Augmented algebras}\label{ss:models}
Let us now present the announced characterisation of the category of
algebras of the composite monad $S(T ⊕ T')$ induced by an incremental
structural law, which we call augmented algebras.

\begin{definition}
  \label{d:alt-models}
  Consider any incremental structural law
  $$d_{X,Y}∶ Γ_Y(Σ(X)) → ST (Γ_{STY}(X) + X + Y)$$ over $δ∶ TS → ST$,
  with $S = Σ^*$.
  \begin{itemize}
  \item A \alert{$(δ,d)$-algebra} is an
  object $X$ equipped with morphisms
  \begin{mathpar}
    𝐚∶ SX → X \and 𝐛∶ TX → X \and 𝐜∶ Γ(X,X)  → X\rlap{,}
  \end{mathpar}
  the first two of which are monad algebra structures, making the
  following diagrams commute.
  \begin{center}
    \Diag{%
      \justify{m-1-1}{d1}{m-2-3} %
      }{%
      TSX \& \& STX \\
      TX \& \& SX \\
      \& X
    }{%
      (m-1-1) edge[labela={δ_X}] (m-1-3) %
       edge[labell={T𝐚}] (m-2-1) %
       (m-1-3) edge[labelr={S𝐛}] (m-2-3) %
       (m-2-1) edge[labelbl={𝐛}] (m-3-2) %
       (m-2-3) edge[labelbr={𝐚}] (m-3-2) %
     }
     
    \Diag{%
      \justify{m-1-1}{d2}{m-3-3} %
      }{%
       Γ(ΣX,X) \& \& ST (Γ(X,STX)+X+X) \\
       Γ(SX,X) \& \& ST(Γ(X,X)+X) \\
       Γ(X,X) \& \& ST(X) \\
       \& X
     }{%
       	 (m-1-1) edge[labela={d_{X,X}}] (m-1-3) %
         (m-1-1) edge[labell={Γ(η_{Σ,X},X)}] (m-2-1) %
         (m-1-3) edge[labelr={ST (Γ(X,𝐚∘S𝐛)+[X,X])}] (m-2-3) %
         (m-2-1) edge[labell={Γ(𝐚,X)}] (m-3-1) %
         (m-2-3) edge[labelr={ST[𝐜,X]}] (m-3-3) %
         (m-3-1) edge[labelbl={𝐜}] (m-4-2) %
         (m-3-3) edge[labelbr={𝐚∘S𝐛}] (m-4-2) %
       }
  \end{center}
\item   A \alert{$(δ,d)$-algebra morphism} is a morphism between underlying objects which is
  an $S$-, $T$-, and $ΓΔ$-algebra morphism.
\item 
  We let $(δ,d)\Alg$ denote the category of $(δ,d)$-algebras and morphisms between them.
\end{itemize}
\end{definition}

\begin{theorem}\label{thm:models}
  Consider any incremental structural law
  $$d_{X,Y}∶ Γ_Y(Σ(X)) → ST (Γ_{STY}(X) + X + Y)$$ over $δ∶ TS → ST$,
  with $S = Σ^*$.
  Letting $T' = Γ_S^\fleur$ as before, there is an isomorphism
  $S (T⊕T')\Alg ≅ (δ,d)\Alg$ of categories over the base category $𝐂$.

  More precisely, for any $S∘(T⊕T')$-algebra $X$, the canonical morphisms from
  $SX$, $TX$, and $Γ(X,X)$ into $S ((T⊕T')(X))$ equip $X$ with $(δ,d)$-algebra
  structure.  This assignment extends to a functor
  $S (T⊕T')\Alg → (δ,d)\Alg$ over $𝐂$, which is an isomorphism of
  categories.
\end{theorem}

\subsection{Independent extensions}\label{ss:indep}
Let us end the theoretical part of this section by briefly showing how
to combine independent incremental structural laws.

\begin{proposition}\label{prop:independent}
  For any set $I$,
    any $I$-indexed family
    $$dⁱ_{X,Y}∶ Γⁱ_Y(Σ(X)) → ST(Γⁱ_{ST(Y)}(X) + X + Y)$$
    of incremental structural laws over $δ∶ TS → ST$, with $i ∈ I$ and
    $S ≔ Σ^*$, induce a distributive law
  $$(T⊕ ⨁ᵢTᵢ)S → S(T⊕ ⨁ᵢTᵢ)\rlap{,}$$
  where $Tᵢ = (Γⁱ_S)^\fleur$ and $T⊕⨁ᵢTᵢ$ is constant-free.  We
  thus get an admissible morphism
  $$S → S (T⊕⨁ᵢTᵢ).$$
\end{proposition}
\begin{proof}
  Taking $Γ_Y(X) = ∑ᵢ Γⁱ_Y(X)$ and observing that
  $⨁ᵢTᵢ = Γ_S^\fleur$, we obtain the desired result directly by
  applying Theorem~\ref{main:thm:incremental} to the cotupling of all
  \begin{center}
    \hfill
    $Γⁱ_Y(Σ(X))
    \xto{dⁱ_{X,Y}}
    ST(Γⁱ_{ST(Y)}(X) + X + Y)
    \xto{ST((inᵢ)_{X,ST(Y)} + X + Y)}
     ST(Γ_{ST(Y)}(X)+ X + Y)$ \qedhere
\end{center}



    
  \end{proof}

\subsection{Application: differential $λ$-calculus}\label{ss:diff}
In this section, we sketch an imperfect treatment of differential
$λ$-calculus, to illustrate incremental structural laws.
\subsubsection{Standard definition}
Let us first introduce the calculus in the usual, informal way.  The
syntax is:
  $$
  \begin{array}{rcl}
    \llap{Simple terms $∋$ } e,f & ::= & x ｜ e\ M  ｜ λx.e ｜ 𝐃e·f \\
    \llap{Multiterms $∋$ } M,N & ::= & 0 ｜ e+M\rlap{,}
  \end{array}$$
  where simple and multiterms are considered equivalent modulo the following equations.
  \begin{equation}
    \begin{array}{rcl}
    𝐃 (𝐃 e · f) · g &≡& 𝐃 (𝐃 e · g) · f \\
    e+f+M & ≡ & f+e+M    
    \end{array}
\label{eq:diff}
\end{equation}
  Structural operations are then defined as follows:
  \begin{enumeratewide}
  \item First of all, operations are extended to multiterms
    by induction, as in Figure~\ref{fig:diff}.A.
    \begin{figure}[thp]
      \centering

      $\begin{array}[t]{rcl}
         \multicolumn{3}{c}{\text{\underline{A. Extended operations}}} \\[.5em] 
        0 + N & = & N \\
        (e+M) + N & = & e + (M+N) \\ \hline
        0\ N & = & 0 \\
        (e+M)\ N & = & (e\ N) + M\ N \\ \hline
        λx.0 & = & 0 \\
        λx.(e+M) & = & λx.e + λx.M \\ \hline
        𝐃(e)·0 & = & 0 \\
        𝐃(e)·(f+N) & = & 𝐃(e)·f + 𝐃(e)·N \\ \hline
        𝐃(0)·N & = & 0 \\
        𝐃(e+M)·N & = & 𝐃(e)·N + 𝐃(M)·N.
      \end{array}$
      $\hfil
    \begin{array}[t]{rcl}
         \multicolumn{3}{c}{\text{\underline{B. Capture-avoiding substitution}}} \\[.5em] 
      x[x ↦ M] & = & M \\
      y[x ↦ M] & = & y+0 \mbox{\ \ (when $x≠y$)}  \\
      (λy.e)[x ↦ M] & = & λy.(e[x ↦ M]) \\
      \multicolumn{3}{r}{\mbox{\ \ ($y$ fresh for $x$ and $M$)}} \\
      (e\ N)[x ↦ M] & = & e[x ↦ M]\ N[x ↦ M] \\
      (𝐃e·f)[x ↦ M] & = & 𝐃e[x ↦ M]·f[x ↦ M] \\
      0[x ↦ M] & = & 0 \\
      (e+N)[x ↦ M] & = & e[x ↦ M] + N[x ↦ M].
    \end{array}$\\[.5em]

    $\begin{array}{rcl}
       \multicolumn{3}{c}{\text{\underline{C. Partial differentiation}}} \\[.5em] 
       \frac{∂x}{∂x}·M & = & M \Bstrutbas \\
       \frac{∂y}{∂x}·M & = & 0 \mbox{\ \ (when $x≠y$)} \Bstrutbas \\
       \frac{∂ (e\ N)}{∂x}·M & = &
                                   \left( \frac{∂e}{∂x}·M \right)\ N
                                   + \left(𝐃e· \left( \frac{∂N}{∂x}·M \right) \right)\ N \Bstrutbas \\
       \frac{∂ λy.e}{∂x}·M & = & λy.\left(\frac{∂e}{∂x} ·M\right)
                                 \mbox{\ \ ($y$ fresh for $x$ and $M$)}\Bstrutbas \\
       \frac{∂𝐃e·f}{∂x}·M & = &
                                𝐃(\frac{∂e}{∂x}·M)·f + 𝐃e· (\frac{∂f}{∂x}·M) \Bstrutbas \\ \hline
       \frac{∂0}{∂x}·M & = & 0 \Bstrutbas\Tstrut \\
       \frac{∂(e+N)}{∂x}·M & = & \frac{∂e}{∂x}·M + \frac{∂N}{∂x}·M.\Tstrut
     \end{array}$
     \caption{Auxiliary functions for differential $λ$-calculus}
    \label{fig:diff}
    \end{figure}
      The first two lines extend $e+M$ to take a
      multiterm as its first argument.
      The next two extend
      $e\ M$ similarly.
      The next two extend $λ$-abstraction.
      Matters then get slightly subtle:
      $𝐃{-}·{-}$ is first extended to take a multiterm
      as its second argument, and then the obtained
      extension is further extended to take a multiterm
      as its first argument.
    \item Then, relying on this, capture-avoiding substitution of a
      variable by a multiterm in a simple or multiterm is defined by
      induction in Figure~\ref{fig:diff}.B, the result being a multiterm.
  \item Finally, partial differentiation is also defined inductively,
    relying on extended operations, in Figure~\ref{fig:diff}.C.
  \end{enumeratewide}

  \subsubsection{Using incremental structural laws}
  Let us now model this syntax and auxiliary operations using
  incremental structural laws, but ignoring equations for this paper.

  We have two syntactic categories, simple terms and multiterms, but
  only one sort for variables, to be replaced with multiterms.  As
  in~§\ref{ss:evconts}, we model this by working with the category
  $[𝐒𝐞𝐭,𝐒𝐞𝐭²]_f$ of finitary functors $F∶ 𝐒𝐞𝐭 → 𝐒𝐞𝐭²$, or equivalently
  functors $𝔽 → 𝐒𝐞𝐭²$. For a change, this time, we emphasise the
  presentation as functors $𝐒𝐞𝐭 → 𝐒𝐞𝐭²$. Furthermore, we write $𝐬$ and
  $𝐦$ for the elements of $2$, and think of $F(X)_𝐬$ as a set of
  simple terms with free variables in $X$, and of $F(X)_𝐦$ as a set of
  multiterms with free variables in $X$. We sometimes write such
  functors as pairs $(F_𝐬,F_𝐦)$ of set-functors.

  The basic endofunctor is then defined as follows.
  $$
  \begin{array}{rcc@{}c@{}c@{}c@{}c@{}c@{}c}
    Σ(F)(X)_𝐬 & = & X  & \,+\, & F(X)_𝐬 × F(X)_𝐦 & \,+\, &
                     F(X+1)_𝐬 & \,+\, & F(X)_𝐬² \\
    e,f   & ::= & x &｜&  e\ M &｜& λ(e) &｜& 𝐃e·f \\
    Σ(F)(X)_𝐦 & = & 1  &+& F(X)_𝐬 × F(X)_𝐦 \\
    M   & ::= & 0 &｜&  e + M
    
  \end{array}$$

  \begin{notation}
    Writing $𝐲_𝐬 = (1,∅)$ and $𝐲_𝐦 = (∅,1)$ (viewing
    $𝐒𝐞𝐭²$ as the category of presheaves over the discrete category
    $\ens{𝐬,𝐦}$), we equivalently have
  $$\begin{array}{r@{\,}c@{\,}l}
    Σ(F)(X)  & = &  (X   +  F(X)_𝐬 × F(X)_𝐦  + 
    F(X+1)_𝐬  +  F(X)_𝐬²) · 𝐲_𝐬 \\
    && {} + ( 1  + F(X)_𝐬 × F(X)_𝐦) · 𝐲_𝐦.    
  \end{array}$$
    
  \end{notation}

  Let us now define extended operations.  We start by defining the
  first four layers of Figure~\ref{fig:diff}.A independently, using
  Proposition~\ref{prop:independent} (with $T=\id$). The relevant
  arities are
  \begin{mathpar}
    Γ^{\plus}(F,G)(X) = (F(X)_𝐦 × G(X)_𝐦) · 𝐲_𝐦
    \and
    Γ^{\app}(F,G)(X) = (F(X)_𝐦 × G(X)_𝐦) · 𝐲_𝐦
    \and
    Γ^{\abs}(F,G)(X) = F(X+1)_𝐦 · 𝐲_𝐦
    \and 
    Γ^{\lapp₀}(F,G)(X) = (G(X)_𝐬 × F(X)_𝐦) · 𝐲_𝐦\rlap{,}
  \end{mathpar}
  and the equations in the table may be read as defining the desired
  simple structural laws
  $$Γⁱ(Σ(F),G) → S (Γⁱ(F,S(G)) + F + G)\rlap{,}$$
  for $i ∈ \ens{\plus,\app,\abs,\lapp₀}$.  We get a distributive law
  $$δ₄∶ T₄S → ST₄\rlap{,}$$ where $T₄ = (∑ᵢ Γⁱ_S)^\fleur$.
  \begin{remark}
    The operation in the fourth layer is defined by induction on the
    second argument, so the main (inductive) argument for $Γ^{\lapp₀}$
    goes to the right of the product.
  \end{remark}
  
  For the last layer in Figure~\ref{fig:diff}.A, we define an
  incremental structural law over $δ₄$, with arity
  $$Γ^{\lapp} (F,G) (X) = (F(X)_𝐦 × G(X)_𝐦) · 𝐲_𝐦.$$
  Again the equations of the last layer may be read as defining this
  incremental law.
  By Theorem~\ref{main:thm:incremental}, we get a distributive law
  $δ₅∶ T₅S → S T₅$, where $T₅ = T₄ ⊕ T_{\lapp}$, with
  $T_{\lapp} = (Γ^{\lapp}_{S})^\fleur$.

  We now define capture-avoiding substitution.  The arity is reminiscent
  of substitution monoidal structures,
  so we write it as a tensor
  product. Because the result of a substitution is a multiterm, we
  readily put $(F⊗G)(X)_𝐬 = ∅$. Then, we define
  $$(F⊗G) (X)_𝐦 = F(G(X)_𝐦)_𝐬 + F (G (X)_𝐦)_𝐦\rlap{,}$$
  reflecting the fact that we substitute multiterms for variables in
  both simple and multiterms.

  The two layers of Figure~\ref{fig:diff}.B may be read as
  defining the two components of the desired incremental
  structural law
  $$Σ(F)⊗G → ST₅( F⊗ST₅G + F + G )\rlap{,}$$
  which yields by Theorem~\ref{main:thm:incremental} a distributive law
  $δ₆∶ T₆S → ST₆$, where $T₆ = T₅ ⊕ T_{\subst}$, with
  $T_{\subst} = ({-}⊗{S-})^\fleur$.

  Finally, we define partial differentiation.
  The arity is again empty at $𝐬$, with
  $$Γ^{\diff}(F,G)(X)_𝐦 =
  F(X+1)_𝐬 × G(X)_𝐦 + F(X+1)_𝐦×G(X)_𝐦.$$ The additional variable in
  the first argument models the distinguished variable $x$ along which
  we differentiate.  Again the equations may be read as defining an
  incremental structural law and we get a distributive law
  $δ₇∶ T₇S → ST₇$, where $T₇ = T₆ ⊕ T_{\diff}$, with
  $T_{\diff} = (Γ^{\diff}_{S})^\fleur$.

  One could be content with this, prove the desired commutation
  lemmas~\cite[§6.1.4]{Iouri}, and then show that everything is
  compatible with Equations~\eqref{eq:diff}, all by hand.  But of
  course, it would be better to derive both results automatically.  In
  the next section, we will explain how to derive commutation lemmas
  (though on a simpler example), leaving compatibility
  with equations for further work.

\subsection{Application: capture-avoiding substitution, De Bruijn style}\label{ss:db}
Most applications of the paper take place in the setting of so-called
presheaf models, i.e., mild generalisations of the category of
finitary functors $𝐒𝐞𝐭 → 𝐒𝐞𝐭$~\cite{fiore:presheaf}.  Our framework,
however, applies just as well in other settings.  To illustrate this,
in this subsection, we transpose the example of capture-avoiding
substitution to the De Bruijn-style setting of~\cite{HHLM}.

To summarise the idea: just as the presheaf-based approach equipes the
nested datatypes representation~\cite{bird_paterson_1999} with initial
algebra semantics, the setting of~\cite{HHLM} does the same for De
Bruijn representation~\cite{DB}.  In the presheaf-based approach,
terms are indexed by sets $n$ of potential free variables, and by
convention the bound variable is always the greatest one, typically
$x_{n+1}$ for $λₙ∶ X(n+1) → X(n)$.  Indexing is thus made necessary by
the binding convention.  In De Bruijn representation, the binding
convention is the opposite: the bound variable is always $0$. This
makes indexing unnecessary, but some operations become less intuitive
(to many, at least~\cite{DBLP:journals/entcs/BerghoferU07}).

In order to specify capture-avoiding substitution in the De Bruijn
setting, we will need several layers. Indeed, the presheaf-based
approach features not only indexing, but also built-in renaming, which
is not the case in the De Bruijn setting. We thus need to define a
first layer for renaming, and then a second one for substitution.

We take $𝐒𝐞𝐭$ as ambient category, and the basic syntax is specified
by the endofunctor
$$Σ(X)  =  ℕ + X² + X.$$
For the first layer, letting $S = Σ^*$ again, the auxiliary bifunctor
is
$$Γ(X,Y)  =  X × ℕ^ℕ\rlap{,}$$
and renaming is specified by the following simple structural law
(omitting $η^S$ for readability):
$$
\begin{array}{rcl}
  d_{X,Y}∶ Σ(X) × ℕ^ℕ & → & S(X × ℕ^ℕ+X+Y)  \\
  (xᵢ,ρ) & ↦ & x_{ρ(i)} \\
  (λ(e),ρ) & ↦ & λ(in₁(e,1+ρ)) \\
  (e\ f, ρ) & ↦ & in₁(e,ρ)\ in₁(f,ρ) \\
\end{array}$$
where
  $1+ρ∶ ℕ → ℕ$ is defined by
  $$\begin{array}{rcl}
      (1+ρ)(0) & = & 0 \\
      (1+ρ)(1+n) & = & 1+ρ(n).
  \end{array}$$
Letting $T = Γ_S^*$, we get a distributive law $δ∶ TS → ST$ by
Theorem~\ref{thm:simple}.
\begin{notation}
  We denote by $e[ρ]$ the renaming operation of $T$.
\end{notation}
We then specify substitution by taking as
auxiliary bifunctor $Θ(X,Y) = X × Y^ℕ$, with incremental structural
law defined by
$$\begin{array}{rcl}
    Σ(X) × Y^ℕ & → & ST (X × (STY)^ℕ + X + Y) \\
    (xᵢ,σ) & ↦ & in₃(σ(i)) \\
    (λ(e),σ) & ↦ & λ(in₁(e,⇑σ)) \\
    (e\ f, σ) & ↦ & in₁(e,σ)\ in₁(f,σ)\rlap{,}
  \end{array}$$
  where $⇑σ∶ ℕ → STY$ is defined by
  $$\begin{array}{rcl}
      ⇑σ∶ ℕ & → & STY \\
      0 & ↦ & x₀ \\
      n+1 & ↦ & σ(n)[↑]\rlap{,}
    \end{array}$$
    letting $↑∶ ℕ → ℕ$ denote the successor map.

\section{Benign equations}\label{s:benign}
In this section, we examine what we call ``benign equations'', i.e.,
equations that are automatically satisfied by the initial model.  The
initial model is thus initial in a potentially smaller category.

We start in~§\ref{ss:benign:ex} by presenting a simple example to
motivate this investigation.  In~§\ref{ss:eqsys}, along with recalling
Fiore and Hur's~[\citeyear{FioreHurEquational}] \alert{equational
  systems} and some basic facts about them, we then introduce a
mathematical definition of a system of benign equations, called a
\alert{benign} equational
system.  
In~§\ref{ss:ses}, we present a notion of signature for benign
equational systems, called \alert{structural equational systems}, which
we apply in~§\ref{ss:assocsubst} to prove associativity of
capture-avoiding substitution.

\subsection{On a simple example}\label{ss:benign:ex}
In the case of addition, as defined by structural recursion
in~§\ref{ss:addition}, an example of benign equation is associativity:
\begin{equation}
(x+y)+z = x+(y+z).\label{eq:assoc}
\end{equation}
This is typically proven by induction on the first variable $x$. We
explain how to derive this result from the results
of~§\ref{ss:addition} and~§\ref{s:incr-struct-laws}, which leads to
structural equational systems.

We start by briefly sketching the idea.  By the construction
of~§\ref{s:incr-struct-laws}, we define a new, auxiliary ternary
operation $\myop(x,y,z)$ by the structurally recursive equations
\begin{eqnarray}
  \label{eq:assocdefi}
  \myop(0,y,z) & = & y+z \\
  \label{eq:assocdefii}  
  \myop(s(x),y,z) & = & s(\myop(x,y,z)).
\end{eqnarray}
By Theorem~\ref{main:thm:incremental}, the syntax, $S∅$, admits
a unique such operation.
But on the other hand, one easily proves that taking $\myop(x,y,z) = x+ (y+z)$ or
$\myop(x,y,z) = (x+y) + z$ yields two such operations. Indeed, we have
  \begin{mathpar}
    s(x) + (y+z) = s(x + (y+z)) \and 0 + (y+z) = y+z
  \end{mathpar}
  and
  \begin{mathpar}
    (s(x)+y) + z = s (x+y) + z = s ((x+y)+z)
    \and
    (0+y)+z = y+z
  \end{mathpar}
  by the defining equations of addition.  By uniqueness, both
  definitions of $\myop$ must thus coincide, which
  proves~\eqref{eq:assoc}.

  More formally, we start from the distributive law $δ∶TS→ST$
  of~§\ref{ss:addition}.  We introduce a bifunctor $Ψ∶ 𝐒𝐞𝐭² → 𝐒𝐞𝐭$
  mapping $(X,Y)$ to $X×Y×Y$, for the arity of the ternary
  operation. Again, $X$ corresponds to the decreasing argument in the
  recursive definition of $\myop$, while $Y$ accounts for other
  arguments.  As in~§\ref{s:incr-struct-laws}, the recursive
  definition is modelled by an incremental structural law
$$
\begin{array}{rcll}
  d_{X,Y}∶ Ψ(ΣX,Y) & → & ST (Ψ(X,STY)+X+Y) \\
  (0,y₁,y₂) & ↦ & in₃(y₁) + in₃(y₂) \\
  (s(x),y₁,y₂) & ↦ & s(in₁(x, y₁, y₂)).
\end{array}
$$
Equation~\eqref{eq:assoc} induces a pair
$$L_X,R_X∶ Ψ(X,X) → STX$$
of natural transformations, respectively mapping any triple
$(x,y,z) ∈ X³$ to $(x + y) + z$ and $x + (y + z)$, and we want to show
that the algebra structure $μ^{ST}_{∅}∶ STST∅ → ST∅$ ($≅ S∅$) of the
initial $ST$-algebra coequalises $L_{ST∅}$ and $R_{ST∅}$, as in
  \begin{center}
    \diag(1,1.7){%
      Ψ(ST∅,ST∅) \& STST∅ \& ST∅.
    }{%
      (m-1-1) edge[bend left=10,labela={L_{ST∅}}] (m-1-2) %
      (m-1-1) edge[bend right=10,labelb={R_{ST∅}}] (m-1-2) %
      (m-1-2) edge[onto,labela={μ^{ST}_{∅}}] (m-1-3) 
    }
  \end{center}
  By Theorem~\ref{thm:models}, models of the incremental structural
  law $d$ are $ST$-algebras equipped with an additional ternary
  operation making the diagram~(d2) commute (with $S ≔ ST$ and $Γ ≔ Ψ$), and
  furthermore, by Theorem~\ref{thm:incremental}, $ST∅$ is an initial
  model.  By initiality, we thus merely need to show that both
  induced composites $Ψ(ST∅,ST∅) → ST∅$ make~(d2) commute.

  In the next subsections, we propose an abstract version of this
  idea, with a simple sufficient condition for it to work.

\subsection{Equations}\label{ss:eqsys}
In this subsection, loosely
following~\cite{FioreHurEquational,HHLlong}, we introduce an abstract
notion of equational system, and define what it means for such an
equational system to be benign.  
\begin{definition}
  Given a finitary monad $T$ on a locally finitely presentable
  category $𝐂$, an \alert{equational system} consists of a finitary
  monad $G$ on $𝐂$, together with two monad morphisms $G → T$.
\end{definition}

\begin{example}
  Taking $𝐂 = 𝐒𝐞𝐭$, we model associativity of a binary operation by
  taking $T$ and $G$ to be the free monads on the endofunctors
  $Σ(X) = X²$ and $Θ(X) = X³$ on sets, respectively.  The monad
  morphisms $L,R∶ G → T$ are induced by universal property of $G=Θ^*$
  from the natural transformations $L⁰,R⁰∶ Θ → T$ defined by
  $$\begin{array}{rcl}
    L⁰_X(x₁,x₂,x₃) & = & x₁ + (x₂ + x₃) \\
    R⁰_X(x₁,x₂,x₃) & = & (x₁ + x₂) + x₃.
    \end{array}$$
\end{example}

\begin{definition}
  The \alert{quotient} $E^⋆$ of a finitary monad $T$ by an equational
  system $E = (G,L,R)$ is the coequaliser
  \begin{center}
    \diag(1,1.7){%
      G \& T \& E^⋆
    }{%
      (m-1-1) edge[bend left=10,labela={L}] (m-1-2) %
      (m-1-1) edge[bend right=10,labelb={R}] (m-1-2) %
      (m-1-2) edge[onto,labela={q}] (m-1-3) 
    }
  \end{center}
  of $L,R$ in $𝐌𝐧𝐝_f(𝐂)$.
\end{definition}
\begin{remark}
  The coequaliser exists because the category of finitary monads on a
  locally finitely presentable category is itself locally finitely
  presentable~\cite{MonadicityMonads}, hence in particular cocomplete.
\end{remark}
\begin{remark}
  \label{rk:coeq-monads}
  The coequaliser is also a coequaliser in the category $𝐌𝐧𝐝(𝐂)$ of
  (not necessarily finitary) monads on $𝐂$, since colimits are
  preserved by the embedding $𝐌𝐧𝐝_f(𝐂) → 𝐌𝐧𝐝(𝐂)$ by~\citet[Proposition
  5.6]{blackwell_phd}.
\end{remark}
\begin{definition}
  An equational system is \alert{benign} when the universal
  coequalising morphism $T → E^⋆$ is admissible.
\end{definition}

Although coequalisers of finitary monads may not be intuitive to all
readers, the monad $E^⋆$ admits the following nice characterisation by
its algebras.
\begin{definition}\label{def:Ealg}
  Given a finitary monad $T$ and an equational system $E= (G,L,R)$ on
  it, a $T$-algebra $a∶ TX → X$ \alert{satisfies} $E$ iff $a$
  coequalises $L_X$ and $R_X$, i.e., $a ∘ L_X = a ∘ R_X$.

  Such $T$-algebras are called \alert{$E$-algebras}, and we denote by
  $E\alg$ the full subcategory of $T\Alg$ spanned by them.
\end{definition}
\begin{remark}
  Equivalently, $E\alg$ is the equaliser in $𝐂𝐀𝐓$ of the induced
  functors $T\Alg → G\Alg$, which \emph{a priori} might differ from
  the equaliser in $𝐌𝐨𝐧𝐚𝐝𝐢𝐜_f/𝐂$, the full subcategory of $𝐌𝐨𝐧𝐚𝐝𝐢𝐜/𝐂$
  spanned by finitary monadic functors.  They in fact coincide, by the
  following Proposition.
\end{remark}

\begin{proposition}\label{prop:monadiceq}
  For any finitary monad $T$ on a locally finitely presentable
  category $𝐂$ and equational system $E= (G,L,R)$ on it, we have an
  isomorphism $E\alg ≅ E^⋆\Alg$ over $𝐂$. Otherwise said, the
  forgetful functor $E\alg → 𝐂$ is finitary and monadic, and the
  associated monad is isomorphic to $E^⋆$.
\end{proposition}
\begin{proof}
  By Remark~\ref{rk:coeq-monads}, $E^⋆$ is a coequaliser of 
  $L$ and $R$, as monad morphisms. By~\citet[Proposition 26.3]{KellyUnified}, its category of
  algebras is thus computed as the equaliser in $𝐂𝐀𝐓/𝐂$ of the functors
  $T\Alg → G\Alg ≅ Θ\alg$, as claimed.
\end{proof}

We may derive from this the following useful characterisation of
benign equational systems.
\begin{proposition}\label{prop:benign}
  An equational system is benign iff the initial $T$-algebra satisfies
  it, i.e., $μᵀ_{∅} ∘ L_{T∅} = μᵀ_{∅} ∘ R_{T∅}$.
\end{proposition}
\begin{proof}
  If the initial algebra satisfies an equational system $E$, then it
  is \emph{a fortiori} initial in $E\alg$, hence in $E^⋆\Alg$ by
  Proposition~\ref{prop:monadiceq}. The forgetful functor
  $E^⋆\Alg → T\Alg$ thus preserves the initial object, and we conclude
  by Proposition~\ref{prop:admissible:street}.

  Conversely, if an equational system $E$ is benign, then,
  by Proposition~\ref{prop:admissible:street},
  the initial $T$-algebra admits a unique $E^⋆$-algebra structure $e∶ E^⋆T∅ → T∅$ making
  the following diagram commute
  \begin{center}
    \diag{%
      T T ∅ \& \& E^⋆ T∅ \\
      \& T∅\rlap{,}
    }{%
      (m-1-1) edge[labela={q_{T ∅}}] (m-1-3) %
      edge[labelbl={μᵀ_∅}] (m-2-2) %
      (m-1-3) edge[labelbr={e}] (m-2-2) %
    }
  \end{center}
  and $e$ makes $T∅$ into an initial $E^⋆$-algebra.  But $q_{T∅}$
  coequalises $L_{T∅}$ and $R_{T∅}$, hence so does $μᵀ_∅$, as desired.
\end{proof}

Let us conclude this subsection with the following observation on
combining equations.
\begin{definition}
  For any family $(Eᵢ)_{i ∈ I}$ of equational systems on a given
  finitary monad $T$ on a locally finitely presentable category $𝐂$,
  with  $Eᵢ = (Gᵢ,Lᵢ,Rᵢ)$ for all $i ∈ I$,
  let $∑ᵢEᵢ$ denote the equational system
  defined by the cotuplings
  \begin{center}
    $[Lᵢ]_{i ∈ I}∶ ⨁_{i ∈ I} Gᵢ → T$ \qquad and \qquad
    $[Rᵢ]_{i ∈ I}∶ ⨁_{i ∈ I} Gᵢ → T$.
  \end{center}
\end{definition}
\begin{proposition}\label{prop:benign:systems}
  For any family $(Eᵢ)_{i ∈ I}$ of equational systems on a given
  finitary monad $T$ on a locally finitely presentable category $𝐂$,
  $∑ᵢ Eᵢ$ is benign iff each $Eᵢ$ is.
\end{proposition}
\begin{proof}
  By Proposition~\ref{prop:benign} and universal property of coproduct.
\end{proof}

\subsection{Structural equational systems}\label{ss:ses}
In this subsection, we introduce a notion of signature for benign equational
systems.

We fix a free monad $S = Σ^*$ on a locally finitely presentable
category $𝐂$, with $Σ$ (hence $S$) finitary.
\begin{definition}
  \   \hfill
  \begin{itemize}
  \item A \alert{structural interpretation} of an incremental
    structural law $$d_{X,Y}∶ Θ(ΣX,Y) → ST (Θ(X,STY)+X+Y)$$ over some
    given distributive law $δ∶ TS → ST$ is a natural transformation
    $K_X∶ Θ(X,X) → STX$ making the coherence diagram of
    Figure~\ref{fig:interp} commute.
  \begin{figure}[htp]
    \centering
    \diag{%
      Θ(ΣX,X) \& \& ST (Θ(X,STX)+X+X) \\
      Θ(SX,SX) \& \& ST (Θ(STX,STX)+X) \\
      STSX \& \& ST(STSTX+X) \\
      SSTX \& \& ST (STX) \\
      \& STX %
    }{%
      (m-1-1) edge[labell={Θ(η_{Σ,X},η^S_X)}] (m-2-1) %
      edge[labela={d_{X,X}}] (m-1-3) %
      (m-2-1) edge[labell={K_{SX}}] (m-3-1) %
      (m-3-1) edge[labell={Sδ_X}] (m-4-1) %
      (m-4-1) edge[labelbl={μ^S_{TX}}] (m-5-2) %
      (m-1-3) edge[labelr={ST (Θ(η^{ST}_X,STX)+[X,X])}] (m-2-3)
      (m-2-3) edge[labelr={ST(K_{STX}+X)}] (m-3-3) %
      (m-3-3) edge[labelr={ST[μ^{ST}_X,η^{ST}_X]}] (m-4-3) %
      (m-4-3) edge[labelbr={μ^{ST}_X}] (m-5-2) %
    }
    \caption{Coherence diagram for structural interpretations}
    \label{fig:interp}
  \end{figure}
\item A \alert{structural equational system} over a distributive law
  $δ∶ TS → ST$ consists of
  \begin{itemize}
  \item an incremental structural law $(Θ,d)$ over $δ$, together with
  \item a pair of structural interpretations $L,R∶ ΘΔ → ST$ of $d$.
  \end{itemize}
\end{itemize}
\end{definition}

We now associate an equational system to any structural equational
system, and prove that it is benign.
\begin{definition}
  For any structural interpretation $K∶ ΘΔ → ST$
  of
  $$d_{X,Y}∶ Θ(ΣX,Y) → ST (Θ(X,STY) + X + Y)$$
  over $δ∶ TS → ST$, let $\tilde{K}∶ Θ_S^\fleur → ST$ denote the monad
  morphism induced by universal property of $Θ_S^\fleur$ from the
  composite
  $$Θ_{SX} X \xto{Θ_{SX} η^S_X} Θ_{SX} SX \xto{K_{SX}} STSX
  \xto{Sδ_X} SSTX \xto{μ^S_{TX}} STX.$$ For any structural equational
  system $E = (d,L,R)$ over $δ$, the equational system
  $\tilde{E}$
  \alert{induced} by $E$ is the pair
  $\tilde{L},\tilde{R}∶ Θ_S^\fleur → ST$.
\end{definition}
\begin{notation}
  We often conflate $E$ and the associated equational system
  $\tilde{E}$.  In particular, recalling Definition~\ref{def:Ealg} and
  Proposition~\ref{prop:monadiceq}, we speak of $E$-algebras, which
  form a category $E\alg ≅ E^⋆\Alg$.
\end{notation}

\begin{theorem}\label{thm:benign}
  Consider any monad distributive law $δ∶ TS → ST$ in a locally
  finitely presentable category $𝐂$ such that $T$ is
  constant-free. Then for any structural equational system $E$ over
  $δ$, the quotient morphism $ST → \tilde{E}^⋆$ is admissible.
\end{theorem}

\begin{proof}[Proof sketch (see~§\ref{app:benign} for more detail)]
  Let $$d_{X,Y}∶ Θ(ΣX,Y) → ST (Θ(X,STY) + X + Y)$$ denote the given
  incremental structural law, and $L$ and $R$ denote the two
  structural interpretations.

  We first prove that each any structural interpretation $K$ induces
  $ΘΔ$-algebra structure on any $ST$-algebra $X$, given by
  $$Θ(X,X) \xto{K_X} STX → X\rlap{,}$$
  and furthermore that this $ΘΔ$-algebra structure satisfies~(d2), hence
  by Theorem~\ref{thm:models} makes $X$ into an $S(T⊕T')$-algebra.

  The given structural interpretations $L$ and $R$ thus induce
  two extensions of the $ST$-algebra structure of $S∅ ≅ ST∅$
  to $S(T⊕T')$-algebra structure.
  But by
  Theorem~\ref{main:thm:incremental} and
  Proposition~\ref{prop:admissible:street}, $S∅$ has a unique
  compatible $ΘΔ$-algebra structure making it into a
  $S (T⊕T')$-algebra structure, so both algebra structures derived
  from $L$ and $R$ must agree with the canonical one. In particular
  the diagram
  \begin{center}
    \diag(1,1.7){%
      Θ({S∅},{S∅}) \& ST{S∅} \& {S∅}
    }{%
      (m-1-1) edge[bend left=10,labela={L_{S∅}}] (m-1-2) %
      (m-1-1) edge[bend right=10,labelb={R_{S∅}}] (m-1-2) %
      (m-1-2) edge[onto,labela={}] (m-1-3) 
    }
  \end{center}
  commutes, hence $S∅$ satisfies the induced equational system, and is
  thus initial in $S (T⊕T')\Alg$.
\end{proof}

Let us conclude this subsection with the following direct consequence
of Theorem~\ref{thm:benign} and Proposition~\ref{prop:benign:systems}.
\begin{corollary}\label{cor:benign:systems}
  Consider any monad distributive law $δ∶ TS → ST$ in a locally
  finitely presentable category $𝐂$ such that $T$ is constant-free.
  Then for any family $(Eᵢ)_{i ∈ I}$ of structural equational systems
  over $δ$, the quotient morphism $ST → (∑ᵢ\tilde{Eᵢ})^⋆$ is
  admissible.
\end{corollary}

\subsection{Application: associativity of substitution}\label{ss:assocsubst}
In this section, we continue the development of~§\ref{ss:lambda}.
  There, we
specified the syntax of pure $λ$-calculus by an endofunctor $Σ$ on
$[𝐒𝐞𝐭,𝐒𝐞𝐭]_f$, and substitution by a simple structural law
$$d_{X,Y}∶ Γ(ΣX,Y) → S (Γ_{SY}(X) + X + Y)\rlap{,}$$
where $Γ(X,Y) = X ⊗ Y$, thus generating an admissible morphism
$S → ST$, where $S = Σ^*$ and $T = Γ_S^\fleur$.

We now want to show that the $ΓΔ$-algebra structure of the syntactic
model $S∅$, a.k.a.\ capture-avoiding substitution, is associative.
Again, this will be subsumed by~§\ref{ss:ptstrengths}.

We define a structural equational system with the following
components:
\begin{itemize}
\item the arity functor is $Θ(X,Y) = (X⊗Y)⊗Y$;
\item the incremental structural law is defined by
  $$
  \begin{array}{rcl}
    d'_{X,Y}(x,σ,θ) & = & σ(x)[θ] \\
    d'_{X,Y} (e\ f,σ,θ) & = & (e,σ,θ)\ (f,σ,θ) \\
    d'_{X,Y} (λ(e),σ,θ) & = & λ(e,σ^↑,θ^↑)\rlap{,}
  \end{array}$$
  where
  \begin{itemize}
  \item $σ∶ p → ST(Y)(q)$, 
  \item $θ∶ q → ST(Y)(n)$, 
  \item ${-₁}[-₂]$ denotes the formal ($=$ explicit) substitution
    operation of $T$, and
  \item we omit coproduct injections and $η^{ST}$ for readability;
\end{itemize}
\item the structural interpretations are defined by
  \begin{center}
    $L_X(e,σ,θ) = e[σ][θ]$ and  $R_X(e,σ,θ) = e[σ[θ]]$,
  \end{center}
  where by definition $σ[θ](xᵢ) ≔ σ(xᵢ)[θ]$ (writing $xᵢ$ for the
  $i$th element of $p$ to emphasise that it is thought of as a
  variable).
\end{itemize}
Checking the coherence condition of Figure~\ref{fig:interp}
essentially amounts to checking each case of the usual induction, separately.
The most interesting case is that of abstraction, with the right-hand
side $R$:
\begin{itemize}
\item $d$ maps any triple $(λ(e),σ,θ)$ to $λ(e,σ^↑,θ^↑)$, which the
  right-hand composite then maps to $λ(e[σ^↑[θ^↑]])$, while
\item the left-hand composite maps the triple to $λ(e[σ[θ]^↑])$.
\end{itemize}
We thus need to prove $σ^↑[θ^↑] = σ[θ]^↑$:
\begin{itemize}
\item on $x_{p+1}$, we directly have
  $σ[θ]^↑(x_{p+1}) = x_{n+1}$, and, slightly less directly,
  $$σ^↑[θ^↑](x_{p+1}) = σ^↑(x_{p+1})[θ^↑] = x_{q+1}[θ^↑] = x_{n+1}\rlap{;}$$
\item on $xᵢ$ for $i ∈ p$, we have 
  $σ[θ]^↑(xᵢ) = wₙ · (σ(xᵢ)[θ]) = σ(xᵢ)[STX(wₙ)∘θ]$, where $wₙ∶ n ↪ n+1$ denotes the inclusion,
  while
  $$\begin{array}{rcl}
      σ^↑[θ^↑](xᵢ) & = & σ^↑(xᵢ)[θ^↑]  \\ 
                   & = &  (w_q · σ(xᵢ))[θ^↑] \\ 
                   & = &  σ(xᵢ)[θ^↑ ∘ w_q]. 
    \end{array}$$ But the following square
    commutes by definition of $θ^↑$, hence the result.
    \begin{center}
      \diag{%
        q \& q+1 \\
        STX(n) \& STX(n+1) %
      }{%
        (m-1-1) edge[labela={w_q}] (m-1-2) %
        edge[labell={θ}] (m-2-1) %
        (m-2-1) edge[labelb={STX(wₙ)}] (m-2-2) %
        (m-1-2) edge[labelr={θ^↑}] (m-2-2) %
      }
    \end{center}
  \end{itemize}
Theorem~\ref{thm:benign} then tells us that the usual substitution
lemma is satisfied in the syntax $S(∅)$.

\subsection{Embedding presheaf-based models}\label{ss:ptstrengths}
In this section, we show how the general framework of pointed strong
endofunctors~\cite{fiore:presheaf,DBLP:conf/lics/Fiore08} embeds into
ours. More precisely, for any pointed strong endofunctor $Σ$ on a
monoidal, locally finitely presentable category $(𝐂,⊗,I,α,λ,ρ)$
satisfying standard additional axioms (see Definition~\ref{def:lfpmon}
below):
\begin{itemize}
\item We define a simple structural law, whose initial algebra
  $S∅ = (I + Σ)^*∅$ is the desired syntax, and whose category of
  algebras is a relaxed variant of Fiore et al.'s;
\item We then define two structural equational systems $E₁$ and $E₂$,
  whose joint algebras in the sense of Corollary~\ref{cor:benign:systems} are
  precisely those of Fiore et al.'s.
\end{itemize}
We thus recover the admissible morphism
$$S → ST → (\tilde{E}₁ + \tilde{E}₂)^⋆ ≅ Σ^⊛$$
of Example~\ref{ex:fpt}.

\begin{remark}
  The motivation for this subsection is one of connecting to other
  people's work. In applications, it will probably be easier to
  directly define the desired structural laws and equational systems.
\end{remark}

We start by recalling the notion of pointed strong endofunctor, and
its associated category of models.

  \begin{definition}\label{def:pts}
  A \alert{pointed strength} on an endofunctor $Σ∶ 𝐂 → 𝐂$ on a
  monoidal category $(𝐂,⊗,I,α,λ,ρ)$ is a family of morphisms
  $st_{C,(D,v)}∶ Σ(C)⊗D → Σ (C⊗D)$, natural in $C ∈ 𝐂$ and
  $(D,v∶ I → D) ∈ I/𝐂$, the coslice category below $I$, making the
  following diagrams commute,
      \begin{center}
    \diag{%
      \& Σ(A) \\
      Σ(A)⊗I \& \& Σ (A⊗I) %
    }{%
      (m-1-2) edge[labelal={ρ_{Σ(A)}}] (m-2-1) %
      edge[labelar={Σ(ρ_A)}] (m-2-3) %
      (m-2-1) edge[labelb={st_{A,(I,\id)}}] (m-2-3) %
    }
    \diag(.6,1.6){%
      (Σ(A)⊗X)⊗Y \& Σ(A⊗X)⊗Y \& Σ ((A⊗X)⊗Y) \\
      Σ(A)⊗(X⊗Y) \& \& Σ (A⊗(X⊗Y))
    }{%
      (m-1-1) edge[labela={st_{A,(X,v_X)}⊗Y}] (m-1-2) %
      edge[labell={α_{Σ(A),X,Y}}] (m-2-1) %
      (m-1-2) edge[labela={st_{A⊗X,(Y,v_Y)}}] (m-1-3) %
      (m-1-3) edge[labelr={Σ (α_{A,X,Y})}] (m-2-3) %
      (m-2-1) edge[labelb={st_{A,(X⊗Y,v_{X⊗Y})}}] (m-2-3) %
    }
  \end{center}
  where $v_X∶ I → X$ and $v_Y∶ I → Y$ are the given points, and
  $v_{X⊗Y}$ denotes the composite
  $$I \xto{ρ_I^{-1}} I ⊗ I \xto{v_X ⊗ v_Y} X ⊗ Y.$$
\end{definition}

\begin{definition}\label{def:SMon}
  For any pointed strong endofunctor $Σ$ on $𝐂$, a \alert{$Σ$-monoid}
  is an object $X$ equipped with $Σ$-algebra and monoid structure, say
  $a∶ Σ(X) → X$, $s∶ X⊗X→X$, and $v∶ I → X$, such that the following
  pentagon commutes.
  \begin{equation}
    \diag{%
      Σ(X)⊗X \& Σ (X⊗X) \& Σ(X) \\
      X⊗X \& \& X %
    }{%
      (m-1-1) edge[labela={st_{X,(X,v)}}] (m-1-2) %
      edge[labell={a⊗X}] (m-2-1) %
      (m-1-2) edge[labela={Σ(s)}] (m-1-3) %
      (m-2-1) edge[labelb={s}] (m-2-3) %
      (m-1-3) edge[labelr={a}] (m-2-3) %
    }
    \label{eq:pentagon:Fmonoids}
  \end{equation}
  A morphism of $Σ$-monoids is a morphism in $𝐂$ which is a morphism
  both of $Σ$-algebras and of monoids.  We let $Σ\Mon$ denote the
  category of $Σ$-monoids and morphisms between them.
\end{definition}

Let us now show how any pointed strong endofunctor $Σ$ gives rise to a
simple structural law, and how the coherence laws of $Σ$-monoids may
be enforced by benign equations.

Following~\cite{fiore:presheaf,DBLP:conf/lics/Fiore08}, we assume that
our category $𝐂$ is leftist, in the following sense.
\begin{definition}
  \label{def:lfpmon}
  A monoidal category $(𝐂,⊗,I,α,λ,ρ)$ is \alert{leftist} iff $⊗$ is
  cocontinuous in its first argument, and finitary in its second
  argument.
\end{definition}

\begin{definition}
  The simple structural law associated to any pointed strong
  endofunctor $(Σ,st)$ is defined as follows.
  \begin{itemize}
  \item We take as basic functor $Σ⁺(X) = I + Σ(X)$, and
  \item as auxiliary bifunctor $Γ(X,Y) = X⊗Y$.
  \item We then take as simple structural law $d_{X,Y}$
    the composite
\[\begin{tikzcd}[ampersand replacement=\&]
	{Σ⁺(X) ⊗ SY} \& {I⊗SY + Σ(X)⊗SY} \& {SY + Σ(X ⊗ SY)} \\
	{Σ⁺(X) ⊗ Y} \&\& {S(X⊗SY + X + Y)\rlap{,}}
	\arrow["{Σ⁺(X) ⊗ η^S_Y}", from=2-1, to=1-1]
	\arrow["\sim", Rightarrow, no head, from=1-1, to=1-2]
;	\arrow["{λ_{SY} + st_{X,(SY,v_Y)}}"{above=.5em}, from=1-2, to=1-3]
	\arrow["{[Sin₃,S(in₁)∘η_{Σ⁺,X⊗SY}∘in₂]}", from=1-3, to=2-3]
\end{tikzcd}\]
where
$S = (Σ⁺)^*$ and $v$ denotes the composite $I → Σ⁺ → S$.
  \end{itemize}
\end{definition}

By Theorem~\ref{thm:simple}, the initial $Σ⁺$-algebra has a unique
model structure, which makes it initial, and furthermore, by
Theorem~\ref{thm:models}, models are $Σ$-algebras $a∶ ΣX → X$,
equipped with morphisms $v∶ I → X$ and $s∶ X⊗X → X$ making the
pentagon
\begin{equation}
\begin{tikzcd}[ampersand replacement=\&]
	{\Sigma^+(X)\otimes X} \&\& {S(X \otimes SX + X + X)} \\
	\&\& {S(X\otimes X+X+X)} \\
	{ X \otimes X} \&\& SX \\
	\& X
	\arrow["{[v,a]\otimes X}"', from=1-1, to=3-1]
	\arrow["{d_{X,X}}", from=1-1, to=1-3]
	\arrow["{S(X\otimes \overline{[v,a]} + X + X)}", from=1-3, to=2-3]
	\arrow["s"', from=3-1, to=4-2]
	\arrow["{S[s,X,X]}", from=2-3, to=3-3]
	\arrow["{\overline{[v,a]}}", from=3-3, to=4-2]
\end{tikzcd}
      \label{eq:pentaplus}
\end{equation}
commute, where $\overline{[v,a]}$ is induced by universal property of $S(X)$.
\begin{proposition}\label{prop:pentas}
  Commutation of~\eqref{eq:pentaplus} is equivalent to joint
  commutation of the pentagon~\eqref{eq:pentagon:Fmonoids} and the diagram
  below.
  \begin{equation}
    \begin{tikzcd}[ampersand replacement=\&]
      {I\otimes X} \&\& { X \otimes X} \\
      \& X
      \arrow["{v \otimes X}", from=1-1, to=1-3]
      \arrow["s", from=1-3, to=2-2]
      \arrow["{\lambda_X}"', from=1-1, to=2-2]
    \end{tikzcd}
    \label{eq:left:unital}
  \end{equation}
\end{proposition}
\begin{proof}
  The domain of~\eqref{eq:pentaplus} is a coproduct, and we claim that
  the restriction to each term yields one of the given diagrams.
  For the first term, we get the following diagram.
  \begin{center}
\begin{tikzcd}[ampersand replacement=\&]
	{I \otimes X} \& {I\otimes SX} \& SX \& {S(X \otimes SX + X + X)} \\
	\& X \&\& {S(X\otimes X+X+X)} \\
	{ X \otimes X} \&\&\& SX \\
	\&\& X
	\arrow["{v \otimes X}"', from=1-1, to=3-1]
	\arrow["{S(X\otimes \overline{[v,a]} + X + X)}", from=1-4, to=2-4]
	\arrow["s"', from=3-1, to=4-3]
	\arrow["{S[s,X,X]}", from=2-4, to=3-4]
	\arrow["{\overline{[v,a]}}", from=3-4, to=4-3]
	\arrow["{I\otimes \eta^S_X}", from=1-1, to=1-2]
	\arrow["{Sin_3}", from=1-3, to=1-4]
	\arrow["{\lambda_{SX}}", from=1-2, to=1-3]
	\arrow["{\lambda_X}"', from=1-1, to=2-2]
	\arrow["{\eta^S_X}"', from=2-2, to=1-3]
	\arrow["{\eta^S_X}"{description}, from=2-2, to=3-4]
	\arrow[Rightarrow, no head, from=2-2, to=4-3]
\end{tikzcd}
  \end{center}
  For the second term, we verify that both pentagons are equivalent by
  chasing the following diagram.
  \begin{center} 
    \hfill
\adjustbox{max width=.95\columnwidth}{
\begin{tikzcd}[ampersand replacement=\&,baseline=(\tikzcdmatrixname-4-3.base)]
	{\Sigma(X)\otimes X} \& {\Sigma(X) \otimes S(X)} \& {\Sigma(X \otimes SX)} \& {\Sigma^+(X \otimes SX)} \& {S(X \otimes SX + X + X)} \\
	\& {\Sigma(X)\otimes X} \& {\Sigma(X\otimes X)} \& {\Sigma^+(X\otimes X)} \& {S(X\otimes X+X+X)} \\
	{ X \otimes X} \&\& {\Sigma(X)} \& {\Sigma^+(X)} \& SX \\
	\&\& X
	\arrow["{a\otimes X}"', from=1-1, to=3-1]
	\arrow["{S(X\otimes \overline{[v,a]} + X + X)}"', from=1-5, to=2-5]
	\arrow["s"', from=3-1, to=4-3]
	\arrow["{S[s,X,X]}"', from=2-5, to=3-5]
	\arrow["{\overline{[v,a]}}", from=3-5, to=4-3]
	\arrow["{\Sigma(X) \otimes \eta^S_X}", from=1-1, to=1-2]
	\arrow["{st_{X,(X,v_X)}}", from=1-2, to=1-3]
	\arrow["{in_2}", from=1-3, to=1-4]
	\arrow["{\eta_{\Sigma^+,in_1}}", from=1-4, to=1-5]
	\arrow["{\Sigma(X\otimes \overline{[v,a]})}"', from=1-3, to=2-3]
	\arrow["{\eta_{\Sigma^+,in_1}}"', from=2-4, to=2-5]
	\arrow["{in_2}"', from=2-3, to=2-4]
	\arrow["{\Sigma(X)\otimes \overline{[v,a]}}"', from=1-2, to=2-2]
	\arrow["{st_{X,(X,v)}}"', from=2-2, to=2-3]
	\arrow[curve={height=12pt}, Rightarrow, no head, from=1-1, to=2-2]
	\arrow["{\Sigma(s)}", from=2-3, to=3-3]
	\arrow["{in_2}", from=3-3, to=3-4]
	\arrow["{\eta_{\Sigma^+,X}}", from=3-4, to=3-5]
	\arrow["a", from=3-3, to=4-3]
	\arrow["{\Sigma^+(s)}", from=2-4, to=3-4]
\end{tikzcd}
    } \qedhere
  \end{center}
\end{proof}

However, $s$ and $v$ are not yet required to satisfy the remaining monoid
equations. In order to enforce them, relying on
Corollary~\ref{cor:benign:systems}, we introduce two structural
equational systems, one for each equation.
First, we introduce some notation.
\begin{definition}
 For any object $Z$, we define $j_Z∶ Z⊗Z→ TZ$ to be the composite
    $$Z⊗Z \xto{Z ⊗ η^S_Z} Z ⊗ SZ = Γ_{SZ}Z \xto{η_{Γ_SΔ,Z}} TZ.$$  
  \end{definition}

\begin{definition}
  For our first structural equational system $E₁$, we take
  \begin{itemize}
  \item  as auxiliary bifunctor 
    $Θ₁(X,Y) = (X⊗Y)⊗Y$; 
  \item as incremental structural law, say $d₁$, we take the left-hand 
    composite of Figure~\ref{fig:ssls:fpt};
    \begin{figure}[htp]
      \begin{center}
\begin{tikzcd}[ampersand replacement=\&]
	{\Sigma^+X \otimes Y \otimes Y} \&\& {\Sigma^+X \otimes I} \\
	{I\otimes Y \otimes Y + \Sigma X \otimes Y \otimes Y} \&\& {I\otimes I + \Sigma X \otimes I} \\
	{Y \otimes Y + \Sigma X \otimes SY \otimes SY} \&\& {I + \Sigma(X\otimes I)} \\
	{TY + \Sigma (X \otimes SY) \otimes SY} \&\& {\Sigma^+(X\otimes I)} \\
	{TY + \Sigma(X \otimes SY \otimes SY)} \&\& {S(X\otimes I)} \\
	{TY + \Sigma^+(X \otimes SY \otimes SY)} \&\& {ST(X\otimes I)} \\
	{STY + ST(X \otimes STY \otimes STY)} \&\& {ST(X \otimes I + X + Y)} \\
	{ST(X \otimes STY \otimes STY + X + Y)}
	\arrow[Rightarrow, no head, from=1-1, to=2-1]
	\arrow["{\lambda_Y \otimes Y + X \otimes \eta^S_Y \otimes \eta^S_Y}", from=2-1, to=3-1]
	\arrow["{j_Y + st_{X,(SY,v_Y)} \otimes SY}", from=3-1, to=4-1]
	\arrow["{TY + st_{X\otimes SY,(SY,v_Y)}}", from=4-1, to=5-1]
	\arrow["{[STin_3,STin_1]}", from=7-1, to=8-1]
	\arrow[Rightarrow, no head, from=1-3, to=2-3]
	\arrow["{\lambda_I + st_{X,(I,id_I)}}", from=2-3, to=3-3]
	\arrow[Rightarrow, no head, from=3-3, to=4-3]
	\arrow["{\eta_{\Sigma^+}}", from=4-3, to=5-3]
	\arrow["{S\eta^T_{X \otimes I}}", from=5-3, to=6-3]
	\arrow["{ST in_1}", from=6-3, to=7-3]
	\arrow["{\eta^S_{TY} + \eta_{\Sigma^+}\eta^T (X\otimes S\eta^T_Y \otimes S\eta^T_Y)}", from=6-1, to=7-1]
	\arrow["{TY + in_2}", from=5-1, to=6-1]
\end{tikzcd}
    \end{center}
      \caption{Two incremental structural laws}
      \label{fig:ssls:fpt}
    \end{figure}
  \item as first structural interpretation, say $L₁$, at any $X$, we take 
    $$X⊗X⊗X \xto{j_X ⊗ ηᵀ_X} TX ⊗ TX \xto{j_{TX}} TTX \xto{η^Sμᵀ_X} STX;$$
  \item as second structural interpretation, say $R₁$, at any $X$, we take
    $$X⊗X⊗X \xto{α_{X,X,X}} X ⊗ (X ⊗ X) \xto{ηᵀ_X ⊗ j_X} TX ⊗ TX \xto{j_{TX}} TTX \xto{η^Sμᵀ_X} STX.$$
  \end{itemize}
\end{definition}

\begin{definition}
  For our second structural equational system $E₂$, we take
  \begin{itemize}
  \item  as auxiliary bifunctor 
    $Θ₂(X,Y) = X⊗I$;
  \item as simple structural law, say $d₂$, we take the right-hand 
    composite of Figure~\ref{fig:ssls:fpt};
  \item as first structural interpretation, say $L₂$, at any $X$, we
    take
    $$X ⊗ I \xto{η^S_X ⊗ v_X} SX ⊗ SX \xto{j_{SX}} TSX \xto{δ_X} STX;$$
  \item as second structural interpretation, say $R₂$, at any $X$, we
    simply take
    $$X ⊗ I \xto{ρ_X} X \xto{η^{ST}_X} STX.$$
  \end{itemize}
\end{definition}

\begin{theorem}
  For any locally finitely presentable, leftist monoidal category
  $(𝐂,⊗,I,α,λ,ρ)$ and pointed strong endofunctor $Σ$, the above data
  $E₁$ and $E₂$ indeed form structural equational systems, and
  $(E₁+E₂)$-algebras form a category isomorphic to $Σ\Mon$ over $𝐂$.

  Furthermore, there are unique morphisms
  \begin{mathpar}
    s∶ S0 ⊗ S0 → S0 \and v∶ I → S0
  \end{mathpar}
  rendering the pentagon~\eqref{eq:pentagon:Fmonoids} and
  triangle~\eqref{eq:left:unital} commutative.

  Finally, these morphisms, together with $μ^S₀$, satisfy the
  remaining $Σ$-monoid axioms, and make the initial $Σ⁺$-algebra $S0$
  into an initial $Σ$-monoid.
\end{theorem}
\begin{proof}
  We start by reducing to proving the first claim, by observing that
  \begin{itemize}
  \item the second claim has already been stated, right before
    Proposition~\ref{prop:pentas}, and,
  \item assuming the first claim, the last one follows directly by
    Corollary~\ref{cor:benign:systems}.
  \end{itemize}
  The first claim is proved in~§\ref{app:ptstrengths}.
\end{proof}

\section{Conclusion and perspectives}\label{s:conclu}
We have introduced admissible monad morphisms as a foundation for
syntax with auxiliary functions. We have then shown how to generate
admissible morphisms from monad distributive laws, and defined simple
structural laws and incremental structural laws as basic notions of
signatures to generate such monad distributive laws, hence admissible
morphisms.  We have also defined structural equational systems as a
basic format for ensuring that the generated auxiliary functions
satisfy some (hopefully useful) properties.  We have used these tools
to cover significant examples of auxiliary functions, from addition
and multiplication of natural numbers to binding evaluation contexts,
capture-avoiding substitution, and partial differentiation.  We have
finally shown that the standard framework of~\citet{fiore:presheaf} is
subsumed by ours.

An important question, already raised in the introduction, remains
open: can we devise some further tools to ensure that the auxiliary
functions generated by our signatures are compatible with severe
equations, i.e., equations on the basic syntax?

Finally, our framework is designed to account for auxiliary functions
defined on top of an existing syntax. It would be useful to extend it
to settings where the syntax and functions are mutually dependent, as
in induction-recursion~\cite{DBLP:conf/dagstuhl/DybjerS01}.

\bibliography{bib}


\ifx\french\undefined \def\biling#1#2{#1} \else \def\biling#1#2{#2}
  \fi\ifx\french\undefined \ifx\spanish\undefined \def\triling#1#2#3{#1} \else
  \def\triling#1#2#3{#3} \fi\else \def\triling#1#2#3{#2}
  \fi\ifx\abbrevbib\undefined \def\abbrev#1#2{#1} \else \def\abbrev#1#2{#2} \fi
\begin{thebibliography}{31}


\ifx \showCODEN    \undefined \def \showCODEN     #1{\unskip}     \fi
\ifx \showDOI      \undefined \def \showDOI       #1{#1}\fi
\ifx \showISBNx    \undefined \def \showISBNx     #1{\unskip}     \fi
\ifx \showISBNxiii \undefined \def \showISBNxiii  #1{\unskip}     \fi
\ifx \showISSN     \undefined \def \showISSN      #1{\unskip}     \fi
\ifx \showLCCN     \undefined \def \showLCCN      #1{\unskip}     \fi
\ifx \shownote     \undefined \def \shownote      #1{#1}          \fi
\ifx \showarticletitle \undefined \def \showarticletitle #1{#1}   \fi
\ifx \showURL      \undefined \def \showURL       {\relax}        \fi
\providecommand\bibfield[2]{#2}
\providecommand\bibinfo[2]{#2}
\providecommand\natexlab[1]{#1}
\providecommand\showeprint[2][]{arXiv:#2}

\bibitem[\protect\citeauthoryear{Accattoli}{Accattoli}{2019}]%
        {DBLP:conf/rta/Accattoli19}
\bibfield{author}{\bibinfo{person}{Beniamino Accattoli}.}
  \bibinfo{year}{2019}\natexlab{}.
\newblock \showarticletitle{A Fresh Look at the lambda-Calculus (Invited
  Talk)}. In \bibinfo{booktitle}{\emph{Proc.\ 4th \abbrev{International
  Conference on Formal Structures for Computation and Deduction}{Int. Conf. on
  Formal Structures for Comp. and Deduction}}} \emph{(\bibinfo{series}{Leibniz
  International Proceedings in Informatics (LIPIcs)},
  Vol.~\bibinfo{volume}{131})}, \bibfield{editor}{\bibinfo{person}{Herman
  Geuvers}} (Ed.). \bibinfo{pages}{1:1--1:20}.
\newblock
\showISBNx{978-3-95977-107-8}
\urldef\tempurl%
\url{https://doi.org/10.4230/LIPIcs.FSCD.2019.1}
\showDOI{\tempurl}


\bibitem[\protect\citeauthoryear{Ad{\'a}mek and Rosicky}{Ad{\'a}mek and
  Rosicky}{1994}]%
        {Adamek}
\bibfield{author}{\bibinfo{person}{J. Ad{\'a}mek} {and} \bibinfo{person}{J.
  Rosicky}.} \bibinfo{year}{1994}\natexlab{}.
\newblock \bibinfo{booktitle}{\emph{Locally Presentable and Accessible
  Categories}}.
\newblock \bibinfo{publisher}{Cambridge University Press}.
\newblock
\urldef\tempurl%
\url{https://doi.org/10.1017/CBO9780511600579}
\showDOI{\tempurl}


\bibitem[\protect\citeauthoryear{Allais, Atkey, Chapman, McBride, and
  McKinna}{Allais et~al\mbox{.}}{2018}]%
        {DBLP:journals/pacmpl/AllaisA0MM18}
\bibfield{author}{\bibinfo{person}{Guillaume Allais}, \bibinfo{person}{Robert
  Atkey}, \bibinfo{person}{James Chapman}, \bibinfo{person}{Conor McBride},
  {and} \bibinfo{person}{James McKinna}.} \bibinfo{year}{2018}\natexlab{}.
\newblock \showarticletitle{A type and scope safe universe of syntaxes with
  binding: their semantics and proofs}.
\newblock \bibinfo{journal}{\emph{\abbrev{Proceedings of the ACM on Programming
  Languages}{Proc. ACM Program. Lang.}}} \bibinfo{volume}{2},
  \bibinfo{number}{{ICFP}} (\bibinfo{year}{2018}),
  \bibinfo{pages}{90:1--90:30}.
\newblock
\urldef\tempurl%
\url{https://doi.org/10.1145/3236785}
\showDOI{\tempurl}


\bibitem[\protect\citeauthoryear{Barr}{Barr}{1970}]%
        {Barr1970CoequalizersAF}
\bibfield{author}{\bibinfo{person}{Michael Barr}.}
  \bibinfo{year}{1970}\natexlab{}.
\newblock \showarticletitle{Coequalizers and free triples}.
\newblock \bibinfo{journal}{\emph{Mathematische Zeitschrift}}
  \bibinfo{volume}{116} (\bibinfo{year}{1970}), \bibinfo{pages}{307--322}.
\newblock


\bibitem[\protect\citeauthoryear{Beck}{Beck}{1969}]%
        {BeckDistlaws}
\bibfield{author}{\bibinfo{person}{Jon~M. Beck}.}
  \bibinfo{year}{1969}\natexlab{}.
\newblock \showarticletitle{Distributive laws}. In
  \bibinfo{booktitle}{\emph{Seminar on Triples and Categorical Homology
  Theory}} \emph{(\bibinfo{series}{\abbrev{Lecture Notes in Mathematics}{LNM}},
  Vol.~\bibinfo{volume}{80})}, \bibfield{editor}{\bibinfo{person}{Beno Eckmann}
  {and} \bibinfo{person}{Myles Tierney}} (Eds.). \bibinfo{publisher}{Springer}.
\newblock


\bibitem[\protect\citeauthoryear{Berghofer and Urban}{Berghofer and
  Urban}{2007}]%
        {DBLP:journals/entcs/BerghoferU07}
\bibfield{author}{\bibinfo{person}{Stefan Berghofer} {and}
  \bibinfo{person}{Christian Urban}.} \bibinfo{year}{2007}\natexlab{}.
\newblock \showarticletitle{A Head-to-Head Comparison of de Bruijn Indices and
  Names}.
\newblock \bibinfo{journal}{\emph{\abbrev{Electronic Notes in Theoretical
  Computer Science}{ENTCS}}} \bibinfo{volume}{174}, \bibinfo{number}{5}
  (\bibinfo{year}{2007}), \bibinfo{pages}{53--67}.
\newblock
\urldef\tempurl%
\url{https://doi.org/10.1016/j.entcs.2007.01.018}
\showDOI{\tempurl}


\bibitem[\protect\citeauthoryear{Bird and Paterson}{Bird and Paterson}{1999}]%
        {bird_paterson_1999}
\bibfield{author}{\bibinfo{person}{Richard~S. Bird} {and} \bibinfo{person}{Ross
  Paterson}.} \bibinfo{year}{1999}\natexlab{}.
\newblock \showarticletitle{{D}e Bruijn notation as a nested datatype}.
\newblock \bibinfo{journal}{\emph{\abbrev{Journal of Functional
  Programming}{J.~Funct. Program.}}} \bibinfo{volume}{9}, \bibinfo{number}{1}
  (\bibinfo{year}{1999}), \bibinfo{pages}{77–91}.
\newblock
\urldef\tempurl%
\url{https://doi.org/10.1017/S0956796899003366}
\showDOI{\tempurl}


\bibitem[\protect\citeauthoryear{Blackwell}{Blackwell}{1976}]%
        {blackwell_phd}
\bibfield{author}{\bibinfo{person}{Robert~Oswald Blackwell}.}
  \bibinfo{year}{1976}\natexlab{}.
\newblock \emph{\bibinfo{title}{Some existence theorems in the theory of
  doctrines}}.
\newblock \bibinfo{thesistype}{Ph.D. Dissertation}. \bibinfo{school}{University
  of New South Wales}.
\newblock


\bibitem[\protect\citeauthoryear{Coraglia and {Di Liberti}}{Coraglia and {Di
  Liberti}}{2021}]%
        {CoragliaDiLiberti}
\bibfield{author}{\bibinfo{person}{Greta Coraglia} {and} \bibinfo{person}{Ivan
  {Di Liberti}}.} \bibinfo{year}{2021}\natexlab{}.
\newblock \bibinfo{title}{Context, judgement, deduction}.
  (\bibinfo{year}{2021}).
\newblock
\showeprint[arxiv]{2111.09438}~[math.LO]
\newblock
\shownote{Preprint.}


\bibitem[\protect\citeauthoryear{De~Bruijn}{De~Bruijn}{1972}]%
        {DB}
\bibfield{author}{\bibinfo{person}{N.~G. De~Bruijn}.}
  \bibinfo{year}{1972}\natexlab{}.
\newblock \showarticletitle{Lambda-Calculus Notation with Nameless Dummies, a
  Tool for Automatic Formula Manipulation, with Application to the
  Church-Rosser Theorem}.
\newblock \bibinfo{journal}{\emph{\abbrev{Indagationes
  Mathematicae}{Indagationes Math.}}}  \bibinfo{volume}{34}
  (\bibinfo{year}{1972}), \bibinfo{pages}{381--392}.
\newblock


\bibitem[\protect\citeauthoryear{Dybjer and Setzer}{Dybjer and Setzer}{2001}]%
        {DBLP:conf/dagstuhl/DybjerS01}
\bibfield{author}{\bibinfo{person}{Peter Dybjer} {and} \bibinfo{person}{Anton
  Setzer}.} \bibinfo{year}{2001}\natexlab{}.
\newblock \showarticletitle{Indexed Induction-Recursion}. In
  \bibinfo{booktitle}{\emph{Proof Theory in Computer Science, International
  Seminar}} \emph{(\bibinfo{series}{LNCS}, Vol.~\bibinfo{volume}{2183})},
  \bibfield{editor}{\bibinfo{person}{Reinhard Kahle}, \bibinfo{person}{Peter
  Schroeder{-}Heister}, {and} \bibinfo{person}{Robert~F. St{\"{a}}rk}} (Eds.).
  \bibinfo{publisher}{Springer}, \bibinfo{pages}{93--113}.
\newblock
\showISBNx{3-540-42752-X}
\urldef\tempurl%
\url{https://doi.org/10.1007/3-540-45504-3\_7}
\showDOI{\tempurl}


\bibitem[\protect\citeauthoryear{Ehrhard and Regnier}{Ehrhard and
  Regnier}{2003}]%
        {LambdaDiff}
\bibfield{author}{\bibinfo{person}{Thomas Ehrhard} {and}
  \bibinfo{person}{Laurent Regnier}.} \bibinfo{year}{2003}\natexlab{}.
\newblock \showarticletitle{The differential lambda-calculus}.
\newblock \bibinfo{journal}{\emph{\abbrev{Theoretical Computer Science}{Theor.
  Comp. Sci.}}} \bibinfo{volume}{309}, \bibinfo{number}{1--3}
  (\bibinfo{year}{2003}), \bibinfo{pages}{1--41}.
\newblock
\urldef\tempurl%
\url{https://doi.org/10.1016/S0304-3975(03)00392-X}
\showDOI{\tempurl}


\bibitem[\protect\citeauthoryear{Fiore and Hur}{Fiore and Hur}{2009}]%
        {FioreHurEquational}
\bibfield{author}{\bibinfo{person}{Marcelo Fiore} {and}
  \bibinfo{person}{Chung-Kil Hur}.} \bibinfo{year}{2009}\natexlab{}.
\newblock \showarticletitle{On the construction of free algebras for equational
  systems}.
\newblock \bibinfo{journal}{\emph{\abbrev{Theoretical Computer Science}{Theor.
  Comp. Sci.}}}  \bibinfo{volume}{410} (\bibinfo{year}{2009}),
  \bibinfo{pages}{1704--1729}.
\newblock


\bibitem[\protect\citeauthoryear{Fiore, Plotkin, and Turi}{Fiore
  et~al\mbox{.}}{1999}]%
        {fiore:presheaf}
\bibfield{author}{\bibinfo{person}{Marcelo Fiore}, \bibinfo{person}{Gordon
  Plotkin}, {and} \bibinfo{person}{Daniele Turi}.}
  \bibinfo{year}{1999}\natexlab{}.
\newblock \showarticletitle{Abstract Syntax and Variable Binding}. In
  \bibinfo{booktitle}{\emph{Proc.\ 14th \abbrev{Symposium on Logic in Computer
  Science}{Logic in Comp. Sci.}}} \bibinfo{publisher}{\abbrev{IEEE}{IEEE}}.
\newblock


\bibitem[\protect\citeauthoryear{Fiore}{Fiore}{2008}]%
        {DBLP:conf/lics/Fiore08}
\bibfield{author}{\bibinfo{person}{Marcelo~P. Fiore}.}
  \bibinfo{year}{2008}\natexlab{}.
\newblock \showarticletitle{Second-Order and Dependently-Sorted Abstract
  Syntax}. In \bibinfo{booktitle}{\emph{LICS}}. \abbrev{IEEE}{IEEE},
  \bibinfo{pages}{57--68}.
\newblock
\showISBNx{978-0-7695-3183-0}
\urldef\tempurl%
\url{https://doi.org/10.1109/LICS.2008.38}
\showDOI{\tempurl}


\bibitem[\protect\citeauthoryear{Gabbay and Pitts}{Gabbay and Pitts}{1999}]%
        {PittsAM:newaas}
\bibfield{author}{\bibinfo{person}{Murdoch~J. Gabbay} {and}
  \bibinfo{person}{Andrew~M. Pitts}.} \bibinfo{year}{1999}\natexlab{}.
\newblock \showarticletitle{A New Approach to Abstract Syntax Involving
  Binders}. In \bibinfo{booktitle}{\emph{Proc.\ 14th \abbrev{Symposium on Logic
  in Computer Science}{Logic in Comp. Sci.}}}
  \bibinfo{publisher}{\abbrev{IEEE}{IEEE}}.
\newblock


\bibitem[\protect\citeauthoryear{Gratzer and Sterling}{Gratzer and
  Sterling}{2021}]%
        {GratzerSterling}
\bibfield{author}{\bibinfo{person}{Daniel Gratzer} {and}
  \bibinfo{person}{Jonathan Sterling}.} \bibinfo{year}{2021}\natexlab{}.
\newblock \bibinfo{title}{Syntactic categories for dependent type theories:
  sketching and adequacy}.  (\bibinfo{year}{2021}).
\newblock
\newblock
\shownote{Preprint.}


\bibitem[\protect\citeauthoryear{Hirschowitz, Hirschowitz, and
  Lafont}{Hirschowitz et~al\mbox{.}}{2020}]%
        {HHL}
\bibfield{author}{\bibinfo{person}{Andr{\'{e}} Hirschowitz},
  \bibinfo{person}{Tom Hirschowitz}, {and} \bibinfo{person}{Ambroise Lafont}.}
  \bibinfo{year}{2020}\natexlab{}.
\newblock \showarticletitle{Modules over Monads and Operational Semantics}. In
  \bibinfo{booktitle}{\emph{Proc.\ 5th \abbrev{International Conference on
  Formal Structures for Computation and Deduction}{Int. Conf. on Formal
  Structures for Comp. and Deduction}}} \emph{(\bibinfo{series}{Leibniz
  International Proceedings in Informatics (LIPIcs)},
  Vol.~\bibinfo{volume}{167})}, \bibfield{editor}{\bibinfo{person}{Zena~M.
  Ariola}} (Ed.). \bibinfo{publisher}{Schloss Dagstuhl - Leibniz-Zentrum
  f{\"{u}}r Informatik}, \bibinfo{pages}{12:1--12:23}.
\newblock
\showISBNx{978-3-95977-155-9}
\urldef\tempurl%
\url{https://doi.org/10.4230/LIPIcs.FSCD.2020.12}
\showDOI{\tempurl}


\bibitem[\protect\citeauthoryear{Hirschowitz, Hirschowitz, and
  Lafont}{Hirschowitz et~al\mbox{.}}{2022a}]%
        {HHLlong}
\bibfield{author}{\bibinfo{person}{Andr{\'{e}} Hirschowitz},
  \bibinfo{person}{Tom Hirschowitz}, {and} \bibinfo{person}{Ambroise Lafont}.}
  \bibinfo{year}{2022}\natexlab{a}.
\newblock \bibinfo{title}{Modules over Monads and Operational Semantics}.
  (\bibinfo{year}{2022}).
\newblock
\newblock
\shownote{Submitted expanded version of~\cite{HHL}.}


\bibitem[\protect\citeauthoryear{Hirschowitz, Hirschowitz, Lafont, and
  Maggesi}{Hirschowitz et~al\mbox{.}}{2022b}]%
        {HHLM}
\bibfield{author}{\bibinfo{person}{Andr{\'e} Hirschowitz}, \bibinfo{person}{Tom
  Hirschowitz}, \bibinfo{person}{Ambroise Lafont}, {and} \bibinfo{person}{Marco
  Maggesi}.} \bibinfo{year}{2022}\natexlab{b}.
\newblock \showarticletitle{Variable binding and substitution for (nameless)
  dummies}. In \bibinfo{booktitle}{\emph{Proc.\ 25th \abbrev{Foundations of
  Software Science and Computational Structures}{Foundations Soft. Science
  Comp. Structures}}} \emph{(\bibinfo{series}{LNCS})},
  \bibfield{editor}{\bibinfo{person}{Lutz Schröder} {and}
  \bibinfo{person}{Patricia Bouyer}} (Eds.). \bibinfo{publisher}{Springer}.
\newblock


\bibitem[\protect\citeauthoryear{Hirschowitz, Leroy, and Wells}{Hirschowitz
  et~al\mbox{.}}{2009}]%
        {DBLP:journals/lisp/HirschowitzLW09}
\bibfield{author}{\bibinfo{person}{Tom Hirschowitz}, \bibinfo{person}{Xavier
  Leroy}, {and} \bibinfo{person}{J.~B. Wells}.}
  \bibinfo{year}{2009}\natexlab{}.
\newblock \showarticletitle{Compilation of extended recursion in call-by-value
  functional languages}.
\newblock \bibinfo{journal}{\emph{\abbrev{Higher Order and Symbolic
  Computation}{High. Order Symb. Comput.}}} \bibinfo{volume}{22},
  \bibinfo{number}{1} (\bibinfo{year}{2009}), \bibinfo{pages}{3--66}.
\newblock
\urldef\tempurl%
\url{https://doi.org/10.1007/s10990-009-9042-z}
\showDOI{\tempurl}


\bibitem[\protect\citeauthoryear{Hofmann}{Hofmann}{1999}]%
        {hofmann:presheaf}
\bibfield{author}{\bibinfo{person}{Martin Hofmann}.}
  \bibinfo{year}{1999}\natexlab{}.
\newblock \showarticletitle{Semantical Analysis of Higher-Order Abstract
  Syntax}. In \bibinfo{booktitle}{\emph{Proc.\ 14th \abbrev{Symposium on Logic
  in Computer Science}{Logic in Comp. Sci.}}}
  \bibinfo{publisher}{\abbrev{IEEE}{IEEE}}.
\newblock


\bibitem[\protect\citeauthoryear{Kelly}{Kelly}{1980}]%
        {KellyUnified}
\bibfield{author}{\bibinfo{person}{G.~M. Kelly}.}
  \bibinfo{year}{1980}\natexlab{}.
\newblock \showarticletitle{A unified treatment of transfinite constructions
  for free algebras, free monoids, colimits, associated sheaves, and so on}.
\newblock \bibinfo{journal}{\emph{Bulletin of the Australian Mathematical
  Society}}  \bibinfo{volume}{22} (\bibinfo{year}{1980}),
  \bibinfo{pages}{1--83}.
\newblock


\bibitem[\protect\citeauthoryear{Kelly}{Kelly}{1982}]%
        {Kelly1982Basic}
\bibfield{author}{\bibinfo{person}{G.~M. Kelly}.}
  \bibinfo{year}{1982}\natexlab{}.
\newblock \bibinfo{booktitle}{\emph{Basic concepts of enriched category
  theory}}. \bibinfo{series}{London Mathematical Society Lecture Note Series},
  Vol.~\bibinfo{volume}{64}.
\newblock \bibinfo{publisher}{Cambridge University Press}.
\newblock
\newblock
\shownote{Republished as: \textit{Reprints in Theory and Applications of
  Categories 10} (2005).}


\bibitem[\protect\citeauthoryear{Lack}{Lack}{1997}]%
        {MonadicityMonads}
\bibfield{author}{\bibinfo{person}{Stephen Lack}.}
  \bibinfo{year}{1997}\natexlab{}.
\newblock \showarticletitle{On the monadicity of finitary monads}.
\newblock \bibinfo{journal}{\emph{\abbrev{Journal of Pure and Applied
  Algebra}{J. of Pure and Appl. Algebra}}}  \bibinfo{volume}{140}
  (\bibinfo{year}{1997}), \bibinfo{pages}{65--73}.
\newblock


\bibitem[\protect\citeauthoryear{{Mac Lane}}{{Mac Lane}}{1998}]%
        {MacLane:cwm}
\bibfield{author}{\bibinfo{person}{Saunders {Mac Lane}}.}
  \bibinfo{year}{1998}\natexlab{}.
\newblock \bibinfo{booktitle}{\emph{Categories for the Working Mathematician}
  (\bibinfo{edition}{2nd} ed.)}.
\newblock Number~5 in \bibinfo{series}{Graduate Texts in Mathematics}.
  \bibinfo{publisher}{Springer}.
\newblock


\bibitem[\protect\citeauthoryear{Parigot}{Parigot}{1992}]%
        {Parigot}
\bibfield{author}{\bibinfo{person}{Michel Parigot}.}
  \bibinfo{year}{1992}\natexlab{}.
\newblock \showarticletitle{$\lambda\mu$-calculus: an algorithmic
  interpretation of classical natural deduction}. In
  \bibinfo{booktitle}{\emph{Proc.\ 19th \abbrev{International Colloquium on
  Automata, Languages and Programming}{Int. Col. on Automata, Lang. and
  Progr.}}} \emph{(\bibinfo{series}{LNCS}, Vol.~\bibinfo{volume}{624})}.
  \bibinfo{publisher}{Springer}.
\newblock


\bibitem[\protect\citeauthoryear{Pierce}{Pierce}{2004}]%
        {ATAPL}
\bibfield{editor}{\bibinfo{person}{Benjamin~C. Pierce}} (Ed.).
  \bibinfo{year}{2004}\natexlab{}.
\newblock \bibinfo{booktitle}{\emph{Advanced Topics in Types and Programming
  Languages}}.
\newblock \bibinfo{publisher}{MIT Press}.
\newblock


\bibitem[\protect\citeauthoryear{Reiterman}{Reiterman}{1977}]%
        {Reiterman}
\bibfield{author}{\bibinfo{person}{Jan Reiterman}.}
  \bibinfo{year}{1977}\natexlab{}.
\newblock \showarticletitle{A left adjoint construction related to free
  triples}.
\newblock \bibinfo{journal}{\emph{\abbrev{Journal of Pure and Applied
  Algebra}{J. of Pure and Appl. Algebra}}}  \bibinfo{volume}{10}
  (\bibinfo{year}{1977}), \bibinfo{pages}{57--71}.
\newblock


\bibitem[\protect\citeauthoryear{Sewell, Stoyle, Hicks, Bierman, and
  Wansbrough}{Sewell et~al\mbox{.}}{2008}]%
        {DBLP:journals/jfp/SewellSHBW08}
\bibfield{author}{\bibinfo{person}{Peter Sewell}, \bibinfo{person}{Gareth~Paul
  Stoyle}, \bibinfo{person}{Michael Hicks}, \bibinfo{person}{Gavin~M. Bierman},
  {and} \bibinfo{person}{Keith Wansbrough}.} \bibinfo{year}{2008}\natexlab{}.
\newblock \showarticletitle{Dynamic rebinding for marshalling and update, via
  redex-time and destruct-time reduction}.
\newblock \bibinfo{journal}{\emph{\abbrev{Journal of Functional
  Programming}{J.~Funct. Program.}}} \bibinfo{volume}{18}, \bibinfo{number}{4}
  (\bibinfo{year}{2008}), \bibinfo{pages}{437--502}.
\newblock
\urldef\tempurl%
\url{https://doi.org/10.1017/S0956796807006600}
\showDOI{\tempurl}


\bibitem[\protect\citeauthoryear{Vaux}{Vaux}{2007}]%
        {Iouri}
\bibfield{author}{\bibinfo{person}{Lionel Vaux}.}
  \bibinfo{year}{2007}\natexlab{}.
\newblock \emph{\bibinfo{title}{{$\lambda$}-calcul diff{\'e}rentiel et logique
  classique~: interactions calculatoires}}.
\newblock \bibinfo{thesistype}{Ph.D. Dissertation}.
  \bibinfo{school}{Universit{\'e} Aix-Marseille 2}.
\newblock


\end{thebibliography}

\appendix

\section{Proof of Proposition~\ref{prop:admissible:street}}\label{s:admissible:street}
First, $α_X:SX → TX$ is an $S$-algebra morphism between $SX$ and
$𝐬𝐞𝐦(α)(TX)$. Since $S∅$ is an initial $S$-algebra, $𝐬𝐞𝐦(α)(T∅)$ is
initial if and only if $α_{∅}$ is an isomorphism.  This argument shows
$\ref{item:admissible} ⇔ \ref{item:preserves}$. Morever, since $S∅$ is
an initial $S$-algebra, creation implies preservation, so
$\ref{item:creates} ⇒ \ref{item:preserves}$.

Finally, let us show that preservation implies creation,
$\ref{item:preserves} ⇒ \ref{item:creates}$.  Let
$SX \xrightarrow{x} X$ be an initial algebra. We show that there
is a unique $T$-algebra structure on $X$ that is mapped to $x$ by
$𝐬𝐞𝐦(α)$, and moreover that it is initial as a $T$-algebra. Since, by
hypothesis, $𝐬𝐞𝐦(α)$ preserves the initial object, $𝐬𝐞𝐦(α)(T∅)$ is
also initial, and thus $X$ is isomorphic to $T∅$, as
$S$-algebras. Now, $𝐬𝐞𝐦(α)$ is an isofibration: through this
isomorphism, $X$ inherits a $T$-algebra structure which is mapped to
$SX → X$ by $𝐬𝐞𝐦(α)$, and this isomorphism lifts to the category of
$T$-algebras. Since this $T$-algebra $TX → X$ is isomorphic to $T∅$,
it is initial. It remains to show uniqueness.  Consider an alternative
$T$-algebra structure on $X$, that is also mapped to 
$SX → X$ by $𝐬𝐞𝐦(α)$.  By initiality of $X$ as a $T$-algebra, there is a morphism
$X → X$ between these two $T$-algebras. It is enough to show that its
image by $𝐬𝐞𝐦(α)$ is the identity morphism, which follows from
initiality of $X$ as a $S$-algebra.

\section{Colimits of monads and their algebras}\label{s:monad:colims}
In this section, we state a few useful results on colimits of monads,
and on limits of the corresponding monadic functors, which notably
entail Lemma~\ref{lem:monad:coprod:alg}.

\begin{definition}
  For any functor $F∶ 𝐀 → 𝐌𝐧𝐝(𝐂)$, let $F\Alg$ denote a (choice of)
  limit of the induced diagram
$$\op{𝐀} → \op{𝐌𝐧𝐝(𝐂)} ≃ 𝐌𝐨𝐧𝐚𝐝𝐢𝐜/𝐂 ↪ 𝐂𝐀𝐓/𝐂\rlap{.}$$
\end{definition}
Thus, an $F$-algebra is an object $X ∈ 𝐂$, equipped with compatible
$F_A$-algebra structures, for all $A ∈ 𝐀$.

The next proposition, due to Kelly, relates this with the colimit of
$F$, when the latter exists.

For any functor $F∶ 𝐀 → 𝐌𝐧𝐝(𝐂)$, if $T ∈ 𝐌𝐧𝐝(𝐂)$ is a colimit of $F$,
then the colimiting cocone induces a cone over the diagram 
$$\op{𝐀} → \op{𝐌𝐧𝐝(𝐂)} ≃ 𝐌𝐨𝐧𝐚𝐝𝐢𝐜/𝐂 ↪ 𝐂𝐀𝐓/𝐂\rlap{,}$$
hence a functor $m∶ T\Alg → F\Alg$ over $𝐂$.

\begin{proposition}[{\cite[Proposition~26.3]{KellyUnified}}]
  For any functor $F∶ 𝐀 → 𝐌𝐧𝐝(𝐂)$ with colimit $T$, if $𝐂$ is complete
  (and locally small), then the functor $m∶ T\Alg → F\Alg$ is an
  isomorphism.
\end{proposition}
\begin{remark}
  Kelly explicitly requires $𝐂$ to be locally small, which is implicit
  here by the convention fixed in~§\ref{s:intro}.
\end{remark}

This immediately entails Lemma~\ref{lem:monad:coprod:alg}.

\begin{corollary}\label{cor:alg:coprod:free}
  Consider two monads $S$ and $T$ on a locally finitely presentable
  category $𝐂$, and a monad distributive law $δ∶ TS → ST$.  For any
  bifunctor $Γ∶ 𝐂 × 𝐂 → 𝐂$, letting $T' ≔ Γ_S^\fleur$, the
  following square is a pullback.
  \begin{center}
    \Diag{%
      \stdpbk %
    }{%
      (T⊕T')\Alg \& T\Alg \\
      Γ_S\alg \& 𝐂
    }{%
      (m-1-1) edge[labela={}] (m-1-2) %
      edge[labell={}] (m-2-1) %
      (m-2-1) edge[labelb={}] (m-2-2) %
      (m-1-2) edge[labelr={}] (m-2-2) %
    }%
    \end{center}
  \end{corollary}
  \begin{proof}
    By Lemma~\ref{lem:monad:coprod:alg} and the isomorphism
    $T'\Alg ≅ Γ_S\alg$ over $𝐂$.
  \end{proof}
  \begin{remark}
    Otherwise said, giving a $(T⊕T')$-algebra structure on $X$ is
    equivalent to giving a $T$-algebra structure and a
    $Γ_S$-algebra structure.
  \end{remark}

\section{Proof of Theorem~\ref{thm:incremental}}
\label{app:thm-incremental}

  \begin{definition}
    A bifunctor $Γ∶ 𝐂² → 𝐂$ on a category $𝐂$ is \alert{bipointed}
    if it comes equipped with natural transformations
    $$X \xto{α^Γ_{X,Y}} Γ(X,Y) \xot{β^Γ_{X,Y}} Y.$$
  \end{definition}

  \begin{definition}
    For any bifunctor $Γ$, let $Γ^•(X,Y) ≔ Γ(X,Y) + X + Y$.
  \end{definition}
  \begin{proposition}\label{prop:bipointedalg}
    For any bifunctor $Γ$, the bifunctor $Γ^•$ is bipointed.  (It is
    in fact the free bipointed bifunctor over $Γ$.)
  \end{proposition}
  \begin{proof}
    Straightforward.
  \end{proof}

  \begin{notation}\label{not:bullet}
    When applied to, say, $Γ_F$ for some endofunctor $F$, the notation
    $Γ^•_F$ is ambiguous, as it could denote $(Γ^•)_F$ or
    $(Γ_F)^•$. The issue is even worse for, e.g., $Γ^•_{GF}$.  In
    order to resolve the ambiguity, we write $Γ^•_{G|F}$ for
    $(Γ_G)^•_F$, i.e.,
    $$Γ^•_{G|FY}X ≔ Γ(X,GFY)+X+FY.$$
  \end{notation}

  We now want to show that any incremental structural law
  $$h_{X,Y}∶ Γ_Y(Σ(X)) → ST(Γ_{STY}(X)+X+Y)\rlap{,}$$
  or using Notation~\ref{not:bullet}
  $$h_{X,Y}∶ Γ_YΣX → ST Γ^•_{ST|Y}X\rlap{,}$$  
  induces an incremental lifting of $S$ to $T'\Alg$ along $δ$.  For
  this, we extend $h$ to a natural transformation
  $$h^{•ω}_{X,Y}∶ Γ_{SY}SX → ST(Γ^•_{ST|Y}(X))$$
  making the following triangle commute
  \begin{center}
    \begin{tikzcd}[ampersand replacement=\&]
      {\Gamma_{Y} \Sigma X} \&\& {\Gamma_{SY}SX} \\
      \& {ST\Gamma^\bullet_{ST|Y}X\rlap{,}}
      \arrow["{h^{\bullet\omega}_{X,Y}}"{pos=0.2}, from=1-3, to=2-2]
      \arrow["{\Gamma_{\eta^S_Y} \eta_{\Sigma,X}}", from=1-1, to=1-3]
      \arrow["{h_{X,Y}}"'{pos=0.2}, from=1-1, to=2-2]
    \end{tikzcd}
  \end{center}
  and then show how it induces an incremental lifting.  This route
  being slightly long-winded, let us point to the definition of
  $h^{•ω}$ (Definition~\ref{def:hbulletomega}) and the proof that it
  induces an icremental lifting (Proposition~\ref{prop:mualg}).

  The plan is to successively define families
  \begin{itemize}
  \item $h^•_{X,Y}∶ Γ^•_{S|Y}(Σ(X)) → ST(Γ^•_{ST|Y}(X))$,
  \item $h^∘_{X,Y}∶ Γ^\pt_{STY}ΣX → STΓ^•_{ST|Y}X$,
  \item $\widetilde{h^∘}_{X,Y}∶ ΣX → ℸ^×_{STY} ST O_Y Γ^∘_{STY} X$,
  \item $\widetilde{h^∘}^ω_{X,Y}∶ SX → ℸ^×_{STY} ST O_Y Γ^∘_{STY} X$,
  \item $\widetilde{\widetilde{h^∘}^ω}_{X,Y}∶ Γ^\pt_{STY}SX →  ST O_Y Γ^∘_{STY} X$
  \end{itemize}
  (involving new functors defined along the way), from which we then
  define $h^{•ω}$ and prove that it satisfies the properties stated in
  Lemma~\ref{lem:coh:hbulletomegamu} below. We then use this to
  construct the desired incremental lifting.

  \subsection{The family $h^•$}
  Let us start with $h^•$, which requires the following preliminary
  definition.
  \begin{definition}
    Consider any pointed endofunctor $F$. We let
    $↑^F$ denote the natural transformation defined at any $X,Y ∈ 𝐂$ by
    the composite
    $$X + FY \xto{η^F_X+FY} FX+FY \xto{[Fin₁,Fin₂]}F(X+Y).$$
  \end{definition}
  \begin{remark}
    In particular, for any $F$    
    and $G$ with $F$ pointed,
    we have  $$↑^F_{Γ_{GFY}X + X,Y}∶ Γ^•_{G|FY}X → F Γ^•_{GF|Y}X.$$
  \end{remark}
  \begin{lemma}\label{lem:delta:uparrow}
    For any bifunctor $G∶ 𝐂² → 𝐂$, distributive law $δ∶ TS → ST$, and
    objects $X,Y ∈ 𝐂$, the following diagram commutes.
    \begin{center}
\begin{tikzcd}[ampersand replacement=\&]
	{G^\bullet_{|TSY} X} \& {G^\bullet_{|STY} X} \& {S^\bullet_{S|TY} X} \\
	{TG^\bullet_{T|SY}X} \& {TSG^\bullet_{TS|Y}X} \& {STG^\bullet_{ST|Y}X}
	\arrow["{ \uparrow^T}"', from=1-1, to=2-1]
	\arrow["{ \uparrow^S}", from=1-2, to=1-3]
	\arrow["{G^\bullet_{|\delta_Y} X}", from=1-1, to=1-2]
	\arrow["{T\uparrow^S}"', from=2-1, to=2-2]
	\arrow["{\delta G^\bullet_{\delta|Y}X}"', from=2-2, to=2-3]
	\arrow["{S \uparrow^T}", from=1-3, to=2-3]
      \end{tikzcd}
    \end{center}
  \end{lemma}
  \begin{proof}
    Unfolding the definition, we refine the claim as follows,
    \begin{center}
      \ajustedroit{
\begin{tikzcd}[ampersand replacement=\&]
	{G_{TSY} X + X + TSY} \& {G_{STY} X + X + STY} \& {S(G_{STY} X + X) + STY} \& {S(G_{STY} X + X + TY)} \\
	{T(G_{TSY} X + X) + TSY} \& {T(G_{STY} X + X) + STY} \& {ST(G_{STY} X + X) + STY} \& {S(T(G_{STY} X + X) + TY)} \\
	\& {ST(G_{TSY} X + X) + STY} \\
	\\
	{T(G_{TSY} X + X + SY)} \& {T(S(G_{TSY} X + X) + SY)} \& {TS(G_{TSY}X +X +Y)} \& {ST(G_{STY} X + X + Y)}
	\arrow["{G_{\delta_Y} X + X + \delta_Y}", from=1-1, to=1-2]
	\arrow["{\delta (G_{\delta Y}X + X + Y)}"', from=5-3, to=5-4]
	\arrow["{\eta^T + TSY}"', from=1-1, to=2-1]
	\arrow["{[Tin_1,Tin_2]}"', from=2-1, to=5-1]
	\arrow["{T(\eta^S+ SY)}"', from=5-1, to=5-2]
	\arrow["{T[Sin_1,Sin_2]}"', from=5-2, to=5-3]
	\arrow["{ \eta^S + STY}", from=1-2, to=1-3]
	\arrow["{[Sin_1,Sin_2]}", from=1-3, to=1-4]
	\arrow["{S(\eta^T + TY)}", from=1-4, to=2-4]
	\arrow["{S[Tin_1,Tin_2]}", from=2-4, to=5-4]
	\arrow["{T(G_{\delta_Y} X + X) + \delta_Y}", from=2-1, to=2-2]
	\arrow["{\eta^T + STY}", from=1-2, to=2-2]
	\arrow["{\eta^S + STY}", from=2-2, to=2-3]
	\arrow["{S\eta^T+STY}", from=1-3, to=2-3]
	\arrow["{\eta^S + TSY}"'{pos=0.2}, from=2-1, to=3-2]
	\arrow["{ST(G_{\delta_Y} X + X) + \delta_Y}"'{pos=0.9}, from=3-2, to=2-3]
	\arrow["{[Sin_1,Sin_2]}"', from=2-3, to=2-4]
	\arrow["{[STin_1,STin_2]}"', from=2-3, to=5-4]
\end{tikzcd}
}
\end{center}
where all subdiagrams easily commute except the top right square and
the bottom left polygon.
Both have coproducts as their domains, so we chase them termwise.
The first term of the top right square is chased as follows.
\begin{center}
  \ajustedroit{
\begin{tikzcd}[ampersand replacement=\&]
	{S(G_{STY} X + X)} \&\& {S(G_{STY} X + X + TY)} \\
	\& {ST(G_{STY} X + X)} \&\& {S(T(G_{STY} X + X) + TY)} \\
	{S(G_{STY} X + X) + STY} \&\& {S(G_{STY} X + X + TY)} \\
	\\
	\& {ST(G_{STY} X + X) + STY} \&\& {S(T(G_{STY} X + X) + TY)}
	\arrow["{[Sin_1,Sin_2]}"'{pos=0.2}, from=3-1, to=3-3]
	\arrow["{S(\eta^T + TY)}", from=3-3, to=5-4]
	\arrow["{S\eta^T+STY}"', from=3-1, to=5-2]
	\arrow["{[Sin_1,Sin_2]}"', from=5-2, to=5-4]
	\arrow["{in_1}"', from=1-1, to=3-1]
	\arrow["{S\eta^T}"', from=1-1, to=2-2]
	\arrow["{in_1}"', from=2-2, to=5-2]
	\arrow["{S in_1}"'{pos=0.2}, from=2-2, to=2-4]
	\arrow[Rightarrow, no head, from=2-4, to=5-4]
	\arrow[Rightarrow, no head, from=1-3, to=3-3]
	\arrow["{S in_1}", from=1-1, to=1-3]
	\arrow["{S(\eta^T + TY)}"{pos=0.9}, from=1-3, to=2-4]
\end{tikzcd}}
\end{center}
The second term of the top right square is chased as follows.
\begin{center}
  \ajustedroit{
\begin{tikzcd}[ampersand replacement=\&]
	STY \&\& STY \\
	\& STY \&\& STY \\
	{S(G_{STY} X + X) + STY} \&\& {S(G_{STY} X + X + TY)} \\
	\& {ST(G_{STY} X + X) + STY} \&\& {S(T(G_{STY} X + X) + TY)}
	\arrow["{[Sin_1,Sin_2]}"'{pos=0.2}, from=3-1, to=3-3]
	\arrow["{S(\eta^T + TY)}"{pos=0.8}, from=3-3, to=4-4]
	\arrow["{S\eta^T+STY}"'{pos=0.3}, from=3-1, to=4-2]
	\arrow["{[Sin_1,Sin_2]}"', from=4-2, to=4-4]
	\arrow["{in_2}"', from=1-1, to=3-1]
	\arrow["{in_2}"'{pos=0.7}, from=2-2, to=4-2]
	\arrow["{Sin_2}", from=2-4, to=4-4]
	\arrow["{Sin_2}"'{pos=0.2}, from=1-3, to=3-3]
	\arrow[Rightarrow, no head, from=1-1, to=1-3]
	\arrow[Rightarrow, no head, from=1-1, to=2-2]
	\arrow[Rightarrow, no head, from=1-3, to=2-4]
	\arrow[Rightarrow, no head, from=2-2, to=2-4]
      \end{tikzcd}
    }
\end{center}
The first term of the bottom left polygon is chased as follows.
\begin{center}
  \ajustedroit{
\begin{tikzcd}[ampersand replacement=\&]
	{T(G_{TSY}X + X)} \&\&\& {ST(G_{TSY}X + X)} \&\&\& {ST(G_{STY} X + X)} \\
	\& {T(G_{TSY}X + X)} \&\&\& {TS(G_{TSY}X + X)} \& {ST(G_{TSY}X + X)} \&\& {ST(G_{STY}X + X)} \\
	{T(G_{TSY} X + X) + TSY} \&\&\& {ST(G_{TSY} X + X) + STY} \&\&\& {ST(G_{STY} X + X) + STY} \\
	\& {T(G_{TSY} X + X + SY)} \&\& {T(S(G_{TSY} X + X) + SY)} \&\& {TS(G_{TSY}X +X +Y)} \&\& {ST(G_{STY} X + X + Y)}
	\arrow["{\delta (G_{\delta Y}X + X + Y)}"', from=4-6, to=4-8]
	\arrow["{[Tin_1,Tin_2]}"', from=3-1, to=4-2]
	\arrow["{T(\eta^S+ SY)}"', from=4-2, to=4-4]
	\arrow["{T[Sin_1,Sin_2]}"', from=4-4, to=4-6]
	\arrow["{\eta^S + TSY}"'{pos=0.7}, from=3-1, to=3-4]
	\arrow["{ST(G_{\delta_Y} X + X) + \delta_Y}"{pos=0.7}, from=3-4, to=3-7]
	\arrow["{[STin_1,STin_2]}"'{pos=0}, from=3-7, to=4-8]
	\arrow["{Tin_1}"{pos=0.8}, from=2-2, to=4-2]
	\arrow["{T\eta^S}"'{pos=0.3}, from=2-2, to=2-5]
	\arrow["{Tin_1}"{pos=0.7}, curve={height=-6pt}, from=2-5, to=4-4]
	\arrow["{TSin_1}"'{pos=0.7}, from=2-5, to=4-6]
	\arrow["\delta"', from=2-5, to=2-6]
	\arrow["{ST(G_{\delta_Y}X+X)}"{pos=0.2}, from=2-6, to=2-8]
	\arrow["{STin_1}", from=2-8, to=4-8]
	\arrow["{in_1}", from=1-1, to=3-1]
	\arrow[Rightarrow, no head, from=1-1, to=2-2]
	\arrow["{\eta^S}"', from=1-1, to=1-4]
	\arrow["{in_1}"'{pos=0.7}, from=1-4, to=3-4]
	\arrow["{ST(G_{\delta_Y} X + X)}", from=1-4, to=1-7]
	\arrow["{in_1}"'{pos=0.8}, from=1-7, to=3-7]
	\arrow[Rightarrow, no head, from=1-7, to=2-8]
	\arrow[Rightarrow, no head, from=1-4, to=2-6]
      \end{tikzcd}
      }
    \end{center}
    Finally, the second term of the bottom left polygon is chased as follows.
    \begin{center}
      \hfill
      \ajustedroit[.95]{
\begin{tikzcd}[ampersand replacement=\&,column sep=.5cm,baseline=(\tikzcdmatrixname-4-2.base)]
	TSY \&\& TSY \&\&\&\& STY \\
	\& TSY \&\& TSY \&\& TSY \&\& STY \\
	{T(G_{TSY} X + X) + TSY} \&\& {ST(G_{TSY} X + X) + STY} \&\&\&\& {ST(G_{STY} X + X) + STY} \\
	\& {T(G_{TSY} X + X + SY)} \&\& {T(S(G_{TSY} X + X) + SY)} \&\& {TS(G_{TSY}X +X +Y)} \&\& {ST(G_{STY} X + X + Y)}
	\arrow["{\delta (G_{\delta Y}X + X + Y)}"', from=4-6, to=4-8]
	\arrow["{[Tin_1,Tin_2]}"', from=3-1, to=4-2]
	\arrow["{T(\eta^S+ SY)}"', from=4-2, to=4-4]
	\arrow["{T[Sin_1,Sin_2]}"', from=4-4, to=4-6]
	\arrow["{\eta^S + TSY}"'{pos=0.7}, from=3-1, to=3-3]
	\arrow["{ST(G_{\delta_Y} X + X) + \delta_Y}", from=3-3, to=3-7]
	\arrow["{[STin_1,STin_2]}"'{pos=0}, from=3-7, to=4-8]
	\arrow["{in_2}"', from=1-1, to=3-1]
	\arrow["{Tin_2}"'{pos=0.2}, from=2-2, to=4-2]
	\arrow[Rightarrow, no head, from=1-1, to=2-2]
	\arrow[Rightarrow, no head, from=2-2, to=2-4]
	\arrow["{Tin_2}"'{pos=0.2}, from=2-4, to=4-4]
	\arrow[Rightarrow, no head, from=2-4, to=2-6]
	\arrow["{TSin_2}"'{pos=0.2}, from=2-6, to=4-6]
	\arrow["{\delta_Y}"{pos=0.2}, from=2-6, to=2-8]
	\arrow["{STin_2}", from=2-8, to=4-8]
	\arrow[Rightarrow, no head, from=1-1, to=1-3]
	\arrow["{in_2}"'{pos=0.2}, from=1-3, to=3-3]
	\arrow["{\delta_Y}", from=1-3, to=1-7]
	\arrow["{in_2}"'{pos=0.2}, from=1-7, to=3-7]
	\arrow[Rightarrow, no head, from=1-7, to=2-8]
      \end{tikzcd}
      } \qedhere
    \end{center}
  \end{proof}

  \begin{definition}\label{def:hbullet}
    Let
  $$h^•_{X,Y}∶ Γ^•_{S|Y}(Σ(X)) → ST(Γ^•_{ST|Y}(X))\rlap{,}$$
  be defined by cotupling the following three composites.
  \begin{mathpar}
    Γ_{SY}ΣX \xto{h_{X,SY}} ST Γ^•_{ST|SY}X \xto{ST↑^{S}}
    STSΓ^•_{STS|Y}X \xto{Sδ;μ^S T}
  ST Γ^•_{STS|Y}X \xto{ST Γ^•_{(Sδ;μ^S T)|Y}X} ST Γ^•_{ST|Y}X
  \and
  ΣX \xto{η_{Σ,ηᵀX}} STX \xto{ST α^{Γ^•_{ST|}}_{X,Y}}
    ST(Γ^•_{ST|Y}(X))
    \and
    Y \xto{η^{ST}_Y} STY \xto{ST β^{Γ^•_{ST|}}_{X,Y}}
    ST(Γ^•_{ST|Y}(X))
  \end{mathpar}
  \end{definition}
\begin{lemma}\label{lem:coh:hbullet}
  For all $h$, $h^•$ satisfies the following coherence laws.
\begin{center}
  \diag{%
    Σ(X)  \& ST(X) \\
    Γ^•_{S|Y}(Σ(X)) \& ST (Γ^•_{ST|Y}(X)) %
  }{%
    (m-1-1) edge[labela={η_{Σ,ηᵀ_X}}] (m-1-2) %
    edge[labell={α^{Γ^•_{S|}}_{ΣX,Y}}] (m-2-1) %
    (m-2-1) edge[labelb={h^•_{X,Y}}] (m-2-2) %
    (m-1-2) edge[labelr={ST  α^{Γ^•_{ST|}}_{X,Y}}] (m-2-2) %
  }
  \diag{%
    Y  \& STY \\
    Γ^•_{S|Y}(Σ(X)) \& ST (Γ^•_{ST|Y}(X)) %
  }{%
    (m-1-1) edge[labela={η^{ST}_Y}] (m-1-2) %
    edge[labell={β^{Γ^•_S}_{ΣX,Y}}] (m-2-1) %
    (m-2-1) edge[labelb={h^•_{X,Y}}] (m-2-2) %
    (m-1-2) edge[labelr={ST  β^{Γ^•_{ST|}}_{X,Y}}] (m-2-2) %
  }
  \diag(.6,1.3){%
    Γ^•_{S|SY}ΣX \& SΓ^•_{SS|Y}ΣX     \&     SΓ^•_{S|Y}ΣX \\
    \& \& SST Γ^•_{ST|Y}X \\
    ST Γ^•_{ST|SY}X \& STS Γ^•_{STS|Y}X \& ST Γ^•_{ST|Y}X %
  }{%
    (m-1-1)
    edge[labell={h^•_{X,SY}}] (m-3-1) %
    (m-3-1) edge[labelb={ST ↑^{S}}] (m-3-2) %
    (m-1-1) edge[labela={↑^{S}}] (m-1-2) %
    (m-1-2) edge[labela={SΓ^•_{μ^S|Y}ΣX}] (m-1-3) %
    (m-3-2) edge[labelb={(Sδ;μ^S) Γ^•_{Sδ;μ^S|Y}X}] (m-3-3) %
    (m-1-3) edge[labelr={Sh^•_{X,Y}}] (m-2-3) %
    (m-2-3) edge[labelr={μ^S}] (m-3-3) %
  }
  \diag{%
    Γ_Y Σ X \& Γ_{SY} Σ X \& Γ^•_{S|Y} Σ X \\
    \& ST Γ^•_{ST|Y} X
  }{%
    (m-1-1) edge[labela={Γ_{η^S_Y} Σ X}] (m-1-2) %
    edge[labelbl={h_{X,Y}}] (m-2-2) %
    (m-1-2) edge[labela={in₁}] (m-1-3) %
    (m-1-3) edge[labelbr={h^•_{X,Y}}] (m-2-2) %
  }
\end{center}
\end{lemma}
\begin{proof}
  The first two laws hold by construction.  For the third law, by
  universal property of coproduct, it suffices to verify it on each
  term of the coproduct $Γ^•_{SSY}ΣX = Γ_{SSY}ΣX + ΣX + Y$.  The
  result follows straightforwardly by diagram chasing in each case.
  The first term is chased as follows,
  using the coherence properties of monad distributive laws.
  \begin{center}
    \ajustedroit{
\begin{tikzcd}[ampersand replacement=\&]
	{\Gamma_{SSY}\Sigma X} \&\&\&\& {S\Gamma_{SSY}\Sigma X} \&\& {S\Gamma_{SY}\Sigma X} \\
	\& {ST \Gamma^\bullet_{ST|SSY}X} \&\&\&\& {SST \Gamma^\bullet_{ST|SSY}X} \&\& {SST\Gamma^\bullet_{ST|SY}X} \\
	\&\&\&\&\& {ST \Gamma^\bullet_{ST|SY}X} \\
	\\
	\& {STS\Gamma^\bullet_{STS|SY} X} \&\& {STSS\Gamma^\bullet_{STSS|Y}X} \&\& {STS\Gamma^\bullet_{STS|Y}X} \&\& {SSTS\Gamma^\bullet_{STS|Y}X} \\
	{\Gamma^\bullet_{S|SY}\Sigma X} \&\& {S\Gamma^\bullet_{SS|Y}\Sigma X} \&\&\&\& {S \Gamma^\bullet_{S|Y} \Sigma X} \&\&\& {STS\Gamma^\bullet_{STS|Y}X} \\
	\&\&\&\&\&\&\& {SST\Gamma^\bullet_{ST|Y} X} \\
	\& {ST\Gamma^\bullet_{ST|SY} X} \&\& {STS\Gamma^\bullet_{STS|Y} X} \&\&\&\&\&\& {ST\Gamma^\bullet_{ST|Y} X}
	\arrow["{\uparrow^S}"{pos=0.3}, from=6-1, to=6-3]
	\arrow["{S\Gamma^\bullet_{\mu^S|Y} \Sigma X}"{pos=0.4}, from=6-3, to=6-7]
	\arrow["{Sh^\bullet_{X,Y}}"', from=6-7, to=7-8]
	\arrow["{\mu^S}", from=7-8, to=8-10]
	\arrow["{h^\bullet_{X,SY}}"', from=6-1, to=8-2]
	\arrow["{ST \uparrow^S}"', from=8-2, to=8-4]
	\arrow["{(S\delta;\mu^S T)\Gamma^\bullet_{S\delta;\mu^S T|Y} X}"', from=8-4, to=8-10]
	\arrow["{in_1}"', from=1-1, to=6-1]
	\arrow["{h_{X,SSY}}"', from=1-1, to=2-2]
	\arrow["{ST \uparrow^S}", from=2-2, to=5-2]
	\arrow["{\eta^S}", from=1-1, to=1-5]
	\arrow["{S\Gamma_{\mu^S_Y}\Sigma X}", from=1-5, to=1-7]
	\arrow["{S in_1}", from=1-7, to=6-7]
	\arrow["{Sh_{X,SY}}", from=1-7, to=2-8]
	\arrow["{(S\delta;\mu^S)\Gamma^\bullet_{S\delta;\mu^S|SY} X}"{pos=0.8}, from=5-2, to=8-2]
	\arrow["{SST\uparrow^S}", from=2-8, to=5-8]
	\arrow["{S(S\delta;\mu^ST)\Gamma^\bullet_{S\delta;\mu^S T|Y}X}"{pos=0.8}, from=5-8, to=7-8]
	\arrow["{STS\uparrow^S}", from=5-2, to=5-4]
	\arrow["{(S\delta;\mu^S T)S\Gamma^\bullet_{S\delta S;\mu^S T S|Y} X}"'{pos=0.6}, from=5-4, to=8-4]
	\arrow["{\eta^S}"{pos=0.2}, from=2-2, to=2-6]
	\arrow["{SST \Gamma^\bullet_{ST|\mu^S_Y}X}"{pos=0.2}, from=2-6, to=2-8]
	\arrow["{Sh_{X,SSY}}"'{pos=0}, from=1-5, to=2-6]
	\arrow["{ST\Gamma^\bullet_{ST|\mu^S_Y}X}"', from=2-2, to=3-6]
	\arrow["{\eta^S}"{pos=0.1}, from=3-6, to=2-8]
	\arrow["{\mu^S TS \Gamma^\bullet_{STS|Y}X}", from=5-8, to=6-10]
	\arrow["{(S\delta;\mu^ST)\Gamma^\bullet_{S\delta;\mu^S T|Y}X}", from=6-10, to=8-10]
	\arrow["{ST\uparrow^S}"', from=3-6, to=5-6]
	\arrow["{\eta^S}"{pos=0.2}, from=5-6, to=5-8]
	\arrow[Rightarrow, no head, from=5-6, to=6-10]
	\arrow["{ST\mu^S\Gamma^\bullet_{ST\mu^S|Y}X}", from=5-4, to=5-6]
\end{tikzcd}
    }
  \end{center}
  The second term is chased as follows.
  \begin{center}
    \ajustedroit{
\begin{tikzcd}[ampersand replacement=\&]
	{\Sigma X} \&\&\& {S\Sigma X} \&\&\& {S\Sigma X} \\
	\&\& STX \&\&\&\&\& SSTX \\
	\&\&\&\& STSX \&\&\&\&\& STX \\
	{\Gamma^\bullet_{S|SY}\Sigma X} \&\&\& {S\Gamma^\bullet_{SS|Y}\Sigma X} \&\&\& {S \Gamma^\bullet_{S|Y} \Sigma X} \\
	\&\&\&\&\&\&\& {SST\Gamma^\bullet_{ST|Y} X} \\
	\&\& {ST\Gamma^\bullet_{ST|SY} X} \&\& {STS\Gamma^\bullet_{STS|Y} X} \&\&\&\&\& {ST\Gamma^\bullet_{ST|Y} X}
	\arrow["{\uparrow^S}"{pos=0.3}, from=4-1, to=4-4]
	\arrow["{S\Gamma^\bullet_{\mu^S|Y} \Sigma X}"{pos=0.7}, from=4-4, to=4-7]
	\arrow["{Sh^\bullet_{X,Y}}"', from=4-7, to=5-8]
	\arrow["{\mu^S}", from=5-8, to=6-10]
	\arrow["{h^\bullet_{X,SY}}"', from=4-1, to=6-3]
	\arrow["{ST \uparrow^S}"', from=6-3, to=6-5]
	\arrow["{(S\delta;\mu^S T)\Gamma^\bullet_{S\delta;\mu^ST|Y} X}"', from=6-5, to=6-10]
	\arrow["{in_2}"', from=1-1, to=4-1]
	\arrow["{\eta_{\Sigma, \eta^T_X}}", from=1-1, to=2-3]
	\arrow["{ST in_2}"{pos=0.8}, from=2-3, to=6-3]
	\arrow["{ST\eta^S}"'{pos=0.3}, from=2-3, to=3-5]
	\arrow["{STS in_2}"'{pos=0.8}, from=3-5, to=6-5]
	\arrow["{\eta^S_{\Sigma X}}", from=1-1, to=1-4]
	\arrow["{S in_2}"'{pos=0.8}, from=1-4, to=4-4]
	\arrow[Rightarrow, no head, from=1-4, to=1-7]
	\arrow["{S in_2}"{pos=0.9}, from=1-7, to=4-7]
	\arrow["{S\eta_{\Sigma,\eta^T_X}}", from=1-7, to=2-8]
	\arrow["{SSTin_2}", from=2-8, to=5-8]
	\arrow["{\mu^S}", from=2-8, to=3-10]
	\arrow["{ST in_2}", from=3-10, to=6-10]
	\arrow["{S\delta ; \mu^S T}"'{pos=0.2}, from=3-5, to=3-10]
	\arrow["{\eta^S_{STX}}", from=2-3, to=2-8]
\end{tikzcd}
      }
  \end{center}
  The third and final term is chased as follows.
  \begin{center}
    \ajustedroit
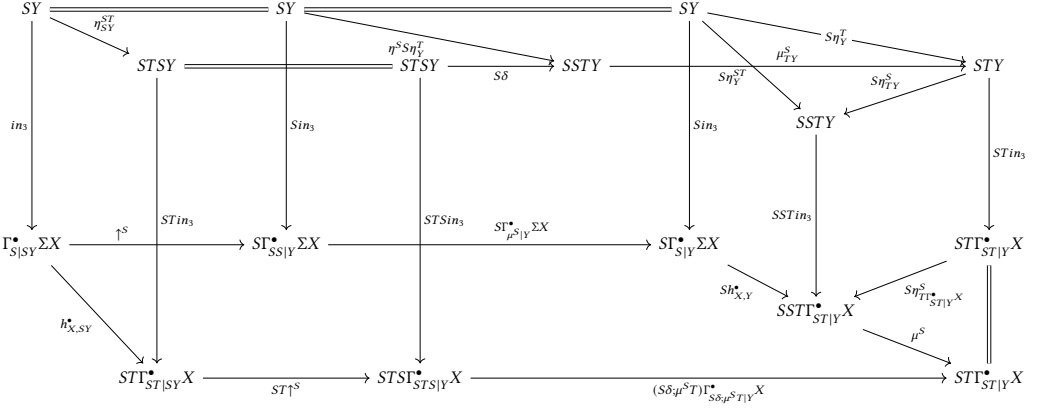
{
\begin{tikzcd}[ampersand replacement=\&]
	SY \&\& SY \&\&\&\& SY \\
	\& STSY \&\& STSY \&\& SSTY \&\&\&\& STY \\
	\&\&\&\&\&\&\& SSTY \\
	\\
	\\
	{\Gamma^\bullet_{S|SY}\Sigma X} \&\& {S\Gamma^\bullet_{SS|Y}\Sigma X} \&\&\&\& {S \Gamma^\bullet_{S|Y} \Sigma X} \&\&\& {ST\Gamma^\bullet_{ST|Y} X} \\
	\&\&\&\&\&\&\& {SST\Gamma^\bullet_{ST|Y} X} \\
	\& {ST\Gamma^\bullet_{ST|SY} X} \&\& {STS\Gamma^\bullet_{STS|Y} X} \&\&\&\&\&\& {ST\Gamma^\bullet_{ST|Y} X}
	\arrow["{\uparrow^S}"{pos=0.3}, from=6-1, to=6-3]
	\arrow["{S\Gamma^\bullet_{\mu^S|Y} \Sigma X}"{pos=0.6}, from=6-3, to=6-7]
	\arrow["{Sh^\bullet_{X,Y}}"', from=6-7, to=7-8]
	\arrow["{\mu^S}", from=7-8, to=8-10]
	\arrow["{h^\bullet_{X,SY}}"', from=6-1, to=8-2]
	\arrow["{ST \uparrow^S}"', from=8-2, to=8-4]
	\arrow["{(S\delta;\mu^S T)\Gamma^\bullet_{S\delta;\mu^S T|Y} X}"', from=8-4, to=8-10]
	\arrow["{in_3}"', from=1-1, to=6-1]
	\arrow["{\eta^{ST}_{SY}}", from=1-1, to=2-2]
	\arrow["{ST in_3}", from=2-2, to=8-2]
	\arrow[Rightarrow, no head, from=2-2, to=2-4]
	\arrow["{STS in_3}", from=2-4, to=8-4]
	\arrow[Rightarrow, no head, from=1-1, to=1-3]
	\arrow["{Sin_3}", from=1-3, to=6-3]
	\arrow[Rightarrow, no head, from=1-3, to=1-7]
	\arrow["{Sin_3}", from=1-7, to=6-7]
	\arrow["{S\eta^T_Y}"{description}, from=1-7, to=2-10]
	\arrow["S\delta"', from=2-4, to=2-6]
	\arrow["{\mu^S_{TY}}", from=2-6, to=2-10]
	\arrow["{\eta^SS\eta^T_Y}"', from=1-3, to=2-6]
	\arrow["{S\eta^{ST}_Y}"', from=1-7, to=3-8]
	\arrow["{SSTin_3}"', from=3-8, to=7-8]
	\arrow["{S\eta^S_{TY}}"', from=2-10, to=3-8]
	\arrow["{STin_3}", from=2-10, to=6-10]
	\arrow["{S\eta^S_{T\Gamma^\bullet_{ST|Y} X}}", from=6-10, to=7-8]
	\arrow[Rightarrow, no head, from=6-10, to=8-10]
\end{tikzcd}
      }
  \end{center}
  Finally, for the fourth law,
  $h^•_{X,Y} ∘ in₁$ yields the first composite of Definition~\ref{def:hbullet}.
  But precomposing the latter with $Γ_{η^S_Y} Σ X$, we obtain by naturality of $h$
  the composite
  $$Γ_Y Σ X \xto{h_{X,Y}} ST Γ^•_{ST|Y} X
  \xto{ST Γ^•_{ST|η^S_Y}X} ST Γ^•_{ST|SY} X \xto{ST↑^S} STS Γ^•_{STS|Y}
  X \xto{(Sδ;μ^S T)Γ^•_{(Sδ;μ^S T)|Y}X}ST Γ^•_{ST|Y}X\rlap{,}$$ which
  we now prove equal to $h_{X,Y}$, i.e., that
  the last three morphisms compose to the identity.
  \begin{center}
    \ajustedroit{
      \hfill
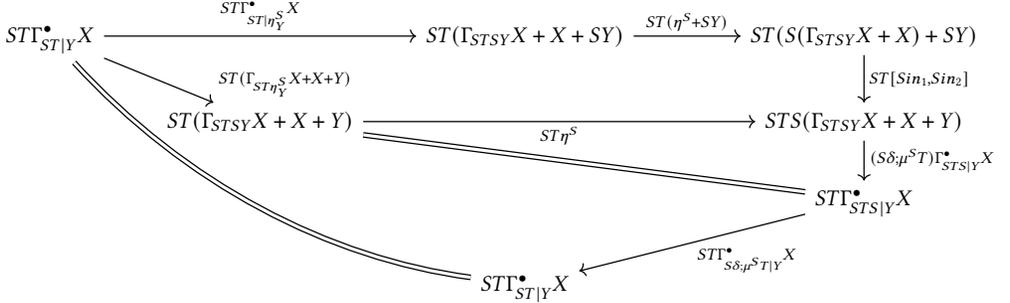
\begin{tikzcd}[ampersand replacement=\&,baseline=(\tikzcdmatrixname-4-3.base)]
	{ST \Gamma^\bullet_{ST|Y}X} \&\& {ST (\Gamma_{STSY}X + X + SY)} \&\& {ST (S(\Gamma_{STSY}X + X) + SY)} \\
	\& {ST (\Gamma_{STSY}X + X + Y)} \&\&\& {STS(\Gamma_{STSY}X + X + Y)} \\
	\&\&\&\& {ST \Gamma^\bullet_{STS|Y}X} \\
	\&\& {ST \Gamma^\bullet_{ST|Y}X}
	\arrow["{ST\Gamma^\bullet_{ST|\eta^S_Y} X}", from=1-1, to=1-3]
	\arrow["{ST (\eta^S + SY)}", from=1-3, to=1-5]
	\arrow["{ST[Sin_1,Sin_2]}", from=1-5, to=2-5]
	\arrow["{ST(\Gamma_{ST\eta^S_Y}X + X + Y)}"{pos=1}, from=1-1, to=2-2]
	\arrow["{ST\eta^S}"', from=2-2, to=2-5]
	\arrow["{(S\delta;\mu^S T) \Gamma^\bullet_{STS|Y} X}", from=2-5, to=3-5]
	\arrow["{ST\Gamma^\bullet_{S\delta;\mu^S T|Y} X}", from=3-5, to=4-3]
	\arrow[Rightarrow, no head, from=2-2, to=3-5]
	\arrow[curve={height=30pt}, Rightarrow, no head, from=1-1, to=4-3]
\end{tikzcd} \qedhere
      }
  \end{center}
\end{proof}

\subsection{The family $h^∘$}

We now want to make the $Σ$ in the domain of $h^•$ into an $S$.
For this, as an intermediate step, we emphasise the less pointed variant of $Γ$ defined as follows.
\begin{definition}
  A bifunctor $Γ∶ 𝐂² → 𝐂$ is \alert{main-pointed} iff it comes equipped
  with a natural transformation $α^Γ_{X,Y}∶ X → Γ(X,Y)$.


\end{definition}

\begin{definition}
  For any $𝐂$ with binary coproducts and bifunctor $Γ∶ 𝐂² → 𝐂$,
  let $Γ^\pt(X,Y) ≔ Γ(X,Y) + X$.
\end{definition}

\begin{proposition}
  For any $𝐂$ with binary coproducts and $Γ∶ 𝐂² → 𝐂$, the bifunctor
  $Γ^\pt$ is main-pointed
  .
\end{proposition}

By construction, we have:
\begin{lemma}
  For any $𝐂$ with binary coproducts, $F,G∶ 𝐂 → 𝐂$, and $Γ∶ 𝐂² → 𝐂$, we have
  $Γ^•_{G|FY} = O_{FY} Γ^\pt_{GFY}$.
\end{lemma}

By the last lemma, $h^•_{X,Y}$ equivalently has type
$$Γ^•_{S|Y}ΣX → ST O_Y Γ^\pt_{STY}X.$$
\begin{definition}
  Let $h^∘_{X,Y}$ denote the natural transformation with components
  $$      Γ^\pt_{STY}ΣX 
      \xto{[in₁,in₂]}
      Γ^•_{S|TY}ΣX 
      \xto{h^•_{X,TY}}
      ST Γ^•_{ST|TY}X
      \xto{ST↑^{T}}
      STT Γ^•_{STT|Y}X
      \xto{Sμᵀ Γ^•_{Sμᵀ|Y}X}
      STΓ^•_{ST|Y}X.$$

\end{definition}

\begin{lemma}
  \label{lem:coh:hcirc}
  For all $h$, $h^\pt$ satisfies the following coherence laws.
\begin{center}
  \diag{%
    Σ(X)  \& ST(X) \\
    Γ^\pt_{STY}(Σ(X)) \& ST (Γ^•_{ST|Y}(X)) %
  }{%
    (m-1-1) edge[labela={η_{Σ,ηᵀ_X}}] (m-1-2) %
    edge[labell={α^{Γ^\pt}_{ΣX,STY}}] (m-2-1) %
    (m-2-1) edge[labelb={h^\pt_{X,Y}}] (m-2-2) %
    (m-1-2) edge[labelr={ST  α^{Γ^•_{ST|}}_{X,Y}}] (m-2-2) %
  }
  \diag{%
    Γ^\pt_{STSY}ΣX \&      \&     Γ^\pt_{STY}ΣX \\
    ST Γ^•_{ST|SY}X \& STS Γ^•_{STS|Y}X \& ST Γ^•_{ST|Y}X %
  }{%
    (m-1-1)
    edge[labell={h^\pt_{X,SY}}] (m-2-1) %
    (m-2-1) edge[loinb={ST ↑^{S}_{X,Y}}] (m-2-2) %
    (m-1-1) edge[labela={Γ^\pt_{(S δ; μ^S T)_Y}ΣX}] (m-1-3) %
    (m-2-2) edge[loinb={(Sδ;μ^S T) Γ^•_{S δ ; μ^S T | Y}X}] (m-2-3) %
    (m-1-3) edge[labelr={h^\pt_{X,Y}}] (m-2-3) %
  }
  \diag{%
    Γ_Y Σ X \& Γ_{STY}ΣX \& Γ^\pt_{STY}ΣX \\
    \& STΓ^•_{ST|Y}X %
  }{%

    (m-1-1) edge[labela={Γ_{η^{ST}_Y} Σ X}] (m-1-2) %
    edge[labelbl={h_{X,Y}}] (m-2-2) %
    (m-1-2) edge[labela={in₁}] (m-1-3) %
    (m-1-3) edge[labelbr={h^∘_{X,Y}}] (m-2-2) %
  }
\ajustedroit{
\begin{tikzcd}[ampersand replacement=\&]
	{\Gamma^\circ_{STTY}\Sigma X} \&\&\&\&\&\&\& {\Gamma^\circ_{STY}\Sigma X} \\
	{ ST O_{TY} \Gamma^\circ_{STTY}X} \&\& { STT O_{Y} \Gamma^\circ_{STTY}X} \&\&\& { ST O_{Y} \Gamma^\circ_{STTY}X} \&\& { ST O_{Y} \Gamma^\circ_{STY}X}
	\arrow["{{h^\circ}_{X,TY}}"', from=1-1, to=2-1]
	\arrow["{ ST\uparrow^T}"', from=2-1, to=2-3]
	\arrow["{ S\mu^T O_{Y} \Gamma^\circ_{STTY}X}"', from=2-3, to=2-6]
	\arrow["{ ST O_{Y} \Gamma^\circ_{S\mu^T_Y}}"', from=2-6, to=2-8]
	\arrow["{{h^\circ}_{X,Y}}", from=1-8, to=2-8]
	\arrow["{\Gamma^\circ_{S\mu^T_Y}\Sigma X}", from=1-1, to=1-8]
\end{tikzcd}
}
\end{center}
\end{lemma}
\begin{proof}
  The first law holds by chasing the following diagram.
  \begin{center}
\begin{tikzcd}[ampersand replacement=\&]
	{\Sigma X} \&\& STX \&\& STX \\
	\&\&\& STTX \\
	{\Gamma^\circ_{STY}\Sigma X} \& {\Gamma^\bullet_{S|TY}\Sigma X} \& {ST \Gamma^\bullet_{ST|TY} X} \& {STT\Gamma^\bullet_{STT|Y}X} \& {ST \Gamma^\bullet_{ST|Y}X}
	\arrow["{[in_1,in_2]}"', from=3-1, to=3-2]
	\arrow["{h^\bullet_{X,TY}}"', from=3-2, to=3-3]
	\arrow["{ST \uparrow^T}"', from=3-3, to=3-4]
	\arrow["{S \mu^T \Gamma^\bullet_{S\mu^T|Y}X}"', from=3-4, to=3-5]
	\arrow["{in_2}"', from=1-1, to=3-1]
	\arrow["{in_2}"', from=1-1, to=3-2]
	\arrow["{\eta_{\Sigma} \eta^T_X}", from=1-1, to=1-3]
	\arrow["{ST in_2}"', from=1-3, to=3-3]
	\arrow["{ST \eta^T}"', from=1-3, to=2-4]
	\arrow["{STT in_2}"', from=2-4, to=3-4]
	\arrow["{S\mu^T_X}"', from=2-4, to=1-5]
	\arrow["{ST in_2}", from=1-5, to=3-5]
	\arrow[Rightarrow, no head, from=1-3, to=1-5]
\end{tikzcd}
  \end{center}
  For the second one, we proceed by diagram chasing, as follows.
  \begin{center}
    \ajustedroit
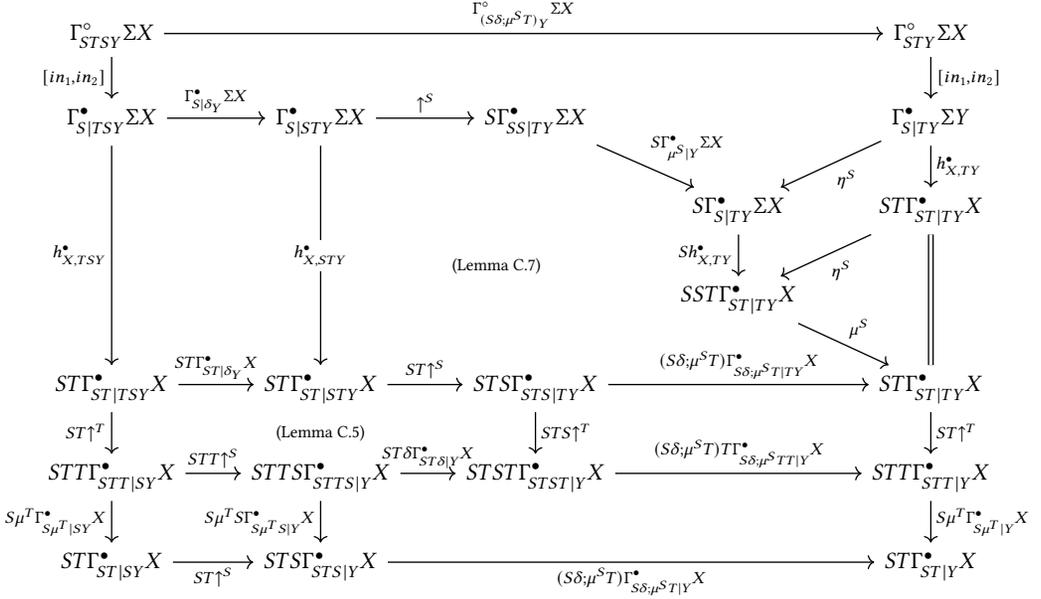
{
\begin{tikzcd}[ampersand replacement=\&]
	{\Gamma^\circ_{STSY}\Sigma X} \&\&\&\& {\Gamma^\circ_{STY} \Sigma X} \\
	{\Gamma^\bullet_{S|TSY}\Sigma X} \& {\Gamma^\bullet_{S|STY}\Sigma X} \& {S\Gamma^\bullet_{SS|TY}\Sigma X} \&\& {\Gamma^\bullet_{S|TY}\Sigma Y} \\
	\&\&\& {S\Gamma^\bullet_{S|TY}\Sigma X} \& {ST\Gamma^\bullet_{ST|TY} X} \\
	\&\&\& {SST\Gamma^\bullet_{ST|TY} X} \\
	{ST\Gamma^\bullet_{ST|TSY} X} \& {ST\Gamma^\bullet_{ST|STY} X} \& {STS\Gamma^\bullet_{STS|TY} X} \&\& {ST\Gamma^\bullet_{ST|TY} X} \\
	{STT\Gamma^\bullet_{STT|SY}X} \& {STTS \Gamma^\bullet_{STTS|Y}X} \& {STST\Gamma^\bullet_{STST|Y} X} \&\& {STT\Gamma^\bullet_{STT|Y}X} \\
	{ST\Gamma^\bullet_{ST|SY}X} \& {STS\Gamma^\bullet_{STS|Y}X} \&\&\& {ST\Gamma^\bullet_{ST|Y}X}
	\arrow["{[in_1,in_2]}"', from=1-1, to=2-1]
	\arrow["{h^\bullet_{X,TSY}}"', from=2-1, to=5-1]
	\arrow["{ST \uparrow^T}"', from=5-1, to=6-1]
	\arrow["{S\mu^T \Gamma^\bullet_{S\mu^T|SY} X}"', from=6-1, to=7-1]
	\arrow["{\Gamma^\circ_{(S\delta;\mu^ST)_Y}\Sigma X}", from=1-1, to=1-5]
	\arrow["{[in_1,in_2]}", from=1-5, to=2-5]
	\arrow["{ST \uparrow^T}", from=5-5, to=6-5]
	\arrow["{S\mu^T \Gamma^\bullet_{S\mu^T|Y} X}", from=6-5, to=7-5]
	\arrow["{ST \uparrow^S}"', from=7-1, to=7-2]
	\arrow["{(S\delta;\mu^ST)\Gamma^\bullet_{S\delta;\mu^ST|Y} X}"', from=7-2, to=7-5]
	\arrow["{\Gamma^\bullet_{S|\delta_Y} \Sigma X}", from=2-1, to=2-2]
	\arrow["{\uparrow^S}", from=2-2, to=2-3]
	\arrow[""{name=0, anchor=center, inner sep=0}, "{S\Gamma^\bullet_{\mu^S|Y} \Sigma X}", from=2-3, to=3-4]
	\arrow["{\eta^S}", from=2-5, to=3-4]
	\arrow["{Sh^\bullet_{X,TY}}"', from=3-4, to=4-4]
	\arrow["{\mu^S}", from=4-4, to=5-5]
	\arrow["{h^\bullet_{X,TY}}", from=2-5, to=3-5]
	\arrow["{\eta^S}", from=3-5, to=4-4]
	\arrow[Rightarrow, no head, from=3-5, to=5-5]
	\arrow["{h^\bullet_{X,STY}}"{description}, from=2-2, to=5-2]
	\arrow["{ST \uparrow^S}", from=5-2, to=5-3]
	\arrow["{(S\delta;\mu^S T) \Gamma^\bullet_{S\delta;\mu^S T|TY}X}", from=5-3, to=5-5]
	\arrow["{ST\Gamma^\bullet_{ST|\delta_Y} X}", from=5-1, to=5-2]
	\arrow["{STS\uparrow^T}", from=5-3, to=6-3]
	\arrow["{STT\uparrow^S}", from=6-1, to=6-2]
	\arrow["{ST \delta \Gamma^\bullet_{ST\delta|Y}X}", from=6-2, to=6-3]
	\arrow["{\text{(Lemma~\ref{lem:delta:uparrow})}}"{description}, draw=none, from=5-2, to=6-2]
	\arrow["{S\mu^T S\Gamma^\bullet_{S\mu^TS|Y}X}"', from=6-2, to=7-2]
	\arrow["{(S\delta;\mu^S T) T\Gamma^\bullet_{S\delta;\mu^S TT|Y}X}", from=6-3, to=6-5]
	\arrow["{\text{(Lemma~\ref{lem:coh:hbullet})}}"{description}, Rightarrow, draw=none, from=5-2, to=0]
      \end{tikzcd}
    }
  \end{center}
Let us now consider the third claim.
It follows by chasing the following diagram,
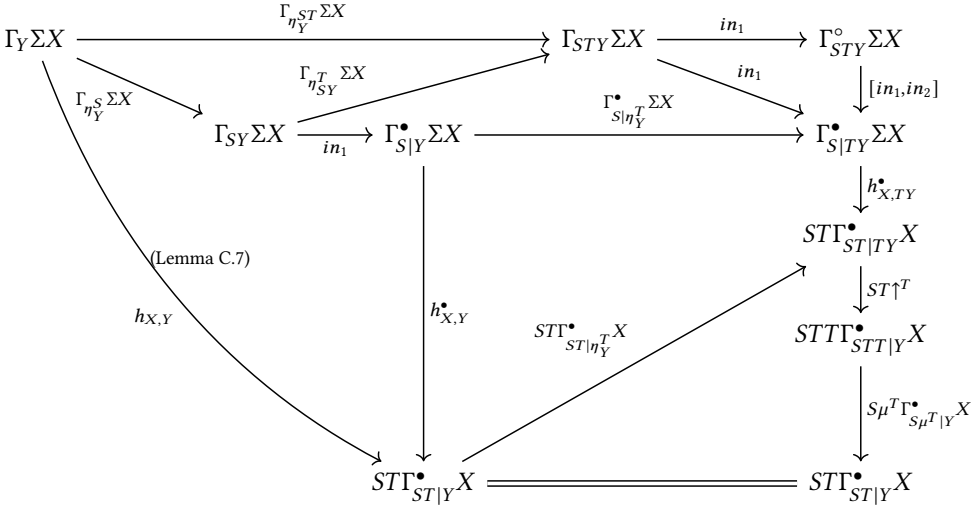
\begin{center}
\begin{tikzcd}[ampersand replacement=\&]
	{\Gamma_Y \Sigma X} \&\&\&\& {\Gamma_{STY} \Sigma X} \&\& {\Gamma^\circ_{STY} \Sigma X} \\
	\&\& {\Gamma_{SY} \Sigma X} \& {\Gamma^\bullet_{S|Y}\Sigma X} \&\&\& {\Gamma^\bullet_{S|TY} \Sigma X} \\
	\&\&\&\&\&\& {ST\Gamma^\bullet_{ST|TY} X} \\
	\&\&\&\&\&\& {STT \Gamma^\bullet_{STT|Y}X} \\
	\\
	\&\&\& {ST \Gamma^\bullet_{ST|Y}X} \&\&\& {ST \Gamma^\bullet_{ST|Y}X}
	\arrow["{\Gamma_{\eta^{ST}_Y} \Sigma X}", from=1-1, to=1-5]
	\arrow["{in_1}", from=1-5, to=1-7]
	\arrow["{[in_1,in_2]}", from=1-7, to=2-7]
	\arrow["{h^\bullet_{X,TY}}", from=2-7, to=3-7]
	\arrow["{ST \uparrow^T}", from=3-7, to=4-7]
	\arrow[""{name=0, anchor=center, inner sep=0}, "{h_{X,Y}}"', curve={height=30pt}, from=1-1, to=6-4]
	\arrow["{\Gamma_{\eta^S_Y}\Sigma X}"', from=1-1, to=2-3]
	\arrow["{\Gamma_{\eta^T_{SY}} \Sigma X}"{pos=0.3}, from=2-3, to=1-5]
	\arrow["{in_1}"', from=2-3, to=2-4]
	\arrow[""{name=1, anchor=center, inner sep=0}, "{h^\bullet_{X,Y}}", from=2-4, to=6-4]
	\arrow["{in_1}", from=1-5, to=2-7]
	\arrow["{\Gamma^\bullet_{S|\eta^T_Y}\Sigma X}", from=2-4, to=2-7]
	\arrow["{S\mu^T \Gamma^\bullet_{S\mu^T|Y}X}", from=4-7, to=6-7]
	\arrow[Rightarrow, no head, from=6-4, to=6-7]
	\arrow["{ST \Gamma^\bullet_{ST|\eta^T_Y}X}", from=6-4, to=3-7]
	\arrow["{\text{(Lemma~\ref{lem:coh:hbullet})}}"{pos=0.3}, curve={height=-18pt}, Rightarrow, draw=none, from=0, to=1]
      \end{tikzcd}
    \end{center}
whose bottom right triangle commutes by chasing the following diagram.
\begin{center}
\begin{tikzcd}[ampersand replacement=\&,baseline=(\tikzcdmatrixname-5-3.base),column sep=1.2cm]
	{T(\Gamma_{STY}X + X + Y)} \&\& {T(\Gamma_{STTY} X + X + TY)} \\
	\& {T(\Gamma_{STTY}X + X + Y)} \& {T(T(\Gamma_{STTY} X + X) + TY)} \\
	\&\& {TT \Gamma^\bullet_{STT|Y}X} \\
	\&\& {T\Gamma^\bullet_{STT|Y}X} \\
	\&\& {T \Gamma^\bullet_{ST|Y}X}
	\arrow["{T \Gamma^\bullet_{ST|\eta^T_Y}X}", from=1-1, to=1-3]
	\arrow["{T (\eta^T + TY)}", from=1-3, to=2-3]
	\arrow["{T[Tin_1,Tin_2]}", from=2-3, to=3-3]
	\arrow[curve={height=18pt}, Rightarrow, no head, from=1-1, to=5-3]
	\arrow["{\mu^T \Gamma^\bullet_{STT|Y}X}", from=3-3, to=4-3]
	\arrow["{T \Gamma^\bullet_{S\mu^T|Y}X}"{pos=0.3}, from=4-3, to=5-3]
	\arrow["{T\eta^T}"'{pos=.6}, from=2-2, to=3-3]
	\arrow["{T(\eta^T + \eta^T_Y)}", from=2-2, to=2-3]
	\arrow[curve={height=6pt}, Rightarrow, no head, from=2-2, to=4-3]
	\arrow["{T(\Gamma_{ST\eta^T_Y}X + X + Y)}"{pos=.7}, from=1-1, to=2-2]
      \end{tikzcd}
\end{center}
Finally, we prove the last claim by diagram chasing as follows,
\begin{center}
\ajustedroit[.95]{
\begin{tikzcd}[ampersand replacement=\&,baseline=(\tikzcdmatrixname-5-1.base)]
	{\Gamma^\circ_{STTY}\Sigma X} \&\&\&\&\&\&\& {\Gamma^\circ_{STY}\Sigma X} \\
	{\Gamma^\bullet_{STTY}\Sigma X} \&\&\&\&\&\&\& {\Gamma^\bullet_{STY}\Sigma X} \\
	{ST O_{TTY} \Gamma^\circ_{STTTY} X} \&\&\&\&\&\&\& {ST O_{TY} \Gamma^\circ_{STTY} X} \\
	{STTO_{TY} \Gamma^\circ_{STTTY} X} \&\& {STTTO_{Y} \Gamma^\circ_{STTTY} X} \&\&\&\&\& {STTO_{Y} \Gamma^\circ_{STTY} X} \\
	{ ST O_{TY} \Gamma^\circ_{STTY}X} \&\& { STT O_{Y} \Gamma^\circ_{STTY}X} \&\&\& { ST O_{Y} \Gamma^\circ_{STTY}X} \&\& { ST O_{Y} \Gamma^\circ_{STY}X}
	\arrow["{ ST\uparrow^T}"', from=5-1, to=5-3]
	\arrow["{ S\mu^T O_{Y} \Gamma^\circ_{STTY}X}"', from=5-3, to=5-6]
	\arrow["{ ST O_{Y} \Gamma^\circ_{S\mu^T_Y}}"', from=5-6, to=5-8]
	\arrow["{\Gamma^\circ_{S\mu^T_Y}\Sigma X}", from=1-1, to=1-8]
	\arrow["{[in_1,in_2]}", from=1-8, to=2-8]
	\arrow["{h^\bullet_{X,TY}}", from=2-8, to=3-8]
	\arrow["{ST \uparrow^T}", from=3-8, to=4-8]
	\arrow["{S\mu^TO_{Y} \Gamma^\circ_{S\mu^T_Y} X}", from=4-8, to=5-8]
	\arrow["{[in_1,in_2]}"', from=1-1, to=2-1]
	\arrow["{h^\bullet_{X,TTY}}"', from=2-1, to=3-1]
	\arrow["{ST \uparrow^T}"', from=3-1, to=4-1]
	\arrow["{S\mu^TO_{TY} \Gamma^\circ_{S\mu^T_{TY}} X}"', from=4-1, to=5-1]
	\arrow["{\Gamma^\bullet_{S\mu^T_Y}\Sigma X}", from=2-1, to=2-8]
	\arrow["{ST O_{\mu^T_Y} \Gamma^\circ_{ST\mu^T_Y} X}", from=3-1, to=3-8]
	\arrow["{STT \uparrow^T}", from=4-1, to=4-3]
	\arrow["{S\mu^TTO_{Y} \Gamma^\circ_{S\mu^T_{TY}} X}", from=4-3, to=5-3]
	\arrow["{ST\mu^TO_{Y} \Gamma^\circ_{ST\mu^T_Y} X}", from=4-3, to=4-8]
      \end{tikzcd}}
\end{center}
where
\begin{itemize}
\item both top rectangles, as well as the bottom left and bottom right
  rectangles, commute by naturality, and
\item the third rectangle commutes by naturality and
  Lemma~\ref{lem:option:mu:inside} below. \qedhere
\end{itemize}
\end{proof}

\begin{lemma}\label{lem:option:mu:inside}
  For any monad $T$ on a category with binary coproducts, for all
  objects $X$ and $Y$, the following diagram commutes.
  \begin{center}
\begin{tikzcd}[ampersand replacement=\&]
	{ O_{TTY}X} \&\& { O_{TY} X} \\
	{TO_{TY}X} \& {TTO_{Y}X} \& {TO_{Y} X}
	\arrow["{ \uparrow^T}", from=1-3, to=2-3]
	\arrow["{ \uparrow^T}"', from=1-1, to=2-1]
	\arrow["{ O_{\mu^T_Y} X}", from=1-1, to=1-3]
	\arrow["{T \uparrow^T}"', from=2-1, to=2-2]
	\arrow["{\mu^TO_{Y} X}"', from=2-2, to=2-3]
\end{tikzcd}
  \end{center}
\end{lemma}
\begin{proof}
  This hold by diagram chasing, as follows.
  \begin{center}
    \hfill
    \ajustedroit[.95]{
  \begin{tikzcd}[ampersand replacement=\&,baseline=(\tikzcdmatrixname-5-1.base)]
    {X + TTY} \&\&\&\&\&\& {X+TY} \\
    {TX + TTY} \\
    \&\&\& {TTX + TTY} \&\&\& {TX + TY} \\
    \\
    {T(X + TY)} \&\& {T(TX + TY)} \&\& {TT(X+Y)} \&\& {T(X+Y)}
    \arrow["{X + \mu^T_Y}", from=1-1, to=1-7]
    \arrow["{\mu^T}"', from=5-5, to=5-7]
    \arrow["{\eta^T_X + TY}", from=1-7, to=3-7]
    \arrow["{[Tin_1,Tin_2]}", from=3-7, to=5-7]
    \arrow["{\eta^T_X + TTY}"', from=1-1, to=2-1]
    \arrow["{[Tin_1,Tin_2]}"', from=2-1, to=5-1]
    \arrow["{T(\eta^T_X + TY)}"', from=5-1, to=5-3]
    \arrow["{T[Tin_1,Tin_2]}"', from=5-3, to=5-5]
    \arrow["{T\eta^T_X + TTY}"', from=2-1, to=3-4]
    \arrow["{[Tin_1,Tin_2]}"'{pos=1}, from=3-4, to=5-3]
    \arrow["{[TTin_1,TTin_2]}"{pos=1}, from=3-4, to=5-5]
    \arrow["{\mu^T_X + \mu^T_Y}"', from=3-4, to=3-7]
    \arrow["{TX + \mu^T_Y}", from=2-1, to=3-7]
  \end{tikzcd}}
  \qedhere
  \end{center}
\end{proof}
\subsection{The family $\widetilde{h^∘}$}
We now turn to defining our next family of morphisms,
$\widetilde{h^∘}$.  For this, we need to exploit the cocontinuity
hypothesis in the first argument of $Γ$.

Because we have assumed $𝐂$ to be locally finitely presentable, each
$Γ_Y$ has a right adjoint, which we denote by $ℸ_Y$.  Furthermore, we
have:
\begin{lemma}\label{lem:circ:times}
  Each $Γ^\pt_Y$ has as right adjoint the functor $ℸ^×_Y$ defined by
  $ℸ^×_Y(X) = ℸ_Y(X) × X$.
\end{lemma}
\begin{proof}
  We have
  \begin{center}
    \hfill $\begin{array}[b]{rcl}
    𝐂(Γ^\pt_Y(X),Z) & = & 𝐂(Γ_Y(X) + X,Z) \\
                    & ≅ & 𝐂(Γ_Y(X),Z) × 𝐂(X,Z) \\
                    & ≅ & 𝐂(X,ℸ_Y Z) × 𝐂(X,Z) \\
                    & ≅ & 𝐂(X,ℸ_Y (Z) × Z) \\
                    & ≅ & 𝐂(X,ℸ^×_Y (Z)).
  \end{array}$ \qedhere
\end{center}
\end{proof}

Furthermore, by~\cite[Theorem~IV.7.3]{MacLane:cwm}, we readily get:
\begin{proposition}\label{prop:paramadj}
  For all bifunctors $L∶ 𝐂² → 𝐂$ such that
  each $L_Y$ has a right adjoint $R_Y$,
  the functors $R_Y$ assemble into a bifunctor $𝐂 × \op{𝐂} → 𝐂$
  making the bijection
  $$𝐂(L_Y(X),Z) ≅ 𝐂(X,R_Y (Z))$$
  natural in all variables.
\end{proposition}
Here, naturality in $Y$ means that, e.g., for all morphisms $f∶ Y → Y'$ and
$u∶ L_{Y'}X → Z$, the transpose of
\begin{center}
  $L_YX \xto{L_f X} L_{Y'} X \xto{u} Z$
  \qquad is \qquad
  $X \xto{\tilde{u}} R_{Y'}Z \xto{R_fZ} R_Y Z\rlap{,}$
\end{center}
where $\tilde{u}$ denotes the transpose of $u$.

\begin{corollary}\label{cor:transpose}
  For all $f∶ Y → Z$, $g∶ B → C$,
  $u∶ A → R_YB$, and $v∶ A → R_ZC$, letting $\tilde{u}$ and
  $\tilde{v}$ denote the transposed morphisms, the square below left
  commutes iff the one below right does.
  \begin{center}
    \diag{%
      A \& R_Y B \\
      R_Z C \& R_Y C %
    }{%
      (m-1-1) edge[labela={u}] (m-1-2) %
      edge[labell={v}] (m-2-1) %
      (m-2-1) edge[labelb={R_f C}] (m-2-2) %
      (m-1-2) edge[labelr={R_Y g}] (m-2-2) %
    }
    \hfil
    \diag{%
    L_Y A \& B \\
    L_Z A \& C %
  }{%
    (m-1-1) edge[labela={\tilde{u}}] (m-1-2) %
    edge[labell={L_f A}] (m-2-1) %
    (m-2-1) edge[labelb={\tilde{v}}] (m-2-2) %
    (m-1-2) edge[labelr={g}] (m-2-2) %
    }%
  \end{center}
  In particular, for all $A$ and $f∶ B→C$,
  the following diagrams commute.
  \begin{center}
    \diag{%
      L_B R_C A \& L_C R_C A \\
      L_B R_B A \& A
    }{%
      (m-1-1) edge[labela={L_f R_C A}] (m-1-2) %
      edge[labell={L_B R_f A}] (m-2-1) %
      (m-2-1) edge[labelb={ε_B A}] (m-2-2) %
      (m-1-2) edge[labelr={ε_C A}] (m-2-2) %
    }
    \ 
    \diag{%
      A \& R_C L_C A \\
      R_B L_B A \& R_B L_C A
    }{%
      (m-1-1) edge[labela={η_C A}] (m-1-2) %
      edge[labell={η_B A}] (m-2-1) %
      (m-2-1) edge[labelb={R_B L_f A}] (m-2-2) %
      (m-1-2) edge[labelr={R_f L_C A}] (m-2-2) %
    }
  \end{center}
\end{corollary}

Let us return to the construction of the next family.
\begin{definition}
  Let $\widetilde{h^∘}_{X,Y}$ denote the
  transpose
  $$\widetilde{h^∘}_{X,Y}∶ ΣX → ℸ^×_{STY} ST O_Y Γ^∘_{STY} X$$
  of $h^∘_{X,Y}$.
\end{definition}
\begin{lemma}\label{lem:nat:hcirctilde}
  The family $\widetilde{h^∘}$ is natural in $X$, i.e., for all $Y$ and
  $f∶ X → Z$, the following square commutes.
  \begin{center}
    \diag(.6,2.7){%
      ΣX \& ΣZ \\
      ℸ^×_{STY} ST O_Y Γ^\pt_{STY} X \&
      ℸ^×_{STY} ST O_Y Γ^\pt_{STY} Z
    }{%
      (m-1-1) edge[labela={Σf}] (m-1-2) %
      edge[labell={\widetilde{h^∘}_{X,Y}}] (m-2-1) %
      (m-2-1) edge[labelb={ℸ^×_{STY} ST O_Y Γ^\pt_{STY} f}] (m-2-2) %
      (m-1-2) edge[labelr={\widetilde{h^∘}_{Z,Y}}] (m-2-2) %
    }
  \end{center}
\end{lemma}
\begin{proof}
  By naturality and bijectivity of transposition.
\end{proof}

\begin{lemma}\label{lem:hcirctilde:extranatural}
  The family $\widetilde{h^∘}$ is extranatural~\cite{Kelly1982Basic}
  in $Y$, i.e., for all $X$ and $f∶ Y → Z$, the following square commutes.
  \begin{center}
    \diag(.6,2.7){%
      ΣX \& ℸ^×_{STY} ST O_Y Γ^∘_{STY} X \\
      ℸ^×_{STZ} ST O_{Z} Γ^∘_{STZ} X \& ℸ^×_{STY} ST O_{Z} Γ^∘_{STZ} X
    }{%
      (m-1-1) edge[labela={\widetilde{h^∘}_{X,Y}}] (m-1-2) %
      edge[labell={\widetilde{h^∘}_{X,Z}}] (m-2-1) %
      (m-2-1) edge[labelb={ℸ^×_{STf} ST O_{Z} Γ^∘_{STZ} X}] (m-2-2) %
      (m-1-2) edge[labelr={ℸ^×_{STY} ST O_{f} Γ^∘_{STf} X}] (m-2-2) %
    }
  \end{center}
\end{lemma}
\begin{proof}
By naturality of $h^∘$ and Corollary~\ref{cor:transpose}.
\end{proof}

\begin{lemma}\label{lem:coh:hcirctilde}
  The family $\widetilde{h^∘}$ satisfies the following laws.
  \begin{center}
    \diag{%
      ΣX \& STX \\
      ℸ^×_{STY} ST Γ^•_{ST|Y} X  \& ST Γ^•_{ST|Y} X %
    }{%
    (m-1-1) edge[labela={η_{Σ,ηᵀ_X}}] (m-1-2) %
    edge[labell={\widetilde{h^∘}_{X,Y}}] (m-2-1) %
    (m-2-1) edge[labelb={π₂}] (m-2-2) %
    (m-1-2) edge[labelr={ST  α^{Γ^•_{ST|}}_{X,Y}}] (m-2-2) %
  }

  \diag(.6,1.6){%
    ΣX \& \& ℸ^×_{STY} ST Γ^•_{ST|Y} X \\
    \& \& ℸ^×_{SSTY} ST Γ^•_{STY} X \\ 
    ℸ^×_{STSY} ST Γ^•_{ST|SY} X \&
    ℸ^×_{STSY} STS Γ^•_{STS|Y} X \&
            ℸ^×_{STSY} ST Γ^•_{ST|Y} X %
          }{%
            (m-1-1) edge[labela={\widetilde{h^∘}_{X,Y}}] (m-1-3) %
            edge[labell={\widetilde{h^∘}_{X,SY}}] (m-3-1) %
            (m-3-1) edge[loinb={ℸ^×_{STSY} ST ↑^{S}_{X,Y}}] (m-3-2) %
            (m-3-2) edge[loinb={ℸ^×_{STSY} (Sδ;μ^S T) Γ^•_{S δ; μ^S T|Y}X}] (m-3-3) %
            (m-1-3) edge[labelr={ℸ^×_{μ^S_{TY}} ST Γ^•_{STY} X}] (m-2-3)
            (m-2-3) edge[labelr={ℸ^×_{Sδ_Y} ST Γ^•_{STY} X}] (m-3-3) %
  }

\begin{tikzcd}[ampersand replacement=\&]
	{\Sigma X} \&\& {\daleth^\times_{STY} ST \Gamma^\bullet_{ST|Y}X} \\
	\&\& {\daleth_{STY} ST \Gamma^\bullet_{ST|Y}X} \\
	\&\& {\daleth_Y ST \Gamma^\bullet_{ST|Y}X}
	\arrow["{\widetilde{h}_{X,Y}}"', from=1-1, to=3-3]
	\arrow["{\widetilde{h^\circ}_{X,Y}}", from=1-1, to=1-3]
	\arrow["{\pi_1}", from=1-3, to=2-3]
	\arrow["{\daleth_{\eta^{ST}_Y}ST \Gamma^\bullet_{ST|Y}X}"{pos=0.4}, from=2-3, to=3-3]
\end{tikzcd}      

\ajustedroit{
\begin{tikzcd}[ampersand replacement=\&]
	{\Sigma A} \&\&\&\&\& {\daleth^\times_{STY} ST O_{Y} \Gamma^\circ_{STY}A} \\
	{\daleth^\times_{STTY} ST O_{TY} \Gamma^\circ_{STTY}A} \&\& {\daleth^\times_{STTY} STT O_{Y} \Gamma^\circ_{STTY}A} \&\&\& {\daleth^\times_{STTY} ST O_{Y} \Gamma^\circ_{STY}A}
	\arrow["{\widetilde{h^\circ}_{\daleth^\times_{STY} ST \Gamma^\bullet_{ST|Y} X,Y}}", from=1-1, to=1-6]
	\arrow["{\daleth^\times_{S\mu^T_Y} ST O_{Y} \Gamma^\circ_{STY}A}", from=1-6, to=2-6]
	\arrow["{\widetilde{h^\circ}_{\daleth^\times_{STY} ST \Gamma^\bullet_{ST|Y} X,TY}}"', from=1-1, to=2-1]
	\arrow["{\daleth^\times_{STTY} ST\uparrow^T}"', from=2-1, to=2-3]
	\arrow["{\daleth^\times_{STTY} S\mu^T O_{Y} \Gamma^\circ_{S\mu^T_Y}A}"', from=2-3, to=2-6]
\end{tikzcd}}
  \end{center}
\end{lemma}
\begin{proof}
  The first diagram holds by Lemma~\ref{lem:coh:hcirc} and
  construction of $ℸ^×_{STY}$.  The second one holds by
  Lemma~\ref{lem:coh:hcirc} and Corollary~\ref{cor:transpose}.
  The third diagram holds by Lemma~\ref{lem:coh:hcirc}, observing
  that by construction of the adjunction of Lemma~\ref{lem:circ:times},
  $π₁ ∘ \widetilde{h^\circ}_{X,Y}$ is the transpose of
  $$Γ_{STY}Σ X \xto{in₁} Γ^\pt_{STY}Σ X \xto{h^∘_{X,Y}} ST Γ^•_{STY}X\rlap{,}$$
  hence, by Proposition~\ref{prop:paramadj}, $ℸ_{η^{ST}_Y}ST \Gamma^\bullet_{ST|Y}X ∘ π₁ ∘ \widetilde{h^∘}_{X,Y}$
  is the transpose of
  \begin{center}
    $Γ_YΣ X \xto{Γ_{η^{ST}_Y} Σ X} Γ_{STY}Σ X \xto{in₁} Γ^\pt_{STY}Σ X \xto{h^∘_{X,Y}} ST Γ^•_{STY}X\rlap{.}$
  \end{center}
  Finally, the last diagram is precisely the transpose of the last
  diagram of Lemma~\ref{lem:coh:hcirc}.
\end{proof}

\subsection{The family $\widetilde{h^∘}^ω$}
We now want to define the next family, $\widetilde{h^∘}^ω$.  For this, we
start by observing the following: We have:
\begin{lemma}
  For all $Y$, the composite functor $ℸ^×_{STY} ST O_Y Γ^∘_{STY}$
  is a monad.
\end{lemma}
\begin{proof}
  We already know that $ST$ is a monad.  Furthermore, precomposing any
  monad by any option functor $O_Z$ again yields a monad, by the
  following (folklore) lemma.  We conclude by the adjunction
  $ℸ^×_{STY} ⊢ Γ^∘_{STY}$.
\end{proof}

Furthermore, the following result is folklore.
\begin{lemma}\label{lem:refOmonad}
  For all monads $T$ on a category with binary coproducts, and for all
  objects $E$ therein, the natural transformation
  $\liftleft{T}_E∶ O_E T → T O_E$ defined at any $X$ by
  $$TX + E \xto{[TX,ηᵀ_E]} T(X) + T(E) \xto{[T(in₁),T(in₂)]}
  T(X+E)$$ forms a monad distributive law, thus equipping
  $T O_E$ with monad structure.
\end{lemma}

We may now define $\widetilde{h^∘}^ω$.
\begin{definition}
  For all $X,Y ∈ 𝐂$, let
  $\widetilde{h^∘}^ω_{X,Y}∶ SX → ℸ^×_{STY} ST O_Y Γ^∘_{STY} X$ denote
  the unique monad morphism
  making the triangle
  \begin{center}
    \begin{tikzcd}[ampersand replacement=\&]
      {\Sigma X} \& SX \\
      \& {\daleth^\times_{STY} ST \Gamma^\bullet_{ST|Y}X}
      \arrow["{\eta_{\Sigma,X}}", from=1-1, to=1-2]
      \arrow["{\widetilde{h^\circ}^\omega_{X,Y}}", from=1-2, to=2-2]
      \arrow["{\widetilde{h^\circ}_{X,Y}}"'{pos=0.3}, from=1-1, to=2-2]
    \end{tikzcd}
  \end{center}
  commute, obtained by universal property of $S$.
\end{definition}
By construction, we have:
\begin{lemma}
  The family $\widetilde{h^∘}^ω_{X,Y}$ is natural in $X$.
\end{lemma}

\begin{notation}\label{not:LKR}
  Let $L = Γ^\pt_{ST}∶ 𝐂² → 𝐂$, $R = ℸ^×_{ST}∶ 𝐂×\op{𝐂} → 𝐂$, and
  $K_Y(X) = ST O_Y X$ (hence $K∶ 𝐂² → 𝐂$). We thus have for all $Y$
  that $L_Y$ is left adjoint to $R_Y$, and that $K_Y$ is a monad.
\end{notation}

Furthermore, we observe the following:
    \begin{lemma}\label{lem:RB:liftsto:KBalg}
      For any $B ∈ 𝐂$ and $A ∈ K_B\alg$, $R_B A$ has a canonical
      $Σ$-algebra structure given by
      $$Σ R_B A 
      \xto{\widetilde{h^∘}_{R_B A, B}} R_B K_B L_B R_B A \xto{R_B K_B ε_B
        A} R_B K_B A → R_B A.$$ This defines a functor $K_B\alg → Σ\alg$
      over $R_B$, for all $B$. 
    \end{lemma}
    \begin{proof}
      Functoriality is straightforward.
    \end{proof}
    Unfolding the definition, the family $\widetilde{h^∘}^ω_{X,Y}$ is
    in fact constructed in this way, with $A = K_Y L_Y X$, whose
    $K_Y$-algebra structure is given by $μ^{K_Y}_{L_Y X}$.


    Let us prove that the functorial assignment of
    Lemma~\ref{lem:RB:liftsto:KBalg} is in fact also functorial in
    $B$, in the following sense.
\begin{lemma}\label{lem:aux:hcirctilde}
  For any morphism $f∶ A → B$ and object $Z$, equipping $K_B Z$ with
  the $K_A$-algebra structure
  $K_A K_B Z \xto{K_f K_B Z} K_B K_B Z \xto{\mu^{K_B}_Z} K_B
  Z\rlap{,}$ the morphism
  $$R_{B} K_B Z \xto{R_f K_B Z} R_{A} K_B Z$$
  is a $Σ$-algebra morphism.
\end{lemma}
\begin{proof}
  By commutativity of the following diagram
  \begin{center}\hfill
    \Diag|baseline=(m-5-1.base)|(1,2.1){%
      \justify{m-4-1}{interchange}{m-5-3} %
      \justify{m-2-1}{Lemma~\ref{lem:hcirctilde:extranatural}}{m-3-2} %
      \justify{m-2-1}{Lemma~\ref{lem:nat:hcirctilde}}{m-1-3} %
      \justify{m-4-1}{interchange}{m-3-2} %
      \justify{m-3-2}{Corollary~\ref{cor:transpose}}{m-2-3} %
    }{%
      \& \& Σ R_{A} K_B Z \\
      Σ R_{B} K_B Z \& R_{A} K_A L_{A} R_{B} K_B Z
      \& R_{A} K_A L_{A} R_{A} K_B Z \\
      R_{B} K_B L_{B} R_{B} K_B Z \&
      R_{A} K_B L_{B} R_{B} K_B Z \&
      R_{A} K_B K_B Z \\
      R_{B} K_B K_B Z \& \& \\
      R_{B} K_B Z \& \& R_{A} K_B Z %
    }{%
      (m-2-1) edge[labela={\widetilde{h^∘}_{R_{B} K_B Z,A}}] (m-2-2) %
      edge[bend left=10,labelal={Σ R_{f} K_B Z}] (m-1-3) %
      (m-2-2) edge[label={[above=.5em]{$\scriptstyle R_{A} K_A L_{B} R_{f} K_B Z$}}]
      (m-2-3) %
      (m-1-3) edge[labelon={\widetilde{h^∘}_{R_{A} K_B Z,A}}] (m-2-3) %
      (m-2-1) edge[labelon={\widetilde{h^∘}_{R_{B} K_B Z,B}}] (m-3-1) %
      (m-3-1) edge[label={[below=.4em]{$\scriptstyle R_f K_B L_{B} R_{B} K_B Z$}}] (m-3-2) %
      (m-2-2) edge[labelon={R_{A} K_f L_f R_{B} K_B Z}] (m-3-2) %
      (m-3-2) edge[labela={R_{A} K_B ε_{B} K_B Z}] (m-3-3) %
      (m-2-3) edge[labelon={R_{A} K_f ε_{A} K_B Z}] (m-3-3) %
      (m-3-1) edge[labelon={R_{B} K_B ε_{B}  K_B Z}] (m-4-1) %
      (m-4-1) edge[labelbr={R_f K_B K_B Z}] (m-3-3) %
      edge[labelon={R_{B} μ^{K_A}_Z}] (m-5-1) %
      (m-3-3) edge[labelon={R_{A} μ^{K_B}_Z}] (m-5-3) %
      (m-5-1) edge[labelb={R_f K_B Z}] (m-5-3) %
    }
    \qedhere
  \end{center}
\end{proof}

Moreover, we have:
\begin{lemma}\label{lem:hcirctildeomega:extranatural}
  The family $\widetilde{h^∘}^ω_{X,Y}$ is extranatural in $Y$, i.e.,
  the following square commutes for all $X$ and $f∶ Y → Z$.
  \begin{center}
    \diag(.6,2.7){%
      SX \& ℸ^×_{STY} ST O_Y Γ^∘_{STY} X \\
      ℸ^×_{STZ} ST O_{Z} Γ^∘_{STZ} X \& ℸ^×_{STY} ST O_{Z} Γ^∘_{STZ} X
    }{%
      (m-1-1) edge[labela={\widetilde{h^∘}^ω_{X,Y}}] (m-1-2) %
      edge[labell={\widetilde{h^∘}^ω_{X,Z}}] (m-2-1) %
      (m-2-1) edge[labelb={ℸ^×_{STf} ST O_{Z} Γ^∘_{STZ} X}] (m-2-2) %
      (m-1-2) edge[labelr={ℸ^×_{STY} ST O_{f} Γ^∘_{STf} X}] (m-2-2) %
    }
  \end{center}
\end{lemma}

For proving this, we will need the following result.

\begin{lemma}\label{lem:liftLR}
  For any monad $S$ on a category with binary coproducts, the following diagram
  commutes for all objects $C$ and $D$.
  \begin{center}
\begin{tikzcd}[ampersand replacement=\&]
	\& {SA+SB} \\
	{S(SA+B)} \&\& {S(A+SB)} \\
	{SS(A+B)} \&\& {SS(A+B)} \\
	\& {S(A+B)}
	\arrow["{\mu^S}", from=3-3, to=4-2]
	\arrow["{\mu^S}"', from=3-1, to=4-2]
	\arrow["{\uparrow^S}"', from=1-2, to=2-1]
	\arrow["{S({{}^{S}\uparrow})}"', from=2-1, to=3-1]
	\arrow["{{}^{S}\uparrow}", from=1-2, to=2-3]
	\arrow["{S \uparrow^S}", from=2-3, to=3-3]
\end{tikzcd}
  \end{center}
\end{lemma}
\begin{proof}
By diagram chasing, termwise, but doing only one term as the other follows by symmetry.
\begin{center}
  \ajustedroit[.95]{
\begin{tikzcd}[ampersand replacement=\&,baseline=(\tikzcdmatrixname-5-1.base)]
	SA \& {SA+SB} \&\&\&\& {SA+SSB} \\
	\&\&\&\&\& {S(A+SB)} \\
	{SA+SB} \& SSA \&\&\& SSA \& {S(SA+SB)} \\
	{SSA+SB} \&\&\&\& SA \& {SS(A+B)} \\
	{S(SA+B)} \&\& {S(SA+SB)} \&\& {SS(A+B)} \& {S(A+B)}
	\arrow["{[S in_1,  S in_2]}", from=1-6, to=2-6]
	\arrow["{S(\eta^S_A + SB)}", from=2-6, to=3-6]
	\arrow["{S[Sin_1, S in_2]}", from=3-6, to=4-6]
	\arrow["{\mu^S}", from=4-6, to=5-6]
	\arrow["{\eta^S_{SA} + SB}"', from=3-1, to=4-1]
	\arrow["{[S in_1, S in_2]}"', from=4-1, to=5-1]
	\arrow["{S(SA + \eta^S_B)}"', from=5-1, to=5-3]
	\arrow["{S[S in_1, S in_2]}"', from=5-3, to=5-5]
	\arrow["{\mu^S}"', from=5-5, to=5-6]
	\arrow["{\eta^S}", from=1-1, to=3-2]
	\arrow["{\mu^S}", from=3-2, to=4-5]
	\arrow["{Sin_1}", from=4-5, to=5-6]
	\arrow[Rightarrow, no head, from=1-1, to=4-5]
	\arrow["{S\eta^S}", from=1-1, to=3-5]
	\arrow["{S in_1}", from=1-1, to=2-6]
	\arrow["{SSin_1}", from=3-5, to=4-6]
	\arrow["{\mu^S}", from=3-5, to=4-5]
	\arrow["{in_1}", from=1-1, to=1-2]
	\arrow["{in_1}"', from=1-1, to=3-1]
	\arrow["{SA + \eta^S_{SB}}", from=1-2, to=1-6]
	\arrow["{S in_1}"', from=3-2, to=5-3]
	\arrow["{SSin_1}"{pos=0.8}, from=3-2, to=5-5]
      \end{tikzcd}
    } \qedhere
    \end{center}
\end{proof}
\begin{proof}[Proof of Lemma~\ref{lem:hcirctildeomega:extranatural}]
  Using the notation, we want to prove that the following square commutes.
  \begin{center}
    \diag(.6,2){%
      SX \& R_Y K_Y L_Y X \\
      R_Z K_Z L_Z X \& R_Y K_Z L_Z X
    }{%
      (m-1-1) edge[labela={\widetilde{h^∘}^ω_{X,Y}}] (m-1-2) %
      edge[labell={\widetilde{h^∘}^ω_{X,Z}}] (m-2-1) %
      (m-2-1) edge[labelb={R_{f} K_Z L_Z X}] (m-2-2) %
      (m-1-2) edge[labelr={R_Y K_{f} L_{f} X}] (m-2-2) %
    }
  \end{center}

  We first observe that both morphisms have the same restriction to
  $ΣX$, by definition of $\widetilde{h^∘}^ω$ and
  Lemma~\ref{lem:hcirctilde:extranatural}.  Thus, by
  \begin{itemize}
  \item universal property of
    $SX$ as a free $Σ$-algebra over $X$, and
  \item the fact that for any morphism $α∶ T₁ → T₂$ of monads, any
    $α_X∶ T₁X → T₂X$ is a $T₁$-algebra morphism, equipping
    $T₂X$ with the $T₁$-algebra structure given by
    $$T₁T₂X \xto{α_{T₂X}} T₂T₂X \xto{μ^{T₂}X} T₂X\rlap{,}$$
    so that each $\widetilde{h^∘}^ω_{X,Y}$ is an $S$-algebra morphism,
    or equivalently a $Σ$-algebra morphism,
  \end{itemize}
  it suffices to equip $R_Y K_Z L_Z X$ with $Σ$-algebra structure and
  show that the bottom and right morphisms above, i.e.,
  $R_f K_Z L_Z X$ and $R_Y K_f L_f X$, are $Σ$-algebra morphisms.  For
  the $Σ$-algebra structure on $R_Y K_Z L_Z X$, we apply
  Lemma~\ref{lem:RB:liftsto:KBalg} with the following $K_Y$-algebra
  structure on $K_Z L_Z X$:
  $$K_Y K_Z L_Z X \xto{K_f K_Z L_Z X} K_Z K_Z L_Z X \xto{μ^{K_Z}} K_Z L_Z X.$$
  The bottom morphism then lifts to $Σ\alg$ by
  Lemma~\ref{lem:aux:hcirctilde}. Finally, the fact that the
  right-hand morphism $R_Y K_f L_f X$ is a $Σ$-algebra morphism
  will thus follow from $K_f L_f X$ being a $K_Y$-algebra morphism,
  which in turn follows by chasing the following diagram.
  \begin{center}
    \hfill
    \begin{tikzcd}[ampersand replacement=\&,baseline=(\tikzcdmatrixname-3-1.base)]
      {K_Y K_Y L_Y X} \&\&\&\& {K_Y K_Z L_Z X} \\
      \&\&\&\& {K_Z K_Z L_Z X} \\
      {K_Y L_Y X} \&\& {K_Z L_Y X} \&\& {K_Z L_Z X}
      \arrow["{\mu^{K_Y}}", from=1-1, to=3-1]
      \arrow["{K_f L_Y X}", from=3-1, to=3-3]
      \arrow["{K_Z L_f X}", from=3-3, to=3-5]
      \arrow["{K_f K_Z L_Z X}", from=1-5, to=2-5]
      \arrow["{\mu^{K_Z}}", from=2-5, to=3-5]
      \arrow["{K_Y K_f L_f X}", from=1-1, to=1-5]
    \end{tikzcd} \qedhere
  \end{center}
\end{proof}

\begin{lemma}\label{lem:coh:hcirctildeomega}
  The family $\widetilde{h^∘}^ω$ satisfies the following laws.
  \begin{center}
    \diag{%
      SX \& STX \\
      ℸ^×_{STY} ST O_Y Γ^\pt_{STY} X  \& ST O_Y Γ^\pt_{STY} X %
    }{%
    (m-1-1) edge[labela={Sηᵀ_X}] (m-1-2) %
    edge[labell={\widetilde{h^∘}^ω_{X,Y}}] (m-2-1) %
    (m-2-1) edge[labelb={π₂}] (m-2-2) %
    (m-1-2) edge[labelr={ST  α^{Γ^•_{ST|}}_{X,Y}}] (m-2-2) %
  }

  \diag(.6,2.2){%
    SX \& \& ℸ^×_{STY} ST Γ^•_{ST|Y} X \\
    \& \& ℸ^×_{SSTY} ST Γ^•_{ST|Y} X \\
    ℸ^×_{STSY} ST Γ^•_{ST|SY} X \&
    ℸ^×_{STSY} STS Γ^•_{STS|Y} X \&
            ℸ^×_{STSY} ST Γ^•_{ST|Y} X %
          }{%
            (m-1-1) edge[labela={\widetilde{h^∘}^ω_{X,Y}}] (m-1-3) %
            edge[labell={\widetilde{h^∘}^ω_{X,SY}}] (m-3-1) %
            (m-3-1) edge[labelb={ℸ^×_{STSY} ↑^{S}_{X,Y}}] (m-3-2) %
            (m-3-2) edge[loinb={ℸ^×_{STSY} (Sδ;μ^S T) Γ^•_{S δ;μ^S T |Y}X}] (m-3-3) %
            (m-1-3) edge[labelr={ℸ^×_{μ^S_{TY}} ST Γ^•_{ST|Y} X}] (m-2-3)
            (m-2-3) edge[labelr={ℸ^×_{Sδ_Y} ST Γ^•_{ST|Y} X}] (m-3-3)
  }

\begin{tikzcd}[ampersand replacement=\&]
	{\Sigma X} \&\& SX \\
	\&\& {\daleth^\times_{STY} ST \Gamma^\bullet_{ST|Y}X} \\
	\&\& {\daleth_{STY} ST \Gamma^\bullet_{ST|Y}X} \\
	\&\& {\daleth_Y ST \Gamma^\bullet_{ST|Y}X}
	\arrow["{\widetilde{h}_{X,Y}}"', from=1-1, to=4-3]
	\arrow["{\pi_1}", from=2-3, to=3-3]
	\arrow["{\daleth_{\eta^{ST}_Y}ST \Gamma^\bullet_{ST|Y}X}"{pos=0.4}, from=3-3, to=4-3]
	\arrow["{\widetilde{h^\circ}^\omega_{X,Y}}", from=1-3, to=2-3]
	\arrow["{\eta_{\Sigma,X}}", from=1-1, to=1-3]
      \end{tikzcd}

\begin{tikzcd}[ampersand replacement=\&]
	SX \&\&\&\& {\daleth^\times_{STTY}ST\Gamma^\bullet_{ST|TY}X} \\
	\&\&\&\& {\daleth^\times_{STTY}STT\Gamma^\bullet_{STT|Y}X} \\
	{\daleth^\times_{STY}ST\Gamma^\bullet_{ST|Y}X} \&\&\&\& {\daleth^\times_{STTY}ST\Gamma^\bullet_{ST|Y}X}
	\arrow["{{\widetilde{h^\circ}^\omega}_{X,TY}}", from=1-1, to=1-5]
	\arrow["{\daleth^\times_{STTY}ST \uparrow^T}", from=1-5, to=2-5]
	\arrow["{\daleth^\times_{STTY}S\mu^T \Gamma^\bullet_{S\mu^T|Y}X}", from=2-5, to=3-5]
	\arrow["{{\widetilde{h^\circ}^\omega}_{X,Y}}"', from=1-1, to=3-1]
	\arrow["{\daleth^\times_{S\mu^T_{Y}}ST \Gamma^\bullet_{ST|Y} X}"', from=3-1, to=3-5]
\end{tikzcd}
    \end{center}
\end{lemma}
\begin{proof}
  The third law is direct by Lemma~\ref{lem:coh:hcirctilde}. For the
  first two, we proceed similarly to the proof of
  Lemma~\ref{lem:hcirctildeomega:extranatural}, by showing that all
  the needed morphisms are $Σ$-algebra morphisms.  For the first
  diagram, the top and right-hand morphisms are free $Σ$-algebra
  morphisms, so we reduce to proving that $π₂$ is a $Σ$-algebra
  morphism. This holds by commutativity of the diagram in
  Figure~\ref{fig:coh:hcirctildeomega:i} (recalling
  Notation~\ref{not:LKR}),
\begin{figure}[htp]
  \ajustetourne{
  \Diag(1.5,1.2){%
    \justify[.4]{m-3-2}{Lemma~\ref{lem:coh:hcirctilde}}{m-1-3} %
    \justify[.5]{m-4-1}{naturality of $π₂$}{m-3-2} %
    \justify[.5]{m-4-4}{dist. law axiom}{mimi} %
    \justify[.3]{mimu}{monad law}{m-2-5} %
    \justify[.3]{m-4-4}{definition of $μ^{ST}$}{m-5-3} %
    \justify[.3]{m-1-3}{interchange}{m-2-5} %
    \justify[.5]{m-4-2}{defn of $μ^{K_Y}$}{midmust} %
    \justify[.5]{m-3-2}{?}{m-3-3} %
  }{%
      \&\& Σ R_Y K_Y L_Y X \& \& Σ K_Y L_Y X \\
      \& R_Y K_Y L_Y R_Y K_Y L_Y X
      \& ST R_Y K_Y L_Y X
      \& \& S K_Y L_Y X = S ST O_Y L_Y X \\
      R_Y K_Y K_Y  L_Y X \&
      K_Y  L_Y R_Y K_Y  L_Y X
         \& ST K_Y L_Y X \& \\
         R_Y K_Y L_Y X \&  K_Y K_Y L_Y X \& \&
         SSTT O_Y L_Y X \& SST O_Y L_Y X \\
            \& \& K_Y L_Y X  %
    }{%
      (m-1-3) edge[labelal={\widetilde{h^∘}_{R_Y K_Y L_Y X,Y}}] (m-2-2) %
      edge[labelon={η_Σ ηᵀ R_Y K_Y L_Y X}] (m-2-3) %
      (m-2-2) edge[labell={π₂}] (m-3-2) %
      (m-2-3) edge[labelon={ST α^{O_Y L_Y} R_Y K_Y L_Y X}]
         (m-3-2) %
      (m-2-2) edge[labelal={R_Y K_Y ε_Y K_Y L_Y X}] (m-3-1) %
      (m-3-1) edge[labell={R_Y μ^{K_Y} L_Y X}] (m-4-1) %
      %
      (m-1-3) edge[labela={Σ π₂}] (m-1-5)%
      (m-1-5) edge[labelon={η_Σ K_Y L_Y X}] (m-2-5) %
      (m-2-5) edge[labelal={S ηᵀ K_Y L_Y X}] %
         node[coordinate](mimi){} (m-3-3)
      (m-2-3) edge[labelon={ST π₂}] (m-3-3) %
      (m-3-2) edge[labelon={K_Y ε_Y K_Y L_Y X}] (m-4-2) %
      (m-4-2) edge[labelon={μ^{K_Y}L_Y X}] (m-5-3) %
      (m-3-3) edge[labelon={S δ_{T O_Y L_Y X}}] (m-4-4) %
      (m-4-4) edge[labela={SS μᵀ_{O_Y L_Y X}}] %
         node[coordinate](mimu){} (m-4-5) %
      (m-2-5) edge[labelon={SSηᵀ TO_YL_Y X}] (m-4-4) %
      (m-2-5) edge[identity] (m-4-5) %
      (m-4-1) edge[labelbl={π₂}] (m-5-3) %
      (m-4-5) edge[bend left=10,labelbr={μ^S TO_YL_YX}] (m-5-3) %
      (m-3-3) edge[labelon={μ^{ST} O_Y L_Y X}] %
      node[coordinate](midmust){} (m-5-3) %
      edge[labelon={ST in₁}] (m-4-2) %
    }
    }
    \caption{Diagram for proof of Lemma~\ref{lem:coh:hcirctildeomega},
      first statement}
    \label{fig:coh:hcirctildeomega:i}
    \end{figure}
    where the question-marked polygon commutes because the following
    diagram does, for all $A$.
    \begin{center}
      \diag(.6,2){%
        R_Y A \& A \\
        L_Y R_Y A \\
        O_Y L_Y R_Y A \& O_Y A
      }{%
        (m-1-1) edge[labela={π₂}] (m-1-2) %
        edge[labelon={in₂}] (m-2-1) %
        edge[bend right,labell={α^{O_Y L_Y}_{R_YA}}, %
        out=-90,in=-90] (m-3-1)
        (m-3-1) edge[labelb={O_Y ε_Y A}] (m-3-2) %
        (m-1-2) edge[labelr={in₁}] (m-3-2) %
        (m-2-1) edge[labelbr={ε_Y A}] (m-1-2) %
        edge[labelon={in₁}] (m-3-1) %
      }
    \end{center}
    The top triangle commutes by construction of
    $$ε_YA∶ Γ^\pt_{STY} ℸ^×_{STY} A → A\rlap{:}$$
    the domain
    $Γ^\pt_{STY} ℸ^×_{STY} A$ is the coproduct
    $$Γ_{STY} (ℸ_{STY}(A) × A) \quad + \quad ℸ_{STY}(A) × A\rlap{,}$$
    and the second component of $ε_YA$
    is
    $π₂∶ ℸ_{STY}(A)×A → A.$

    Let us now turn to the second diagram. We start by equipping the
    codomain with $Σ$-algebra structure, by applying
    Lemma~\ref{lem:RB:liftsto:KBalg} with $B = SY$ and
    $A = ST Γ^•_{ST|Y}X = K_Y Γ^\pt_{STY}X$, whose $K_{SY}$-algebra
    structure is given by either of the following three equal
    composites.
    \begin{center}
\begin{tikzcd}[ampersand replacement=\&]
	{K_{SY}K_Y L_Y X} \&\& {K_{SY}K_Y L_Y X} \&\& {K_{SY} K_Y L_Y X} \\
	{ST (ST O_Y L_Y X + SY)} \&\& {ST (ST O_Y L_Y X + SY)} \&\& {ST (ST O_Y L_Y X + SY)} \\
	{STSO_Y ST O_YL_Y X} \&\& {ST ST (O_Y L_Y X + SY)} \&\& {ST ST (O_Y L_Y X + SY)} \\
	{STO_Y ST O_YL_Y X} \&\& {ST ST S(O_Y L_Y X + Y)} \&\& {ST(O_Y L_Y X + SY)} \\
	{ST O_YL_Y X} \&\& {STTS(L_Y X + Y + Y)} \&\& {STS(L_YX + Y + Y)} \\
	{K_YL_YX} \&\& {STS(L_Y X + Y + Y)} \&\& {ST(L_YX + Y)} \\
	\&\& {SST(L_YX + Y)} \&\& {K_YL_YX} \\
	\&\& {ST(L_YX + Y)} \\
	\&\& {K_Y L_Y X}
	\arrow[Rightarrow, no head, from=1-3, to=2-3]
	\arrow["{ST( {}^{ST}\uparrow)}", from=2-3, to=3-3]
	\arrow["{STST \uparrow^S}", from=3-3, to=4-3]
	\arrow["{S\delta ; \mu^S}", from=4-3, to=5-3]
	\arrow["{S\mu^T}", from=5-3, to=6-3]
	\arrow[Rightarrow, no head, from=1-1, to=2-1]
	\arrow["{ST \uparrow^S}", from=2-1, to=3-1]
	\arrow["{(S \delta ; \mu^S T)}", from=3-1, to=4-1]
	\arrow["{\mu^{K_Y}}", from=4-1, to=5-1]
	\arrow[Rightarrow, no head, from=5-1, to=6-1]
	\arrow[Rightarrow, no head, from=1-5, to=2-5]
	\arrow["{ST( {}^{ST}\uparrow)}", from=2-5, to=3-5]
	\arrow["{\mu^{ST}}", from=3-5, to=4-5]
	\arrow["{ST \uparrow^S}", from=4-5, to=5-5]
	\arrow["{(S \delta ; \mu^S T)(L_Y X + [Y,Y]}", from=5-5, to=6-5]
	\arrow[Rightarrow, no head, from=6-5, to=7-5]
	\arrow["{S\delta(L_YX + [Y,Y])}", from=6-3, to=7-3]
	\arrow[Rightarrow, no head, from=8-3, to=9-3]
	\arrow["{\mu^S}", from=7-3, to=8-3]
      \end{tikzcd}
    \end{center}

  \begin{lemma}\label{lem:deux:KSY:alg:struct}
    The above $K_{SY}$-algebra structures are indeed equal.
  \end{lemma}
  \begin{proof}
    The latter equivalence is clear. For the former, we proceed by
    diagram chasing as in Figure~\ref{fig:deux:KSY:alg:struct}, where
    $A = L_YX$.
    \begin{figure}[htp]
      \ajustetourne{
\begin{tikzcd}[ampersand replacement=\&]
	{ ST (  ST (A +Y) + SY)} \&\&\&\& { ST (  ST (A +Y) + STSY)} \&\&\& { ST   ST(A+Y+SY)} \&\&\&\& { ST   ST(S(A+Y)+SY)} \\
	\\
	\&\&\&\&\&\&\&\& { ST(A+Y+SY)} \&\&\& { ST   STS(A+Y+Y)} \\
	\&\& { ST (  ST (A +Y) + STY)} \&\& { ST (  ST (A +Y) + STSTY)} \&\& { ST ST (A +Y + STY)} \& {ST( A +Y + STY)} \&\& {ST(S(A+Y)+SY)} \&\& { S STTS(A+Y+Y)} \\
	{ ST (  SST (A+Y) + SY)} \&\& { ST (  STST (A +Y) + STY)} \&\&\&\&\&\& {ST( ST(A+Y) + STY)} \&\& {STS(A+Y+Y)} \& { S TTS(A+Y+Y)} \\
	{ ST S( ST (A+Y)+Y)} \&\& { ST ST(ST (A +Y) + Y)} \&\& {\text{(Lemma~\ref{lem:liftLR})}} \&\&\&\& {STST( A+Y + Y)} \&\&\& { ST   S(A + Y + Y)} \\
	{ SST ( ST (A+Y)+Y)} \&\&\&\&\&\&\&\&\&\&\& { S ST(A + Y)} \\
	{ ST(  ST (A+Y) + Y)} \&\&\&\& { ST(ST (A +Y) + STY)} \&\& { STST( A +Y + Y)} \&\& {ST( A +Y + Y)} \&\&\& { ST (A+Y)}
	\arrow["{ S \delta(A + [Y,Y])}", from=6-12, to=7-12]
	\arrow["{ \mu^ST (A + Y)}", from=7-12, to=8-12]
	\arrow["S\delta"', from=6-1, to=7-1]
	\arrow["{ \mu^S T ( ST (A+Y)+Y)}"', from=7-1, to=8-1]
	\arrow["{ ST [STin_1,STin_2]}", from=1-5, to=1-8]
	\arrow["{STST(\eta^S_{A+Y}+SY)}", from=1-8, to=1-12]
	\arrow["S\delta", from=3-12, to=4-12]
	\arrow["{ ST[Sin_1,Sin_2]}"', from=5-1, to=6-1]
	\arrow["{ST(\eta^S_{ST(A+Y)}+SY)}"', from=1-1, to=5-1]
	\arrow["{\mu^S}", from=4-12, to=5-12]
	\arrow["{S \mu^T}", from=5-12, to=6-12]
	\arrow["{ST(ST(A+Y) + S \eta^T_Y)}", from=1-1, to=4-3]
	\arrow["{ST(η^{ST} + STY)}", from=4-3, to=5-3]
	\arrow["{ST[ S T in₁, S T in₂ ]}", from=5-3, to=6-3]
	\arrow["{μ^{ST}}", from=6-3, to=8-1]
	\arrow["{ ST(ST (A +Y) + η^{ST}_Y)}"', from=8-1, to=8-5]
	\arrow["{ST [ S T in₁, S T in₂ ]}"', from=8-5, to=8-7]
	\arrow["{μ^{ST}}"', from=8-7, to=8-9]
	\arrow["{ST(ST(A+Y) + η^{ST}_Y)}", from=4-3, to=4-5]
	\arrow["{ST [ S T in₁, S T in₂ ]}", from=4-5, to=4-7]
	\arrow["{μ^{ST}}", from=4-7, to=4-8]
	\arrow["{ST( η^{ST} + STY)}"'{pos=0.1}, from=4-8, to=5-9]
	\arrow["{ST [ ST in_1, ST in_2 ]}"'{pos=0.4}, from=5-9, to=6-9]
	\arrow["{\mu^{ST}}", from=6-9, to=8-9]
	\arrow["{\mu^{ST}}", from=1-8, to=3-9]
	\arrow["{ ST(A+Y+S\eta^T_Y)}"'{pos=1}, from=3-9, to=4-8]
	\arrow["{ST( A +[Y, Y])}"', from=8-9, to=8-12]
	\arrow["{\mu^{ST}}", from=1-12, to=4-10]
	\arrow["{ST(S\eta^T_{A+Y}+S\eta^T_Y)}"'{pos=0.9}, from=4-10, to=5-9]
	\arrow["{ ST(\eta^S_{A+Y}+SY)}"{pos=0.9}, from=3-9, to=4-10]
	\arrow["{\mu^{ST}}"', from=3-12, to=5-11]
	\arrow["{STS\eta^T_{A+Y+Y}}", from=5-11, to=6-9]
	\arrow["{ST[S in_1, S in_2]}"{pos=0.9}, from=4-10, to=5-11]
	\arrow[Rightarrow, no head, from=5-11, to=6-12]
	\arrow["{\mu^S}"', from=7-12, to=8-9]
	\arrow["{ ST (  ST (A +Y) + STS\eta^T_Y)}", from=1-5, to=4-5]
	\arrow["{ ST   ST(A+Y+S\eta^T_Y)}"', from=1-8, to=4-7]
	\arrow["{ ST (  S\eta^T + S\eta^T_Y)}", from=5-1, to=5-3]
	\arrow["{ ST S \eta^T}", from=6-1, to=6-3]
	\arrow["{ ST (  ST (A +Y) + \eta^{ST})}", from=1-1, to=1-5]
	\arrow["{STST[Sin_1,Sin_2]}", from=1-12, to=3-12]
      \end{tikzcd}}
      \caption{Proof of Lemma~\ref{lem:deux:KSY:alg:struct}}
      \label{fig:deux:KSY:alg:struct}
    \end{figure}
  \end{proof}

  Returning to the second law of Lemma~\ref{lem:coh:hcirctildeomega},
  the bottom row is the image by $ℸ^×_{STSY}$ of a morphism that lifts
  to $K_{SY}\alg$, by commutation of the diagram in
  Figure~\ref{fig:coh:hcirctildeomega:bottom:lifts}.
  \begin{figure}[htp]
    \centering
    \ajustetourne{
\begin{tikzcd}[ampersand replacement=\&]
	{K_{SY} K_{SY} L_{STSY} X} \& {K_{SY}  ST(SL_{STSY} X + SY)} \&\&\&\& {K_{SY}  STS(L_{STSY} X + Y)} \&\&\& {K_{SY}  ST(L_{STSY} X + Y)} \&\&\&\& {K_{SY} K_Y L_{STY} X} \\
	\&\&\&\&\&\&\&\& {K_{SY}  STS(L_{STY} X + Y)} \\
	\& {ST ST O_{SY} O_{SY} L_{STSY} X} \&\& {ST  ST O_{SY} O_{SY} SL_{STSY} X} \&\& {STSTO_{SY} S(O_YL_{STSY} X)} \&\&\&\&\& {ST ST O_{SY}  O_Y L_{STY} X} \&\& {STS O_Y K_Y L_{STY} X} \\
	\&\&\&\&\& {STSTSO_{Y} O_YL_{STSY} X} \&\&\& {ST ST O_{SY} SO_Y L_{STY} X} \&\& {ST ST S O_{Y}  O_Y L_{STY} X} \\
	\&\&\&\&\&\&\&\& {ST ST SO_{Y} O_Y L_{STY} X} \&\& {ST  S O_{Y}   L_{STY} X} \&\& {ST O_Y K_Y L_{STY} X} \\
	\&\&\&\&\&\&\&\& {ST S O_Y L_{STY} X} \\
	{K_{SY}  L_{STSY} X} \&\& {ST(SL_{STSY} X + SY)} \&\& {STS(L_{STSY} X + Y)} \&\&\&\& {ST(L_{STSY} X + Y)} \&\&\&\& {ST(L_{STY} X + Y)}
	\arrow["{K_{SY}  ST(\eta^S + SY)}", from=1-1, to=1-2]
	\arrow["{K_{SY}  ST [S in_1, S in_2]}", from=1-2, to=1-6]
	\arrow["{K_{SY} ( S\delta ; \mu^S T )}", from=1-6, to=1-9]
	\arrow["{K_{SY}  ST(L_{S\delta;\mu^S T} X + Y)}", from=1-9, to=1-13]
	\arrow["{ST [S in_1, S in_2]}", from=7-3, to=7-5]
	\arrow["{S\delta ; \mu^S T}", from=7-5, to=7-9]
	\arrow["{ST(L_{S\delta;\mu^S T} X + Y)}", from=7-9, to=7-13]
	\arrow["{ST \uparrow^S}", from=1-13, to=3-13]
	\arrow["{S\delta ; \mu^S T}", from=3-13, to=5-13]
	\arrow["{\mu^{K_Y}}", from=5-13, to=7-13]
	\arrow["{ST ({}^{ST}\uparrow_{SY})}", from=1-13, to=3-11]
	\arrow["{STST \uparrow^S}", from=3-11, to=4-11]
	\arrow["{\mu^{ST} S \mu^{O_Y}}", from=4-11, to=5-11]
	\arrow["{S\delta ; \mu^S T}", from=5-11, to=7-13]
	\arrow["{ST ({}^{ST}\uparrow_{SY})}", from=1-1, to=3-2]
	\arrow[from=1-1, to=7-1]
	\arrow["{ ST(\eta^S + SY)}", from=7-1, to=7-3]
	\arrow["{\mu^{ST} \mu^{O_{SY}}}", from=3-2, to=7-1]
	\arrow["{ST ST O_{SY} O_{SY} \eta^S}", from=3-2, to=3-4]
	\arrow["{ST ({}^{ST}\uparrow_{SY})}", from=1-6, to=3-6]
	\arrow["{ST  ST O_{SY} [Sin_1,Sin_2]}", from=3-4, to=3-6]
	\arrow["{\mu^{ST} \mu^{O_{SY}}}", from=3-4, to=7-3]
	\arrow["{STST [Sin_1, S in_2]}", from=3-6, to=4-6]
	\arrow["{\mu^{ST} S \mu^{O_Y}}", from=4-6, to=7-5]
	\arrow["{K_{SY} STS O_Y L_{S\delta;\mu^S T}X}"', from=1-6, to=2-9]
	\arrow["{K_{SY}  (S \delta ; \mu^S T)}", from=2-9, to=1-13]
	\arrow["{ST ({}^{ST}\uparrow_{SY})}", from=2-9, to=4-9]
	\arrow["{STST[S in_1, S in_2]}", from=4-9, to=5-9]
	\arrow["{\mu^{ST} S \mu^{O_Y}}", from=5-9, to=6-9]
	\arrow["{STS(L_{S\delta ; \mu^S T} X + Y)}", from=7-5, to=6-9]
	\arrow["{S\delta ; \mu^S T}", from=6-9, to=7-13]
	\arrow["{\text{(termwise)}}"{description}, draw=none, from=3-2, to=7-3]
	\arrow["{\text{(proved in Coq)}}"{description}, draw=none, from=3-11, to=4-9]
\end{tikzcd}
    }
    \caption{Lemma~\ref{lem:coh:hcirctildeomega}: bottom morphism lifts to $Σ\alg$}
    \label{fig:coh:hcirctildeomega:bottom:lifts}
  \end{figure}
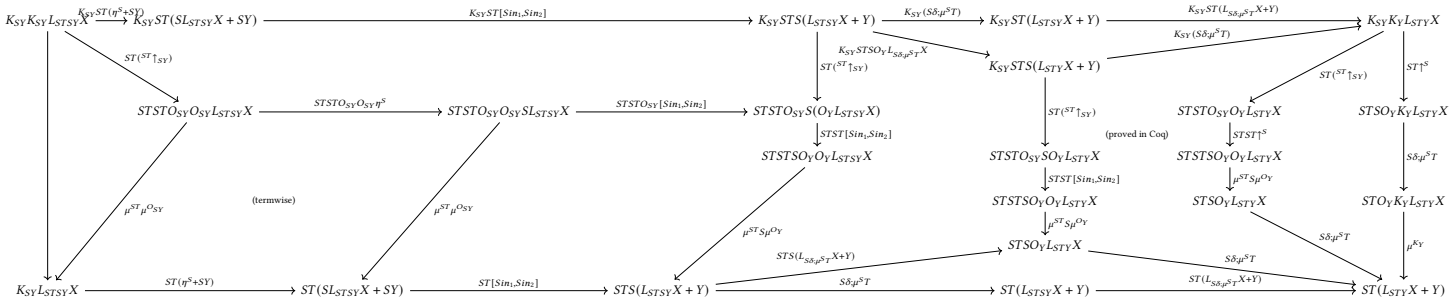
  It thus itself lifts to $Σ\alg$ by Lemma~\ref{lem:RB:liftsto:KBalg}.
  
  Finally, we prove that the right-hand composite is a $Σ$-algebra morphism by
  diagram chasing as follows.
  \begin{center}
    \ajustedroit[.95]{
\begin{tikzcd}[ampersand replacement=\&,baseline=(\tikzcdmatrixname-7-1.base)]
{\Sigma R_Y K_Y Z} \&\&\&\&\&\& {\Sigma R_{SY} K_Y Z} \\
	\&\&\& {R_{SY} K_{SY} L_{SY} R_Y K_Y Z} \&\&\& {R_{SY} K_{SY} L_{SY} R_{SY} K_Y Z} \\
	\&\&\& {R_{SY} K_{SY} L_{Y} R_Y K_Y Z} \&\&\& {R_{SY} K_{SY} K_Y Z} \\
	\&\&\& {R_{SY} STS O_{Y} L_{Y} R_Y K_Y Z} \&\&\& {R_{SY} STS O_{Y} K_Y Z} \\
	{R_Y K_Y L_Y R_Y K_Y Z} \&\&\& {R_{SY} ST O_{Y} L_{Y} R_Y K_Y Z} \&\&\& {R_{SY} ST O_{Y} K_Y Z} \\
	{R_Y K_Y K_Y Z} \\
	{R_Y K_Y Z} \&\&\&\&\&\& {R_{SY} K_Y Z}
	\arrow["{\widetilde{h^\circ}_{R_Y K_Y Z, Y}}"', from=1-1, to=5-1]
	\arrow["{R_Y K_Y \epsilon_Y}"', from=5-1, to=6-1]
	\arrow["{R_Y \mu^{K_Y}}"', from=6-1, to=7-1]
	\arrow["{\widetilde{h^\circ}_{R_{SY}K_YZ,SY}}", from=1-7, to=2-7]
	\arrow["{R_{SY} K_{SY} \epsilon_{SY}}", from=2-7, to=3-7]
	\arrow["{R_{SY} ST \uparrow^S}", from=3-7, to=4-7]
	\arrow["{R_{SY} (S \delta ; \mu^S T)}", from=4-7, to=5-7]
	\arrow["{R_{SY} \mu^{K_Y}}", from=5-7, to=7-7]
	\arrow["{\widetilde{h^\circ}_{ R_Y K_Y Z, SY}}"', from=1-1, to=2-4]
	\arrow[""{name=0, anchor=center, inner sep=0}, "{R_{SY} K_{SY} L_{SY} \daleth^\times_{S \delta ; \mu^S T} K_Y Z}", from=2-4, to=2-7]
	\arrow["{\Sigma \daleth^\times_{S \delta ; \mu^S T} K_Y Z}", from=1-1, to=1-7]
	\arrow["{ \daleth^\times_{S \delta ; \mu^S T} K_Y Z}"', from=7-1, to=7-7]
	\arrow["{R_{SY} K_{SY} \Gamma^\circ_{S \delta ; \mu^S T} R_Y K_Y Z}"', from=2-4, to=3-4]
	\arrow[""{name=1, anchor=center, inner sep=0}, "{R_{SY} K_{SY} \epsilon_Y K_Y Z}", from=3-4, to=3-7]
	\arrow["{R_{SY} ST \uparrow^S}", from=3-4, to=4-4]
	\arrow["{R_{SY} (S \delta ; \mu^S T)}", from=4-4, to=5-4]
	\arrow["{R_{SY} ST O_{Y} \epsilon_Y K_Y Z}", from=5-4, to=5-7]
	\arrow["{\daleth^\times_{S \delta ; \mu^S T} K_Y L_Y R_Y K_Y Z}", from=5-1, to=5-4]
	\arrow["{\daleth^\times_{S \delta ; \mu^S T}  K_Y K_Y Z}"', from=6-1, to=5-7]
	\arrow["{\text{(Lemma~\ref{lem:coh:hcirctilde})}}"{description}, draw=none, from=5-1, to=2-4]
	\arrow["{\text{(Corollary~\ref{cor:transpose})}}"{description}, Rightarrow, draw=none, from=0, to=1]
      \end{tikzcd}} 
  \end{center}

  We then consider the fourth law, for which we proceed similarly: it
  suffices to show that the bottom morphism and right-hand composite
  are  $Σ$-algebra morphisms.
  We first prove that the bottom morphism is a $Σ$-algebra morphism by diagram chasing as in Figure~\ref{fig:coh:hcirctildeomegatilde:mu}.
  \begin{figure}[htp]
    \ajustetourne{
\begin{tikzcd}[ampersand replacement=\&]
	{\Sigma \daleth^\times_{STY} ST \Gamma^\bullet_{ST|Y} X} \&\&\&\&\&\&\&\&\&\& {\Sigma \daleth^\times_{STTY} ST \Gamma^\bullet_{ST|Y} X} \\
	{\daleth^\times_{STY} ST O_{Y} \Gamma^\circ_{STY} \daleth^\times_{STY} ST \Gamma^\bullet_{ST|Y} X} \&\&\& {\daleth^\times_{STTY} ST O_{TY} \Gamma^\circ_{STTY} \daleth^\times_{STY} ST \Gamma^\bullet_{ST|Y} X} \&\&\&\&\&\&\& {\daleth^\times_{STTY} ST O_{TY} \Gamma^\circ_{STTY} \daleth^\times_{STTY} ST \Gamma^\bullet_{ST|Y} X} \\
	\&\&\& {\daleth^\times_{STTY} STT O_{Y} \Gamma^\circ_{STTY} \daleth^\times_{STY} ST \Gamma^\bullet_{ST|Y} X} \&\&\&\&\&\&\& {\daleth^\times_{STTY} STO_{TY}STO_Y\Gamma^\circ_{STY}X} \\
	\&\&\& {\daleth^\times_{STTY} ST O_{Y} \Gamma^\circ_{STTY} \daleth^\times_{STY} ST \Gamma^\bullet_{ST|Y} X} \&\&\&\& {\daleth^\times_{STTY} ST TO_{Y} \Gamma^\circ_{STTY} \daleth^\times_{STTY} ST \Gamma^\bullet_{ST|Y} X} \&\&\& {\daleth^\times_{STTY} STTO_{Y}STO_Y\Gamma^\circ_{STY}X} \\
	\&\& {\daleth^\times_{STTY} ST O_{Y} \Gamma^\circ_{STY} \daleth^\times_{STY} ST \Gamma^\bullet_{ST|Y} X} \&\&\&\&\& {\daleth^\times_{STTY} S TO_{Y} \Gamma^\circ_{STTY} \daleth^\times_{STTY} ST \Gamma^\bullet_{ST|Y} X} \\
	{\daleth^\times_{STTY} STO_{Y}STO_Y\Gamma^\circ_{STY}X} \&\&\&\&\&\&\&\&\&\& {\daleth^\times_{STTY} STO_{Y}STO_Y\Gamma^\circ_{STY}X} \\
	{\daleth^\times_{STY} ST\Gamma^\bullet_{ST|Y}X} \&\&\&\&\&\&\&\&\&\& {\daleth^\times_{STTY} ST\Gamma^\bullet_{ST|Y}X}
	\arrow["{\daleth^\times_{STTY}  ST \uparrow^T}", from=3-11, to=4-11]
	\arrow["{\daleth^\times_{STTY} S\mu^T}", from=4-11, to=6-11]
	\arrow["{\daleth^\times_{STTY} \mu^{K_Y}}", from=6-11, to=7-11]
	\arrow["{\widetilde{h^\circ}_{\daleth^\times_{STTY} ST \Gamma^\bullet_{ST|Y} X,TY}}", from=1-11, to=2-11]
	\arrow["{\daleth^\times_{STTY} ST O_{TY} \epsilon_{STTY} ST \Gamma^\bullet_{ST|Y} X}", from=2-11, to=3-11]
	\arrow["{\daleth^\times_{STTY} \mu^{K_Y}}"', from=6-1, to=7-1]
	\arrow["{\widetilde{h^\circ}_{\daleth^\times_{STY} ST \Gamma^\bullet_{ST|Y} X,Y}}"', from=1-1, to=2-1]
	\arrow["{\daleth^\times_{STY} ST O_{Y} \epsilon_{STY} ST \Gamma^\bullet_{ST|Y} X}"', from=2-1, to=6-1]
	\arrow["{\daleth^\times_{S\mu^T_Y} ST\Gamma^\bullet_{ST|Y}X}"', from=7-1, to=7-11]
	\arrow["{\daleth^\times_{S\mu^T_Y} STO_{Y}STO_Y\Gamma^\circ_{STY}X}"', from=6-1, to=6-11]
	\arrow["{\Sigma \daleth^\times_{S\mu^T_Y} ST \Gamma^\bullet_{ST|Y} X}", from=1-1, to=1-11]
	\arrow["{\daleth^\times_{S\mu^T_Y} ST O_{Y} \Gamma^\circ_{STY} \daleth^\times_{STY} ST \Gamma^\bullet_{ST|Y} X}", from=2-1, to=5-3]
	\arrow["{\widetilde{h^\circ}_{\daleth^\times_{STY} ST \Gamma^\bullet_{ST|Y} X,TY}}"', from=1-1, to=2-4]
	\arrow["{\daleth^\times_{STTY} ST\uparrow^T}", from=2-4, to=3-4]
	\arrow["{\daleth^\times_{STTY} ST O_{TY} \Gamma^\circ_{STTY} \daleth^\times_{S\mu^T_Y} ST \Gamma^\bullet_{ST|Y} X}", from=2-4, to=2-11]
	\arrow["{\daleth^\times_{STTY} ST O_{Y} \epsilon_{STY} }"'{pos=0.2}, from=5-3, to=6-11]
	\arrow["{\daleth^\times_{STTY}  ST \uparrow^T}"'{pos=0.6}, from=2-11, to=4-8]
	\arrow["{\daleth^\times_{STTY} ST TO_{Y} \epsilon_{STTY} ST \Gamma^\bullet_{ST|Y} X}", from=4-8, to=4-11]
	\arrow["{\daleth^\times_{STTY} S\mu^T}", from=4-8, to=5-8]
	\arrow["{\daleth^\times_{STTY} S TO_{Y} \epsilon_{STTY} }", from=5-8, to=6-11]
	\arrow["{\daleth^\times_{STTY} S\mu^T O_{Y} \Gamma^\circ_{STTY} \daleth^\times_{STY} ST \Gamma^\bullet_{ST|Y} X}", from=3-4, to=4-4]
	\arrow["{\daleth^\times_{STTY} ST O_{Y} \Gamma^\circ_{S\mu^T_Y}}"'{pos=0.8}, from=4-4, to=5-3]
	\arrow["{\daleth^\times_{STTY} ST O_{Y} \Gamma^\circ_{STTY} \daleth^\times_{S\mu^T_{TY}} ST \Gamma^\bullet_{ST|Y} X}"'{pos=0.2}, from=4-4, to=5-8]
\end{tikzcd}}    
    \caption{Proof of Lemma~\ref{lem:coh:hcirctildeomegatilde}, fourth law, bottom morphism}
    \label{fig:coh:hcirctildeomegatilde:mu}
  \end{figure}
  Finally, we show that the right-hand
  composite is a $Σ$-algebra morphism, by
  Lemma~\ref{lem:RB:liftsto:KBalg}: we verify that it applies by
  proving that the composite is a morphism of $K_{TY}$-algebras, as we
  show by diagram chasing in
  Figure~\ref{fig:coh:hcirctildeomega:eta},
  \begin{figure}[htp]
    \ajustetourne{
\begin{tikzcd}[ampersand replacement=\&]
	{STO_{TY}STO_{TY}\Gamma^\circ_{STTY}X} \&\&\&\& {STO_{TY}STTO_{Y}\Gamma^\circ_{STTY}X} \&\&\&\&\&\&\&\&\&\&\& {STO_{TY}STO_Y\Gamma^\circ_{STY}X} \\
	\& {STSTO_{TY}O_{TY}\Gamma^\circ_{STTY}X} \&\&\& {STSTO_{TY}TO_{Y}\Gamma^\circ_{STTY}X} \&\& {STSTTO_{TY}O_{Y}\Gamma^\circ_{STTY}X} \&\&\& {STSTO_{TY}O_Y\Gamma^\circ_{STY}X} \&\&\&\&\&\& {STTO_{Y}STO_Y\Gamma^\circ_{STY}X} \\
	\&\&\&\&\&\&\&\& {STO_{TY}O_Y\Gamma^\circ_{STY}X} \&\& {STSTTO_{Y}O_Y\Gamma^\circ_{STY}X} \&\&\&\& {STTSTO_{Y}O_Y\Gamma^\circ_{STY}X} \\
	\&\& {STO_{TY}O_{TY}\Gamma^\circ_{STTY}X} \&\& {STO_{TY}TO_{Y}\Gamma^\circ_{STTY}X} \&\& {STTO_{TY}O_{Y}\Gamma^\circ_{STTY}X} \&\&\& {STTO_{Y}O_Y\Gamma^\circ_{STY}X} \&\& {STSTTO_{Y}O_Y\Gamma^\circ_{STY}X} \&\&\&\& {STO_{Y}STO_Y\Gamma^\circ_{STY}X} \\
	\&\&\&\& {STTO_{Y}TO_{Y}\Gamma^\circ_{STTY}X} \&\&\& {STTTO_{Y}O_{Y}\Gamma^\circ_{STTY}X} \&\&\& {STSTO_{Y}O_Y\Gamma^\circ_{STY}X} \&\& {SSTTTO_{Y}O_Y\Gamma^\circ_{STY}X} \&\& {STSTO_{Y}O_Y\Gamma^\circ_{STY}X} \\
	\&\&\&\&\& {STTTO_{Y}O_{Y}\Gamma^\circ_{STTY}X} \&\& {STTO_{Y}O_{Y}\Gamma^\circ_{STTY}X} \&\&\&\& {SSTTO_{Y}O_Y\Gamma^\circ_{STY}X} \&\& {SSTTO_{Y}O_Y\Gamma^\circ_{STY}X} \\
	\&\&\&\&\&\&\&\&\&\&\&\& {SSTO_{Y}O_Y\Gamma^\circ_{STY}X} \\
	\&\&\&\&\&\&\&\&\& {STO_{Y}O_Y\Gamma^\circ_{STY}X} \\
	{STO_{TY}\Gamma^\circ_{STTY}X} \&\&\&\& {STT\Gamma^\bullet_{STT|Y}X} \&\&\&\&\&\&\&\&\&\&\& {ST\Gamma^\bullet_{ST|Y}X}
	\arrow["{ST \uparrow^T}", from=9-1, to=9-5]
	\arrow["{S\mu^T \Gamma^\bullet_{S\mu^T|Y}X}", from=9-5, to=9-16]
	\arrow["{ ST \uparrow^T}", from=1-16, to=2-16]
	\arrow["{S\mu^T}", from=2-16, to=4-16]
	\arrow["{\mu^{K_Y}}", from=4-16, to=9-16]
	\arrow["{\mu^{STO_{TY}}}", from=1-1, to=9-1]
	\arrow["{STO_{TY}S\mu^TO_{Y}\Gamma^\circ_{S\mu^T_Y}X}", from=1-5, to=1-16]
	\arrow["{STO_{TY}ST\uparrow^T}", from=1-1, to=1-5]
	\arrow["{ST({{}^{ST}\uparrow})}", from=1-1, to=2-2]
	\arrow["{ST({{}^{ST}\uparrow})}", from=1-5, to=2-5]
	\arrow["{STSTO_{TY}\uparrow^T}", from=2-2, to=2-5]
	\arrow["{ST({}^{ST}\uparrow)}", from=1-16, to=2-10]
	\arrow["{STST\uparrow^T}", from=2-10, to=3-11]
	\arrow["{STT({}^{ST}\uparrow)}", from=2-16, to=3-15]
	\arrow["ST\delta", from=3-15, to=4-12]
	\arrow["{ST({}^{ST}\uparrow)}", from=4-16, to=5-15]
	\arrow["{S\mu^T}", from=3-15, to=5-15]
	\arrow["{STS\mu^T}", from=4-12, to=5-11]
	\arrow["S\delta", from=4-12, to=5-13]
	\arrow["{SS\mu^T}", from=5-13, to=6-14]
	\arrow["S\delta", from=5-15, to=6-14]
	\arrow["{SST\mu^T}", from=5-13, to=6-12]
	\arrow["S\delta", from=5-11, to=6-12]
	\arrow["{SS\mu^T}", from=6-12, to=7-13]
	\arrow["{SS\mu^T}", from=6-14, to=7-13]
	\arrow["{STS\mu^T}", from=3-11, to=5-11]
	\arrow["{STST({}^{T}\uparrow)}", from=2-5, to=2-7]
	\arrow["{STS\mu^T O_{TY} O_Y \Gamma^\circ_{S\mu^T_Y} X}", from=2-7, to=2-10]
	\arrow["{\mu^{ST}\mu^{O_{TY}}}", from=2-2, to=9-1]
	\arrow["{\mu^{ST}}", from=2-2, to=4-3]
	\arrow["{ST\mu^{O_{TY}}}", from=4-3, to=9-1]
	\arrow["{\mu^{ST}}", from=2-5, to=4-5]
	\arrow["{STO_{TY}\uparrow^T}", from=4-3, to=4-5]
	\arrow["{\mu^{ST}}", from=2-10, to=3-9]
	\arrow["{ST({}^{T}\uparrow)}", from=4-5, to=4-7]
	\arrow["{\mu^{ST}}", from=2-7, to=4-7]
	\arrow["{ST\uparrow^T}", from=3-9, to=4-10]
	\arrow["{\mu^{ST}}", from=3-11, to=4-10]
	\arrow["{S\mu^T}", from=4-10, to=8-10]
	\arrow["{ST\mu^{O_Y}}", from=8-10, to=9-16]
	\arrow["{S\mu^TO_{TY}O_{Y}\Gamma^\circ_{S\mu^T_Y}X}", from=4-7, to=3-9]
	\arrow["{STT\uparrow^T}", from=4-7, to=5-8]
	\arrow["{S\mu^T T O_Y O_Y \Gamma^\circ_{S\mu^T_Y} X}", from=5-8, to=4-10]
	\arrow["{ST\mu^T}", from=5-8, to=6-8]
	\arrow["{ST\uparrow^T}", from=4-5, to=5-5]
	\arrow["{STT({}^{T}\uparrow)}", from=5-5, to=6-6]
	\arrow["{ST\mu^T}", from=6-6, to=6-8]
	\arrow["{S\mu^T O_Y O_Y \Gamma^\circ_{S\mu^T_Y} X}", from=6-8, to=8-10]
	\arrow["{STT\mu^{O_{Y}}}", from=6-8, to=9-5]
	\arrow["{\mu^S}", from=7-13, to=8-10]
      \end{tikzcd}
    }
    \caption{Lemma~\ref{lem:coh:hcirctildeomega}, fourth diagram, first two morphisms}
    \label{fig:coh:hcirctildeomega:eta}
  \end{figure}
  where the bottom left heptagon commutes by Lemma~\ref{lem:option:mu} below.
\end{proof}

\begin{lemma}\label{lem:option:mu}
  For all objects $A$ and $Y$, the following diagram commutes.
  \begin{center}
    \begin{tikzcd}[ampersand replacement=\&]
      {O_{TY}O_{TY}A} \& {O_{TY}TO_{Y}A} \& {TO_{Y}TO_{Y}A} \& {TTO_{Y}O_{Y}A} \\
      \&\&\& {TO_{Y}O_{Y}A} \\
      {O_{TY}A} \&\&\& {TO_YA}
      \arrow["{ \uparrow^T}"', from=3-1, to=3-4]
      \arrow["{\mu^{O_{TY}}}"', from=1-1, to=3-1]
      \arrow["{O_{TY}\uparrow^T}", from=1-1, to=1-2]
      \arrow["{\uparrow^T}", from=1-2, to=1-3]
      \arrow["{T({}^{T}\uparrow)}", from=1-3, to=1-4]
      \arrow["{\mu^T}", from=1-4, to=2-4]
      \arrow["{T\mu^{O_{Y}}}", from=2-4, to=3-4]
    \end{tikzcd}
  \end{center}
\end{lemma}

\begin{proof}
  The domain is a coproduct, so we proceed termwise, by diagram chasing.
  The first term is chased as follows.
  \begin{center}
    \ajustedroit{
\begin{tikzcd}[ampersand replacement=\&]
	\&\& TA \& {\ \ \ \ \ \ \ \ \ \ \ \ \ \ \ \ \ \ \ \  \ \ \ \ \ \ \ \  } \& {T(A+Y)} \&\& {TT(A+Y)} \\
	{A+ TY + TY} \&\& {TA+TY+TY} \&\& {T(A+Y)+TY} \&\& {TT(A+Y)+TY} \&\& {T(T(A+Y)+Y)} \&\& {T(T(A+Y)+TY)} \&\& {TT(A+Y+Y)} \\
	\&\&\&\&\&\&\& {T(A+Y)} \&\& {\ \ \ \ \ \ \ \ \ \ \ \ \ \ \ \  } \\
	\&\&\&\&\& {T(A+Y)} \&\&\&\&\&\&\&\& {T(A+Y+Y)} \\
	\&\& {A+TY} \&\& {TA+TY} \&\&\&\&\&\&\&\&\& {T(A+Y)}
	\arrow["{A+[TY,TY]}", from=2-1, to=5-3]
	\arrow["{\mu^T}", from=2-13, to=4-14]
	\arrow["{T(A+[Y,Y])}", from=4-14, to=5-14]
	\arrow["{\eta^T + TY}", from=2-5, to=2-7]
	\arrow["{[Tin_1,Tin_2]}", from=2-7, to=2-9]
	\arrow["{\eta^T_A+ TY + TY}", from=2-1, to=2-3]
	\arrow["{[Tin_1,Tin_2]+TY}"{pos=0.7}, from=2-3, to=2-5]
	\arrow["{T(T(A+Y)+\eta^T_Y)}", from=2-9, to=2-11]
	\arrow["{T[Tin_1,Tin_2]}"', from=2-11, to=2-13]
	\arrow["{\eta^T_A + TY}", from=5-3, to=5-5]
	\arrow["{[Tin_1,Tin_2]}", from=5-5, to=5-14]
	\arrow["{TA+[TY,TY]}"', from=2-3, to=5-5]
	\arrow["{in_1}"', from=1-3, to=2-3]
	\arrow["{Tin_1}", from=1-3, to=1-5]
	\arrow["{in_1}"', from=1-5, to=2-5]
	\arrow["{\eta^T}", from=1-5, to=1-7]
	\arrow["{in_1}"', from=1-7, to=2-7]
	\arrow["{Tin_1}"{pos=0.7}, from=1-7, to=2-9]
	\arrow["{in_1}", from=1-3, to=5-5]
	\arrow["{Tin_1}", curve={height=6pt}, from=1-3, to=4-6]
	\arrow[Rightarrow, no head, from=4-6, to=5-14]
	\arrow["{TTin_1}", curve={height=-6pt}, from=1-7, to=2-13]
	\arrow["{\mu^T}"{pos=0.9}, curve={height=-6pt}, from=1-7, to=3-8]
	\arrow["{Tin_1}", from=3-8, to=4-14]
	\arrow[Rightarrow, no head, from=3-8, to=5-14]
	\arrow[curve={height=18pt}, Rightarrow, no head, from=1-5, to=3-8]
\end{tikzcd}
      }
  \end{center}
  The second term is chased as follows.
  \begin{center}
    \ajustedroit{
\begin{tikzcd}[ampersand replacement=\&]
	TY \&\&\& TY \&\&\& TTY \\
	\&\&\& {T(A+Y)} \&\&\& {TT(A+Y)} \\
	{TA+TY+TY} \&\&\& {T(A+Y)+TY} \&\&\& {TT(A+Y)+TY} \&\& {T(T(A+Y)+Y)} \&\& {T(T(A+Y)+TY)} \&\& {TT(A+Y+Y)} \\
	\&\&\&\&\&\&\&\&\& {\ \ \ \ \ \ \ \ \ \ \ \ \ \ \ \ \ \  } \\
	\&\&\&\&\& TY \&\&\&\&\&\&\& {T(A+Y+Y)} \\
	\&\&\& {TA+TY} \&\&\&\&\&\&\&\&\& {T(A+Y)}
	\arrow["{\mu^T}", from=3-13, to=5-13]
	\arrow["{T(A+[Y,Y])}", from=5-13, to=6-13]
	\arrow["{\eta^T + TY}", from=3-4, to=3-7]
	\arrow["{[Tin_1,Tin_2]}"', from=3-7, to=3-9]
	\arrow["{[Tin_1,Tin_2]+TY}"{pos=0.6}, from=3-1, to=3-4]
	\arrow["{T(T(A+Y)+\eta^T_Y)}"', from=3-9, to=3-11]
	\arrow["{T[Tin_1,Tin_2]}"', from=3-11, to=3-13]
	\arrow["{[Tin_1,Tin_2]}", from=6-4, to=6-13]
	\arrow["{TA+[TY,TY]}"', from=3-1, to=6-4]
	\arrow["{in_2}", from=1-1, to=3-1]
	\arrow["{in_2}"{pos=0.6}, curve={height=6pt}, from=1-1, to=6-4]
	\arrow[Rightarrow, no head, from=1-1, to=1-4]
	\arrow["{\eta^T_{TY}}", from=1-4, to=1-7]
	\arrow["{Tin_2}"', from=1-4, to=2-4]
	\arrow["{in_1}"', from=2-4, to=3-4]
	\arrow["{TTin_2}", from=1-7, to=2-7]
	\arrow["{\eta^T}", from=2-4, to=2-7]
	\arrow["{in_1}", from=2-7, to=3-7]
	\arrow["{Tin_1}"'{pos=0.4}, from=2-7, to=3-9]
	\arrow["{TTin_1}"{pos=0.2}, from=2-7, to=3-13]
	\arrow["{TTin_2}", from=1-7, to=3-13]
	\arrow["{\mu^T_Y}"{pos=0.7}, curve={height=18pt}, from=1-7, to=5-6]
	\arrow["{Tin_2}", from=5-6, to=5-13]
	\arrow["{Tin_2}", from=5-6, to=6-13]
	\arrow[Rightarrow, no head, from=1-4, to=5-6]
\end{tikzcd}
    }
  \end{center}
  Finally, the third term is chased as follows.
  \begin{center}
    \ajustedroit
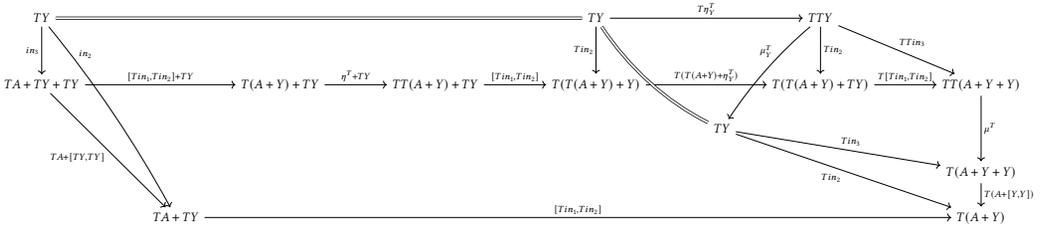
{
\begin{tikzcd}[ampersand replacement=\&]
	TY \&\&\&\&\&\&\& TY \&\&\& TTY \\
	\\
	{TA+TY+TY} \&\&\& {T(A+Y)+TY} \&\& {TT(A+Y)+TY} \&\& {T(T(A+Y)+Y)} \&\&\& {T(T(A+Y)+TY)} \&\& {TT(A+Y+Y)} \\
	\&\&\&\&\&\&\&\&\& TY \\
	\&\&\&\&\&\&\&\&\&\&\&\& {T(A+Y+Y)} \\
	\&\& {TA+TY} \&\&\&\&\&\&\&\&\&\& {T(A+Y)}
	\arrow["{\mu^T}", from=3-13, to=5-13]
	\arrow["{T(A+[Y,Y])}", from=5-13, to=6-13]
	\arrow["{\eta^T + TY}", from=3-4, to=3-6]
	\arrow["{[Tin_1,Tin_2]}", from=3-6, to=3-8]
	\arrow["{[Tin_1,Tin_2]+TY}", from=3-1, to=3-4]
	\arrow["{T(T(A+Y)+\eta^T_Y)}", from=3-8, to=3-11]
	\arrow["{T[Tin_1,Tin_2]}", from=3-11, to=3-13]
	\arrow["{[Tin_1,Tin_2]}", from=6-3, to=6-13]
	\arrow["{TA+[TY,TY]}"', from=3-1, to=6-3]
	\arrow["{in_3}"', from=1-1, to=3-1]
	\arrow["{in_2}"{pos=0.2}, curve={height=-6pt}, from=1-1, to=6-3]
	\arrow[Rightarrow, no head, from=1-1, to=1-8]
	\arrow["{Tin_2}"', from=1-8, to=3-8]
	\arrow["{T\eta^T_Y}", from=1-8, to=1-11]
	\arrow["{Tin_2}", from=1-11, to=3-11]
	\arrow["{TTin_3}", from=1-11, to=3-13]
	\arrow["{\mu^T_Y}"'{pos=0.4}, curve={height=6pt}, from=1-11, to=4-10]
	\arrow["{Tin_3}", from=4-10, to=5-13]
	\arrow["{Tin_2}"', from=4-10, to=6-13]
	\arrow[curve={height=12pt}, Rightarrow, no head, from=1-8, to=4-10]
\end{tikzcd}      
      }
  \end{center}
\end{proof}

\subsection{The family $\widetilde{\widetilde{h^∘}^ω}$}
Transposing back, we obtain a family
$$\widetilde{\widetilde{h^∘}^ω}_{X,Y}∶ Γ^\pt_{STY}SX →  ST O_Y Γ^∘_{STY} X$$
of morphisms.

By extranaturality of $\widetilde{h^∘}^ω$ in $Y$, naturality in $X$,
and Corollary~\ref{cor:transpose}, we get:
\begin{lemma}\label{lem:naturalomega}
  The family $\widetilde{\widetilde{h^∘}^ω}$ is natural in both
  variables.
\end{lemma}

\begin{lemma}\label{lem:coh:hcirctildeomegatilde}
  The family $\widetilde{\widetilde{h^∘}^ω}$ satisfies the following
  coherence laws.
\begin{mathpar}
  \diag{%
    S(X)  \& ST(X) \\
    Γ^\pt_{STY}(S(X)) \& ST (Γ^•_{ST|Y}(X)) %
  }{%
    (m-1-1) edge[labela={S(ηᵀ_X)}] (m-1-2) %
    edge[labell={α^{Γ^\pt}_{SX,STY}}] (m-2-1) %
    (m-2-1) edge[labelb={\widetilde{\widetilde{h^∘}^ω}_{X,Y}}] (m-2-2) %
    (m-1-2) edge[labelr={ST  α^{Γ^•_{ST|}}_{X,Y}}] (m-2-2) %
  } \and
  \diag(.6,1.5){%
    Γ^\pt_{STSY}SX \&  Γ^\pt_{SSTY}SX    \&     Γ^\pt_{STY}SX \\
    ST Γ^•_{ST|SY}X \& STS Γ^•_{STS|Y}X \& ST Γ^•_{ST|Y}X %
  }{%
    (m-1-1)
    edge[labell={\widetilde{\widetilde{h^∘}^ω}_{X,SY}}] (m-2-1) %
    (m-2-1) edge[labelb={ST ↑^{S}_{X,Y}}] (m-2-2) %
    (m-1-1) edge[labela={Γ^\pt_{Sδ_Y}SX}] (m-1-2) %
    (m-1-2) edge[labela={Γ^\pt_{μ^S_{TY}}SX}] (m-1-3) %
    (m-2-2) edge[loinb={(Sδ;μ^S T) Γ^•_{Sδ ; μ^S T | Y}X}] (m-2-3) %
    (m-1-3) edge[labelr={\widetilde{\widetilde{h^∘}^ω}_{X,Y}}] (m-2-3) %
  } \and
  \diag{%
    Γ^\pt_{STY}X \& ST Γ^\pt_{STY}X \\
    Γ^\pt_{STY}SX \& ST Γ^•_{ST|Y} X \\
  }{%
    (m-1-1) edge[labela={η^{ST}_{Γ^\pt_{STY}X}}] (m-1-2) %
    edge[labell={Γ^\pt_{STY} η^S_X}] (m-2-1) %
    (m-2-1) edge[labelb={\widetilde{\widetilde{h^∘}^ω}}] (m-2-2) %
    (m-1-2) edge[labelr={ST in₁}] (m-2-2) %
  } \and
\begin{tikzcd}[ampersand replacement=\&]
	{\Gamma^\circ_{STTY}SX} \&\&\&\& {ST\Gamma^\bullet_{ST|TY}X} \\
	\&\&\&\& {STT\Gamma^\bullet_{STT|Y}X} \\
	{\Gamma^\circ_{STY}SX} \&\&\&\& {ST\Gamma^\bullet_{ST|Y}X}
	\arrow["{\Gamma^\circ_{S\mu^T_{Y}} SX}"', from=1-1, to=3-1]
	\arrow["{\widetilde{\widetilde{h^\circ}^\omega}_{X,TY}}", from=1-1, to=1-5]
	\arrow["{ST \uparrow^T}", from=1-5, to=2-5]
	\arrow["{S\mu^T \Gamma^\bullet_{S\mu^T|Y}X}", from=2-5, to=3-5]
	\arrow["{\widetilde{\widetilde{h^\circ}^\omega}_{X,Y}}"', from=3-1, to=3-5]
\end{tikzcd}  
  \and
  \diag(.6,2){%
    Γ^\pt_{STY} SSX \& ST O_Y Γ^\pt_{STY} SX
    \& ST O_Y ST O_Y Γ^\pt_{STY} X \\
    Γ^\pt_{STY} SX \& \& ST O_Y Γ^\pt_{STY} X %
  }{%
    (m-1-1) edge[labela={\widetilde{\widetilde{h^∘}^ω}_{SX,Y}}]
    (m-1-2) %
    edge[labell={Γ^\pt_{STY} μ^S_X}] (m-2-1) %
    (m-2-1) edge[labelb={\widetilde{\widetilde{h^∘}^ω}_{X,Y}}] (m-2-3) %
    (m-1-2) edge[labela={ST O_Y
      \widetilde{\widetilde{h^∘}^ω}_{X,Y}}] (m-1-3) %
    (m-1-3) edge[labelr={μ^{ST O_Y}_{Γ^\pt_{STY} X}}] (m-2-3) %
  }
\and
\begin{tikzcd}[ampersand replacement=\&]
	{\Gamma_Y\Sigma X} \& {\Gamma_{STY}SX} \& {\Gamma^\circ_{STY}SX} \\
	\& {ST \Gamma^\bullet_{ST|Y}X}
	\arrow["{\Gamma_{\eta^{ST}_Y} \eta_{\Sigma,X}}", from=1-1, to=1-2]
	\arrow["{in_1}", from=1-2, to=1-3]
	\arrow["{\widetilde{\widetilde{h^\circ}^\omega}_{X,Y}}"{pos=0.3}, from=1-3, to=2-2]
	\arrow["{h_{X,Y}}"'{pos=0.4}, from=1-1, to=2-2]
\end{tikzcd}
\end{mathpar}
\end{lemma}
\begin{proof}
  By Lemma~\ref{lem:coh:hcirctildeomega},
  Corollary~\ref{cor:transpose}, and the fact that
  $\widetilde{h^∘}^ω_{X,Y}$ is a monad morphism by contruction.
  E.g., the third law holds iff the transposed morphisms coincide. This gives the following diagram,
  \begin{center}
\begin{tikzcd}[ampersand replacement=\&]
	X \& {R_Y L_Y X} \&\& {R_Y ST L_Y X} \\
	SX \&\&\& {R_Y K_Y L_Y X}
	\arrow["{\eta^S_X}", from=1-1, to=2-1]
	\arrow["{\widetilde{h^\circ}^\omega_{X,Y}}", from=2-1, to=2-4]
	\arrow["{\eta_Y X}", from=1-1, to=1-2]
	\arrow["{R_Y \eta^{ST}_{L_Y X}}", from=1-2, to=1-4]
	\arrow["{R_Y ST in_1}", from=1-4, to=2-4]
\end{tikzcd}
  \end{center}
  which is precisely commutation of $\widetilde{h^∘}^ω$ with unit.
  Similarly, the penultimate law holds iff the transposed morphisms
  coincide. This gives the following diagram,
  \begin{center}
\begin{tikzcd}[ampersand replacement=\&]
	SSX \&\& {R_Y K_Y L_Y SX} \&\&\& {R_Y K_Y L_Y R_Y K_Y L_Y X} \\
	\&\&\&\&\& {R_Y K_Y K_Y L_Y X} \\
	SX \&\&\&\&\& {R_Y K_Y L_Y X}
	\arrow["{\mu^S_X}"', from=1-1, to=3-1]
	\arrow["{\widetilde{h^\circ}^\omega_{X,Y}}"', from=3-1, to=3-6]
	\arrow["{\widetilde{h^\circ}^\omega_{SX,Y}}", from=1-1, to=1-3]
	\arrow["{R_Y K_Y L_Y \widetilde{h^\circ}^\omega_{X,Y}}", from=1-3, to=1-6]
	\arrow["{R_Y K_Y \epsilon_Y}", from=1-6, to=2-6]
	\arrow["{R_Y \mu^{K_Y}}", from=2-6, to=3-6]
\end{tikzcd}
  \end{center}
  which precisely commutation of $\widetilde{h^∘}^ω$ with
  multiplication.
\end{proof}

\subsection{The family $h^{•ω}$}

\begin{definition}\label{def:hbulletomega}
  Let $h^{•ω}_{X,Y}∶ Γ_{SY} SX → ST Γ^•_{ST|Y} X$ denote
  the composite
  $$Γ_{SY} SX \xto{Γ_{Sηᵀ_Y}SX} Γ_{STY}SX \xto{in₁} 
  Γ^\pt_{STY}SX
  \xto{\widetilde{\widetilde{h^∘}^ω}_{X,Y}} ST Γ^•_{ST|Y} X.$$
\end{definition}

\begin{lemma}\label{lem:hbulletomega:piml}
  The natural transformation
  $h^{•ω}_{X,Y}∶ Γ^•_{SY} SX → ST Γ^•_{ST|Y} X$ is natural in both
  variables.
\end{lemma}
\begin{proof}
  This follows straightforwardly from Lemma~\ref{lem:naturalomega}.
\end{proof}

\begin{lemma}\label{lem:coh:hbulletomegamu}
  The natural transformation $h^{•ω}$ makes all diagrams below commute,
  for all $X,Y ∈ 𝐂$.
  \begin{center}
\begin{tikzcd}[ampersand replacement=\&]
	{\Gamma_{SSY} S X} \&\&\&\&\& {\Gamma_{SY}SX} \\
	{ST \Gamma^\bullet_{ST|SY}X} \& {STS\Gamma^\bullet_{STS|Y}X} \&\& {SST \Gamma^\bullet_{SST|Y}X} \&\& {ST \Gamma^\bullet_{ST|Y}X}
	\arrow["{h^{\bullet\omega}_{X,SY}}"', from=1-1, to=2-1]
	\arrow["{ST\uparrow^S}"', from=2-1, to=2-2]
	\arrow["{S\delta \Gamma^\bullet_{S\delta|Y}X}"', from=2-2, to=2-4]
	\arrow["{\mu^S T \Gamma^\bullet_{\mu^S T|Y}X}"', from=2-4, to=2-6]
	\arrow["{\Gamma_{\mu^S Y}SX}", from=1-1, to=1-6]
	\arrow["{h^{\bullet \omega}_{X,Y}}", from=1-6, to=2-6]
      \end{tikzcd}

\begin{tikzcd}[ampersand replacement=\&]
	{\Gamma_{SY}SSX} \& {ST(\Gamma_{STY}SX + SX + Y)} \&\& {ST(ST\Gamma^\bullet_{ST|TY}X+ STX + STY)} \\
	\&\&\& {ST(ST(\Gamma_{STTY}X + X + TY)+STX + STY)} \\
	\&\&\& {ST(STT(\Gamma_{STTY}X + X + Y)+STX + STY)} \\
	\&\&\& {ST(ST(\Gamma_{STY}X + X + Y)+STX + STY)} \\
	\&\&\& {STST(\Gamma_{STY}X + X + Y + X + Y)} \\
	\&\&\& {ST(\Gamma_{STY}X + X + Y)} \\
	{\Gamma_{SY}SX} \&\&\& {ST\Gamma^\bullet_{ST|Y}X}
	\arrow["{ST(h^{\bullet\omega}_{X,TY} + S\eta^T_X + \eta^{ST}_Y)}", from=1-2, to=1-4]
	\arrow[Rightarrow, no head, from=1-4, to=2-4]
	\arrow["{ST(ST\uparrow^T + STX + STY)}"', from=2-4, to=3-4]
	\arrow["{ST(S\mu^T(\Gamma_{S\mu^TY}X + X + Y)+STX + STY)}"', from=3-4, to=4-4]
	\arrow["{ST[ST in_1, ST in_2, ST in_3]}"', from=4-4, to=5-4]
	\arrow["{\mu^{ST} [in_1,in_2,in_3,in_2,in_3]}"', from=5-4, to=6-4]
	\arrow[Rightarrow, no head, from=6-4, to=7-4]
	\arrow["{\Gamma_{SY} \mu^S_X}"', from=1-1, to=7-1]
	\arrow["{h^{\bullet\omega}_{X,Y}}"', from=7-1, to=7-4]
	\arrow["{h^{\bullet\omega}_{SX,Y}}", from=1-1, to=1-2]
      \end{tikzcd}

\begin{tikzcd}[ampersand replacement=\&]
	{\Gamma_{SY} X} \& {\Gamma_{STY} X} \& {\Gamma^\bullet_{ST|Y} X} \\
	\\
	{\Gamma_{SY}SX} \&\& {ST\Gamma^\bullet_{ST|Y}X}
	\arrow["{h^{\bullet\omega}_{X,Y}}", from=3-1, to=3-3]
	\arrow["{\Gamma_{SY} \eta^S_X}"', from=1-1, to=3-1]
	\arrow["{\Gamma_{S\eta^T_X} X}", from=1-1, to=1-2]
	\arrow["{in_1}", from=1-2, to=1-3]
	\arrow["{\eta^{ST}}", from=1-3, to=3-3]
\end{tikzcd}

\begin{tikzcd}[ampersand replacement=\&]
	{\Gamma_{Y} \Sigma X} \&\& {\Gamma_{SY}SX} \\
	\& {ST\Gamma^\bullet_{ST|Y}X}
	\arrow["{h^{\bullet\omega}_{X,Y}}"{pos=0.2}, from=1-3, to=2-2]
	\arrow["{\Gamma_{\eta^S_Y} \eta_{\Sigma,X}}", from=1-1, to=1-3]
	\arrow["{h_{X,Y}}"'{pos=0.2}, from=1-1, to=2-2]
\end{tikzcd}  \end{center}
\end{lemma}
\begin{proof}
  The last diagram is direct by Lemma~\ref{lem:coh:hcirctildeomegatilde}.
  The third one follows by chasing as follows,
  \begin{center}
\begin{tikzcd}[ampersand replacement=\&]
	{\Gamma_{SY} X} \&\& {\Gamma_{STY} X} \& {\Gamma^\circ_{STY}X} \&\& {\Gamma^\bullet_{ST|Y} X} \\
	\&\&\&\& {ST\Gamma^\circ_{STY}X} \\
	{\Gamma_{SY}SX} \&\& {\Gamma_{STY}SX} \& {\Gamma^\circ_{STY}SX} \&\& {ST\Gamma^\bullet_{ST|Y}X}
	\arrow["{\Gamma_{SY} \eta^S_X}"', from=1-1, to=3-1]
	\arrow["{\Gamma_{S\eta^T_X} X}", from=1-1, to=1-3]
	\arrow["{\eta^{ST}}", from=1-6, to=3-6]
	\arrow["{\Gamma_{S\eta^T_Y}SX}"', from=3-1, to=3-3]
	\arrow["{in_1}"', from=3-3, to=3-4]
	\arrow["{\widetilde{\widetilde{h^\circ}^\omega}_{X,Y}}"', from=3-4, to=3-6]
	\arrow["{\Gamma_{STY}\eta^S_X}"', from=1-3, to=3-3]
	\arrow["{in_1}", from=1-3, to=1-4]
	\arrow["{in_1}", from=1-4, to=1-6]
	\arrow["{\Gamma^\circ_{STY} \eta^S_X}"', from=1-4, to=3-4]
	\arrow["{\eta^{ST}}"', from=1-4, to=2-5]
	\arrow["{ST in_1}"'{pos=0.2}, from=2-5, to=3-6]
      \end{tikzcd}
    \end{center}
using Lemma~\ref{lem:coh:hcirctildeomegatilde}.
  The first statement follows by chasing the next diagram.
\[\begin{tikzcd}[ampersand replacement=\&]
	{\Gamma_{SSY} S X} \&\&\&\&\& {\Gamma_{SY}SX} \\
	{\Gamma_{STSY}SX} \& {\Gamma_{SSTY}SX} \&\&\&\& {\Gamma_{STY}SX} \\
	{\Gamma^\circ_{STSY}SX} \& {\Gamma^\circ_{SSTY}SX} \&\&\&\& {\Gamma^\circ_{STY}SX} \\
	{ST \Gamma^\bullet_{ST|SY}X} \& {STS\Gamma^\bullet_{STS|Y}X} \&\& {SST \Gamma^\bullet_{SST|Y}X} \&\& {ST \Gamma^\bullet_{ST|Y}X}
	\arrow["{ST\uparrow^S}"', from=4-1, to=4-2]
	\arrow[""{name=0, anchor=center, inner sep=0}, "{S\delta \Gamma^\bullet_{S\delta|Y}X}"', from=4-2, to=4-4]
	\arrow["{\mu^S T \Gamma^\bullet_{\mu^S T|Y}X}"', from=4-4, to=4-6]
	\arrow["{\Gamma_{\mu^S_Y}SX}", from=1-1, to=1-6]
	\arrow["{\widetilde{\widetilde{h^\circ}^\omega}_{X,SY}}"', from=3-1, to=4-1]
	\arrow["{\Gamma_{S\eta^T_{SY}} S X}"', from=1-1, to=2-1]
	\arrow["{in_1}"', from=2-1, to=3-1]
	\arrow["{\widetilde{\widetilde{h^\circ}^\omega}_{X,Y}}", from=3-6, to=4-6]
	\arrow["{in_1}", from=2-6, to=3-6]
	\arrow["{\Gamma_{S\eta^T_Y} S X}", from=1-6, to=2-6]
	\arrow["{\Gamma_{S\delta_Y}SX}", from=2-1, to=2-2]
	\arrow["{\Gamma_{\mu^S_{TY}}SX}", from=2-2, to=2-6]
	\arrow["{\Gamma_{SS\eta^T_Y}SX}"{pos=1}, curve={height=-6pt}, from=1-1, to=2-2]
	\arrow["{\Gamma^\circ_{S\delta_Y}SX}", from=3-1, to=3-2]
	\arrow[""{name=1, anchor=center, inner sep=0}, "{\Gamma^\circ_{\mu^S TY}SX}", from=3-2, to=3-6]
	\arrow["{\text{(Lemma~\ref{lem:coh:hcirctildeomegatilde})}}"{description}, Rightarrow, draw=none, from=1, to=0]
\end{tikzcd}\]  
  For the second statement, we first reduce to the rightmost subdiagram in Figure~\ref{fig:qgammaSmu},
  \begin{figure*}[pth]
  \adjustbox{max width=\columnwidth,max height=.95\textheight}{%
    \begin{turn}{90}
\begin{tikzcd}[ampersand replacement=\&]
	{\Gamma_{SY}SSX} \&\& {\Gamma_{STY}SSX} \& {\Gamma^{\circ}_{STY}SSX} \&\& {ST(\Gamma_{STY}SX + SX + Y)} \&\& {ST(ST\Gamma^\bullet_{ST|TY}X+ STX + STY)} \\
	\&\&\&\&\&\&\& {ST(ST(\Gamma_{STTY}X + X + TY)+STX + STY)} \\
	\&\&\&\&\&\&\& {ST(STT(\Gamma_{STTY}X + X + Y)+STX + STY)} \\
	\&\&\&\&\&\&\& {ST(ST(\Gamma_{STY}X + X + Y)+STX + STY)} \\
	\&\&\&\&\& {ST(ST\Gamma^\bullet_{ST|Y}X + Y)} \&\& {STST(\Gamma_{STY}X + X + Y + X + Y)} \\
	\&\&\&\&\& {ST O_Y ST O_Y \Gamma_{STY}^\circ X} \&\& {ST(\Gamma_{STY}X + X + Y)} \\
	{\Gamma_{SY}SX} \&\& {\Gamma_{STY}SX} \& {\Gamma^{\circ}_{STY}SX} \&\&\&\& {ST\Gamma^\bullet_{ST|Y}X}
	\arrow["{ST(h^{\bullet\omega}_{X,TY} + S\eta^T_X + \eta^{ST}_Y)}", from=1-6, to=1-8]
	\arrow[Rightarrow, no head, from=1-8, to=2-8]
	\arrow["{ST(ST\uparrow^T + STX + STY)}"', from=2-8, to=3-8]
	\arrow["{ST(S\mu^T(\Gamma_{S\mu^TY}X + X + Y)+STX + STY)}"', from=3-8, to=4-8]
	\arrow["{ST[ST in_1, ST in_2, ST in_3]}"', from=4-8, to=5-8]
	\arrow["{\mu^{ST} [in_1,in_2,in_3,in_2,in_3]}"', from=5-8, to=6-8]
	\arrow[Rightarrow, no head, from=6-8, to=7-8]
	\arrow["{\Gamma_{SY} \mu^S_X}"', from=1-1, to=7-1]
	\arrow["{\widetilde{\widetilde{h^{\circ}}^\omega}_{SX,Y}}", from=1-4, to=1-6]
	\arrow["{\Gamma_{S\eta^T_Y}SSX}", from=1-1, to=1-3]
	\arrow["{in_1}", from=1-3, to=1-4]
	\arrow["{\widetilde{\widetilde{h^{\circ}}^\omega}_{X,Y}}"', from=7-4, to=7-8]
	\arrow["{\Gamma_{S\eta^T_Y}SX}"', from=7-1, to=7-3]
	\arrow["{in_1}"', from=7-3, to=7-4]
	\arrow[""{name=0, anchor=center, inner sep=0}, "{\Gamma^\circ_{STY}\mu^S_X}"', from=1-4, to=7-4]
	\arrow[""{name=1, anchor=center, inner sep=0}, "{ST(\widetilde{\widetilde{h^\circ}^\omega}_{X,Y}+Y)}"', from=1-6, to=5-6]
	\arrow[Rightarrow, no head, from=5-6, to=6-6]
	\arrow["{\mu^{STO_Y}_{\Gamma^\circ_{STY}X}}"{description}, from=6-6, to=7-8]
	\arrow["{\text{(interchange)}}"', draw=none, from=1-1, to=7-4]
	\arrow["{\text{(Lemma~\ref{lem:coh:hcirctildeomegatilde})}}"', Rightarrow, draw=none, from=0, to=1]
\end{tikzcd}
    \end{turn}}
    \caption{Proof of Lemma~\ref{lem:coh:hbulletomegamu}}
    \label{fig:qgammaSmu}
  \end{figure*}
  which further reduces as in Figure~\ref{fig:qgammaSmudeux}.
  \begin{figure*}[pth]
  \adjustbox{max width=\columnwidth,max height=.95\textheight}{%
    \begin{turn}{90}
\begin{tikzcd}[ampersand replacement=\&]
	{ST(\Gamma_{STY}SX + SX + Y)} \&\&\& {ST(\Gamma_{STTY}SX + SX + Y)} \&\&\& {ST(\Gamma^\circ_{STTY}SX + SX + Y)} \&\&\&\& {ST(ST\Gamma^\bullet_{ST|TY}X+ STX + STY)} \\
	\&\&\& {ST(\Gamma^\circ_{STY}SX + SX + Y)} \&\&\&\&\&\&\& {ST(ST(\Gamma_{STTY}X + X + TY)+STX + STY)} \\
	\&\&\&\&\&\& {ST(\Gamma^\circ_{STY}SX + SX + Y)} \&\&\&\& {ST(STT(\Gamma_{STTY}X + X + Y)+STX + STY)} \\
	\&\&\&\&\&\&\&\&\&\& {ST(ST(\Gamma_{STY}X + X + Y)+STX + STY)} \\
	\&\&\&\&\&\& {ST (ST O_Y \Gamma^\circ_{STY} X + STX + Y)} \&\&\&\& {STST(\Gamma_{STY}X + X + Y + X + Y)} \\
	\&\&\&\&\&\& {ST  (ST O_Y \Gamma^\circ_{STY} X +  Y)} \&\& {ST (ST O_Y \Gamma^\circ_{STY} X +  STY)} \& {ST ST ( O_Y \Gamma^\circ_{STY} X +  Y)} \& {STST(\Gamma_{STY}X + X + Y)} \\
	\\
	{ST(ST(\Gamma^\circ_{STY}X + Y) + Y)} \&\&\&\&\&\&\&\&\&\& {ST\Gamma^\bullet_{ST|Y}X}
	\arrow["{ST(h^{\bullet\omega}_{X,TY} + S\eta^T_X + \eta^{ST}_Y)}", curve={height=-30pt}, from=1-1, to=1-11]
	\arrow["{ST(S\mu^T(\Gamma_{S\mu^TY}X + X + Y)+STX + STY)}", from=3-11, to=4-11]
	\arrow["{ST[ST in_1, ST in_2, ST in_3]}", from=4-11, to=5-11]
	\arrow["{STST [in_1,in_2,in_3,in_2,in_3]}", from=5-11, to=6-11]
	\arrow["{\mu^{ST}}", from=6-11, to=8-11]
	\arrow["{ST(\Gamma_{S\eta^T_{TY}}SX + SX + Y)}"', from=1-1, to=1-4]
	\arrow["{ST(in_1+SX+Y)}", from=1-4, to=1-7]
	\arrow["{ST(\widetilde{\widetilde{h^\circ}^\omega}_{X,TY} + S\eta^T_X + \eta^{ST}_YY)}"', from=1-7, to=1-11]
	\arrow[Rightarrow, no head, from=1-11, to=2-11]
	\arrow["{ST(ST\uparrow^T + STX + STY)}", from=2-11, to=3-11]
	\arrow["{ST(in_1+SX+Y)}"'{pos=0.3}, curve={height=12pt}, from=1-1, to=2-4]
	\arrow["{ST(\Gamma^\circ_{S\eta^T_{TY}}+SX+Y)}"', from=2-4, to=1-7]
	\arrow["{ST(\Gamma^\circ_{S\mu^T_Y}SX + SX + Y)}", from=1-7, to=3-7]
	\arrow[Rightarrow, no head, from=2-4, to=3-7]
	\arrow["{ST(\widetilde{\widetilde{h^\circ}^\omega}_{X,Y}+S\eta^T_X+Y)}"', from=3-7, to=5-7]
	\arrow["{ST O_Y (ST O_Y \Gamma^\circ_{STY} X + STX + \eta^{ST}_Y)}", from=5-7, to=4-11]
	\arrow["{ST ([ST O_Y \Gamma^\circ_{STY} X, STin_2] + Y)}", from=5-7, to=6-7]
	\arrow["{ST (ST O_Y \Gamma^\circ_{STY} X +  \eta^{ST}_Y)}"', from=6-7, to=6-9]
	\arrow["{ST[ST in_1, ST in_2]}"', from=6-9, to=6-10]
	\arrow["{STST(\Gamma^\circ_{STY}X + [Y,Y])}"', from=6-10, to=6-11]
	\arrow["{ST(\widetilde{\widetilde{h^\circ}^\omega}_{X,Y} + Y)}"', from=1-1, to=8-1]
	\arrow[Rightarrow, no head, from=8-1, to=6-7]
	\arrow["{\mu^{STO_Y}}"', from=8-1, to=8-11]
	\arrow["{\text{(Lemma~\ref{lem:coh:hcirctildeomegatilde})}}"{description}, draw=none, from=5-7, to=1-11]
      \end{tikzcd}
    \end{turn}}
    \caption{Proof of Lemma~\ref{lem:coh:hbulletomegamu}, part 2}
    \label{fig:qgammaSmudeux}
  \end{figure*}
\end{proof}

\subsection{The incremental lifting}
\begin{proposition}\label{prop:mualg}
  For any object $X$ and $(T⊕T')$-algebra structure
  $$TX \xto{𝐚} X \xot{𝐛} Γ_{SX}X\rlap{,}$$
  the $Γ_S$-algebra structure on $SX$ given by
  $$Γ_{SSX}SX \xto{Γ_{μ^S_X}SX} Γ_{SX}SX \xto{h^{•ω}_{X,X}}
  ST Γ^•_{ST|X} X \xto{ST (Γ_{S𝐚}X + [X,X])} ST (Γ_{SX}X + X) \xto{ST
    [𝐛,X]}STX \xto{S𝐚} SX$$ defines an incremental lifting of $S$ to
  $T'\Alg$ along $δ$.
\end{proposition}
\begin{proof}
  The given composite readily equips $SX$ with $(T⊕T')$-algebra
  structure by Lemma~\ref{lem:monad:coprod:alg} and the isomorphism
  $Γ_S\alg ≅ T'\Alg$.  It remains to verify that the unit and
  multiplication are $Γ_S$-algebra morphisms. The former follows
  easily by Lemma~\ref{lem:coh:hbulletomegamu}, and the latter
  by Figure~\ref{fig:qmualg}. There,  the middle, unlabelled polygon
  is  chased
  as follows,
  \begin{center}
    \ajustedroit{
      \begin{tikzcd}[ampersand replacement=\&]
	{STO_{SX} O_{SX} STA} \& {STO_{SX} SO_{SX} STA} \&\& {STSO_{X} O_{SX} STA} \&\&\& {SSTO_{STX} O_{STX} STA} \& {STO_{STX} O_{STX} STA} \&\&\&\& {STSTO_{X} O_{X} A} \\
	\&\&\&\&\&\&\&\&\&\& {SSTSTO_{X} O_{X} A} \\
	\& {STO_{STX} SO_{STX} STA} \&\& {STO_{SSTX} SO_{STX} STA} \&\& {STSO_{STX} O_{STX} STA} \&\& {STSO_{STX} ST O_X A} \&\& {STSSTO_{X} O_{X} A} \& {SSSTTO_{X} O_{X} A} \& {SSTTO_{X} O_{X} A} \\
	{STO_{STX} O_{STX} ST A} \&\&\& {STO_{STX} STO_{X}  A} \&\&\&\&\&\& {STSTO_{X} O_{X}  A} \& {SSTTO_{X} O_{X} A} \& {STO_{X} O_{X} A}
	\arrow["{S\delta O_{\eta^{ST}_X} O_{S\eta^T_X}}", from=1-4, to=1-7]
	\arrow["{\mu^S}", from=1-7, to=1-8]
	\arrow["{ST[ST in_1, ST in_2, ST in_3]}", from=1-8, to=1-12]
	\arrow["{ST O_{S\eta^T_X} O_{S\eta^T_X}}"', from=1-1, to=4-1]
	\arrow["{ST O_{SX} \eta^S}", from=1-1, to=1-2]
	\arrow["{ST[S in_1, S in_2]}", from=1-2, to=1-4]
	\arrow["{ST O_{S\eta^T_X} S O_{S\eta^T_X} ST O_X \Gamma^\circ_{STX} X}", from=1-2, to=3-2]
	\arrow["{STO_{STX}  \eta^S}"{pos=0.3}, from=4-1, to=3-2]
	\arrow["{ST O_{STX} [ST in_1, ST in_2]}", from=4-1, to=4-4]
	\arrow["{ST[ST in_1, ST in_2]}", from=4-4, to=4-10]
	\arrow["{STO_{S\eta^S_{TX}} SO_{STX} STA}", from=3-2, to=3-4]
	\arrow["{ST[S in_1, S in_2]}", from=3-4, to=3-6]
	\arrow["S\delta", from=3-6, to=1-7]
	\arrow["S\delta", from=3-10, to=2-11]
	\arrow["{\mu^S}", from=2-11, to=1-12]
	\arrow["{ST \mu^S}", from=3-10, to=4-10]
	\arrow["{STSO_{STX} [ST in_1, ST in_2]}", from=3-6, to=3-8]
	\arrow["{STS[ST in_1, ST in_2]}", from=3-8, to=3-10]
	\arrow["S\delta", from=1-12, to=3-12]
	\arrow["{\mu^S\mu^T}", from=3-12, to=4-12]
	\arrow["S\delta", from=4-10, to=4-11]
	\arrow["{\mu^S\mu^T}", from=4-11, to=4-12]
	\arrow["SS\delta", from=2-11, to=3-11]
	\arrow["{\mu^S}", from=3-11, to=3-12]
	\arrow[Rightarrow, no head, from=4-11, to=3-12]
      \end{tikzcd}      
      }
  \end{center}
  where the middle, bottom polygon is the image by $ST$ of a
  diagram whose domain is a coproduct, and whose commutativity may be checked termwise.
  E.g., the second term is chased as follows.
  \begin{center}
    \hfill
    \ajustedroit[.95]{
\begin{tikzcd}[ampersand replacement=\&,baseline=(\tikzcdmatrixname-5-3.base)]
	STX \&\&\&\&\& SSTX \&\&\&\&\&\&\&\&\&\&\& {SST(A+X)} \\
	\&\&\&\&\& {S(STA + STX)} \&\&\&\& {S(STA + STX)} \\
	{O_{STX} O_{STX} ST A} \&\&\&\&\& {O_{STX} SO_{STX} STA} \&\&\&\& {O_{SSTX} SO_{STX} STA} \&\&\& {SO_{STX} O_{STX} STA} \&\&\&\& {SO_{STX} ST O_X A} \&\&\& {SSTO_{X} O_{X} A} \\
	\&\& {ST(A+X)} \&\&\&\&\&\&\& STX \\
	\&\& {O_{STX}ST O_{X}  A} \&\&\&\&\&\&\&\&\&\&\&\&\&\&\&\&\&\& {STO_{X} O_{X}  A}
	\arrow["{O_{STX}  \eta^S}"'{pos=0.6}, from=3-1, to=3-6]
	\arrow["{O_{STX} [ST in_1, ST in_2]}"'{pos=0.2}, from=3-1, to=5-3]
	\arrow["{[ST in_1, ST in_2]}", from=5-3, to=5-21]
	\arrow["{O_{S\eta^S_{TX}} SO_{STX} STA}"{pos=.3}, from=3-6, to=3-10]
	\arrow["{[S in_1, S in_2]}", from=3-10, to=3-13]
	\arrow["{ \mu^S}", from=3-20, to=5-21]
	\arrow["{SO_{STX} [ST in_1, ST in_2]}", from=3-13, to=3-17]
	\arrow["{S[ST in_1, ST in_2]}", from=3-17, to=3-20]
	\arrow["{in_2}"', from=1-1, to=3-1]
	\arrow["{\eta^S}", from=1-1, to=1-6]
	\arrow["{Sin_2}", from=1-6, to=2-6]
	\arrow["{in_1}", from=2-6, to=3-6]
	\arrow[Rightarrow, no head, from=2-6, to=2-10]
	\arrow["{in_1}", from=2-10, to=3-10]
	\arrow["{Sin_2}", from=1-6, to=2-10]
	\arrow["{S in_1}", from=2-10, to=3-13]
	\arrow["{SST in_2}", from=1-6, to=1-17]
	\arrow["{Sin_1}", from=1-17, to=3-17]
	\arrow["{S[STin_1,STin_2]}"', from=2-10, to=1-17]
	\arrow["{SSTin_1}", from=1-17, to=3-20]
	\arrow["{\mu^S}"'{pos=0.9}, fore, curve={height=-6pt}, from=1-6, to=4-10]
	\arrow["{STin_2}", from=4-10, to=5-21]
	\arrow[curve={height=34pt}, Rightarrow, no head, from=1-1, to=4-10]
	\arrow["{ST in_2}", from=1-1, to=4-3]
	\arrow["{in_1}", from=4-3, to=5-3]
\end{tikzcd}      
      }
    \qedhere
  \end{center}
\end{proof}

\begin{figure*}[p]
  \centering
  \adjustbox{max width=\columnwidth,max height=.95\textheight}{%
    \begin{turn}{90}
\begin{tikzcd}[ampersand replacement=\&]
	{\Gamma_{SSSX}SSX} \&\&\&\&\&\& {\Gamma_{SSX} SSX} \&\&\&\&\& {\Gamma_{SSX}SX} \\
	{\Gamma_{SSX}SSX} \&\&\&\&\&\& {\Gamma_{SX}SSX} \&\&\&\&\& {\Gamma_{SX}SX} \\
	\\
	{ST \Gamma^\bullet_{ST|SX}SX} \&\& {STS \Gamma^\bullet_{STS|X} SX} \& {STS \Gamma^\bullet_{SST|X} SX} \& {SST \Gamma^\bullet_{SST|X}SX} \\
	{ST(\Gamma_{STSX}SX + 2\cdot SX)} \\
	{ST(\Gamma_{SSTX}SX + 2\cdot SX)} \&\& {ST(\Gamma_{STX}SX + 2 \cdot SX)} \& {STS\Gamma^\bullet_{ST|X}SX } \& {SST \Gamma^\bullet_{ST|X}SX} \&\& {ST \Gamma^\bullet_{ST|X}SX} \\
	{ST(\Gamma_{SSX} SX + 2 \cdot SX)} \\
	{ST(\Gamma_{SX}SX + 2 \cdot SX)} \& {} \& {ST(ST\Gamma^\bullet_{ST|TX}X + 2 \cdot SX)} \& {STS(ST \Gamma^\bullet_{ST|TX}X + SX + X)} \& {SST(ST \Gamma^\bullet_{ST|TX}X + STX + STX)} \&\& {ST(ST\Gamma^\bullet_{ST|TX}X + STX+STX)} \& {ST(STT\Gamma^\bullet_{STT|X} X + 2 \cdot STX)} \& {ST(ST\Gamma^\bullet_{ST|X} X + 2 \cdot STX)} \& {STST(\Gamma^\bullet_{ST|X}X + X + X)} \\
	\&\& {ST (ST(\Gamma_{STX}X + X + TX) + 2 \cdot SX)} \&\& {ST (STT(\Gamma_{STX}X + 2\cdot X) + 2 \cdot SX)} \&\& {ST (ST(\Gamma_{STX}X + 2\cdot X) + 2 \cdot STX)} \&\& {ST ST(\Gamma_{STX}X + 2\cdot X + 2 \cdot X)} \& {ST(\Gamma^\bullet_{ST|X}X + X + X)} \&\& {ST \Gamma^\bullet_{ST|X}X} \& {} \\
	{ST(ST \Gamma^\bullet_{ST|X}X + 2 \cdot SX)} \&\& {ST (ST(\Gamma_{SX}X + X + TX) + 2 \cdot SX)} \&\& {ST (STT(\Gamma_{SX}X + 2\cdot X) + 2 \cdot SX)} \&\& {ST (ST(\Gamma_{SX}X + 2\cdot X) + 2 \cdot STX)} \&\& {ST ST(\Gamma_{SX}X + 2\cdot X + 2 \cdot X)} \&\&\& {ST \Gamma^\bullet_{S|X} X} \\
	\&\& {ST (ST(X + X + TX) + 2 \cdot SX)} \&\& {ST (STT(2\cdot X + X) + 2 \cdot SX)} \&\& {ST (ST(2\cdot X + X) + 2 \cdot STX)} \&\& {ST ST(2\cdot X + X + 2 \cdot X)} \&\&\& {ST X} \\
	{ST(ST \Gamma^\bullet_{S|X}X  + 2 \cdot SX)} \\
	{ST(STX + 2 \cdot SX)} \&\& {STS(T(2\cdot X + TX)+2\cdot X)} \&\&\&\& {STS(T(X + 2 \cdot X) + 2 \cdot TX)} \&\& {SSTT(2\cdot X + X + 2 \cdot X)} \\
	{ST(SX + 2 \cdot SX)} \& {STS(TX + 2 \cdot X)} \& {SST(T(2\cdot X + TX) + 2\cdot X)} \&\&\&\& {SST(T(X + 2 \cdot X) + 2 \cdot TX)} \\
	\& {STS (X+2\cdot X)} \& {ST(T(2\cdot X + TX) + 2\cdot X)} \&\& {ST(TT(2\cdot X + X) + 2\cdot X)} \&\& {ST(T(2\cdot X + X) + 2\cdot TX)} \& {STT(2\cdot X + X + 2 \cdot X)} \& STTX \\
	\&\& {ST(T(2\cdot X + X) + 2\cdot X)} \&\&\&\& {ST(TX + 2\cdot X)} \\
	STSX \\
	SSTX \&\&\&\&\&\&\&\& STX \\
	SSX \&\&\&\&\&\&\&\&\&\&\& SX
	\arrow[""{name=0, anchor=center, inner sep=0}, "{\Gamma_{S\mu^S_X}SSX}", from=1-1, to=1-7]
	\arrow[""{name=1, anchor=center, inner sep=0}, "{\Gamma_{SSX} \mu^S_X}", from=1-7, to=1-12]
	\arrow["{\Gamma_{\mu^S_X}SX}", from=1-12, to=2-12]
	\arrow["{\bar{h}_{X,X}}", color={rgb,255:red,153;green,92;blue,214}, from=2-12, to=9-12]
	\arrow["{ST(\Gamma_{S\mathbf{a}}X + 2\cdot X)}", color={rgb,255:red,250;green,56;blue,69}, from=9-12, to=10-12]
	\arrow["{ST[\mathbf{b},X,X]}", color={rgb,255:red,1;green,90;blue,24}, from=10-12, to=11-12]
	\arrow["{S\mathbf{a}}", color={rgb,255:red,250;green,56;blue,69}, from=11-12, to=19-12]
	\arrow["{\Gamma_{\mu^S_{SX}} SSX}"', from=1-1, to=2-1]
	\arrow["{\bar{h}_{SX,SX}}"', color={rgb,255:red,153;green,92;blue,214}, from=2-1, to=4-1]
	\arrow[Rightarrow, no head, from=4-1, to=5-1]
	\arrow["{ST(\Gamma_{S\delta_X}SX + 2\cdot SX)}"', from=5-1, to=6-1]
	\arrow["{ST(\Gamma_{SS\mathbf{a}}SX + 2 \cdot SX)}"', color={rgb,255:red,250;green,56;blue,69}, from=6-1, to=7-1]
	\arrow["{ST(\Gamma_{\mu^S_X} SX + 2\cdot SX)}"', from=7-1, to=8-1]
	\arrow["{ST(\bar{h}_{X,X} + 2\cdot SX)}"', color={rgb,255:red,153;green,92;blue,214}, from=8-1, to=10-1]
	\arrow["{ST(ST(\Gamma_{S\mathbf{a}}X + 2\cdot X + 2 \cdot SX))}"', color={rgb,255:red,250;green,56;blue,69}, from=10-1, to=12-1]
	\arrow["{ST(ST[\mathbf{b},X,X] + 2 \cdot SX)}"', color={rgb,255:red,1;green,90;blue,24}, from=12-1, to=13-1]
	\arrow["{ST(S\mathbf{a} + 2 \cdot SX)}"', color={rgb,255:red,250;green,56;blue,69}, from=13-1, to=14-1]
	\arrow["{ST[SX,SX,SX]}"', from=14-1, to=17-1]
	\arrow["{S\delta_X}"', from=17-1, to=18-1]
	\arrow["{SS\mathbf{a}}"', color={rgb,255:red,250;green,56;blue,69}, from=18-1, to=19-1]
	\arrow[""{name=2, anchor=center, inner sep=0}, "{\mu^S_X}"', from=19-1, to=19-12]
	\arrow[""{name=3, anchor=center, inner sep=0}, "{\Gamma_{\mu^S_X}SSX}"', from=2-1, to=2-7]
	\arrow["{\Gamma_{\mu^S_X} SSX}"', from=1-7, to=2-7]
	\arrow[""{name=4, anchor=center, inner sep=0}, "{\Gamma_{SX} \mu^S_X}"', from=2-7, to=2-12]
	\arrow[""{name=5, anchor=center, inner sep=0}, "{S\delta \Gamma^\bullet_{SST | X} SX}"', from=4-4, to=4-5]
	\arrow["{\mu^S T \Gamma^\bullet_{\mu^S T |X} SX}"{description}, from=4-5, to=6-7]
	\arrow["{\bar{h}_{SX,X}}"', color={rgb,255:red,153;green,92;blue,214}, from=2-7, to=6-7]
	\arrow["{ST(\bar{h}_{X,TX} + S\eta^T X + \eta^{ST}X)}", color={rgb,255:red,153;green,92;blue,214}, from=6-7, to=8-7]
	\arrow["{ST(ST\uparrow^T + 2\cdot STX)}", from=8-7, to=8-8]
	\arrow["{ST(S \mu^T \Gamma^\bullet_{S\mu^T|X} X + 2 \cdot STX)}", from=8-8, to=8-9]
	\arrow["{ST[STin_1,STin_2,STin_3]}", from=8-9, to=8-10]
	\arrow[""{name=6, anchor=center, inner sep=0}, "{ST(\Gamma_{\mu^S_{TX}}SX + 2 \cdot SX)}"{pos=0.7}, from=6-1, to=6-3]
	\arrow[""{name=7, anchor=center, inner sep=0}, "{ST(\Gamma_{S\mathbf{a}}SX + 2\cdot SX)}"{description}, color={rgb,255:red,250;green,56;blue,69}, from=6-3, to=8-1]
	\arrow["{ST(\bar{h}_{X,TX}+2\cdot SX)}"{description}, color={rgb,255:red,153;green,92;blue,214}, from=6-3, to=8-3]
	\arrow[""{name=8, anchor=center, inner sep=0}, "{ST(ST\Gamma^\bullet_{ST|\mathbf{a}}X + 2 \cdot SX)}"'{pos=0.6}, color={rgb,255:red,250;green,56;blue,69}, from=8-3, to=10-1]
	\arrow["{ST(\uparrow^S)}", from=6-3, to=6-4]
	\arrow[""{name=9, anchor=center, inner sep=0}, "{S\delta \Gamma^\bullet_{ST|X}SX}"', from=6-4, to=6-5]
	\arrow[""{name=10, anchor=center, inner sep=0}, "{\mu^S T \Gamma^\bullet_{ST|X}SX}"', from=6-5, to=6-7]
	\arrow["{ST(\uparrow^S)}", from=8-3, to=8-4]
	\arrow["{S \delta (ST \Gamma^\bullet_{ST|TX} X + S\eta^T_X + \eta^{ST}X)}", from=8-4, to=8-5]
	\arrow["{\mu^S T (ST \Gamma^\bullet_{ST|TX} X + TX + STX)}"{pos=0.4}, from=8-5, to=8-7]
	\arrow["{ST(ST(\Gamma_{S\mu^TX}X + X + TX) + 2 \cdot SX)}", from=8-3, to=9-3]
	\arrow["{ST(ST\uparrow^T + 2 \cdot SX)}", from=9-3, to=9-5]
	\arrow["{ST (S\mu^T(\Gamma_{STX}X + 2\cdot X) + 2 \cdot S\eta^TX)}", from=9-5, to=9-7]
	\arrow["{ST[STin_1,STin_2,STin_3]}", from=9-7, to=9-9]
	\arrow["{ST (ST(\Gamma_{S\mathbf{a}}X + X + TX) + 2 \cdot SX)}", color={rgb,255:red,250;green,56;blue,69}, from=9-3, to=10-3]
	\arrow["{ST(ST\uparrow^T + 2 \cdot SX)}", from=10-3, to=10-5]
	\arrow["{ST (S\mu^T(\Gamma_{SX}X + 2\cdot X) + 2 \cdot S\eta^TX)}", from=10-5, to=10-7]
	\arrow["{ST[STin_1,STin_2,STin_3]}", from=10-7, to=10-9]
	\arrow[""{name=11, anchor=center, inner sep=0}, "{\mu^{ST}(\Gamma_{SX}X + \nabla)}", from=10-9, to=10-12]
	\arrow["{ST (ST(\mathbf{b} + X + TX) + 2 \cdot SX)}", color={rgb,255:red,1;green,90;blue,24}, from=10-3, to=11-3]
	\arrow["{ST(ST\uparrow^T + 2 \cdot SX)}", from=11-3, to=11-5]
	\arrow["{ST (S\mu^T(2\cdot X + X) + 2 \cdot S\eta^TX)}"', from=11-5, to=11-7]
	\arrow["{ST[STin_1,STin_2,STin_3]}", from=11-7, to=11-9]
	\arrow[""{name=12, anchor=center, inner sep=0}, "{\mu^{ST}[X,X,X,X,X]}"', from=11-9, to=11-12]
	\arrow[""{name=13, anchor=center, inner sep=0}, "{ST (ST(\Gamma_{SX}X + X + \mathbf{a}) + 2 \cdot SX)}"'{pos=0.6}, color={rgb,255:red,250;green,56;blue,69}, from=10-3, to=12-1]
	\arrow[""{name=14, anchor=center, inner sep=0}, "{ST(ST[X,X,\mathbf{a}]+ 2 \cdot SX)}"{description, pos=0.4}, color={rgb,255:red,250;green,56;blue,69}, from=11-3, to=13-1]
	\arrow["{ST[Sin_1,Sin_2,Sin_3]}", from=11-3, to=13-3]
	\arrow["S\delta", from=13-3, to=14-3]
	\arrow["{\mu^S}", from=14-3, to=15-3]
	\arrow["{ST(T(2\cdot X + \mathbf{a}) + 2\cdot X)}", color={rgb,255:red,250;green,56;blue,69}, from=15-3, to=16-3]
	\arrow["{ST(T[X,X,X] + 2\cdot X)}", from=16-3, to=16-7]
	\arrow["{ST[\mathbf{a},X,X]}", color={rgb,255:red,250;green,56;blue,69}, from=16-7, to=18-9]
	\arrow["{S\mathbf{a}}", from=18-9, to=19-12]
	\arrow["{ST(T \uparrow^T + 2 \cdot X)}", from=15-3, to=15-5]
	\arrow["{ST(\mu^T(2\cdot X + X) + 2\cdot \eta^T_X)}", from=15-5, to=15-7]
	\arrow["{STT[X,X,X,X,X]}", from=15-8, to=15-9]
	\arrow[""{name=15, anchor=center, inner sep=0}, "{STS(T[X,X,\mathbf{a}] + 2 \cdot X)}"{description}, color={rgb,255:red,250;green,56;blue,69}, from=13-3, to=14-2]
	\arrow["{STS(\mathbf{a} + 2 \cdot X)}", color={rgb,255:red,250;green,56;blue,69}, from=14-2, to=15-2]
	\arrow[""{name=16, anchor=center, inner sep=0}, "{ST[Sin_1,Sin_2,Sin_3]}"{pos=0.9}, from=13-1, to=14-2]
	\arrow[""{name=17, anchor=center, inner sep=0}, "{ST[Sin_1,Sin_2,Sin_3]}"{description}, from=14-1, to=15-2]
	\arrow["{STS[X,X,X]}"{description}, from=15-2, to=17-1]
	\arrow[""{name=18, anchor=center, inner sep=0}, "{\mu^S_{TX}}"{description}, from=18-1, to=18-9]
	\arrow["{ST [Tin_1,Tin_2,Tin_3]}", from=15-7, to=15-8]
	\arrow["{ST[Sin_1,Sin_2,Sin_3]}"'{pos=0.7}, from=11-7, to=13-7]
	\arrow["S\delta", from=13-7, to=14-7]
	\arrow["{\mu^S}", from=14-7, to=15-7]
	\arrow[""{name=19, anchor=center, inner sep=0}, "{STS[Tin_1,Tin_2,Tin_3]}"{description}, from=13-7, to=11-9]
	\arrow["S\delta"{description}, from=11-9, to=13-9]
	\arrow[""{name=20, anchor=center, inner sep=0}, "{\mu^S}"{description}, from=13-9, to=15-8]
	\arrow[""{name=21, anchor=center, inner sep=0}, "{S\mu^T_X}", from=15-9, to=11-12]
	\arrow["{\text{(proved in Coq)}}"{description}, draw=none, from=15-9, to=18-9]
	\arrow["{SS T \Gamma^\bullet_{\mu^S T |X} SX}"{description}, from=4-5, to=6-5]
	\arrow["{ST \uparrow^S}"', from=4-1, to=4-3]
	\arrow["{STS\Gamma^\bullet_{S\delta | X} SX}"', from=4-3, to=4-4]
	\arrow["{STS \Gamma^\bullet_{\mu^S T |X} SX}"', from=4-4, to=6-4]
	\arrow["{\text{(easy)}}"{description}, draw=none, from=4-3, to=6-3]
	\arrow["{STST(\mathbf{b}+2\cdot X + 2 \cdot X)}"', color={rgb,255:red,1;green,90;blue,24}, from=10-9, to=11-9]
	\arrow["{\text{(proved in coq)}}"{description}, draw=none, from=11-5, to=15-5]
	\arrow["{\text{(easy)}}"{description}, Rightarrow, draw=none, from=9-7, to=10-7]
	\arrow["{\text{(easy)}}"{description}, Rightarrow, draw=none, from=10-7, to=11-7]
	\arrow["{\mu^{ST}}", from=8-10, to=9-10]
	\arrow["{[\Gamma^\bullet_{ST|X},\alpha,\beta]}"', from=9-10, to=9-12]
	\arrow["{\mu^{ST}}"', from=9-9, to=9-10]
	\arrow["{\text{(easy)}}"'{pos=0.6}, curve={height=24pt}, Rightarrow, draw=none, from=6, to=7]
	\arrow["{\text{(easy)}}"'{pos=0.4}, Rightarrow, draw=none, from=9, to=8-5]
	\arrow["{\text{(easy)}}"', Rightarrow, draw=none, from=7, to=8]
	\arrow["{\text{(easy)}}"', Rightarrow, draw=none, from=8, to=13]
	\arrow["{\text{(easy)}}"', Rightarrow, draw=none, from=13, to=14]
	\arrow["{\text{(easy)}}"{description}, Rightarrow, draw=none, from=14, to=15]
	\arrow["{\text{(easy)}}"{description}, Rightarrow, draw=none, from=16, to=17]
	\arrow["{\text{(easy)}}"{description, pos=0.3}, Rightarrow, draw=none, from=17, to=17-1]
	\arrow["{\text{(naturality of } S\delta ; \mu^S \text{)}}"{description}, Rightarrow, draw=none, from=16-3, to=18]
	\arrow["{\text{(naturality of } \mu^S \text{)}}"{description}, Rightarrow, draw=none, from=18, to=2]
	\arrow["{\text{(easy)}}"{description}, curve={height=-6pt}, Rightarrow, draw=none, from=10, to=4-5]
	\arrow["{\text{(naturality of } \delta \text{)}}"{description}, Rightarrow, draw=none, from=5, to=9]
	\arrow["{\text{(Lemma~\ref{lem:coh:hbulletomegamu})}}"{description, pos=0.6}, Rightarrow, draw=none, from=3, to=4-4]
	\arrow["{\text{(associativity of } \mu^S \text{)}}"{description}, Rightarrow, draw=none, from=0, to=3]
	\arrow["{\text{(functoriality of } \Gamma \text{)}}"{description}, Rightarrow, draw=none, from=1, to=4]
	\arrow["{\text{(Lemma~\ref{lem:coh:hbulletomegamu})}}"{description}, Rightarrow, draw=none, from=4, to=8-9]
	\arrow["{\text{(easy)}}"{description}, Rightarrow, draw=none, from=11-7, to=19]
	\arrow["{\text{(easy)}}"{description}, Rightarrow, draw=none, from=11, to=12]
	\arrow["{\text{(easy)}}"{description}, Rightarrow, draw=none, from=11-9, to=21]
	\arrow["{\text{(naturality of } S\delta ; \mu^S \text{)}}"{description}, Rightarrow, draw=none, from=19, to=20]
\end{tikzcd}
    \end{turn}}%
  \caption{Multiplication is an algebra morphism}
  \label{fig:qmualg}
\end{figure*}
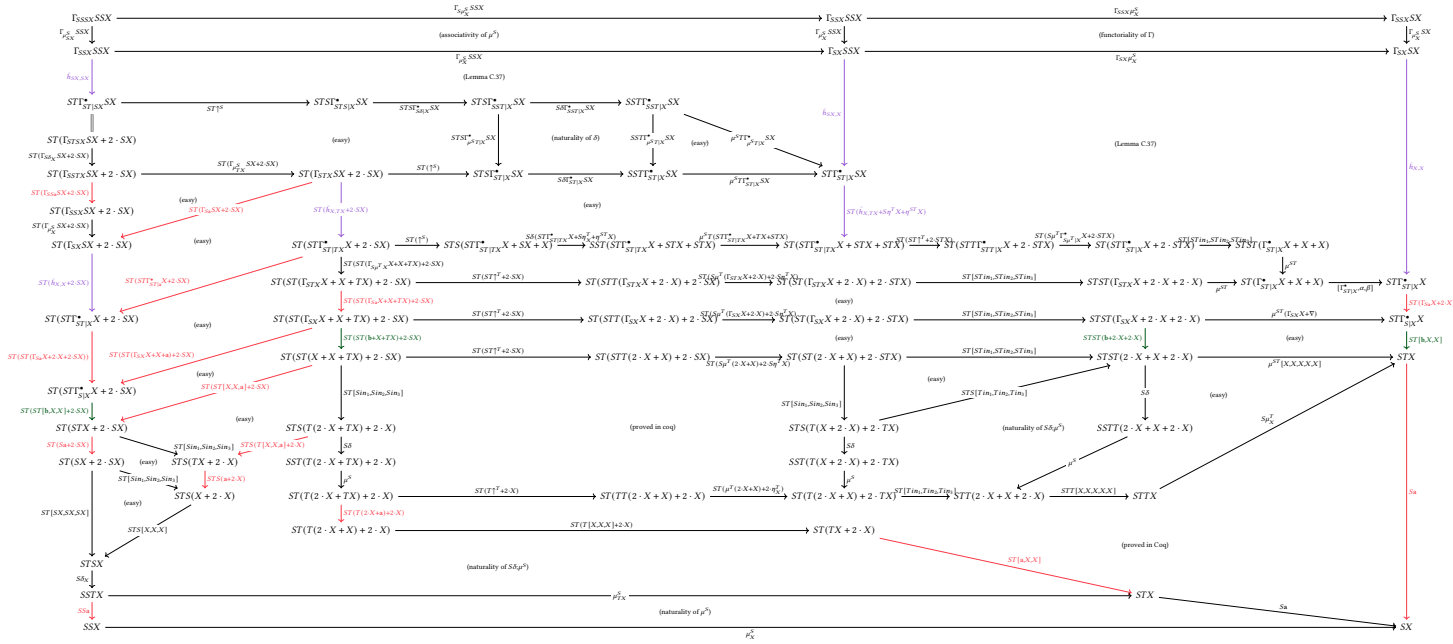

It remains to show:
\begin{proposition}\label{prop:d:delta}
  The distributive law $d_{/δ}∶ (T⊕T')S → S(T⊕T')$ satisfies~\eqref{eq:d:delta},
  which we repeat here for convenience.
  \begin{center}
\begin{tikzcd}[ampersand replacement=\&]
	{Γ_X ΣX} \&\& {ST(Γ_{STX}X + X + X)} \\
	\&\& {ST(Γ_{STX}TX + X)} \\
	{Γ_{SX}SX} \&\& {ST(T'TX + X)} \\
	{Γ_{SSX}SX} \&\& {ST((T ⊕ T')(T⊕T')X + X)} \\
	{T'SX} \&\& {S(T⊕T')(T⊕T')X} \\
	{(T ⊕ T')SX} \&\& {S(T⊕T')X}
	\arrow["{d_{X,X}}", from=1-1, to=1-3]
	\arrow["{Γ_{η^S_X} η_{Σ,X}}"', from=1-1, to=3-1]
	\arrow["{Γ_{η^S_{SX}}SX}"', from=3-1, to=4-1]
	\arrow["{d_{/δ}X}"', from=6-1, to=6-3]
	\arrow["{Sμ^{T⊕T'}_X}", from=5-3, to=6-3]
	\arrow["{η_{Γ_S,SX}}"', from=4-1, to=5-1]
	\arrow["{in_2 SX}"', from=5-1, to=6-1]
	\arrow["{Sin₁[μ^{T⊕T'}_X, η^{T⊕T'}X]}", from=4-3, to=5-3]
	\arrow["{ST(in₂ in₁ X + X)}", from=3-3, to=4-3]
	\arrow["{ST(Γ_{STX} ηᵀ_X + [X, X])}", from=1-3, to=2-3]
	\arrow["{ST(\eta_{\Gamma_S,TX} + X)}", from=2-3, to=3-3]
      \end{tikzcd}
    \end{center}  
\end{proposition}

We need the following intermediate result.
\begin{lemma}\label{lem:d:delta:Y}
  For any object $Y$, the following diagram commutes.
  \begin{center}
\begin{tikzcd}[ampersand replacement=\&]
	{Γ_{SSY} SY} \&\& {Γ_{SY} SY} \&\& {ST(Γ_{STY} Y + Y + Y)} \\
	{T'SY} \&\&\&\& {ST(Γ_{STY} TY + Y)} \\
	\&\&\&\& {ST(T'TY + Y)} \\
	\&\&\&\& {S(T⊕T')(T⊕T')Y} \\
	{(T ⊕ T')SY} \&\&\&\& {S(T⊕T')Y}
	\arrow["{d^{•ω}_{Y,Y}}", from=1-3, to=1-5]
	\arrow["{ST(Γ_{STY} \eta^T_Y + [Y, Y])}", from=1-5, to=2-5]
	\arrow["{in₂SY}"', from=2-1, to=5-1]
	\arrow["{d_{/δ}{Y}}"', from=5-1, to=5-5]
	\arrow["{Γ_{μ^SY} SY}", from=1-1, to=1-3]
	\arrow["{η_{Γ_S,SY}}"', from=1-1, to=2-1]
	\arrow["{ST(η_{Γ_S,TY} + Y)}", from=2-5, to=3-5]
	\arrow["{S\mu^{T\oplus T'}_Y}", from=4-5, to=5-5]
	\arrow["{Sin_1[(\mu^{T \oplus T'}_Y \circ in₂in_1),\eta^{T\oplus T'}_Y]}", from=3-5, to=4-5]
      \end{tikzcd}
    \end{center}
\end{lemma}
\begin{proof}
  Up to some easy rewriting of the right-hand composite, this is
  proved in Figure~\ref{fig:d:delta:Y}.
\end{proof}
\begin{figure}[p]
  \adjustbox{max width=\columnwidth,max height=.95\textheight}{
  \begin{turn}{90}
\begin{tikzcd}[ampersand replacement=\&]
	{Γ_{SSY} SY} \&\&\& {Γ_{SY} SY} \&\&\& {ST(Γ_{STY} Y + Y + Y)} \\
	{T'SY} \\
	\\
	\& {\Gamma_{S(T\oplus T')SY} (T \oplus T')SY} \&\& {\Gamma_{SS(T\oplus T')Y} S(T \oplus T')Y} \& {\Gamma_{S(T\oplus T')Y} S(T \oplus T')Y} \\
	\& {T'(T \oplus T')SY} \&\&\& {ST(\Gamma_{ST(T\oplus T')Y} (T \oplus T')Y + (T \oplus T')Y + (T \oplus T')Y)} \&\& {ST(Γ_{S(T\oplus T')Y} (T\oplus T')Y + Y)} \\
	\& {(T \oplus T')(T \oplus T')SY} \&\&\& {ST(\Gamma_{S(T\oplus T')Y} (T \oplus T')Y + (T \oplus T')Y)} \\
	\\
	\&\&\&\& {ST(T'(T \oplus T')Y + (T \oplus T')Y)} \&\& {ST(T'(T\oplus T')Y + Y)} \\
	\&\&\&\& {ST((T \oplus T')(T \oplus T')Y + (T \oplus T')Y)} \\
	\&\&\&\& {ST(T \oplus T')Y} \&\& {S(T⊕T')(T⊕T')Y} \\
	\&\&\&\& {S(T \oplus T')(T \oplus T')Y} \\
	\\
	{(T ⊕ T')SY} \&\&\&\&\& {S(T⊕T')Y}
	\arrow["{d^{•ω}_{Y,Y}}"', from=1-4, to=1-7]
	\arrow["{ST(Γ_{S in_1} \eta^{T\oplus T'}_Y + [Y, Y])}"{description}, from=1-7, to=5-7]
	\arrow["{in₂SY}"', from=2-1, to=13-1]
	\arrow["{\eta_\Gamma}", color={rgb,255:red,33;green,131;blue,33}, from=4-2, to=5-2]
	\arrow["{in_2}"{description}, from=5-2, to=6-2]
	\arrow["\mu"{description}, from=6-2, to=13-1]
	\arrow["{\Gamma_{Sd_{/\delta} Y}d_{/\delta}Y}"{description}, from=4-2, to=4-4]
	\arrow["{\Gamma_\mu}", from=4-4, to=4-5]
	\arrow["{d^{\bullet \omega}}", from=4-5, to=5-5]
	\arrow["{ST(\Gamma_{S(in_1;\mu)Y} (T \oplus T')Y + [ (T \oplus T')Y,(T \oplus T')Y])}"', from=5-5, to=6-5]
	\arrow["{S in_1}"', from=10-5, to=11-5]
	\arrow["\eta"{description}, curve={height=18pt}, from=2-1, to=5-2]
	\arrow["\eta"{description}, from=1-4, to=4-5]
	\arrow["\eta"{description}, from=1-7, to=5-5]
	\arrow[""{name=0, anchor=center, inner sep=0}, "\eta"{description}, from=5-7, to=6-5]
	\arrow[""{name=1, anchor=center, inner sep=0}, "{d_{/δ}{Y}}"', from=13-1, to=13-6]
	\arrow["{Γ_{μ^SY} SY}"', from=1-1, to=1-4]
	\arrow["{η_{Γ_S,SY}}"', color={rgb,255:red,33;green,131;blue,33}, from=1-1, to=2-1]
	\arrow["\eta"', from=1-1, to=4-2]
	\arrow["\eta"{description}, curve={height=24pt}, from=1-1, to=4-4]
	\arrow["{ST[\mu^{T \oplus T'}_Y,(T \oplus T')Y]}"', from=9-5, to=10-5]
	\arrow["{ST(\eta_{\Gamma_S}+(T \oplus T')Y)}"', color={rgb,255:red,33;green,131;blue,33}, from=6-5, to=8-5]
	\arrow["{ST(in_2+(T \oplus T')Y)}"', from=8-5, to=9-5]
	\arrow["{S \mu}"', from=11-5, to=13-6]
	\arrow["{ST(η_{Γ_S,(T\oplus T')Y} + Y)}", color={rgb,255:red,33;green,131;blue,33}, from=5-7, to=8-7]
	\arrow["{S\mu^{T\oplus T'}_Y}", from=10-7, to=13-6]
	\arrow["{Sin_1[(\mu^{T \oplus T'}_Y \circ in₂), \eta^{T\oplus T'}_Y]}", from=8-7, to=10-7]
	\arrow["{\text{(} d_{/ \delta} Y \text{ is a } \Gamma_S\text{-algebra morphism)}}"{description}, Rightarrow, draw=none, from=1, to=4-4]
	\arrow["{\text{(easy)}}"{description}, Rightarrow, draw=none, from=0, to=13-6]
      \end{tikzcd}
    \end{turn}}
\caption{Proof of Lemma~\ref{lem:d:delta:Y}}
\label{fig:d:delta:Y}
\end{figure}

\begin{proof}[Proof of Proposition~\ref{prop:d:delta}]
  By diagram chasing, as follows.
  \begin{center}\hfill
\begin{tikzcd}[ampersand replacement=\&,baseline=(\tikzcdmatrixname-5-1.base)]
	{\Gamma_X \Sigma X} \&\&\&\&\&\& {ST (\Gamma_{STX} X + X + X)} \\
	{\Gamma_{SX} S X} \&\&\&\&\&\& {ST (T'TX + X)} \\
	{\Gamma_{SSX} S X} \&\& {\Gamma_{SX} S X} \&\&\&\& {ST ((T \oplus T')(T \oplus T')X + X)} \\
	{T'SX} \&\&\&\&\&\& {S(T \oplus T') (T \oplus T')X} \\
	{(T \oplus T')SX} \&\&\&\&\&\& {S(T \oplus T') X}
	\arrow["{\Gamma_{\eta^S} \eta_{\Sigma,X}}"', from=1-1, to=2-1]
	\arrow["{\Gamma_{\eta^S_{SX}}SX}"', from=2-1, to=3-1]
	\arrow["{\eta_{\Gamma_S, SX}}"', from=3-1, to=4-1]
	\arrow["{in_2}"', from=4-1, to=5-1]
	\arrow["{d_{X,X}}", from=1-1, to=1-7]
	\arrow["{d^{\bullet\omega}_{X,X}}", from=2-1, to=1-7]
	\arrow["{ST(\eta_{\Gamma_S} \eta^T_X + [X, X])}", from=1-7, to=2-7]
	\arrow["{ST(in_2 in_1 + X)}", from=2-7, to=3-7]
	\arrow["{Sin_1[\mu^{T \oplus T'}_X,\eta^{T\oplus T'}_X]}", from=3-7, to=4-7]
	\arrow["{S\mu^{T \oplus T'}_X}", from=4-7, to=5-7]
	\arrow[""{name=0, anchor=center, inner sep=0}, "{d_{{/} \delta} X}"', from=5-1, to=5-7]
	\arrow["{\Gamma_{\mu^S_X} SX}"', from=3-1, to=3-3]
	\arrow[""{name=1, anchor=center, inner sep=0}, "{d^{\bullet\omega}_{X,X}}", from=3-3, to=1-7]
	\arrow[Rightarrow, no head, from=2-1, to=3-3]
	\arrow["{\text{(Lemma~\ref{lem:d:delta:Y})}}"{description, pos=0.6}, Rightarrow, draw=none, from=1, to=0]
\end{tikzcd} \qedhere
  \end{center}
  
\end{proof}

\section{Proof of Theorem~\ref{thm:models}}
\label{sec:models}
We fix a given  incremental structural law
  $$d∶ Γ_Y(Σ(X)) → ST (Γ_{STY}(X) + X + Y)$$ over $δ∶ TS → ST$,
  and let $T' = Γ_S^\fleur$.

Given any distributive law, algebras for the composite monad
admit the following well-known characterisation.
\begin{definition}
  Given any monad distributive law $δ∶ RS → SR$, a \alert{$δ$-algebra}
  is an object $X ∈ 𝐂$ of the underlying category, equipped with $S$-
  and $R$-algebra structures, say $𝐚∶ SX → X$ and $𝐛∶ RX → X$,
  satisfying the following law.
  \begin{equation}
    \diag{%
      RSX \& \& SRX \\
      RX \& \& SX \\
      \& X
    }{%
      (m-1-1) edge[labela={δ_X}] (m-1-3) %
       edge[labell={R𝐚}] (m-2-1) %
       (m-1-3) edge[labelr={S𝐛}] (m-2-3) %
       (m-2-1) edge[labelbl={𝐛}] (m-3-2) %
       (m-2-3) edge[labelbr={𝐚}] (m-3-2) %
    }
    \label{eq:SRAlg}
  \end{equation}
  A \alert{$δ$-algebra morphism} $X → Y$ is a morphism between
  underlying objects which is both an $S$- and $R$-algebra morphism.

  We let $δ\Alg$ denote the category of $δ$-algebras and morphisms
  between them.
\end{definition}
\begin{lemma}[{\cite[§2]{BeckDistlaws}}]
  Given any monad distributive law $δ∶ RS → SR$, and $SR$-algebra
  $𝐱∶ SRX → X$,
  the derived $S$- and $R$-algebra structures
  \begin{mathpar}
    SX \xto{Sηᵀ_X} SRX \xto{𝐱} X \and RX \xto{η^S_{RX}} SRX \xto{𝐱} X
  \end{mathpar}
  equip $X$ with $δ$-algebra structure.  Furthermore, this underlies
  an isomorphism $SR\Alg → δ\Alg$ of categories over $𝐂$.
\end{lemma}

In our situation, applying this to the distributive law
$d_{/δ}∶ (T⊕T')S → S(T⊕T')$ from Theorem~\ref{thm:incremental}, we
obtain:
\begin{corollary}\label{cor:TTialg}
  An $S(T⊕T')$-algebra is equivalently an $S$-algebra $𝐚∶ SX → X$,
  equipped with $(T⊕T')$-algebra structure satisfying the
  pentagon~\eqref{eq:SRAlg} with $R=T⊕T'$.
\end{corollary}

But we can say more, by Corollary~\ref{cor:alg:coprod:free}:
\begin{replemma}{lem:TTialg}
  An $S(T⊕T')$-algebra is equivalently an objet $X$, equipped with
  morphisms
  \begin{mathpar}
    𝐚∶ SX → X \and 𝐛∶ TX → X \and 𝐜∶ Γ(X,SX) → X\rlap{,}
  \end{mathpar}
  the first two of which are monad algebra structures, satisfying the
  pentagon~\eqref{eq:SRAlg} with $R=T$, together with the following
  diagram,
  \begin{center}
    \Diag(.6,1.8){%
      \justify{m-1-1}{d2$'$}{m-3-4} %
      }{%
      Γ_{SSX}SX \& \&  \& Γ_{SX} X \\
      Γ_{SX}SX \\
      ST Γ^•_{ST|X} X \& STX \& SX
      \& X
    }{%
      (m-1-1) edge[labela={Γ_{S𝐚}𝐚}] (m-1-4) %
      edge[labelr={Γ_{μ^S_X}SX}] (m-2-1) %
      (m-2-1) edge[labelr={d^{•ω}_{X,X}}] (m-3-1) %
      (m-3-1) edge[labelb={ST (𝐛⊳𝐜)}] (m-3-2) %
      (m-3-2) edge[labelb={S𝐛}] (m-3-3) %
      (m-3-3) edge[labelb={𝐚}] (m-3-4) %
      (m-1-4) edge[labell={𝐜}] (m-3-4) %
    }
  \end{center}
recalling $𝐛⊳𝐜$ from Definition~\ref{def:triangleright}.
\end{replemma}
\begin{proof}
  The pentagon~\eqref{eq:SRAlg} with $R=T⊕T'$ equivalently states that
  the $S$-algebra structure $𝐚∶ SX → X$ is a morphism of
  $(T⊕T')$-algebras.  But by Lemma~\ref{lem:monad:coprod:alg} this is
  equivalent to being both a $T$- and $Γ_SΔ$-algebra morphism.  The
  former is taken care of by~\eqref{eq:SRAlg} with $R=T$, the latter
  by the given diagram.
\end{proof}

  It remains to show that, given an object $X$ equipped with $S$- and
  $T$-algebra structures $𝐚∶ SX → X$ and $𝐛∶ TX → X$ satisfying (d1),
  the following are equivalent:
  \begin{itemize}
  \item a morphism $𝐜∶ Γ(X,X) → X$ satisfying (d2), and
  \item a morphism $𝐜∶ Γ(X,SX) → X$ satisfying (d2$'$);
  \end{itemize}
  and furthermore the notions of morphisms agree.
  
  For any $𝐜∶ Γ(X,X) → X$ satisfying (d2), we define $\check{𝐜}$ to
  be the following composite
  $$Γ(X,SX) \xto{Γ(X,𝐚)} Γ(X,X) \xto{𝐜} X\rlap{,}$$
  and conversely for any $𝐜∶ Γ(X,SX) → X$ satisfying (d2$'$), we
  define $\hat{𝐜}$ to be
  $$Γ(X,X) \xto{Γ(X,η^S_X)} Γ(X,SX) \xto{𝐜} X.$$

It thus suffices to prove:
  \begin{enumerata}
  \item \label{item:inverse} the assignments $𝐜 ↦ \check{𝐜}$ and $𝐜↦\hat{c}$ are mutually inverse, 
  \item \label{item:check} for any
    $𝐜∶ Γ(X,X) → X$ satisfying (d2), $\check{𝐜}$ satisfies (d2$'$),
    and conversely
  \item \label{item:hat} for any $𝐜∶ Γ(X,SX) → X$ satisfying (d2$'$), $\hat{𝐜}$
    satisfies (d2);
  \item \label{item:checkmor} for any $(d,δ)$-algebras $X$ and $Y$, a
    morphism $f∶ X → Y$ which is both a $T$-algebra morphism, an
    $S$-algebra morphism, and a $Γ$-algebra morphism is a morphism
    between the corresponding $Γ_S$-algebras, and
  \item \label{item:hatmor} conversely for any $d_{/δ}$-algebras $X$
    and $Y$, a morphism $f∶ X → Y$ which is both a $T$-algebra
    morphism, an $S$-algebra morphism, and a $Γ_S$-algebra morphism is
    a morphism between the corresponding $Γ$-algebras.
  \end{enumerata}
  Statements~\ref{item:checkmor} and~\ref{item:hatmor} follow easily
  by naturality of $η^S$ and the fact that the considered morphism is
  an algebra morphism.
  
Statement~\ref{item:check} follows from commutation of the diagram in
Figure~\ref{fig:check}, and $d^{•ω}_{X,Y}$ as
in Definition~\ref{def:hbulletomega}.
\begin{figure}[htp]
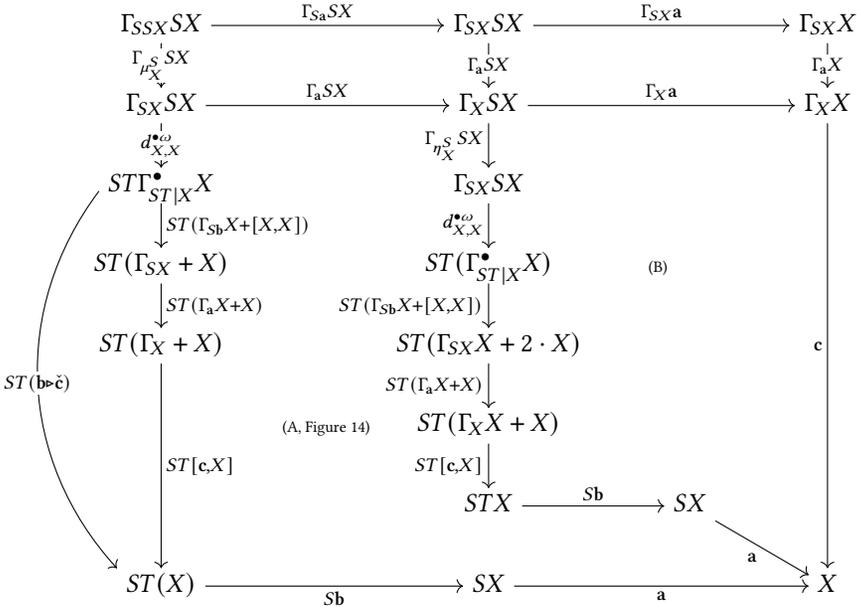

  \Diag{%
    \justify{m-5-1}{A, Figure~\ref{fig:subdiagramA}}{m-7-3} %
    \justify{m-1-5}{B}{m-7-3} %
  }{%
    Γ_{SSX}SX  \& \& Γ_{SX} SX \& \& Γ_{SX} X \\
    Γ_{SX}SX \& \& Γ_X SX  \& \& Γ_X X \\
    ST Γ^•_{ST|X}X \& \& Γ_{SX} SX \\
    ST (Γ_{SX} + X) \& \& ST (Γ^•_{ST|X} X) \\
    ST (Γ_{X} + X) \& \& ST (Γ_{SX} X + 2·X) \\    
    \& \& ST (Γ_X X + X) \\
    \& \& ST X \& SX \\
    ST (X) \&  \& SX \&  \& X %
  }{%
    (m-1-1) edge[labela={Γ_{S𝐚} SX}] (m-1-3) %
    (m-1-3) edge[labela={Γ_{SX} 𝐚}] (m-1-5) %
    (m-2-1) edge[labela={Γ_𝐚 SX}] (m-2-3) %
    (m-2-3) edge[labela={Γ_X 𝐚}] (m-2-5) %
    (m-1-1) edge[labelon={Γ_{μ^S_X}SX}] (m-2-1) %
    (m-1-3) edge[labelon={Γ_𝐚 SX}] (m-2-3) %
    (m-1-5) edge[labelon={Γ_𝐚 X}] (m-2-5) %
    (m-2-1) edge[labelon={d^{•ω}_{X,X}}] (m-3-1) %
    (m-3-1.base west) edge[bend right=40,labelon={ST (𝐛 ⊳ \check{𝐜})}] (m-8-1.north west) %
    (m-3-1) edge[labelr={ST (Γ_{S𝐛} X + [X,X])}] (m-4-1) %
    (m-4-1) edge[labelr={ST (Γ_{𝐚} X + X)}] (m-5-1) %
    (m-5-1) edge[labelr={ST [𝐜, X]}] (m-8-1) %
    (m-2-3) edge[labell={Γ_{η^S_X} SX}] (m-3-3) %
    (m-3-3) edge[labell={d^{•ω}_{X,X}}] (m-4-3) %
    (m-4-3) edge[labell={ST (Γ_{S𝐛}X + [X,X])}] (m-5-3) %
    (m-5-3) edge[labell={ST (Γ_𝐚 X+X)}] (m-6-3) %
    (m-6-3) edge[labell={ST [𝐜,X]}] (m-7-3) %
    (m-7-3) edge[labela={S𝐛}] (m-7-4) %
    (m-7-4) edge[labelbl={𝐚}] (m-8-5) %
    (m-8-1) edge[labelb={S𝐛}] (m-8-3) %
    (m-8-3) edge[labelb={𝐚}] (m-8-5) %
    (m-2-5) edge[labell={𝐜}] (m-8-5) %
  }
      \caption{Main diagram for proving (d2$'$) for $\check{𝐜}$}
  \label{fig:check}
\end{figure}
Subdiagram~(A) commutes by chasing as in Figure~\ref{fig:subdiagramA}.
\begin{figure}[htp]
    \adjustbox{max width=\columnwidth,max height=.95\textheight}{%
  \begin{turn}{90}
  \begin{tikzcd}[ampersand replacement=\&]
    {\Gamma_{SX}SX} \&\& {\Gamma_{SX}SX} \&\&\&\&\& {ST\Gamma^\bullet_{ST|X}X} \\
    \& {\Gamma_{SSX}SX} \& {ST\Gamma^\bullet_{ST|SX}X} \& {STS\Gamma^\bullet_{STS|X}X} \& {SST\Gamma^\bullet_{SST|X}X} \&\&\& {ST(\Gamma_{SX}X + X)} \\
    \&\&\&\&\&\&\& {ST(\Gamma_X + X)} \\
    \&\&\&\&\&\&\& STX \\
    \&\&\&\&\&\&\& SX \\
    {\Gamma_X SX} \& {\Gamma_{SX}SX} \& {ST\Gamma^\bullet_{ST|X}X} \& {ST(\Gamma_{SX}X + X)} \& {ST(\Gamma_{X}+X)} \& STX \& SX \& X
    \arrow["{\Gamma_{\eta^S_{SX}}SX}"{pos=0.9}, from=1-1, to=2-2]
    \arrow[Rightarrow, no head, from=1-1, to=1-3]
    \arrow[""{name=0, anchor=center, inner sep=0}, "{d^{\bullet\omega}_{X,X}}", from=1-3, to=1-8]
    \arrow["{d^{\bullet\omega}_{X,X}}", curve={height=-30pt}, from=1-1, to=1-8]
    \arrow["{d^{\bullet\omega}_{X,SX}}"', from=2-2, to=2-3]
    \arrow["{\Gamma_{\mu^S_X}SX}"'{pos=0.9}, from=2-2, to=1-3]
    \arrow["{ST \uparrow^S}"', from=2-3, to=2-4]
    \arrow["{S\delta \Gamma^\bullet_{S\delta|X}X}"', from=2-4, to=2-5]
    \arrow["{\mu^S T \Gamma^\bullet_{\mu^S T|X}X}"{description}, from=2-5, to=1-8]
    \arrow["{\Gamma_{S\mathbf{a}}SX}"{description}, from=2-2, to=6-2]
    \arrow["{\Gamma_{\mathbf{a}}SX}"{description}, from=1-1, to=6-1]
    \arrow["{\Gamma_{\eta^S_X}SX}"', from=6-1, to=6-2]
    \arrow["{ST\Gamma^\bullet_{ST|\mathbf{a}}X}"', from=2-3, to=6-3]
    \arrow["{d^{\bullet\omega}_{X,X}}"', from=6-2, to=6-3]
    \arrow["{ST(\Gamma_{S\mathbf{b}}X + [X,X])}"', from=6-3, to=6-4]
    \arrow["{ST(\Gamma_{\mathbf{a}}X + X)}"', from=6-4, to=6-5]
    \arrow["{ST(\Gamma_{S\mathbf{b}}X + [X,X])}", from=1-8, to=2-8]
    \arrow["{ST(\Gamma_{\mathbf{a}}X + X)}"', from=2-8, to=3-8]
    \arrow["{ST[\mathbf{c},X]}"', from=6-5, to=6-6]
    \arrow["{ST[\mathbf{c},X]}", from=3-8, to=4-8]
    \arrow["{S\mathbf{b}}", from=4-8, to=5-8]
    \arrow["{S\mathbf{b}}"', from=6-6, to=6-7]
    \arrow["{\mathbf{a}}"', from=6-7, to=6-8]
    \arrow["{\mathbf{a}}", from=5-8, to=6-8]
    \arrow["{\text{(A', Figure~\ref{fig:subdiagramAprime})}}"{description}, draw=none, from=2-5, to=6-5]
    \arrow["{\text{(Lemma~\ref{lem:coh:hbulletomegamu})}}"{description}, Rightarrow, draw=none, from=2-4, to=0]
  \end{tikzcd}
\end{turn}
}
  \caption{Chasing subdiagram A}
  \label{fig:subdiagramA}
\end{figure}
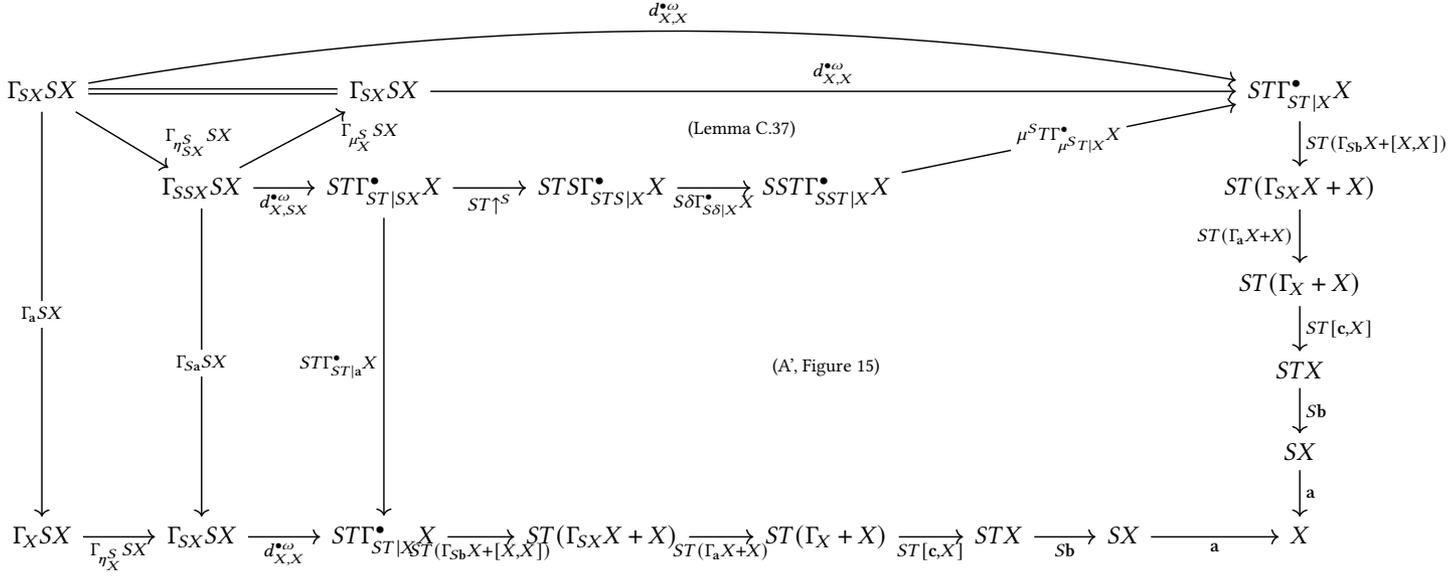
Subdiagram (A') commutes as shown in Figure~\ref{fig:subdiagramAprime}.
  \begin{figure*}[pth]
    \adjustbox{max width=\columnwidth,max height=.95\textheight}{%
      \begin{turn}{90}
\begin{tikzcd}[ampersand replacement=\&]
	{ST\Gamma^\bullet_{ST|SX}X} \&\&\& {ST(S\Gamma_{STSX}^\circ X + SX)} \&\& {STS\Gamma^\bullet_{STS|X}X} \&\& {SST\Gamma^\bullet_{SST|X}X} \&\& {ST\Gamma^\bullet_{ST|X}X} \\
	\&\&\& {ST(S\Gamma_{STX}^\circ X + SX)} \&\& {STS\Gamma^\bullet_{ST|X}X} \& {STS\Gamma^\bullet_{SST|X}X} \& {STS\Gamma^\bullet_{ST|X}X} \\
	\& {ST(\Gamma^\circ_{STX}X + SX)} \&\&\&\&\& {STS\Gamma^\bullet_{SS|X}X} \&\& {STS(\Gamma_{SX}X+X)} \& {ST(\Gamma_{SX}+X)} \\
	\&\& {ST(S(\Gamma^\circ_{STX}X)+X)} \&\& {ST(S\Gamma_{SX}^\circ X + SX)} \&\& {STS\Gamma^\bullet_{S|X}X} \\
	\&\& {ST(S(\Gamma^\circ_{SX}X)+X)} \&\& {ST(S\Gamma_{X}^\circ X + SX)} \&\&\&\& {STS(\Gamma_X X + X)} \& {ST(\Gamma_{X}X+X)} \\
	\&\&\&\& {ST(S(X+X)+SX)} \&\& {ST(SX+SX)} \& STSX \& SSTX \& STX \\
	\&\&\&\&\&\&\&\& SSX \& SX \\
	\&\& {ST(S(\Gamma^\circ_{X}X)+X)} \& {} \&\& {ST(S(X+X)+X)} \& {ST(SX+X)} \& STX \& SX \\
	\&\& {ST(\Gamma_{SX}X + X + X)} \& {ST(\Gamma_{X}X+ X +X)} \& {ST(X+X+X)} \& {ST(X+X)} \\
	{ST\Gamma^\bullet_{ST|X}X} \&\& {ST(\Gamma_{SX}X + X)} \&\& {ST(\Gamma_{X}X+X)} \& STX \&\& SX \&\& X
	\arrow["{S\delta \Gamma^\bullet_{S\delta|X}X}", from=1-6, to=1-8]
	\arrow["{\mu^S T \Gamma^\bullet_{\mu^S T|X}X}", from=1-8, to=1-10]
	\arrow["{ST\Gamma^\bullet_{ST|\mathbf{a}}X}"', from=1-1, to=10-1]
	\arrow["{ST(\Gamma_{S\mathbf{b}}X + [X,X])}"', from=10-1, to=10-3]
	\arrow["{ST (\eta^S + SX)}", from=1-1, to=1-4]
	\arrow["{ST[Sin_1,Sin_2]}", from=1-4, to=1-6]
	\arrow["{STS \Gamma^\bullet_{S\delta |X}X}"{pos=1}, from=1-6, to=2-7]
	\arrow["{STS\Gamma^\bullet_{\mu^S T | X}X}"', from=2-7, to=2-8]
	\arrow[""{name=0, anchor=center, inner sep=0}, "{STS(\Gamma_{S\mathbf{b}}X + [X,X])}"{pos=1}, from=2-8, to=3-9]
	\arrow["{ST(\Gamma_{S\mathbf{b}}X + [X,X])}"{description}, from=1-10, to=3-10]
	\arrow["{STS(\Gamma_{\mathbf{a}} + X)}", from=3-9, to=5-9]
	\arrow["{STS\Gamma^\bullet_{ST\mathbf{a}|X}X}"', from=1-6, to=2-6]
	\arrow["{STS\Gamma^\bullet_{S\mathbf{b}|X}X}"'{pos=0.2}, curve={height=18pt}, from=2-6, to=4-7]
	\arrow["{STS\Gamma^\bullet_{SS\mathbf{b}|X}X}", from=2-7, to=3-7]
	\arrow[""{name=1, anchor=center, inner sep=0}, "{STS(\Gamma_{\mu^S|X}X+[X,X])}"', from=3-7, to=3-9]
	\arrow["{STS\Gamma^\bullet_{S\mathbf{a}|X}X}"', from=3-7, to=4-7]
	\arrow["{\text{\eqref{eq:SRAlg}}}"{description}, draw=none, from=3-7, to=2-6]
	\arrow[""{name=2, anchor=center, inner sep=0}, "{STS(\Gamma_{\mathbf{a}}X+[X,X])}"'{pos=0.1}, from=4-7, to=5-9]
	\arrow["{ST(S\Gamma_{ST\mathbf{a}}^\circ X + SX)}"', from=1-4, to=2-4]
	\arrow["{ST[Sin_1,Sin_2]}", from=2-4, to=2-6]
	\arrow["{ST(S\Gamma_{S\mathbf{b}}^\circ X + SX)}"{pos=0.8}, color={rgb,255:red,10;green,124;blue,4}, from=2-4, to=4-5]
	\arrow["{ST[Sin_1,Sin_2]}"', from=4-5, to=4-7]
	\arrow["{ST(S\Gamma_{\mathbf{a}}^\circ X + SX)}"', color={rgb,255:red,169;green,4;blue,4}, from=4-5, to=5-5]
	\arrow["{ST[S \Gamma^\circ_{X} X,Sin_2]}"', from=5-5, to=5-9]
	\arrow["{ST(\Gamma^\circ_{ST\mathbf{a}}X + SX)}", from=1-1, to=3-2]
	\arrow[""{name=3, anchor=center, inner sep=0}, "{ST(\Gamma^\circ_{STX}X + \mathbf{a})}"{description, pos=0.3}, from=3-2, to=10-1]
	\arrow["{ST(\eta^S+SX)}", from=3-2, to=2-4]
	\arrow["{ST[\mathbf{c},X]}"', from=10-5, to=10-6]
	\arrow["{S\mathbf{b}}"', from=10-6, to=10-8]
	\arrow["{\mathbf{a}}"', from=10-8, to=10-10]
	\arrow[""{name=4, anchor=center, inner sep=0}, "{ST(\Gamma_{\mathbf{a}}X + X)}"', from=10-3, to=10-5]
	\arrow["{ST(\Gamma_{\mathbf{a}} + X)}", from=3-10, to=5-10]
	\arrow["{S\delta ; \mu^S}"', from=5-9, to=5-10]
	\arrow["{ST[\mathbf{c},X]}", color={rgb,255:red,41;green,8;blue,120}, from=5-10, to=6-10]
	\arrow["{S\mathbf{b}}", from=6-10, to=7-10]
	\arrow["{\mathbf{a}}", from=7-10, to=10-10]
	\arrow["{ST(S(\Gamma^\circ_{STX}X)+\mathbf{a})}", color={rgb,255:red,169;green,4;blue,4}, from=2-4, to=4-3]
	\arrow["{ST(S\Gamma^\circ_{S\mathbf{b}}X + X)}", color={rgb,255:red,10;green,124;blue,4}, from=4-3, to=5-3]
	\arrow[""{name=5, anchor=center, inner sep=0}, "{ST(S(\Gamma^\circ_{\mathbf{a}}X)+X)}"{pos=0.9}, color={rgb,255:red,169;green,4;blue,4}, from=5-3, to=8-3]
	\arrow["{ST(S(\mathbf{c}+X)+X)}", color={rgb,255:red,41;green,8;blue,120}, from=8-3, to=8-6]
	\arrow["{ST(\Gamma_{S\mathbf{b}}X + X+X)}", from=10-1, to=9-3]
	\arrow["{ST(\Gamma_{\mathbf{a}}X + X + X)}", from=9-3, to=9-4]
	\arrow["{ST(\mathbf{c}+X+X)}", from=9-4, to=9-5]
	\arrow["{ST[X,X,X]}"'{pos=0}, from=9-5, to=10-6]
	\arrow["{ST(\eta^S+X)}"{pos=0}, from=9-5, to=8-6]
	\arrow["{ST(S(\mathbf{c}+X)+SX)}"', color={rgb,255:red,41;green,8;blue,120}, from=5-5, to=6-5]
	\arrow["{ST(S(X+X)+\mathbf{a})}"'{pos=0.3}, color={rgb,255:red,169;green,4;blue,4}, from=6-5, to=8-6]
	\arrow["{ST(S[X,X]+SX)}", from=6-5, to=6-7]
	\arrow["{STS[\mathbf{c},X]}"', from=5-9, to=6-8]
	\arrow["{ST[SX,SX]}", from=6-7, to=6-8]
	\arrow["{ST\mathbf{a}}", from=6-8, to=8-8]
	\arrow["S\delta", from=6-8, to=6-9]
	\arrow["{\mu^S}", from=6-9, to=6-10]
	\arrow["{SS\mathbf{b}}", from=6-9, to=7-9]
	\arrow["{\mu^S_X}", from=7-9, to=7-10]
	\arrow["{S\mathbf{a}}", from=7-9, to=8-9]
	\arrow["{\mathbf{a}}", from=8-9, to=10-10]
	\arrow["{S\mathbf{b}}", from=8-8, to=8-9]
	\arrow["{ST(S[X,X]+X)}", from=8-6, to=8-7]
	\arrow["{ST([X,X]+X)}", from=9-5, to=9-6]
	\arrow["{ST(\eta^S_X+X)}"{pos=0}, from=9-6, to=8-7]
	\arrow["{ST[X,X]}", from=9-6, to=10-6]
	\arrow["{ST[\mathbf{a},X]}", from=8-7, to=8-8]
	\arrow["{ST[X,X]}"', from=9-6, to=8-8]
	\arrow["{\text{(interchange)}}"{pos=0.4}, Rightarrow, draw=none, from=0, to=1-10]
	\arrow["{\text{(naturality of } \mu \text{)}}"{description}, Rightarrow, draw=none, from=1, to=2-8]
	\arrow["{\text{(monad algebra law)}}"{description}, Rightarrow, draw=none, from=2, to=1]
	\arrow["{\text{(naturality of } \eta \text{)}}", Rightarrow, draw=none, from=3, to=5]
	\arrow["{\text{(easy)}}"{description, pos=0.6}, Rightarrow, draw=none, from=4, to=9-4]
	\arrow["{\text{(easy)}}", Rightarrow, draw=none, from=5, to=6-5]
\end{tikzcd}
      \end{turn}
    }
    \caption{Chasing subsubdiagram (A')}
    \label{fig:subdiagramAprime}
  \end{figure*}
Subdiagram~(B) commutes by Lemma~\ref{lem:pentagon-moyen} below.  
  
Statement~\ref{item:hat} follows from commutation of the diagram in
Figure~\ref{fig:hat}, where
\begin{itemize}
\item the top part commutes by Lemma~\ref{lem:coh:hbulletomegamu} and
\item the bottom right part commutes by
  Lemma~\ref{lem:alg-ST-coh-beta} below.
\end{itemize}
\begin{figure}[htp]
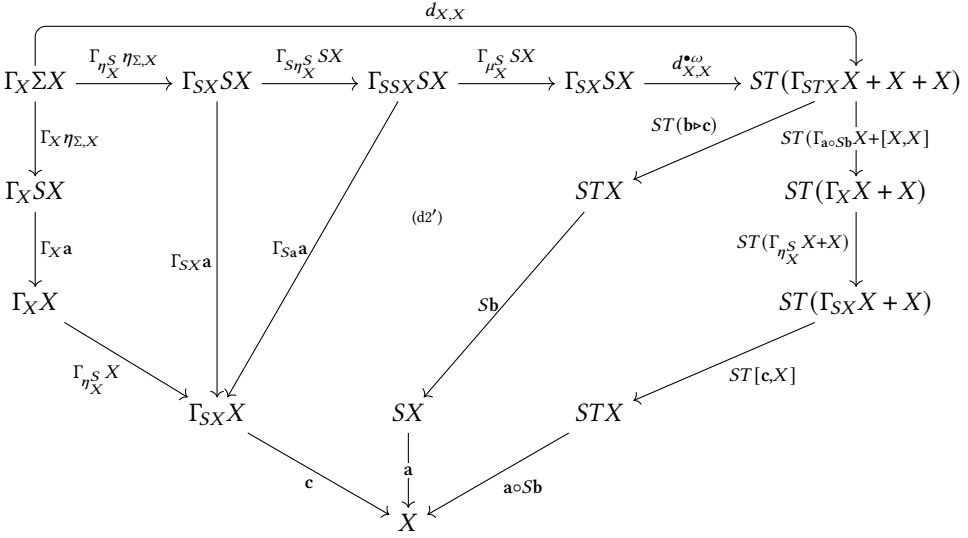

  \noindent
    \Diag(1,1.3){%
    \sqpath{m-1-1}{90}{5}{a}{m-1-5}{labela={d_{X,X}}} %
      \justify{a}{d2$'$}{m-4-3} %
    }{%
      Γ_X Σ X \& Γ_{SX}SX \& Γ_{SSX}SX \& Γ_{SX} SX \& ST(Γ_{STX} X + X + X) \\
      Γ_X SX \&            \&   \&   STX       \& ST (Γ_{X}X+X) \\
      Γ_X X \&   \&  \&  \& ST (Γ_{SX} X + X) \\
      \& Γ_{SX}X \& SX \& STX \\
      \& \& X %
    }{%
      (m-1-1)edge[labela={Γ_{η^S_X} η_{Σ,X}}] (m-1-2) %
      (m-1-2) edge[labela={Γ_{Sη^S_X}SX}] (m-1-3) %
      (m-1-3) edge[twoa={Γ_{μ^S_X}SX}] (m-1-4) %
      (m-1-4) edge[labela={d^{•ω}_{X,X}}] (m-1-5) %
      (m-1-1) edge[labelr={Γ_X η_{Σ,X}}] (m-2-1) %
      (m-2-1) edge[labelr={Γ_X 𝐚}] (m-3-1) %
      (m-3-1) edge[labelbl={Γ_{η^S_X}X}] (m-4-2) %
      (m-4-2) edge[labelbl={𝐜}] (m-5-3) %
      (m-1-3) edge[labell={Γ_{S𝐚} 𝐚}] (m-4-2) %
      (m-1-2) edge[labelbl={Γ_{SX} 𝐚}] (m-4-2) %
      (m-1-5) edge[labelon={ST(Γ_{𝐚 ∘ S𝐛} X + [X,X]}] (m-2-5) %
      (m-1-5) edge[labelal={ST (𝐛 ⊳ 𝐜)}] (m-2-4) %
      (m-2-5) edge[labell={ST (Γ_{η^S_X} X + X)}] (m-3-5) %
      (m-2-4) edge[labell={S𝐛}] (m-4-3) %
      (m-3-5) edge[labelbr={ST [𝐜,X]}] (m-4-4) %
      (m-4-4) edge[labelbr={𝐚 ∘ S𝐛}] (m-5-3) %
      (m-4-3) edge[labelon={𝐚}] (m-5-3) %
    }%
  \caption{Diagram for~\ref{item:hat}}
  \label{fig:hat}
\end{figure}

  Finally, \ref{item:inverse} follows from
  Lemma~\ref{lem:bij-gamma-alg} and~\ref{lem:alg-ST-coh-beta} below.
\begin{lemma}
  \label{lem:bij-gamma-alg}
  Given an algebra $SX \xrightarrow{𝐚} X$, precomposition with morphisms
  $Γ_{S𝐚}$ and $Γ_{η^S}$ induces a bijection (natural in the algebra) between morphisms
  $Γ_X X → X$ and morphisms $Γ_{SX}X → X$ making the following diagram commute.
  \begin{equation}
    \diag{%
      Γ_{SX} X \& Γ_X X \& Γ_{SX} X \\
      \& X
    }{%

      (m-1-1) edge[labela={Γ_𝐚 X}] (m-1-2) %
      edge[labelbl={𝐜}] (m-2-2) %
      (m-1-3) edge[labelbr={𝐜}] (m-2-2) %
      (m-1-2) edge[labela={Γ_{η^S_X} X}] (m-1-3) %
    }
    \label{eq:coherence-beta}
  \end{equation}
\end{lemma}
\begin{proof}
  Straightforward.
\end{proof}

\begin{lemma}
  \label{lem:alg-ST-coh-beta}
  Given an algebra for $S(T⊕T')$ presented as in Lemma~\ref{lem:TTialg}
  by compatible algebra structures
  \begin{mathpar}
    𝐚∶ SX → X \and 𝐛∶ TX → X \and 𝐜∶ Γ_{SX}X → X\rlap{,}
  \end{mathpar}
  the diagram~\eqref{eq:coherence-beta} commutes.
\end{lemma}
\begin{proof}
  We show commutation of the diagram by precomposing both sides  by the split epimorphism $Γ_{S𝐚}𝐚$,
  and observing that they both are equal to the left bottom composite in~(d2$'$).
  For $𝐜$, we readily get the top right part of~(d2$'$), hence the result.
  For the top right part of~\eqref{eq:coherence-beta}, the result follows
  by chasing the diagram below.
  \begin{center}
    \hfill
   \diag|baseline=(m-5-1.base)|(.6,1.5){%
      Γ_{SSX}SX \& \& Γ_{SX}SX \& \& Γ_{SX} X \\
      Γ_{SX}SX \& \& Γ_{X}SX \& \& Γ_{X} X \\
      Γ_{SSX}SX \& \& Γ_{SX}SX \& \& Γ_{SX} X \\
      Γ_{SX}SX \\
      ST Γ^•_{ST|X} X \& STX  \&  \& SX \& X %
    }{%
      (m-1-1) edge[labela={Γ_{S𝐚}SX}] (m-1-3) %
      (m-1-3) edge[labela={Γ_{SX} 𝐚}] (m-1-5) %
      (m-1-1) edge[labelr={Γ_{μ^S_X} SX}] (m-2-1) %
      (m-1-3) edge[labell={Γ_{𝐚} SX}] (m-2-3) %
      (m-1-5) edge[labell={Γ_𝐚 X}] (m-2-5) %
      (m-2-1) edge[labela={Γ_{𝐚}SX}] (m-2-3) %
      (m-2-3) edge[labela={Γ_{X} 𝐚}] (m-2-5) %
      (m-3-1) edge[labela={Γ_{S𝐚}SX}] (m-3-3) %
      (m-3-3) edge[labela={Γ_{SX} 𝐚}] (m-3-5) %
      (m-2-1) edge[labelr={Γ_{η^S_{SX}} SX}] (m-3-1) %
      (m-2-3) edge[labell={Γ_{η^S_{X}} SX}] (m-3-3) %
      (m-2-5) edge[labell={Γ_{η^S_{X}} X}] (m-3-5) %
      (m-3-1) edge[labelr={Γ_{μ^S_X}SX}] (m-4-1) %
      (m-4-1) edge[labelr={d^{•ω}_{X,X}}] (m-5-1) %
      (m-5-1) edge[labelb={ST (𝐛⊳𝐜)}] (m-5-2) %
      (m-5-2) edge[labelb={S𝐛}] (m-5-4) %
      (m-5-4) edge[labelb={𝐚}] (m-5-5) %
      (m-3-5) edge[labell={𝐜}] (m-5-5) %
      (m-2-1.base west) edge[bend right,identity] (m-4-1.north west) %
    } \qedhere
  \end{center}
\end{proof}
\begin{lemma}
  \label{lem:pentagon-moyen}
  For any $(δ,d)$-algebra $X$ with structure induced by
  \begin{mathpar}
    𝐚∶ SX → X \and 𝐛∶ TX → X \and 𝐜∶ Γ_X X → X\rlap{,}
  \end{mathpar}
  the following diagram commutes
  \begin{center}
    \diag{%
      Γ_X SX \& Γ_{SX}SX \& ST Γ^•_{ST|X} X \\
      \& \& ST Γ^\pt_X X \\
      Γ_X X \& \& STX \\
      \& X
    }{%
      (m-1-1) edge[labela={Γ_{η^{S}_X}SX}] (m-1-2) %
      edge[labell={Γ_X 𝐚}] (m-3-1) %
      (m-1-2) edge[labela={d^{•ω}_{X,X}}] (m-1-3) %
       (m-2-3) edge[labelr={ST[𝐜,X]}] (m-3-3) %
       (m-3-1) edge[labelbl={𝐜}] (m-4-2) %
       (m-3-3) edge[labelbr={𝐚∘S𝐛}] (m-4-2) %
       (m-1-3) edge[labelr={ST (Γ_{𝐚∘S𝐛}X+[X,X])}] (m-2-3)
    }
  \end{center}
\end{lemma}
\begin{proof}
  We start by reducing to Subdiagram~(C) as below.
  \begin{center}
    \Diag(1,2){%
      \justify{m-2-2}{C}{m-4-3} %
      }{%
         \& Γ_{SX}SX \\
      Γ_X SX \& Γ^∘_{STX}SX \& ST Γ^•_{ST|X} X \\
      Γ_X X \& Γ^∘_{STX} X \& ST Γ^\pt_X X \\
      \&  Γ^∘_X X  \& STX \\
      \& X
    }{%
      (m-2-1) edge[labela={Γ_{η^{S}_X}SX}] (m-1-2) %
      edge[labell={Γ_X 𝐚}] (m-3-1) %
      edge[labelb={in₁∘(Γ_{η^{ST}_X} SX)}] (m-2-2) %
      (m-2-2) edge[labelb={\widetilde{\widetilde{d^∘}^ω}_{X,X}}] (m-2-3) %
      (m-1-2) edge[labelar={d^{•ω}_{X,X}}] (m-2-3) %
      (m-2-2) edge[labelr={Γ^∘_{STX} 𝐚}] (m-3-2) %
       (m-3-3) edge[labelr={ST[𝐜,X]}] (m-4-3) %
       (m-3-1) edge[labelar={in₁}] (m-4-2) %
       edge[labelbl={𝐜}] (m-5-2) %
       edge[labela={in₁ ∘ (Γ_{η^{ST}_X} X)}] (m-3-2) %
       (m-3-2) edge[labelr={Γ_{𝐚∘S𝐛} X + X}] (m-4-2) %
       (m-4-3) edge[labelbr={𝐚∘S𝐛}] (m-5-2) %
       (m-2-3) edge[labelr={ST (Γ_{𝐚∘S𝐛}X+[X,X])}] (m-3-3)
       (m-4-2) edge[labelr={[𝐜,X]}] (m-5-2) %
    }
  \end{center}
  Commutation of~(C) is then equivalent to commutation of the transposed diagram below.
  \begin{center}
    \Diag(.6,2.2){%

    }{%
      SX \& ℸ^×_{STX} ST Γ^•_{ST|X} X \& ℸ^×_{STX} ST Γ^\pt_X X \\
      X \& \& ℸ^×_{STX} STX \\
      ℸ^×_{STX} Γ^\pt_{STX} X \& ℸ^×_{STX} Γ^\pt_X X \& ℸ^×_{STX} X %
    }{%
      (m-1-1) edge[labela={\widetilde{d^∘}^ω}] (m-1-2) %
      (m-1-2) edge[loina={ℸ^×_{STX} ST (Γ_{𝐚∘S𝐛} X+[X,X])}] (m-1-3) %
      (m-1-1) edge[labelr={𝐚}] (m-2-1) %
      (m-1-3) edge[labell={ℸ^×_{STX} ST [𝐜,X]}] (m-2-3) %
      (m-2-3) edge[labell={ℸ^×_{STX} (𝐚 ∘ S𝐛)}] (m-3-3) %
      (m-2-1) edge[labelr={η_{STX,X}}] (m-3-1) %
      (m-3-1) edge[labelb={ℸ^×_{STX} Γ^\pt_{𝐚 ∘ S𝐛} X }] (m-3-2) %
      (m-3-2) edge[labelb={ℸ^×_{STX} [𝐜,X]}] (m-3-3) %
    }%
    \end{center}
    Because $\widetilde{d^∘}^ω_{X,X}$ is a monad morphism by construction,
    one easily proves that both sides have the same restriction along
    $η^S_X∶ X → SX$, so it suffices to prove that
    both sides are $S$-algebra morphisms.
    By construction, $𝐚$ and $\widetilde{d^∘}^ω_{X,X}$ are, so we focus
    on other morphisms.

    \begin{itemize}
    \item We first recall Lemma~\ref{lem:RB:liftsto:KBalg}, which says
      that $ℸ^×_{STX}$ lifts to a functor $ST O_{X}\Alg → Σ\alg$, by
      considering, for any $ST O_{X}$-algebra $𝐞∶ ST O_{X} A → A$, the
      following $Σ$-algebra structure:
      \begin{center}
      \diag{%
        Σℸ^×_{STX}A \\
        ℸ^×_{STX} ST O_{X} Γ^∘_{STX} ℸ^×_{STX} A \\
        ℸ^×_{STX} ST O_{X} A \\
        ℸ^×_{STX} A.
      }{%
        (m-1-1)
        edge[labelr={\widetilde{d^∘}_{X,X} ℸ^×_{STX} A}] (m-2-1) %
        (m-2-1)
        edge[labelr={ℸ^×_{STX} ST O_{X} ε_{STX} A}] (m-3-1) %
        (m-3-1) edge[labelr={      ℸ^×_{STX} 𝐞}] (m-4-1) %
        }
      \end{center}
    \item For proving that the top right composite is a morphism of
      $Σ$-algebras, it thus suffices to prove that the underlying
      composite $ST Γ^•_{STX} X → X$ is a morphism of
      $STO_{X}$-algebras.  But this is in turn equivalent to being
      both a morphism of $ST$- and $O_{X}$-algebras.  This easy to see
      for $O_X$, and the given composite is a morphism of
      $ST$-algebras as a composite of two free $ST$-algebra morphisms,
      and the $ST$-algebra structure $STX → X$.
    \item For the bottom left composite, we show that it is a
      $Σ$-algebra morphism by chasing the diagram in
      Figure~\ref{fig:etad},
\begin{figure}[htp]
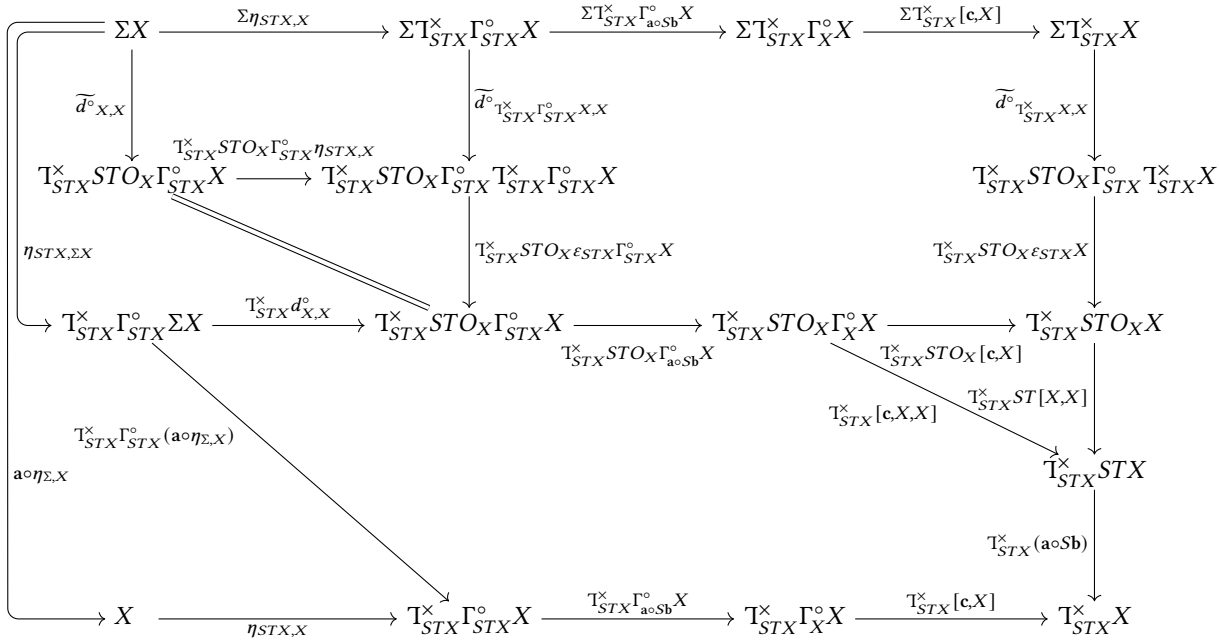

  \centering
  \begin{turn}{90}
  \diagramme[diag={1.5}{1}]{}{%
    \sqpath{m-1-1}{180}{10}{l}{m-3-1}{labelrat={η_{STX,ΣX}}{.75}} %
     \sqpath{m-1-1.160}{180}{8.5}{bidule}{m-5-1.west}{labelrat={𝐚 ∘ η_{Σ,X}}{.7}} %
  }{%
      ΣX \& Σℸ^×_{STX} Γ^∘_{STX} X \& Σℸ^×_{STX}Γ^∘_X X \& Σ ℸ^×_{STX} X \\
      ℸ^×_{STX} ST O_{X} Γ^∘_{STX}X \& ℸ^×_{STX} ST O_{X} Γ^∘_{STX} ℸ^×_{STX} Γ^∘_{STX} X \& \& ℸ^×_{STX} ST O_{X} Γ^∘_{STX} ℸ^×_{STX}  X \\
      ℸ^×_{STX} Γ^∘_{STX} Σ X \& ℸ^×_{STX} ST O_{X} Γ^∘_{STX} X \& ℸ^×_{STX} ST O_{X} Γ^∘_{X} X \& ℸ^×_{STX} ST O_{X} X \\
      \&  \&  \& ℸ^×_{STX} STX \\
     |[text width=3ex]| X \& ℸ^×_{STX} Γ^∘_{STX} X \& ℸ^×_{STX} Γ^∘_{X} X \& ℸ^×_{STX} X %
    }{%
      (m-1-1) edge[labela={Ση_{STX,X}}] (m-1-2) %
      edge[labell={\widetilde{d^∘}_{X,X}}] (m-2-1) %
      (m-1-2) edge[labela={Σ ℸ^×_{STX} Γ^∘_{𝐚∘S𝐛} X}] (m-1-3) %
      (m-3-1) edge[labelblat={ℸ^×_{STX} Γ^∘_{STX} (𝐚 ∘ η_{Σ,X})}{.3}] (m-5-2) %
      edge[labela={ℸ^×_{STX} d^∘_{X,X}}] (m-3-2) %
      (m-5-1) edge[labelb={η_{STX,X}}] (m-5-2) %
      (m-2-1) edge[identity] (m-3-2) %
      edge[label={[above=.5em]{$\scriptstyle ℸ^×_{STX} ST O_{X} Γ^∘_{STX} η_{STX,X}$}}] (m-2-2) %
      (m-1-2) edge[labelr={\widetilde{d^∘}_{ℸ^×_{STX} Γ^∘_{STX} X,X}}] (m-2-2) %
      (m-1-3) edge[labela={Σ ℸ^×_{STX} [𝐜,X]}] (m-1-4) %
      (m-1-4) edge[labell={\widetilde{d^∘}_{ℸ^×_{STX} X,X}}] (m-2-4) %
      (m-2-2) edge[labelr={ℸ^×_{STX} ST O_{X} ε_{STX} Γ^∘_{STX} X}] (m-3-2) %
      (m-2-4) edge[labell={ℸ^×_{STX} ST O_{X} ε_{STX} X}] (m-3-4) %
      (m-3-2) edge[loinb={ℸ^×_{STX} ST O_{X} Γ^∘_{𝐚∘S𝐛} X}] (m-3-3) %
      (m-3-3) edge[label={[below=.5em]{$\scriptstyle ℸ^×_{STX} ST O_{X} [𝐜,X]$}}] (m-3-4) %
      edge[labelbl={ℸ^×_{STX} [𝐜,X,X]}] (m-4-4) %
      (m-5-2) edge[labela={ℸ^×_{STX} Γ^∘_{𝐚∘S𝐛} X}] (m-5-3) %
      (m-5-3) edge[labela={ℸ^×_{STX} [𝐜,X]}] (m-5-4) %
      (m-3-4) edge[labell={ℸ^×_{STX}ST [X,X]}] (m-4-4) %
      (m-4-4) edge[labell={ℸ^×_{STX} (𝐚 ∘ S𝐛)}] (m-5-4) %
    }
  \end{turn}
  \caption{Diagram for Lemma~\ref{lem:pentagon-moyen}}
  \label{fig:etad}
\end{figure}      
where the bottom right subdiagram commutes by Lemma~\ref{lem:p12}
below.
    \end{itemize}
\end{proof}

\begin{lemma}
  \label{lem:p12}
  For any $(δ,d)$-algebra $X$ with structure induced by
  \begin{mathpar}
    𝐚∶ SX → X \and 𝐛∶ TX → X \and 𝐜∶ Γ_X X → X\rlap{,}
  \end{mathpar}
    the following diagram commutes.
    \begin{center}
      \diag{%
        Γ^∘_{STX} Σ X \& \& ST Γ^•_{ST|X} X \\
          Γ^∘_{STX} X \& \& ST Γ^\pt_X X \\
          Γ^∘_{X} X \& \& ST X \\
           \& X
      }{%
        (m-1-1) edge[labela={d^∘_{X,X}}] (m-1-3) %
       (m-1-1) edge[labell={Γ^∘_{STX} (𝐚 ∘ η_{Σ,X})}] (m-2-1) %
       (m-1-3) edge[labelr={ST (Γ_{𝐚∘S𝐛} X + [X,X])}] (m-2-3) %
       (m-2-1) edge[labell={Γ^∘_{𝐚∘S𝐛} X}] (m-3-1) %
       (m-2-3) edge[labelr={ST[𝐜,X]}] (m-3-3) %
       (m-3-1) edge[labelbl={[𝐜,X]}] (m-4-2) %
       (m-3-3) edge[labelbr={𝐚∘S𝐛}] (m-4-2) %
      }
    \end{center}
\end{lemma}
\begin{proof}
  Since the domain is a coproduct, we proceed termwise.  On $ΣX$, the
  result is straightforward.  On $Γ_{STX}ΣX$, we observe that the
   diagram in Figure~\ref{fig:p12:dcircin1} commutes for all
  $X,Y$,
  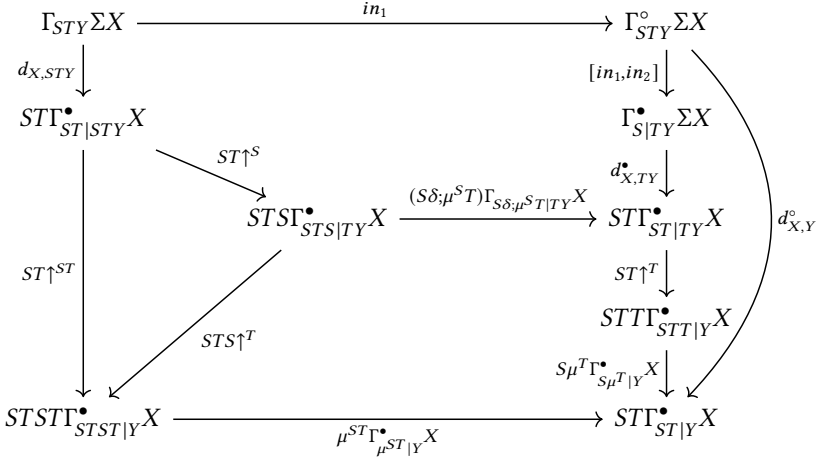
\begin{figure}[ht]
\begin{tikzcd}[ampersand replacement=\&,baseline=(\tikzcdmatrixname-5-1.base)]
	{\Gamma_{STY} \Sigma X} \&\&\&\& {\Gamma^\circ_{STY} \Sigma X} \\
	{ST\Gamma^\bullet_{ST|STY} X} \&\&\&\& {\Gamma^\bullet_{S|TY} \Sigma X} \\
	\& {STS\Gamma^\bullet_{STS|TY}  X} \&\&\& {ST\Gamma^\bullet_{ST|TY} X} \\
	\&\&\&\& {STT\Gamma^\bullet_{STT|Y} X} \\
	{STST\Gamma^\bullet_{STST|Y} X} \&\&\&\& {ST\Gamma^\bullet_{ST|Y} X}
	\arrow["{in_1}", from=1-1, to=1-5]
	\arrow["{[in_1,in_2]}"', from=1-5, to=2-5]
	\arrow["{d^\bullet_{X,TY}}"', from=2-5, to=3-5]
	\arrow["{ST \uparrow^T}"', from=3-5, to=4-5]
	\arrow["{S\mu^T\Gamma^\bullet_{S\mu^T|Y} X}"', from=4-5, to=5-5]
	\arrow["{d_{X,STY}}"', from=1-1, to=2-1]
	\arrow["{ST \uparrow^S}", from=2-1, to=3-2]
	\arrow["{(S\delta;\mu^S T)\Gamma_{S\delta;\mu^S T|TY} X}", from=3-2, to=3-5]
	\arrow["{STS \uparrow^T}", from=3-2, to=5-1]
	\arrow["{\mu^{ST}\Gamma^\bullet_{\mu^{ST}|Y} X}"', from=5-1, to=5-5]
	\arrow["{ST\uparrow^{ST}}"', from=2-1, to=5-1]
	\arrow["{d^\circ_{X,Y}}", curve={height=-50pt}, from=1-5, to=5-5]
      \end{tikzcd}
      \caption{Proof of Lemma~\ref{lem:p12}, first step}
      \label{fig:p12:dcircin1}
    \end{figure}
  so that the lemma reduces (using interchange) to commutation of the
  following diagram,
  \begin{center}
    \ajustedroit
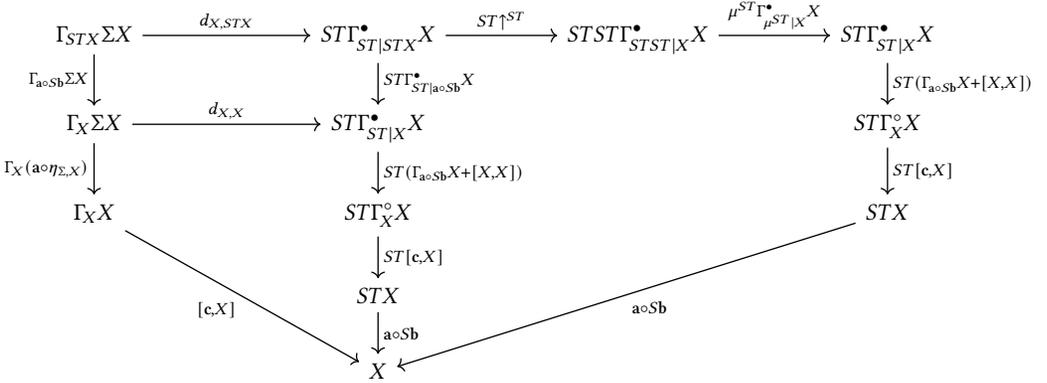
{
\begin{tikzcd}[ampersand replacement=\&]
	{\Gamma_{STX} \Sigma X} \&\&\& {ST\Gamma^\bullet_{ST|STX} X} \&\& {STST\Gamma^\bullet_{STST|X} X} \&\& {ST\Gamma^\bullet_{ST|X} X} \\
	{\Gamma_{X} \Sigma X} \&\&\& {ST\Gamma^\bullet_{ST|X} X} \&\&\&\& {ST\Gamma^\circ_{X} X} \\
	{\Gamma_{X} X} \&\&\& {ST\Gamma^\circ_{X} X} \&\&\&\& {ST X} \\
	\&\&\& {ST X} \\
	\&\&\& X
	\arrow["{d_{X,STX}}", from=1-1, to=1-4]
	\arrow["{\mu^{ST}\Gamma^\bullet_{\mu^{ST}|X} X}", from=1-6, to=1-8]
	\arrow["{ST\uparrow^{ST}}", from=1-4, to=1-6]
	\arrow["{ST [\mathbf{c},X]}", from=3-4, to=4-4]
	\arrow["{\mathbf{a} \circ S \mathbf{b}}", from=4-4, to=5-4]
	\arrow["{\Gamma_{\mathbf{a} \circ S \mathbf{b}} \Sigma X}"', from=1-1, to=2-1]
	\arrow["{\Gamma_{X} (\mathbf{a} \circ \eta_{\Sigma,X})}"', from=2-1, to=3-1]
	\arrow["{[\mathbf{c},X]}"', from=3-1, to=5-4]
	\arrow["{ST\Gamma^\bullet_{ST|\mathbf{a} \circ S \mathbf{b}} X}", from=1-4, to=2-4]
	\arrow["{d_{X,X}}", from=2-1, to=2-4]
	\arrow["{ST(\Gamma_{\mathbf{a} \circ S \mathbf{b}} X + [X,X])}", from=2-4, to=3-4]
	\arrow["{ST(\Gamma_{\mathbf{a} \circ S \mathbf{b}}X + [X,X])}", from=1-8, to=2-8]
	\arrow["{ST [\mathbf{c},X]}", from=2-8, to=3-8]
	\arrow["{\mathbf{a} \circ S \mathbf{b}}", from=3-8, to=5-4]
      \end{tikzcd}}
    \end{center}
    whose bottom left polygon commutes by~(d2).
    The right-hand part is chased as in Figure~\ref{fig:p12}.
    \begin{figure}[htp]
      \ajustetourne{
    \begin{tikzcd}[ampersand replacement=\&,baseline=(\tikzcdmatrixname-7-1.base)]
      {ST(\Gamma_{STSTX} X + X + STX)} \&\&\&\&\&\& {STST\Gamma^\bullet_{STST|X} X} \&\&\&\&\&\& {ST\Gamma^\bullet_{ST|X} X} \\
      \\
      \&\& {ST(\Gamma_{X} X + X + STX)} \&\&\&\&\&\&\&\&\&\& {ST\Gamma^\circ_{X} X} \\
      {ST\Gamma^\bullet_{ST|X} X} \&\& {ST(X + STX)} \&\&\& {ST(STX + STX)} \&\&\& {STST(X + X)} \&\& {ST(X + X)} \&\& {ST X} \\
      \&\&\&\&\& {ST(X + STX)} \&\&\& STSTX \\
      {ST\Gamma^\circ_{X} X} \&\& {ST(X + X)} \&\&\& {ST(X + X)} \\
      {ST X} \&\&\&\&\&\&\&\& STX \&\&\&\& X
      \arrow["{\mu^{ST}\Gamma^\bullet_{\mu^{ST}|X} X}", from=1-7, to=1-13]
      \arrow["{ST\uparrow^{ST}}", from=1-1, to=1-7]
      \arrow["{ST [\mathbf{c},X]}"', from=6-1, to=7-1]
      \arrow["{ST\Gamma^\bullet_{ST|\mathbf{a} \circ S \mathbf{b}} X}"', from=1-1, to=4-1]
      \arrow["{ST(\Gamma_{\mathbf{a} \circ S \mathbf{b}} X + [X,X])}"', from=4-1, to=6-1]
      \arrow["{ST(\Gamma_{\mathbf{a} \circ S \mathbf{b}}X + [X,X])}", from=1-13, to=3-13]
      \arrow["{ST [\mathbf{c},X]}", from=3-13, to=4-13]
      \arrow["{\mathbf{a} \circ S \mathbf{b}}", from=4-13, to=7-13]
      \arrow["{ST(\Gamma_{\mathbf{a} \circ S \mathbf{b} \circ ST \mathbf{a} \circ STS \mathbf{b}} X + X + STX)}"{pos=0.9}, from=1-1, to=3-3]
      \arrow["{ST([\mathbf{c},X] + STX)}"', from=3-3, to=4-3]
      \arrow["{ST \uparrow^{ST}}", curve={height=-24pt}, from=4-3, to=4-9]
      \arrow["{\mu^{ST}}", from=4-9, to=4-11]
      \arrow["{ST[X,X]}", from=4-11, to=4-13]
      \arrow["{ST(X + \mathbf{a} \circ S\mathbf{b})}"', from=4-3, to=6-3]
      \arrow["{ST[X,X]}"', from=6-3, to=7-1]
      \arrow["{ST(\eta^{ST}_X + STX)}"', from=4-3, to=4-6]
      \arrow["{ST[ST in_1, ST in_2]}", from=4-6, to=4-9]
      \arrow["{STST[X,X]}", from=4-9, to=5-9]
      \arrow["{\mu^{ST}_X}"', from=5-9, to=4-13]
      \arrow["{ST(\mathbf{a} \circ S \mathbf{b})}", from=5-9, to=7-9]
      \arrow["{\mathbf{a} \circ S \mathbf{b}}", from=7-9, to=7-13]
      \arrow[Rightarrow, no head, from=7-1, to=7-9]
      \arrow["{ST[STX,STX]}"', from=4-6, to=5-9]
      \arrow["{ST(\mathbf{a} \circ S \mathbf{b} + STX)}", from=4-6, to=5-6]
      \arrow["{ST(X + \mathbf{a} \circ S \mathbf{b})}", from=5-6, to=6-6]
      \arrow["{ST[X,X]}", from=6-6, to=7-9]
      \arrow[Rightarrow, no head, from=6-3, to=6-6]
      \arrow[curve={height=6pt}, Rightarrow, no head, from=4-3, to=5-6]
    \end{tikzcd}
  }
    \caption{Proof of Lemma~\ref{lem:p12}, second step}
    \label{fig:p12}
    \end{figure}
  \end{proof}

\section{Proof of Theorem~\ref{thm:benign}}
\label{app:benign}

\begin{lemma}
  For any structural interpretation $K∶ Θ(X,X) → STX$ of an incremental structural law
  $d_{X,Y}∶ Θ(ΣX,Y) → ST(Θ_{STY}X + X + Y)$ over a monad distributive $δ∶ TS → ST$,
  with $S = Σ^*$, the $Θ$-algebra structure defined on any $ST$-algebra $X$
  with structure maps $𝐚∶ SX → X$ and $𝐛∶ TX → X$ by
  the composite
  $$Θ(X,X) \xto{K_X} STX \xto{S𝐛} SX \xto{𝐚} X$$
  satisfies~(d2).
\end{lemma}
\begin{proof}
  By coherence of $K$, this reduces to commutativity of both of the following diagrams.
  \begin{center}
    \ajustedroit{
  \begin{tikzcd}[ampersand replacement=\&]
    {\Theta(\Sigma X , X)} \&\& {\Theta(SX,SX)} \& STSX \& SSTX \& STX \\
    \&\&\&\& SSX \& SX \\
    {\Theta(SX,X)} \& {\Theta(X,X)} \& STX \& STX \& SX \& X
    \arrow["{\Theta(\eta_{\Sigma,X}, \eta^S_X)}", from=1-1, to=1-3]
    \arrow["{K_{SX}}", from=1-3, to=1-4]
    \arrow["{S\delta_X}", from=1-4, to=1-5]
    \arrow["{\mu^S_{TX}}", from=1-5, to=1-6]
    \arrow["{S\mathbf{b}}", from=1-6, to=2-6]
    \arrow["{\mathbf{a}}", from=2-6, to=3-6]
    \arrow["{\Theta(\eta_{\Sigma,X}, X)}"', from=1-1, to=3-1]
    \arrow["{S\mathbf{b}}"', from=3-4, to=3-5]
    \arrow["{\mathbf{a}}"', from=3-5, to=3-6]
    \arrow[Rightarrow, no head, from=3-3, to=3-4]
    \arrow["{K_X}"', from=3-2, to=3-3]
    \arrow["{\Theta(\mathbf{a},X)}"', from=3-1, to=3-2]
    \arrow["{\Theta(\mathbf{a},\mathbf{a})}"'{pos=0.7}, from=1-3, to=3-2]
    \arrow["{ST\mathbf{a}}"'{pos=0.7}, from=1-4, to=3-3]
    \arrow["{S\mathbf{a}}"', from=2-5, to=3-5]
    \arrow["{SS\mathbf{b}}"', from=1-5, to=2-5]
    \arrow["{\mu^S_X}", from=2-5, to=2-6]
  \end{tikzcd}}
  \end{center}

  \begin{center}
    \hfill
    \ajustedroit[.95]{
  \begin{tikzcd}[ampersand replacement=\&,baseline=(\tikzcdmatrixname-5-1.base)]
    {ST(\Theta(X,STX) + X + X)} \&\&\& {ST(\Theta(STX,STX)+X)} \&\& {ST(STSTX+X)} \\
    {ST(\Theta(X,SX) + X)} \&\&\& {ST(\Theta(SX,SX)+X)} \&\& STSTX \\
    {ST(\Theta(X,X) + X)} \&\&\&\& {ST(STSX+X)} \& STX \\
    {ST(STX + X)} \&\&\&\&\& SX \\
    {ST(SX + X)} \& STX \&\& SX \&\& X
    \arrow["{ST(\Theta(\eta^{ST}_X,STX) + [X,X])}", from=1-1, to=1-4]
    \arrow["{ST(K_{STX} + X)}", from=1-4, to=1-6]
    \arrow["{ST[\mu^{ST}_X, \eta^{ST}_X]}", from=1-6, to=2-6]
    \arrow["{\mu^{ST}_X}", from=2-6, to=3-6]
    \arrow["{S\mathbf{b}}", from=3-6, to=4-6]
    \arrow["{\mathbf{a}}", from=4-6, to=5-6]
    \arrow["{ST(\Theta(X,S\mathbf{b})+[X,X])}"', from=1-1, to=2-1]
    \arrow["{ST(\Theta(X,\mathbf{a})+X)}"', from=2-1, to=3-1]
    \arrow["{ST(K_X+X)}"', from=3-1, to=4-1]
    \arrow["{ST(S\mathbf{b} + X)}"', from=4-1, to=5-1]
    \arrow["{ST[\mathbf{a}, X]}"', from=5-1, to=5-2]
    \arrow["{S\mathbf{b}}"', from=5-2, to=5-4]
    \arrow["{\mathbf{a}}"', from=5-4, to=5-6]
    \arrow["{ST(STS\mathbf{b}+X)}"'{pos=0.8}, from=1-6, to=3-5]
    \arrow["{ST(\Theta(S\mathbf{b},S\mathbf{b}) + X)}", from=1-4, to=2-4]
    \arrow["{ST(\Theta(\mathbf{a},\mathbf{a})+X)}"', from=2-4, to=3-1]
    \arrow["{ST(ST\mathbf{a} + X)}"{pos=0.3}, from=3-5, to=4-1]
  \end{tikzcd}      }
  \qedhere
  \end{center}
\end{proof}

\begin{corollary}
  The assignment of the previous lemma lifts to a section
  $\bar{K}∶ ST\Alg → S(T ⊕ T')\Alg$ of the forgetful functor
  $S(T⊕T')\Alg → ST\Alg$, where $T' = Θ_S^*$.
\end{corollary}

Let us now prove Theorem~\ref{thm:benign}.
Suppose given a structural equational system over a distributive law $δ∶ TS → ST$, i.e.,
an incremental structural law 
$$d_{X,Y}∶ Θ_Y(Σ(X)) → ST(Θ_{ST(Y)}(X) + X + Y)$$ and a pair of structural interpretations $L,R∶ΘΔ → ST$.
In this section, we show that the initial $ST$-algebra
$STST∅ → ST∅$ coequalises $L_{ST∅}$ and $R_{ST∅}$.

To this end, we exploit admissibility of the monad morphism $ST →
S(T⊕T')$ (Remark~\ref{rk:admis-STT})
which entails that the $ST$-algebra structure on $ST∅$ uniquely extends to
an $S(T⊕T')$-algebra structure. By Theorem~\ref{thm:models},
we obtain that there is a unique $ΘΔ$-algebra structure on $ST∅$
satisfying~(d2).

But by the corollary $L$ and $R$ both induce $ΘΔ$-algebra stucture on $ST∅$
which satisfy~(d2), namely the composites
\begin{center}
  \diag(1,2){%
    Θ_{ST∅}ST∅ \& STST∅ \& ST∅.
  }{%
    (m-1-1) edge[bend left=10,labela={L_{ST∅}}] (m-1-2) %
    (m-1-1) edge[bend right=10,labelb={R_{ST∅}}] (m-1-2) %
    (m-1-2) edge[labela={μ^{ST}_{∅}}] (m-1-3)
  }
\end{center}
We thus conclude by uniqueness.


\section{Proofs of §\ref{ss:ptstrengths}}\label{app:ptstrengths}

We need to check coherence (Figure~\ref{fig:interp}) for all of
$L₁$, $R₁$, $L₂$, and $R₂$.

\subsection{Coherence of $L₁$}
We first check coherence of $L₁$ in Figure~\ref{fig:coherence:L1},
which uses the following lemma.
\begin{lemma}\label{lem:eta:gamma:s}
  For all objects $X$, the following diagram commutes.
  \begin{center}
\begin{tikzcd}[ampersand replacement=\&]
	{\Gamma_{SX}X} \&\& {\Gamma_{SSX}X} \&\& {\Gamma_{SSX}SX} \\
	TX \&\&\&\& TSX \\
	\&\& STX
	\arrow["{\eta_{\Gamma_S,X}}"', from=1-1, to=2-1]
	\arrow["{\eta_{\Gamma_S,SX}}", from=1-5, to=2-5]
	\arrow["{\delta_X}", from=2-5, to=3-3]
	\arrow["{\eta^S_{TX}}"', from=2-1, to=3-3]
	\arrow["{\Gamma_{SX} \eta^S_X}", from=1-1, to=1-3]
	\arrow["{\Gamma_{\eta^S_{SX}} SX}", from=1-3, to=1-5]
\end{tikzcd}
  \end{center}
\end{lemma}
\begin{proof}
  By diagram chasing, as follows.
  \begin{center}\hfill
    \ajustedroit[.95]{
\begin{tikzcd}[ampersand replacement=\&,baseline=(\tikzcdmatrixname-11-6.base)]
	{\Gamma_{SX}X} \&\&\&\&\& {\Gamma_{SX}SX} \&\&\&\& {\Gamma_{SSX}SX} \\
	\&\&\&\&\& {\Gamma_{SSX}X} \\
	\\
	\\
	\&\& {S(\Gamma_{SX}X)} \&\& {S(\Gamma_{SSX}X)} \& {S(\Gamma_{SSX}X + X + SX)} \\
	\&\&\&\&\& {S(S(\Gamma_{SSX}X + X) + SX)} \&\& {\Gamma_{SX}SX} \\
	\&\&\&\&\& {SS(\Gamma_{SSX}X + X + X)} \\
	\\
	TX \&\&\&\&\& {S(\Gamma_{SX}X + X + X)} \&\&\&\& TSX \\
	\\
	\&\&\&\&\& STX
	\arrow["{\eta_{\Gamma_S,X}}"', from=1-1, to=9-1]
	\arrow["{\eta_{\Gamma_S,SX}}", from=1-10, to=9-10]
	\arrow["{\delta_X}", from=9-10, to=11-6]
	\arrow["{\eta^S_{TX}}"', from=9-1, to=11-6]
	\arrow[""{name=0, anchor=center, inner sep=0}, "{d^{\bullet\omega}_{X,SX}}"', from=1-10, to=5-6]
	\arrow["{\eta^S \Gamma_{SX}X}"', from=1-1, to=5-3]
	\arrow["{S in_1}"', curve={height=30pt}, from=5-3, to=9-6]
	\arrow["{S(\eta^S + SX)}"', from=5-6, to=6-6]
	\arrow["{S[Sin_1,Sin_2]}"', from=6-6, to=7-6]
	\arrow["{\mu^S(\Gamma_{\mu^S_X}X + X + X)}"', from=7-6, to=9-6]
	\arrow["{S[\eta_{\Gamma_S,X},\eta^T_X,\eta^T_X]}"', from=9-6, to=11-6]
	\arrow["{\Gamma_\mu}", from=1-10, to=6-8]
	\arrow[""{name=1, anchor=center, inner sep=0}, "{d^{\bullet\omega}_{X,X}}", from=6-8, to=9-6]
	\arrow["{\text{(Lemma~\ref{lem:d:delta:Y})}}"{description}, curve={height=18pt}, draw=none, from=6-8, to=9-10]
	\arrow["{\Gamma_{\eta^S_{SX}}X}"', from=1-1, to=2-6]
	\arrow[""{name=2, anchor=center, inner sep=0}, "{\Gamma_{SSX} \eta^S_X}"', from=2-6, to=1-10]
	\arrow["{\Gamma_{SX} \eta^S_X}", from=1-1, to=1-6]
	\arrow["{\Gamma_{\eta^S_{SX}} SX}", from=1-6, to=1-10]
	\arrow["{\eta^Sin_1}"', from=2-6, to=5-6]
	\arrow["{S(\Gamma_{\eta^S_{SX}}X)}", from=5-3, to=5-5]
	\arrow["{Sin_1}", from=5-5, to=5-6]
	\arrow["{\text{(Lemma~\ref{lem:coh:hbulletomegamu})}}"{description}, Rightarrow, draw=none, from=0, to=1]
	\arrow["{\text{(Lemma~\ref{lem:coh:hbulletomegamu})}}"{description}, Rightarrow, draw=none, from=2, to=0]
\end{tikzcd}}      \qedhere
  \end{center}
\end{proof}

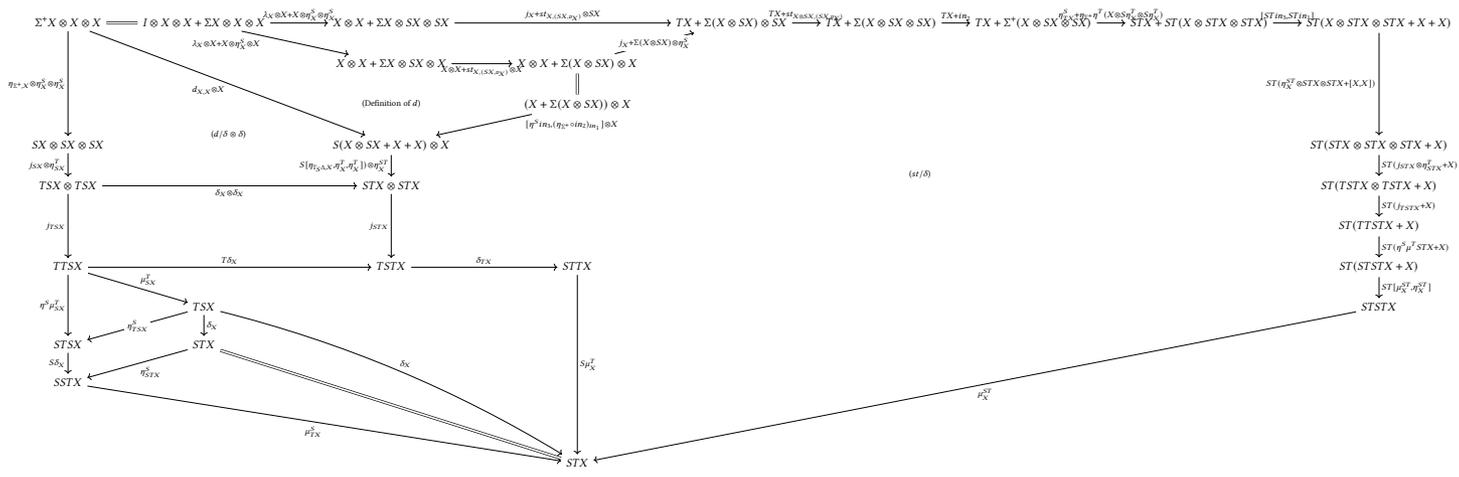
\begin{figure}[htp]
  \ajustetourne{
\begin{tikzcd}[ampersand replacement=\&]
	{\Sigma^+X \otimes X \otimes X} \& {I\otimes X \otimes X + \Sigma X \otimes X \otimes X} \&\& {X\otimes X + \Sigma X \otimes SX \otimes SX} \&\&\& {TX + \Sigma (X \otimes SX) \otimes SX} \& {TX + \Sigma(X \otimes SX \otimes SX)} \& {TX + \Sigma^+(X \otimes SX \otimes SX)} \& {STX + ST(X \otimes STX \otimes STX)} \& {ST(X \otimes STX \otimes STX + X + X)} \\
	\&\&\& {X \otimes X + \Sigma X \otimes SX \otimes X} \&\& {X \otimes X + \Sigma(X\otimes SX) \otimes X} \\
	\&\&\&\&\& {(X + \Sigma(X \otimes SX)) \otimes X} \\
	{SX \otimes SX \otimes SX} \&\&\& {S(X\otimes SX + X+ X)\otimes X} \&\&\&\&\&\&\& {ST(STX \otimes STX \otimes STX + X)} \\
	{TSX\otimes TSX} \&\&\& {STX \otimes STX} \&\&\&\&\&\&\& {ST(TSTX \otimes TSTX + X)} \\
	\&\&\&\&\&\&\&\&\&\& {ST(TTSTX + X)} \\
	TTSX \&\&\& TSTX \&\& STTX \&\&\&\&\& {ST(STSTX + X)} \\
	\& TSX \&\&\&\&\&\&\&\&\& STSTX \\
	STSX \& STX \\
	SSTX \\
	\\
	\\
	\&\&\&\&\& STX
	\arrow[Rightarrow, no head, from=1-1, to=1-2]
	\arrow["{\lambda_X \otimes X + X \otimes \eta^S_X \otimes \eta^S_X}", from=1-2, to=1-4]
	\arrow["{j_X + st_{X,(SX,v_X)} \otimes SX}", from=1-4, to=1-7]
	\arrow["{TX + st_{X\otimes SX,(SX,v_X)}}", from=1-7, to=1-8]
	\arrow["{[STin_3,STin_1]}", from=1-10, to=1-11]
	\arrow["{ST(\eta^{ST}_X \otimes STX \otimes STX + [X,X])}"', from=1-11, to=4-11]
	\arrow["{ST(j_{STX} \otimes \eta^T_{STX} + X)}", from=4-11, to=5-11]
	\arrow["{ST(j_{TSTX} + X)}", from=5-11, to=6-11]
	\arrow["{ST(\eta^S \mu^T STX + X)}", from=6-11, to=7-11]
	\arrow["{ST[\mu^{ST}_X,\eta^{ST}_X]}", from=7-11, to=8-11]
	\arrow[""{name=0, anchor=center, inner sep=0}, "{\mu^{ST}_X}", from=8-11, to=13-6]
	\arrow["{\eta_{\Sigma^+,X} \otimes \eta^S_X \otimes \eta^S_X}"', from=1-1, to=4-1]
	\arrow["{j_{SX} \otimes \eta^T_{SX}}"', from=4-1, to=5-1]
	\arrow["{j_{TSX}}"', from=5-1, to=7-1]
	\arrow["{\eta^S\mu^T_{SX}}"', from=7-1, to=9-1]
	\arrow["{S\delta_X}"', from=9-1, to=10-1]
	\arrow["{\mu^S_{TX}}"', from=10-1, to=13-6]
	\arrow["{\mu^T_{SX}}", from=7-1, to=8-2]
	\arrow["{\delta_X}", curve={height=-18pt}, from=8-2, to=13-6]
	\arrow["{\eta^S_{TSX}}"{description}, from=8-2, to=9-1]
	\arrow["{\delta_X}", from=8-2, to=9-2]
	\arrow["{\eta^S_{STX}}", from=9-2, to=10-1]
	\arrow[Rightarrow, no head, from=9-2, to=13-6]
	\arrow["{T \delta_X}", from=7-1, to=7-4]
	\arrow["{\delta_{TX}}", from=7-4, to=7-6]
	\arrow["{S\mu^T_X}", from=7-6, to=13-6]
	\arrow[""{name=1, anchor=center, inner sep=0}, "{d_{X,X} \otimes X}"', from=1-1, to=4-4]
	\arrow["{S[\eta_{\Gamma_S \Delta,X},\eta^T_X,\eta^T_X])\otimes \eta^{ST}_X}"', from=4-4, to=5-4]
	\arrow["{j_{STX}}"', from=5-4, to=7-4]
	\arrow["{\lambda_X \otimes X + X \otimes \eta^S_X \otimes X}"'{pos=0.2}, from=1-2, to=2-4]
	\arrow["{X\otimes X + st_{X,(SX,v_X)} \otimes X}"', from=2-4, to=2-6]
	\arrow[Rightarrow, no head, from=2-6, to=3-6]
	\arrow["{j_X + \Sigma(X\otimes SX) \otimes \eta^S_X}"{description}, from=2-6, to=1-7]
	\arrow["{\text{(Definition of } d \text{)}}"{description}, draw=none, from=2-4, to=4-4]
	\arrow["{[\eta^S in_3, (\eta_{\Sigma^+} \circ in_2)_{in_1}]\otimes X}"{pos=0.1}, from=3-6, to=4-4]
	\arrow[""{name=2, anchor=center, inner sep=0}, "{\delta_X \otimes \delta_X}"', from=5-1, to=5-4]
	\arrow["{\eta^S_{TX} + \eta_{\Sigma^+}\eta^T (X\otimes S\eta^T_X \otimes S\eta^T_X)}", from=1-9, to=1-10]
	\arrow["{TX + in_2}", from=1-8, to=1-9]
	\arrow["{\text{($d/\delta \otimes \delta$)}}"{description}, Rightarrow, draw=none, from=2, to=1]
	\arrow["{\text{($st/\delta$)}}"{description, pos=0.6}, draw=none, from=0, to=1-8]
\end{tikzcd}
  }
  \caption{Coherence of $L₁$}
  \label{fig:coherence:L1}
\end{figure}
In Figure~\ref{fig:coherence:L1},
the subdiagram marked ($d/δ⊗δ$) commutes by chasing as follows.
\begin{center}
  \adjustbox{max width=\columnwidth}{%
\begin{tikzcd}[ampersand replacement=\&]
	{\Sigma^+X \otimes X \otimes X} \&\&\&\&\& {S(X\otimes SX + X+ X)\otimes X} \\
	{SX \otimes SX \otimes SX} \& {SX \otimes SX \otimes X} \\
	{SX \otimes SSX \otimes SX} \& {SX \otimes SSX \otimes X} \\
	\& {TSX \otimes X} \&\& {STX \otimes X} \\
	{TSX \otimes TSX} \&\&\&\&\& {STX \otimes STX}
	\arrow["{\eta_{\Sigma^+,X} \otimes \eta^S_X \otimes \eta^S_X}"', from=1-1, to=2-1]
	\arrow[""{name=0, anchor=center, inner sep=0}, "{d_{X,X} \otimes X}", from=1-1, to=1-6]
	\arrow["{S[\eta_{\Gamma_S \Delta,X},\eta^T_X,\eta^T_X])\otimes \eta^{ST}_X}"{description}, from=1-6, to=5-6]
	\arrow["{SX \otimes \eta^S_{SX} \otimes SX}"', from=2-1, to=3-1]
	\arrow["{\eta_{\Gamma_S}SX \otimes \eta^T_{SX}}"', from=3-1, to=5-1]
	\arrow["{\eta_{\Sigma^+,X} \otimes \eta^S_X \otimes X}"{pos=0.9}, from=1-1, to=2-2]
	\arrow["{SX \otimes \eta^S_{SX} \otimes X}", from=2-2, to=3-2]
	\arrow["{\eta_{\Gamma_S,SX} \otimes X}", from=3-2, to=4-2]
	\arrow["{\delta_X \otimes X}", from=4-2, to=4-4]
	\arrow["{S[\eta_{\Gamma_S \Delta,X},\eta^T_X,\eta^T_X])\otimes X}"{description}, from=1-6, to=4-4]
	\arrow["{STX \otimes \eta^{ST}_X}"', from=4-4, to=5-6]
	\arrow["{TSX \otimes \eta^T \eta^S X}"{description}, from=4-2, to=5-1]
	\arrow["{\delta_X \otimes \delta_X}"', from=5-1, to=5-6]
	\arrow["{\text{\eqref{eq:simple:d:delta}}}"{description}, Rightarrow, draw=none, from=0, to=4-4]
\end{tikzcd}
    }
    \end{center}
The subdiagram marked ($st/δ$) has a coproduct as domain, so we check
its commutation termwise in Figure~\ref{fig:coherence:L1:st:delta}.

\begin{figure}[htp]
  \adjustbox{max width=\columnwidth,max height=.95\textheight}{%
  \begin{turn}{90}
\begin{tikzcd}[ampersand replacement=\&]
	{X\otimes X} \&\&\&\&\&\&\&\& TX \\
	{SX \otimes X} \& {SX \otimes SX} \&\& TSX \&\& STX \\
	{S(X \otimes SX + X + X) \otimes X} \&\& {STX \otimes STX} \& TSTX \&\& STTX \&\&\& STX
	\arrow["{j_X}", from=1-1, to=1-9]
	\arrow["{\eta^S_{TX}}", from=1-9, to=3-9]
	\arrow["{\eta^S_X \otimes X}"', from=1-1, to=2-1]
	\arrow["{Sin_3 \otimes X}"', from=2-1, to=3-1]
	\arrow["{S[\eta_{\Gamma_S \Delta,X},\eta^T_X,\eta^T_X]\otimes \eta^{ST}_X}"', from=3-1, to=3-3]
	\arrow["{j_{STX}}"', from=3-3, to=3-4]
	\arrow["{\delta_{TX}}"', from=3-4, to=3-6]
	\arrow["{S\mu^T_X}"', from=3-6, to=3-9]
	\arrow["{S\eta^T_X \otimes \eta^{ST}_X}"{description}, from=2-1, to=3-3]
	\arrow["{S_X \otimes \eta^S_X}"{description}, from=2-1, to=2-2]
	\arrow["{j_{SX}}"{description}, from=2-2, to=2-4]
	\arrow["{TS\eta^T_X}"', from=2-4, to=3-4]
	\arrow["{\delta_X}"{description}, from=2-4, to=2-6]
	\arrow["{ST\eta^T_X}"', from=2-6, to=3-6]
	\arrow[curve={height=-6pt}, Rightarrow, no head, from=2-6, to=3-9]
	\arrow["{T \eta^S_X}"{description}, from=1-9, to=2-4]
\end{tikzcd}
  \end{turn}}
\hfil
  \adjustbox{max width=\columnwidth,max height=.95\textheight}{%
  \begin{turn}{90}
\begin{tikzcd}[ampersand replacement=\&]
	{\Sigma(X\otimes SX) \otimes X} \&\&\&\& {\Sigma (X \otimes SX) \otimes SX} \& {\Sigma(X \otimes SX \otimes SX)} \& {\Sigma^+(X \otimes SX \otimes SX)} \& {ST(X \otimes STX \otimes STX)} \&\&\&\& {ST(X \otimes STX \otimes STX + X + X)} \\
	\& {\Sigma^+(X\otimes SX) \otimes X} \&\&\&\&\& {S(X\otimes SX \otimes SX)} \& {ST(TX \otimes STX \otimes STX)} \&\& {ST(STX \otimes SSTX \otimes TSTX)} \&\& {ST(STX \otimes STX \otimes STX + X)} \\
	{S(X\otimes SX + X+ X)\otimes X} \&\&\&\&\& {S(X \otimes SX \otimes SX + X \otimes SX + X)} \&\&\&\& {ST(TSTX \otimes TSTX)} \&\& {ST(STX \otimes SSTX \otimes TSTX + X)} \\
	\&\&\&\&\&\& {S(TX \otimes STX)} \& {ST(TTX \otimes STTX)} \&\& {ST(TSTX \otimes STSTX)} \\
	{STX \otimes STX} \& {\Sigma^+(TX) \otimes TX} \&\&\&\& {S(TX \otimes STX + TX + TX)} \&\& STTTTX \& {\text{($\Gamma;\Gamma$)}} \& STTTSTX \&\& {ST(TSTX \otimes TSTX + X)} \\
	\& {STX \otimes STX} \&\&\&\&\&\&\&\& STSTSTX \&\& {ST(TTSTX + X)} \\
	{STX \otimes SSTX} \&\&\&\&\&\&\& STTTX \&\& STSTX \&\& {ST(STSTX + X)} \\
	\&\&\&\&\&\&\& STTX \&\&\&\& STSTX \\
	\& TSTX \&\&\&\& STTX \&\&\&\&\&\& STX
	\arrow["{st_{X\otimes SX,(SX,v_X)}}", from=1-5, to=1-6]
	\arrow["{STin_1}", from=1-8, to=1-12]
	\arrow["{ST(\eta^{ST}_X \otimes STX \otimes STX + [X,X])}", from=1-12, to=2-12]
	\arrow["{ST(j_{TSTX} + X)}", from=5-12, to=6-12]
	\arrow["{ST(\eta^S \mu^T STX + X)}", from=6-12, to=7-12]
	\arrow["{ST[\mu^{ST}_X,\eta^{ST}_X]}", from=7-12, to=8-12]
	\arrow["{\mu^{ST}_X}", from=8-12, to=9-12]
	\arrow[""{name=0, anchor=center, inner sep=0}, "{\delta_{TX}}"', from=9-2, to=9-6]
	\arrow["{S\mu^T_X}", from=9-6, to=9-12]
	\arrow["{S[\eta_{\Gamma_S \Delta,X},\eta^T_X,\eta^T_X]\otimes \eta^{ST}_X}"', from=3-1, to=5-1]
	\arrow["{\Sigma(X\otimes SX) \otimes \eta^S_X}", from=1-1, to=1-5]
	\arrow["{(\eta_{\Sigma^+} \circ in_2)_{in_1} \otimes X}"', from=1-1, to=3-1]
	\arrow["{(\eta_{\Sigma^+} \circ in_2)_{in_1}}"', from=1-6, to=3-6]
	\arrow["{in_2 \otimes X}", from=1-1, to=2-2]
	\arrow[""{name=1, anchor=center, inner sep=0}, "{d_{X\otimes SX, X}}", from=2-2, to=3-6]
	\arrow["{STX \otimes \eta^S_{STX}}"', from=5-1, to=7-1]
	\arrow["{\eta_{\Gamma_S}}"', from=7-1, to=9-2]
	\arrow["{\Sigma^+(\eta_{\Gamma_S,X}) \otimes \eta^T_X}"', from=2-2, to=5-2]
	\arrow["{S(\eta_{\Gamma_S,X} \otimes S\eta^T_X + \eta_{\Gamma_S,X} + \eta^T_X)}"', from=3-6, to=5-6]
	\arrow[""{name=2, anchor=center, inner sep=0}, "{d_{TX,TX}}"', from=5-2, to=5-6]
	\arrow["{\eta_{\Sigma^+,TX} \otimes \eta^S_{TX}}"', from=5-2, to=6-2]
	\arrow["{STX \otimes \eta^S_{STX}}"', from=6-2, to=7-1]
	\arrow["{S[\eta_{\Gamma_S,TX},\eta^T_{TX},\eta^T_{TX}]}"', from=5-6, to=9-6]
	\arrow["{\eta_{\Sigma^+}\eta^T (X\otimes S\eta^T_X \otimes S\eta^T_X)}", from=1-7, to=1-8]
	\arrow["{in_2}", from=1-6, to=1-7]
	\arrow["{ST(\eta^T_X \otimes STX \otimes STX)}", from=1-8, to=2-8]
	\arrow["{ST(\eta_{\Gamma_S,TX} \otimes S\eta^T_{TX})}"', from=2-8, to=4-8]
	\arrow["{ST(\eta_{\Gamma_S,TTX})}"', from=4-8, to=5-8]
	\arrow["{ST(STX \otimes \eta^S_{STX} \otimes \eta^T_{STX} + X)}", from=2-12, to=3-12]
	\arrow["{ST(\eta_{\Gamma_S,STX} \otimes TSTX + X)}", from=3-12, to=5-12]
	\arrow["{ST(\eta^S_{TX} \otimes \eta^S_{STX} \otimes \eta^T_{STX})}", from=2-8, to=2-10]
	\arrow["{ST(\eta_{\Gamma_S,STX} \otimes TSTX)}"', from=2-10, to=3-10]
	\arrow["{ST(TSTX \otimes \eta^S_{TSTX})}"', from=3-10, to=4-10]
	\arrow["{ST(\eta_{\Gamma_S,TSTX})}", from=4-10, to=5-10]
	\arrow["{ST \eta^S \mu^T STX}", from=5-10, to=6-10]
	\arrow["{ST \mu^{ST}_X}", from=6-10, to=7-10]
	\arrow["{\mu^{ST}_X}", from=7-10, to=9-12]
	\arrow["{S\mu^T_X}"', from=8-8, to=9-12]
	\arrow["{\eta_{\Sigma^+}\circ in_2}", from=1-6, to=2-7]
	\arrow["{S(\eta_{\Gamma_S,X}\otimes S\eta^T_X)}"'{pos=0.8}, from=2-7, to=4-7]
	\arrow["{Sin_1}"', from=2-7, to=3-6]
	\arrow["{S\eta^T (\eta^T_X \otimes S \eta^T_X \otimes S \eta^T_X)}", from=2-7, to=2-8]
	\arrow["{S\eta^T(T\eta^T_X \otimes ST \eta^T_X)}", from=4-7, to=4-8]
	\arrow["{S\eta_{\Gamma_S,TX}}"', from=4-7, to=9-6]
	\arrow["{S\eta^TTT\eta^T_X}"{description}, from=9-6, to=5-8]
	\arrow["{STT \mu^T_X}", from=5-8, to=7-8]
	\arrow["{ST\mu^T_X}", from=7-8, to=8-8]
	\arrow["{S\eta^T_{TTX}}", from=9-6, to=7-8]
	\arrow["{ST in_1}"{description}, from=2-10, to=3-12]
	\arrow["{ST(in_1)}"', from=3-10, to=5-12]
	\arrow["{STin_1}"{description}, from=5-10, to=6-12]
	\arrow[Rightarrow, no head, from=9-6, to=8-8]
	\arrow["{\text{\eqref{eq:simple:d:delta}}}"', Rightarrow, draw=none, from=2, to=0]
	\arrow["{\text{(naturality of $d$)}}"{description}, Rightarrow, draw=none, from=1, to=2]
\end{tikzcd}
  \end{turn}}
  \caption{Subdiagram ($st/δ$) of Figure~\ref{fig:coherence:L1}}
  \label{fig:coherence:L1:st:delta}
\end{figure}

The subdiagram ($Γ;Γ$) commutes as in Figure~\ref{fig:gamma:gamma}.

\begin{figure}[htp]
  \adjustbox{max width=0.95\columnwidth}{
\begin{tikzcd}[ampersand replacement=\&]
	{ST(TX \otimes STX \otimes STX)} \&\&\&\& {ST(STX \otimes SSTX \otimes TSTX)} \\
	\&\&\&\& {ST(TSTX \otimes TSTX)} \\
	{ST(TTX \otimes STTX)} \&\&\& {ST(STTX \otimes TSTX)} \& {ST(TSTX \otimes STSTX)} \\
	\&\& {ST(STTX \otimes STTX)} \&\& STTTSTX \\
	\&\& {ST(STTX \otimes SSTTX)} \\
	\&\& STTSTTX \& STTSTX \& STSTSTX \\
	STTTTX \&\& STSTTTX \& STSTTX \& STSSTTX \\
	\& STTTX \&\& SSTTTX \& STSTTX \\
	STTTX \&\& STTTX \&\& STSTX \\
	\&\&\&\& SSTTX \\
	\&\&\&\& STTX \\
	STTX \&\&\&\& STX
	\arrow["{ST(\eta_{\Gamma_S,TX} \otimes S\eta^T_{TX})}"', color={rgb,255:red,37;green,147;blue,37}, from=1-1, to=3-1]
	\arrow[""{name=0, anchor=center, inner sep=0}, "{ST(\eta^S_{TX} \otimes \eta^S_{STX} \otimes \eta^T_{STX})}", from=1-1, to=1-5]
	\arrow["{ST(\eta_{\Gamma_S,STX} \otimes TSTX)}", color={rgb,255:red,37;green,147;blue,37}, from=1-5, to=2-5]
	\arrow["{ST(TSTX \otimes \eta^S_{TSTX})}", from=2-5, to=3-5]
	\arrow["{ST(\eta_{\Gamma_S,TSTX})}", color={rgb,255:red,37;green,147;blue,37}, from=3-5, to=4-5]
	\arrow["{ST \eta^S \mu^T STX}", from=4-5, to=6-5]
	\arrow["{ST \mu^T_X}"', from=9-1, to=12-1]
	\arrow["{S\mu^T_X}"', from=12-1, to=12-5]
	\arrow["{ST \mu^TST}"{description}, from=4-5, to=6-4]
	\arrow["{ST\delta_{TX}}", from=6-4, to=7-4]
	\arrow["{S\delta_{TTX}}", from=7-4, to=8-4]
	\arrow["{ST(\delta_{TX} \otimes TSTX)}"'{pos=0.8}, from=2-5, to=3-4]
	\arrow["{ST(STTX \otimes \delta_{TX})}"'{pos=0.9}, from=3-4, to=4-3]
	\arrow[""{name=1, anchor=center, inner sep=0}, "{ST(\eta^S_{TTX} \otimes STTX)}"{pos=0.7}, from=3-1, to=4-3]
	\arrow["{ST(STTX \otimes \eta^S_{STTX})}"', from=4-3, to=5-3]
	\arrow["{ST(\delta_{TX} \otimes S\delta_{TX})}"', from=3-5, to=5-3]
	\arrow["{ST(\eta_{\Gamma_S,STTX})}", color={rgb,255:red,37;green,147;blue,37}, from=5-3, to=6-3]
	\arrow["{STT\delta_{TX}}"', from=4-5, to=6-3]
	\arrow["{ST\delta_{TTX}}"', from=6-3, to=7-3]
	\arrow["{STS\mu^T_{TX}}", from=7-3, to=7-4]
	\arrow["{ST(\eta_{\Gamma_S,TTX})}"', color={rgb,255:red,37;green,147;blue,37}, from=3-1, to=7-1]
	\arrow["{STT\mu^T_X}"', from=7-1, to=9-1]
	\arrow[""{name=2, anchor=center, inner sep=0}, "{ST\eta^S_{TTTX}}"', from=7-1, to=7-3]
	\arrow["{\mu^S_{TTTX}}", from=8-4, to=9-3]
	\arrow["{ST\mu^T_{TX}}"{description}, from=7-1, to=8-2]
	\arrow["{ST\eta^S_{TTX}}"{description}, from=8-2, to=7-4]
	\arrow["{S\eta^S_{TTTX}}"{description}, from=8-2, to=8-4]
	\arrow[Rightarrow, no head, from=8-2, to=9-3]
	\arrow["{ST\mu^T_X}"{description}, from=9-3, to=12-1]
	\arrow["{ST\eta^S_{TSTX}}", from=6-4, to=6-5]
	\arrow["{ST S\delta_{TX}}", from=6-5, to=7-5]
	\arrow["{ST\mu^S}", from=7-5, to=8-5]
	\arrow["{ST\eta^S_{STTX}}", from=7-4, to=7-5]
	\arrow[Rightarrow, no head, from=7-4, to=8-5]
	\arrow["{STS\mu^T_X}", from=8-5, to=9-5]
	\arrow["{S\delta_{TX}}", from=9-5, to=10-5]
	\arrow["{\mu^S_{TTX}}", from=10-5, to=11-5]
	\arrow["{S\mu^T_X}", from=11-5, to=12-5]
	\arrow["{ST\mu^T_X}", from=9-3, to=11-5]
	\arrow["{\text{(Lemma~\ref{lem:eta:gamma:s})}}"{description}, Rightarrow, draw=none, from=0, to=1]
	\arrow["{\text{(Lemma~\ref{lem:eta:gamma:s})}}"{description}, Rightarrow, draw=none, from=1, to=2]
      \end{tikzcd}}    
  \caption{Subdiagram ($Γ;Γ$)}
  \label{fig:gamma:gamma}
\end{figure}

\subsection{Coherence of $R₁$}

We now check coherence of $R₁$, in Figure~\ref{fig:coherence:R1}.

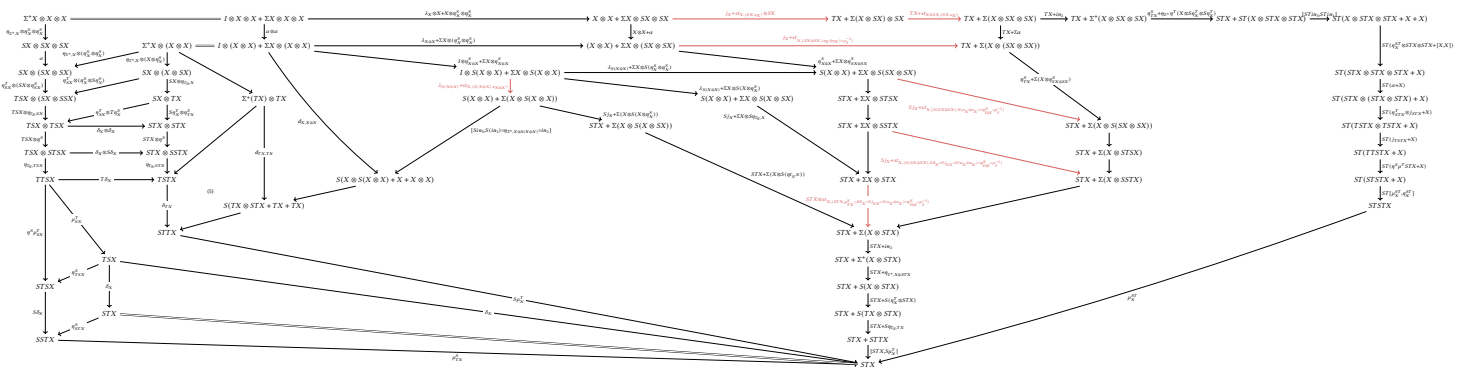
\begin{figure}[htp]
  \ajustetourne{
\begin{tikzcd}[ampersand replacement=\&]
	{\Sigma^+X \otimes X \otimes X} \&\&\& {I\otimes X \otimes X + \Sigma X \otimes X \otimes X} \&\&\& {X\otimes X + \Sigma X \otimes SX \otimes SX} \&\& {TX + \Sigma (X \otimes SX) \otimes SX} \&\& {TX + \Sigma(X \otimes SX \otimes SX)} \& {TX + \Sigma^+(X \otimes SX \otimes SX)} \&\&\& {STX + ST(X \otimes STX \otimes STX)} \& {ST(X \otimes STX \otimes STX + X + X)} \\
	{SX \otimes SX \otimes SX} \&\& {\Sigma^+X \otimes (X \otimes X)} \& {I\otimes (X \otimes X) + \Sigma X \otimes (X \otimes X)} \&\&\& {(X \otimes X) + \Sigma X \otimes (SX \otimes SX)} \&\&\&\& {TX + \Sigma(X \otimes (SX \otimes SX))} \\
	{SX \otimes (SX \otimes SX)} \&\& {SX \otimes (X \otimes SX)} \&\&\& {I\otimes S(X \otimes X) + \Sigma X \otimes S(X \otimes X)} \&\&\& {S(X \otimes X) + \Sigma X \otimes S(SX \otimes SX)} \&\&\&\&\&\&\& {ST(STX \otimes STX \otimes STX + X)} \\
	{TSX \otimes (SX \otimes SSX)} \&\& {SX \otimes TX} \& {\Sigma^+(TX) \otimes TX} \&\& {S(X\otimes X) + \Sigma(X \otimes  S(X\otimes X))} \&\& {S(X\otimes X) + \Sigma X \otimes S(X \otimes SX)} \& {STX + \Sigma X \otimes STSX} \&\&\&\&\&\&\& {ST(STX \otimes (STX \otimes STX) + X)} \\
	{TSX\otimes TSX} \&\& {STX \otimes STX} \&\&\&\& {STX + \Sigma(X \otimes S(X \otimes SX))} \&\& {STX + \Sigma X \otimes SSTX} \&\&\& {STX + \Sigma(X \otimes S(SX \otimes SX))} \&\&\&\& {ST(TSTX \otimes TSTX + X)} \\
	{TSX \otimes STSX} \&\& {STX \otimes SSTX} \&\&\&\&\&\&\&\&\& {STX + \Sigma(X \otimes STSX)} \&\&\&\& {ST(TTSTX + X)} \\
	TTSX \&\& TSTX \&\& {S(X⊗S(X\otimes X) + X + X\otimes X)} \&\&\&\& {STX + \Sigma X \otimes STX} \&\&\& {STX + \Sigma(X \otimes SSTX)} \&\&\&\& {ST(STSTX + X)} \\
	\&\&\& {S(TX \otimes STX + TX + TX)} \&\&\&\&\&\&\&\&\&\&\&\& STSTX \\
	\&\& STTX \&\&\&\&\&\& {STX + \Sigma(X \otimes STX)} \\
	\& TSX \&\&\&\&\&\&\& {STX + \Sigma^+(X \otimes STX)} \\
	STSX \&\&\&\&\&\&\&\& {STX + S(X \otimes STX)} \\
	\& STX \&\&\&\&\&\&\& {STX + S(TX \otimes STX)} \\
	SSTX \&\&\&\&\&\&\&\& {STX + STTX} \\
	\&\&\&\&\&\&\&\& STX
	\arrow[Rightarrow, no head, from=1-1, to=1-4]
	\arrow["{\lambda_X \otimes X + X \otimes \eta^S_X \otimes \eta^S_X}", from=1-4, to=1-7]
	\arrow["{j_X + st_{X,(SX,v_X)} \otimes SX}", color={rgb,255:red,214;green,92;blue,92}, from=1-7, to=1-9]
	\arrow["{TX + st_{X\otimes SX,(SX,v_X)}}", color={rgb,255:red,214;green,92;blue,92}, from=1-9, to=1-11]
	\arrow["{[STin_3,STin_1]}", from=1-15, to=1-16]
	\arrow["{ST(\eta^T_{STX} \otimes j_{STX} + X)}", from=4-16, to=5-16]
	\arrow["{ST(j_{TSTX} + X)}", from=5-16, to=6-16]
	\arrow["{ST(\eta^S \mu^T STX + X)}", from=6-16, to=7-16]
	\arrow["{ST[\mu^{ST}_X,\eta^{ST}_X]}", from=7-16, to=8-16]
	\arrow["{\mu^{ST}_X}", curve={height=-12pt}, from=8-16, to=14-9]
	\arrow["{\eta^S\mu^T_{SX}}"', from=7-1, to=11-1]
	\arrow["{S\delta_X}"', from=11-1, to=13-1]
	\arrow["{\mu^S_{TX}}"', from=13-1, to=14-9]
	\arrow["{\eta^S_{TX} + \eta_{\Sigma^+}\eta^T (X\otimes S\eta^T_X \otimes S\eta^T_X)}", from=1-12, to=1-15]
	\arrow["{TX + in_2}", from=1-11, to=1-12]
	\arrow["{\eta_{\Sigma^+,X} \otimes \eta^S_X \otimes \eta^S_X}"', from=1-1, to=2-1]
	\arrow["\alpha"', from=2-1, to=3-1]
	\arrow["{ST(\eta^{ST}_X \otimes STX \otimes STX + [X,X])}", from=1-16, to=3-16]
	\arrow["{ST(\alpha + X)}", from=3-16, to=4-16]
	\arrow["{\alpha \otimes \alpha}", from=1-4, to=2-4]
	\arrow["{\lambda_{X\otimes X} + \Sigma X \otimes (\eta^S_X \otimes \eta^S_X)}", from=2-4, to=2-7]
	\arrow["{X\otimes X + \alpha}", from=1-7, to=2-7]
	\arrow["{j_X + st_{X,(SX\otimes SX, (v_X \otimes v_X) \circ \rho_I^{-1})}}", color={rgb,255:red,214;green,92;blue,92}, from=2-7, to=2-11]
	\arrow["{TX + \Sigma \alpha}", from=1-11, to=2-11]
	\arrow["{I \otimes \eta^S_{X\otimes X} + \Sigma X \otimes \eta^S_{X\otimes X} }"{pos=1}, from=2-4, to=3-6]
	\arrow["{\lambda_{S(X\otimes X)} + st_{X,(S(X\otimes X),v_{X \otimes X})}}"', color={rgb,255:red,214;green,92;blue,92}, from=3-6, to=4-6]
	\arrow["{[Sin₃,S(in₁)∘η_{Σ⁺,X⊗S(X\otimes X)}∘in₂]}"{pos=0.3}, from=4-6, to=7-5]
	\arrow["{d_{X,X\otimes X}}"', curve={height=12pt}, from=2-4, to=7-5]
	\arrow["{\eta^S_{X\otimes X} + \Sigma X \otimes \eta^S_{SX \otimes SX}}"{pos=1}, from=2-7, to=3-9]
	\arrow["{\lambda_{S(X\otimes X)} + \Sigma X \otimes S(\eta^S_X \otimes \eta^S_X)}"{pos=0.3}, from=3-6, to=3-9]
	\arrow["{Sj_X + st_{X,(S(SX \otimes SX),S(v_X \otimes v_X) \circ \eta^S_{I\otimes I} \circ \rho_I^{-1})}}"'{pos=0.6}, color={rgb,255:red,214;green,92;blue,92}, from=3-9, to=5-12]
	\arrow["{\eta^S_{TX} + \Sigma(X \otimes \eta^S_{SX \otimes SX})}"{description, pos=0.4}, from=2-11, to=5-12]
	\arrow[from=5-12, to=6-12]
	\arrow[from=6-12, to=7-12]
	\arrow[from=7-12, to=9-9]
	\arrow[from=3-9, to=4-9]
	\arrow[from=4-9, to=5-9]
	\arrow["{Sj_X + st_{X,(S(SX \otimes SX),S\delta_X \circ Sj_{SX} \circ S(v_X \otimes v_X) \circ \eta^S_{I\otimes I} \circ \rho_I^{-1})}}"'{pos=0.6}, color={rgb,255:red,214;green,92;blue,92}, from=5-9, to=7-12]
	\arrow["{STX \otimes st_{X,(STX,\mu^S_{TX} \circ S\delta_X \circ Sj_{SX} \circ S(v_X \otimes v_X) \circ \eta^S_{I\otimes I} \circ \rho_I^{-1})}}"{description, pos=0.4}, color={rgb,255:red,214;green,92;blue,92}, from=7-9, to=9-9]
	\arrow[from=5-9, to=7-9]
	\arrow[from=7-5, to=8-4]
	\arrow["{d_{TX,TX}}"{description}, from=4-4, to=8-4]
	\arrow[from=4-4, to=7-3]
	\arrow["{\delta_{TX}}"{description}, from=7-3, to=9-3]
	\arrow[from=8-4, to=9-3]
	\arrow["{S\mu^T_X}"{description}, from=9-3, to=14-9]
	\arrow["{\text{(5)}}"{description}, draw=none, from=7-3, to=8-4]
	\arrow["{\mu^T_{SX}}"{description}, from=7-1, to=10-2]
	\arrow["{\delta_X}"{description}, from=10-2, to=14-9]
	\arrow["{T\delta_X}"{description}, from=7-1, to=7-3]
	\arrow["{\eta^S_{TSX}}"{description}, from=10-2, to=11-1]
	\arrow["{\delta_X}"{description}, from=10-2, to=12-2]
	\arrow["{\eta^S_{STX}}"{description}, from=12-2, to=13-1]
	\arrow[Rightarrow, no head, from=12-2, to=14-9]
	\arrow["{TSX \otimes \eta^S}"', from=5-1, to=6-1]
	\arrow["{\eta_{\Gamma_S,TSX}}"', from=6-1, to=7-1]
	\arrow["{\delta_X \otimes S\delta_X}"{description}, from=6-1, to=6-3]
	\arrow["{\eta_{\Gamma_S,STX}}"', from=6-3, to=7-3]
	\arrow["{\delta_X \otimes \delta_X}"', from=5-1, to=5-3]
	\arrow["{STX \otimes \eta^S}"', from=5-3, to=6-3]
	\arrow["{\eta^T_{SX} \otimes (SX \otimes \eta^S_{SX})}"', from=3-1, to=4-1]
	\arrow["{TSX \otimes \eta_{\Gamma_S,SX}}"', from=4-1, to=5-1]
	\arrow["{\eta^T_{SX}\otimes T \eta^S_X}"{description}, from=4-3, to=5-1]
	\arrow[Rightarrow, no head, from=2-4, to=2-3]
	\arrow["{\eta_{\Sigma^+,X} \otimes (X \otimes \eta^S_X)}"', from=2-3, to=3-3]
	\arrow["{\eta^T_{SX} \otimes (\eta^S_X \otimes S\eta^S_X)}"', from=3-3, to=4-1]
	\arrow["{\eta_{\Sigma^+,X} \otimes (\eta^S_X \otimes \eta^S_X)}"', from=2-3, to=3-1]
	\arrow["{SX \otimes \eta_{\Gamma_S,X}}"{pos=0.2}, from=3-3, to=4-3]
	\arrow[from=2-3, to=4-4]
	\arrow["{S\eta^T_X \otimes \eta^S_{TX}}", from=4-3, to=5-3]
	\arrow["{\lambda_{S(X\otimes X)} + \Sigma X \otimes S(X \otimes \eta^S_X)}"{pos=1}, from=3-6, to=4-8]
	\arrow["{Sj_X  + \Sigma X \otimes S\eta_{\Gamma_S,X}}"'{pos=0.1}, from=4-8, to=7-9]
	\arrow["{Sj_X + \Sigma (X \otimes S(X \otimes \eta^S_X))}"{pos=1}, from=4-6, to=5-7]
	\arrow["{STX + \Sigma(X\otimes S(\eta_{\Gamma_S,X}))}", from=5-7, to=9-9]
	\arrow["{STX + in_2}", from=9-9, to=10-9]
	\arrow["{STX + \eta_{\Sigma^+,X\otimes STX}}", from=10-9, to=11-9]
	\arrow["{STX + S(\eta^T_X \otimes STX)}", from=11-9, to=12-9]
	\arrow["{STX + S\eta_{\Gamma_S,TX}}", from=12-9, to=13-9]
	\arrow["{[STX,S\mu^T_X]}", from=13-9, to=14-9]
\end{tikzcd}    }
  \caption{Coherence of $R₁$}
  \label{fig:coherence:R1}
\end{figure}
There, the top subdiagram with red arrows is the associativity axiom for pointed strengths.
We then use naturality of $st$ three times easily, and a fourth time much less easily.
The point is that we use naturality at the pair of morphisms
\begin{mathpar}
  X \xto{id_X} X \and S(X ⊗ X) \xto{S(X ⊗ η^S_X)} S(X ⊗ SX) \xto{Sη_{Γ_S,X}} STX\rlap{,}
\end{mathpar}
which requires us to prove that the second morphism is a morphism of
pointed objects. We do this by chasing the following diagram,
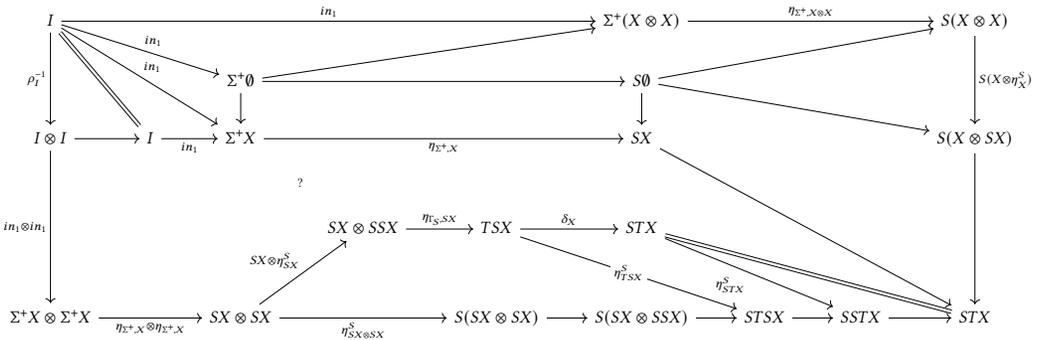
\begin{center}
  \ajustedroit{
\begin{tikzcd}[ampersand replacement=\&]
	I \&\&\&\&\& {\Sigma^+(X\otimes X)} \&\&\& {S(X\otimes X)} \\
	\&\& {\Sigma^+ \emptyset} \&\&\& S\emptyset \\
	{I \otimes I} \& I \& {\Sigma^+X} \&\&\& SX \&\&\& {S(X \otimes SX)} \\
	\\
	\&\&\& {SX\otimes SSX} \& TSX \& STX \\
	\\
	{\Sigma^+X \otimes \Sigma^+X} \&\& {SX \otimes SX} \&\& {S(SX \otimes SX)} \& {S(SX \otimes SSX)} \& STSX \& SSTX \& STX
	\arrow["{\rho_I^{-1}}"', from=1-1, to=3-1]
	\arrow["{in_1}", from=1-1, to=1-6]
	\arrow["{\eta_{\Sigma^+,X\otimes X}}", from=1-6, to=1-9]
	\arrow["{\eta_{\Sigma^+,X} \otimes \eta_{\Sigma^+,X}}"', from=7-1, to=7-3]
	\arrow[from=3-1, to=3-2]
	\arrow[Rightarrow, no head, from=1-1, to=3-2]
	\arrow["{\eta^S_{SX \otimes SX}}"', from=7-3, to=7-5]
	\arrow[from=7-7, to=7-8]
	\arrow[from=7-8, to=7-9]
	\arrow["{in_1}", from=1-1, to=3-3]
	\arrow["{in_1}"', from=3-2, to=3-3]
	\arrow["{\eta_{\Sigma^+,X}}"', from=3-3, to=3-6]
	\arrow[from=3-6, to=7-9]
	\arrow["{SX \otimes \eta^S_{SX}}", from=7-3, to=5-4]
	\arrow["{\eta_{\Gamma_S,SX}}", from=5-4, to=5-5]
	\arrow[from=7-5, to=7-6]
	\arrow[from=7-6, to=7-7]
	\arrow["{\eta^S_{TSX}}"{description}, from=5-5, to=7-7]
	\arrow["{\delta_X}", from=5-5, to=5-6]
	\arrow["{\eta^S_{STX}}"', from=5-6, to=7-8]
	\arrow[Rightarrow, no head, from=5-6, to=7-9]
	\arrow["{S(X \otimes \eta^S_X)}", from=1-9, to=3-9]
	\arrow[from=3-9, to=7-9]
	\arrow["{in_1}", from=1-1, to=2-3]
	\arrow[from=2-3, to=3-3]
	\arrow[from=2-3, to=1-6]
	\arrow[from=2-3, to=2-6]
	\arrow[from=2-6, to=3-6]
	\arrow[from=2-6, to=1-9]
	\arrow[from=2-6, to=3-9]
	\arrow["{in_1 \otimes in_1}"', from=3-1, to=7-1]
	\arrow["{\text{?}}"{description}, draw=none, from=5-4, to=3-3]
      \end{tikzcd}
    }
\end{center}
whose question marked subdiagram commutes as follows.
      \begin{center}
        \ajustedroit{
\begin{tikzcd}[ampersand replacement=\&]
	{I⊗I} \&\&\& {Σ^+X⊗Σ^+X} \&\&\&\&\&\& {SX⊗SX} \\
	\&\&\& {I \otimes SX} \&\&\&\&\&\& {SX⊗SSX} \\
	I \&\&\& {I \otimes SSX} \&\& TSSX \&\&\&\& TSX \\
	{Σ^+X} \&\& SX \& SSX \\
	\&\&\& {SX \otimes SSX + SX + SX} \& {S(SX \otimes SSX + SX + SX)} \\
	\&\&\& TSX \&\& STSX \\
	\&\&\& STX \&\&\&\& SSTX \\
	\\
	SX \&\&\&\&\&\&\&\&\& STX
	\arrow["{in₁ \otimes in₁}", from=1-1, to=1-4]
	\arrow[""{name=0, anchor=center, inner sep=0}, "{\eta_{\Sigma^+,X}⊗\eta_{\Sigma^+,X}}", from=1-4, to=1-10]
	\arrow["{SX⊗η^SSX}", from=1-10, to=2-10]
	\arrow["{\eta_{\Gamma_S,SX}}", from=2-10, to=3-10]
	\arrow["{δ_X}", from=3-10, to=9-10]
	\arrow["{λI}"', from=1-1, to=3-1]
	\arrow["{in_1}"', from=3-1, to=4-1]
	\arrow["{\eta_{\Sigma^+,X}}"', from=4-1, to=9-1]
	\arrow["{Sη^T_X}", from=9-1, to=9-10]
	\arrow[from=1-1, to=2-4]
	\arrow[from=2-4, to=3-4]
	\arrow["\lambda"', from=3-4, to=4-4]
	\arrow[""{name=1, anchor=center, inner sep=0}, "{Sin_3}"', from=4-4, to=5-5]
	\arrow[""{name=2, anchor=center, inner sep=0}, from=2-4, to=3-6]
	\arrow["{\delta_{SX}}"', from=3-6, to=6-6]
	\arrow[from=5-5, to=6-6]
	\arrow[from=3-1, to=4-3]
	\arrow["{\eta^S_{SX}}", from=4-3, to=4-4]
	\arrow["{T\mu^S_X}"{description}, from=3-6, to=3-10]
	\arrow["{in_3}"', from=4-3, to=5-4]
	\arrow[from=5-4, to=6-4]
	\arrow["{\eta^S}"', from=6-4, to=6-6]
	\arrow["{S\delta_X}"', from=6-6, to=7-8]
	\arrow["{\mu^S_{TX}}", from=7-8, to=9-10]
	\arrow["{\delta_X}"', from=6-4, to=7-4]
	\arrow["{\eta^S}"', from=7-4, to=7-8]
	\arrow[Rightarrow, no head, from=7-4, to=9-10]
	\arrow[from=4-1, to=4-3]
	\arrow[Rightarrow, no head, from=9-1, to=4-3]
	\arrow["{\eta^T_{SX}}"', curve={height=12pt}, from=4-3, to=6-4]
	\arrow["{??}"{description}, Rightarrow, draw=none, from=3-6, to=0]
	\arrow["{\text{\eqref{eq:simple:d:delta}}}"{description}, Rightarrow, draw=none, from=1, to=2]
\end{tikzcd}}
      \end{center}
      The double question marked subdiagram of the latter diagram in turn commutes
      as follows.
      \begin{center}
        \ajustedroit
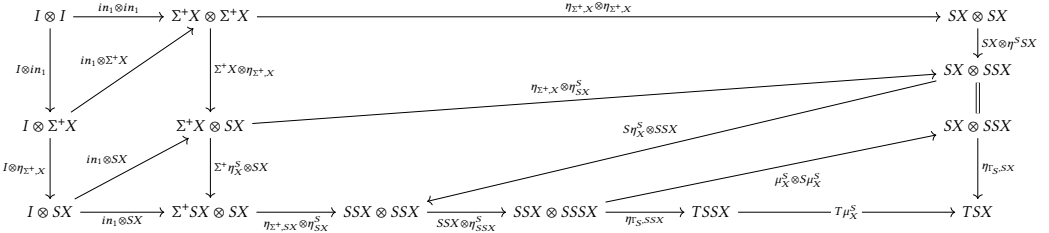
{
          \begin{tikzcd}[ampersand replacement=\&]
            {I⊗I} \&\& {Σ^+X⊗Σ^+X} \&\&\&\&\&\&\&\&\&\&\& {SX⊗SX} \\
            \&\&\&\&\&\&\&\&\&\&\&\&\& {SX⊗SSX} \\
            {I \otimes \Sigma^+X} \&\& {Σ^+X⊗SX} \&\&\&\&\&\&\&\&\&\&\& {SX⊗SSX} \\
            \\
            {I \otimes SX} \&\& {\Sigma^+ SX \otimes SX} \&\& {SSX \otimes SSX} \&\& {SSX \otimes SSSX} \&\& TSSX \&\&\&\&\& TSX
            \arrow["{in₁ \otimes in₁}", from=1-1, to=1-3]
            \arrow["{\eta_{\Sigma^+,X}⊗\eta_{\Sigma^+,X}}", from=1-3, to=1-14]
            \arrow["{SX⊗η^SSX}", from=1-14, to=2-14]
            \arrow["{T\mu^S_X}"{description}, from=5-9, to=5-14]
            \arrow["{in_1 \otimes SX}"', from=5-1, to=5-3]
            \arrow["{\eta_{\Sigma^+,SX} \otimes \eta^S_{SX}}"', from=5-3, to=5-5]
            \arrow["{SSX \otimes \eta^S_{SSX}}"', from=5-5, to=5-7]
            \arrow["{\eta_{\Gamma_S,SSX}}"', from=5-7, to=5-9]
            \arrow["{I \otimes in_1}"', from=1-1, to=3-1]
            \arrow["{I \otimes \eta_{\Sigma^+,X}}"', from=3-1, to=5-1]
            \arrow["{in_1 \otimes \Sigma^+X}", from=3-1, to=1-3]
            \arrow["{in_1 \otimes SX}", from=5-1, to=3-3]
            \arrow["{Σ^+X⊗\eta_{Σ^+,X}}", from=1-3, to=3-3]
            \arrow["{\Sigma^+\eta^S_X \otimes SX}", from=3-3, to=5-3]
            \arrow["{\eta_{\Sigma^+,X} \otimes \eta^S_{SX}}", from=3-3, to=2-14]
            \arrow["{S\eta^S_X \otimes SSX}"', from=2-14, to=5-5]
            \arrow["{\eta_{\Gamma_S,SX}}", from=3-14, to=5-14]
            \arrow["{\mu^S_X \otimes S\mu^S_X}"', from=5-7, to=3-14]
            \arrow[Rightarrow, no head, from=2-14, to=3-14]
          \end{tikzcd}}
      \end{center}
      It remains to prove that both large subdiagrams at the bottom of
      Figure~\ref{fig:coherence:R1} commute. For the left-hand one,
      the first term is almost trivial, and the second term is easy:
      \begin{center}
        \ajustedroit
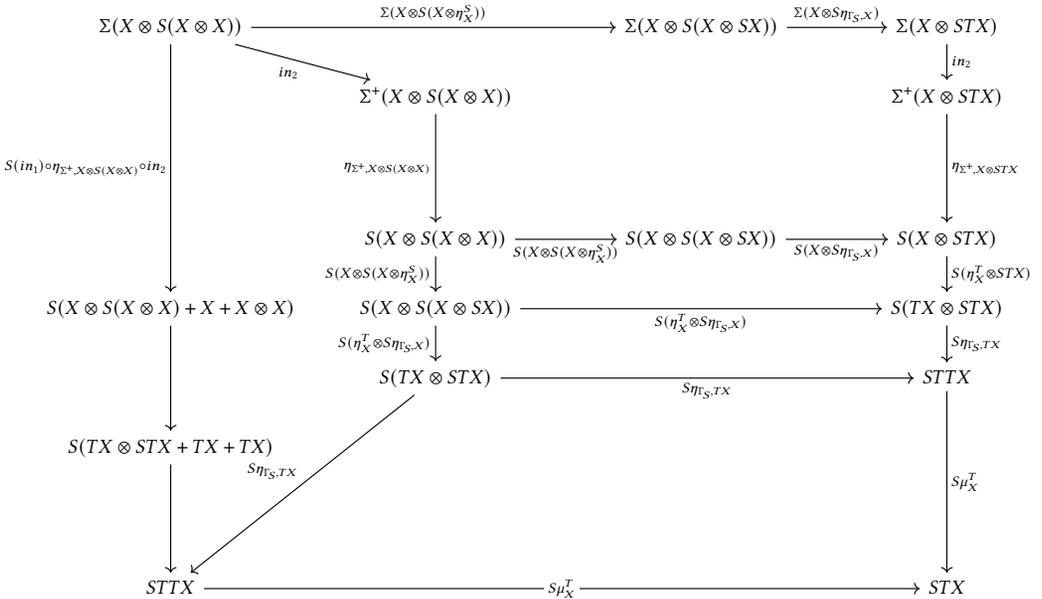
{
\begin{tikzcd}[ampersand replacement=\&]
	{\Sigma(X \otimes  S(X\otimes X))} \&\&\& {\Sigma(X \otimes S(X \otimes SX))} \&\& {\Sigma(X \otimes STX)} \\
	\& {\Sigma^+(X \otimes  S(X\otimes X))} \&\&\&\& {\Sigma^+(X \otimes STX)} \\
	\\
	\\
	\& {S(X \otimes  S(X\otimes X))} \&\& {S(X \otimes S(X \otimes SX))} \&\& {S(X \otimes STX)} \\
	{S(X⊗S(X\otimes X) + X + X\otimes X)} \& {S(X \otimes S(X \otimes SX))} \&\&\&\& {S(TX \otimes STX)} \\
	\& {S(TX \otimes STX)} \&\&\&\& STTX \\
	{S(TX \otimes STX + TX + TX)} \\
	\\
	\\
	STTX \&\&\&\&\& STX
	\arrow["{S(in₁)∘η_{Σ⁺,X⊗S(X\otimes X)}∘in₂}"', from=1-1, to=6-1]
	\arrow[from=6-1, to=8-1]
	\arrow[from=8-1, to=11-1]
	\arrow["{S\mu^T_X}"{description}, from=11-1, to=11-6]
	\arrow["{in_2}"', from=1-1, to=2-2]
	\arrow["{\eta_{\Sigma^+,X \otimes S(X \otimes X)}}"', from=2-2, to=5-2]
	\arrow["{S(X \otimes S(X \otimes \eta^S_X))}"', from=5-2, to=6-2]
	\arrow["{S(\eta^T_X \otimes S\eta_{\Gamma_S,X})}"', from=6-2, to=7-2]
	\arrow["{S\eta_{\Gamma_S,TX}}"', from=7-2, to=11-1]
	\arrow["{in_2}", from=1-6, to=2-6]
	\arrow["{\eta_{\Sigma^+,X \otimes STX}}", from=2-6, to=5-6]
	\arrow["{S(\eta^T_X\otimes STX)}", from=5-6, to=6-6]
	\arrow["{S\eta_{\Gamma_S,TX}}", from=6-6, to=7-6]
	\arrow["{S\mu^T_X}", from=7-6, to=11-6]
	\arrow["{S\eta_{\Gamma_S,TX}}"', from=7-2, to=7-6]
	\arrow["{S(\eta^T_X \otimes S\eta_{\Gamma_S,X})}"', from=6-2, to=6-6]
	\arrow["{S(X \otimes S(X \otimes \eta^S_X))}"', from=5-2, to=5-4]
	\arrow["{\Sigma(X \otimes S(X \otimes \eta^S_X))}", from=1-1, to=1-4]
	\arrow["{\Sigma(X \otimes S\eta_{\Gamma_S,X})}", from=1-4, to=1-6]
	\arrow["{S(X \otimes S\eta_{\Gamma_S,X})}"', from=5-4, to=5-6]
      \end{tikzcd}
    }
      \end{center}
      In order to prove the right-hand one, we need the following
      lemma.
      \begin{lemma}\label{lem:mu:eta:delta}
        For any object $X$, the following diagram commutes.
        
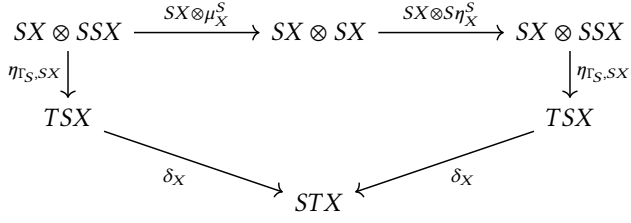
\begin{center}
\begin{tikzcd}[ampersand replacement=\&]
	{SX \otimes SSX} \&\& {SX \otimes SX} \&\& {SX \otimes SSX} \\
	TSX \&\&\&\& TSX \\
	\&\& STX
	\arrow["{SX \otimes \mu^S_X}", from=1-1, to=1-3]
	\arrow["{SX \otimes S\eta^S_X}", from=1-3, to=1-5]
	\arrow["{\eta_{\Gamma_S,SX}}"', from=1-1, to=2-1]
	\arrow["{\delta_X}"', from=2-1, to=3-3]
	\arrow["{\eta_{\Gamma_S,SX}}", from=1-5, to=2-5]
	\arrow["{\delta_X}", from=2-5, to=3-3]
\end{tikzcd}
        \end{center}
      \end{lemma}
      \begin{proof}
        We resort to the definition of $δ$ from $d$, and proceed by
        diagram chasing, as follows.
        \begin{center}\hfill
          \ajustedroit
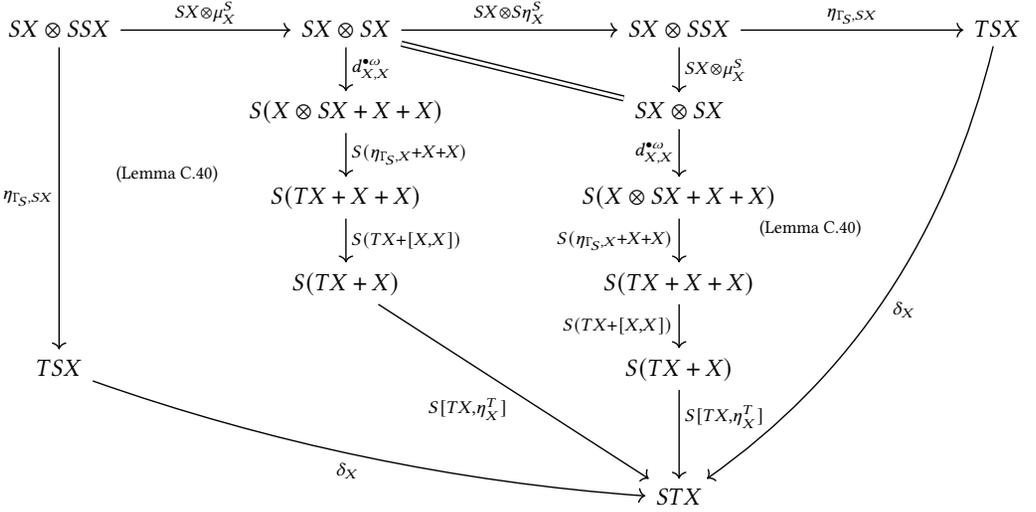
{
\begin{tikzcd}[ampersand replacement=\&,baseline=(\tikzcdmatrixname-7-5.base)]
	{SX \otimes SSX} \&\& {SX \otimes SX} \&\& {SX \otimes SSX} \&\&\& TSX \\
	\&\& {S(X \otimes SX +  X+ X)} \&\& {SX \otimes SX} \\
	\&\& {S(TX + X + X)} \&\& {S(X \otimes SX + X+ X)} \\
	\&\& {S(TX + X)} \&\& {S(TX + X+X)} \\
	TSX \&\&\&\& {S(TX+X)} \\
	\\
	\&\&\&\& STX
	\arrow[""{name=0, anchor=center, inner sep=0}, "{SX \otimes \mu^S_X}", from=1-1, to=1-3]
	\arrow["{SX \otimes S\eta^S_X}", from=1-3, to=1-5]
	\arrow["{\eta_{\Gamma_S,SX}}"', from=1-1, to=5-1]
	\arrow["{\delta_X}"', curve={height=12pt}, from=5-1, to=7-5]
	\arrow["{\eta_{\Gamma_S,SX}}", from=1-5, to=1-8]
	\arrow[""{name=1, anchor=center, inner sep=0}, "{\delta_X}", curve={height=-30pt}, from=1-8, to=7-5]
	\arrow["{d^{\bullet\omega}_{X,X}}", from=1-3, to=2-3]
	\arrow["{S[TX,\eta^T_X]}"', from=4-3, to=7-5]
	\arrow["{SX \otimes \mu^S_X}", from=1-5, to=2-5]
	\arrow["{d^{\bullet\omega}_{X,X}}"', from=2-5, to=3-5]
	\arrow["{S[TX,\eta^T_X]}"{pos=0.3}, from=5-5, to=7-5]
	\arrow["{S(\eta_{\Gamma_S,X} + X+X)}"', from=3-5, to=4-5]
	\arrow["{S(TX + [X,X])}"', from=4-5, to=5-5]
	\arrow["{S(\eta_{\Gamma_S,X} + X+X)}", from=2-3, to=3-3]
	\arrow["{S(TX + [X,X])}", from=3-3, to=4-3]
	\arrow[Rightarrow, no head, from=1-3, to=2-5]
	\arrow["{\text{(Lemma~\ref{lem:d:delta:Y})}}"{description, pos=0.4}, curve={height=-12pt}, Rightarrow, draw=none, from=0, to=5-1]
	\arrow["{\text{(Lemma~\ref{lem:d:delta:Y})}}"{description, pos=0.6}, draw=none, from=2-5, to=1]
      \end{tikzcd}
    } \qedhere
        \end{center}
      \end{proof}

      We may now prove that the bottom right subdiagram of
      Figure~\ref{fig:coherence:R1} commutes.
      \begin{center}
        \ajustedroit{
\begin{tikzcd}[ampersand replacement=\&]
	{\Sigma(X \otimes SX \otimes SX)} \& {\Sigma^+(X \otimes SX \otimes SX)} \&\&\& {ST(X \otimes STX \otimes STX)} \&\&\&\& {ST(X \otimes STX \otimes STX + X + X)} \\
	{\Sigma(X \otimes (SX \otimes SX))} \& {\Sigma^+(X \otimes (SX \otimes SX))} \&\&\&\&\&\&\& {ST(STX \otimes STX \otimes STX + X)} \\
	{\Sigma(X \otimes S(SX \otimes SX))} \& {\Sigma^+(X \otimes S(SX \otimes SX))} \&\& {ST(X \otimes (STX \otimes STX))} \\
	\&\& {ST(X \otimes S(STX \otimes STX))} \&\& {ST(STX \otimes (STX \otimes SSTX) )} \&\&\&\& {ST(STX \otimes (STX \otimes STX) + X)} \\
	{\Sigma(X \otimes STSX)} \& {\Sigma^+(X \otimes STSX)} \& {ST(X \otimes S(STX \otimes SSTX))} \&\& {ST(STX \otimes TSTX)} \&\&\&\& {ST(STX \otimes (STX \otimes SSTX) + X)} \\
	\&\& {ST(X \otimes STSTX)} \&\& {ST(TSTX \otimes TSTX)} \&\&\&\& {ST(STX \otimes TSTX + X)} \\
	{\Sigma(X \otimes SSTX)} \& {\Sigma^+(X \otimes SSTX)} \&\& {ST(STX \otimes STSTX)} \&\&\&\&\& {ST(TSTX \otimes TSTX + X)} \\
	{\Sigma(X \otimes STX)} \& {ST(X \otimes SSTTX)} \&\&\& {ST(TSTX \otimes STSTX)} \&\&\&\& {ST(TSTX \otimes STSTX + X)} \\
	{\Sigma^+(X \otimes STX)} \&\& {ST(STX \otimes SSTTX)} \&\& {ST(TTSTX)} \&\&\&\& {ST(TTSTX + X)} \\
	\&\&\& {ST(STTX \otimes SSTTX)} \&\& STSTSTX \&\&\& {ST(STSTX + X)} \\
	\&\&\&\& STTSTTX \\
	{S(X \otimes STX)} \& {ST(X \otimes STTX)} \\
	{S(TX \otimes STX)} \& {ST(TX \otimes STTX)} \& {ST(STX \otimes STTX)} \& {ST(STTX \otimes STTX)} \\
	STTX \&\&\& {ST(STTX \otimes SSTTX)} \& STTSTTX \& STSTTTX \& STSTTX \& STSSTTX \\
	\&\&\& {ST(STTTX \otimes SSTTTX)} \&\&\&\&\& STSTX \\
	\&\&\& STTSTTTX \\
	\&\&\& STX
	\arrow["{STin_1}", from=1-5, to=1-9]
	\arrow["{ST(\eta^S \mu^T STX + X)}", from=9-9, to=10-9]
	\arrow["{ST[\mu^{ST}_X,\eta^{ST}_X]}", from=10-9, to=15-9]
	\arrow["{\mu^{ST}_X}", curve={height=-24pt}, from=15-9, to=17-4]
	\arrow["{\eta_{\Sigma^+}\eta^T (X\otimes S\eta^T_X \otimes S\eta^T_X)}", from=1-2, to=1-5]
	\arrow["{in_2}", from=1-1, to=1-2]
	\arrow["{ST(\eta^{ST}_X \otimes STX \otimes STX + [X,X])}", from=1-9, to=2-9]
	\arrow["{ST(\alpha + X)}", from=2-9, to=4-9]
	\arrow["{\Sigma \alpha}", from=1-1, to=2-1]
	\arrow["{\Sigma(X \otimes \eta^S_{SX \otimes SX})}"'{pos=0.4}, from=2-1, to=3-1]
	\arrow["{in_2}", from=2-1, to=2-2]
	\arrow[from=1-2, to=2-2]
	\arrow[from=3-1, to=5-1]
	\arrow[from=5-1, to=7-1]
	\arrow[from=7-1, to=8-1]
	\arrow["{in_2}"', from=8-1, to=9-1]
	\arrow["{\eta_{\Sigma^+,X \otimes STX}}"', from=9-1, to=12-1]
	\arrow["{S(\eta^T_X\otimes STX)}"', from=12-1, to=13-1]
	\arrow["{S\eta_{\Gamma_S,TX}}"', from=13-1, to=14-1]
	\arrow["{S\mu^T_X}", curve={height=30pt}, from=14-1, to=17-4]
	\arrow["{\eta_{\Sigma^+}\eta^T (X\otimes (S\eta^T_X \otimes S\eta^T_X))}"{pos=0.6}, from=2-2, to=3-4]
	\arrow["{ST(\alpha_{X,STX,STX})}", from=1-5, to=3-4]
	\arrow["{ST(STX \otimes (STX \otimes \eta^S_{STX}) + X)}", from=4-9, to=5-9]
	\arrow["{ST(STX \otimes \eta_{\Gamma_S,STX} + X)}", from=5-9, to=6-9]
	\arrow["{ST(\eta^T_{STX} \otimes TSTX + X)}", from=6-9, to=7-9]
	\arrow["{ST(TSTX \otimes \eta^S_{TSTX} + X)}", from=7-9, to=8-9]
	\arrow["{ST(\eta_{\Gamma_S,TSTX} + X)}", from=8-9, to=9-9]
	\arrow["{ST\mu^{ST}_X}"{description}, curve={height=-12pt}, from=10-6, to=15-9]
	\arrow["{ST(\eta^S \mu^T STX)}"{pos=1}, from=9-5, to=10-6]
	\arrow["{ST(\eta_{\Gamma_S,TSTX})}", from=8-5, to=9-5]
	\arrow["{ST(TSTX \otimes \eta^S_{TSTX})}", from=6-5, to=8-5]
	\arrow["{ST(\eta^T_{STX} \otimes TSTX)}", from=5-5, to=6-5]
	\arrow["{ST(STX \otimes \eta_{\Gamma_S,STX})}", from=4-5, to=5-5]
	\arrow["{ST(\eta^{ST}_X \otimes (STX \otimes \eta^S_{STX}))}"'{pos=0}, from=3-4, to=4-5]
	\arrow["{ST(X \otimes \eta^S_{STX \otimes STX})}"'{pos=0.9}, from=3-4, to=4-3]
	\arrow["{ST(X \otimes S(STX \otimes \eta^S_{STX}))}", from=4-3, to=5-3]
	\arrow["{ST(X \otimes S(\eta_{\Gamma_S,STX}))}", from=5-3, to=6-3]
	\arrow["{ST(X \otimes S\delta_{TX})}"', from=6-3, to=8-2]
	\arrow["{ST(X \otimes \mu^S_{TTX})}"', from=8-2, to=12-2]
	\arrow["{S\eta^T (X \otimes ST \eta^T_X)}"', from=12-1, to=12-2]
	\arrow["{ST(\eta^{ST}_X \otimes STSTX)}"{pos=0.7}, from=6-3, to=7-4]
	\arrow["{ST(STX \otimes \eta^S_{TSTX})}"'{pos=0.8}, from=5-5, to=7-4]
	\arrow["{ST(\eta^T_{STX}\otimes STSTX)}"{pos=0.8}, from=7-4, to=8-5]
	\arrow["{STS\delta_{TX}}"{description}, from=10-6, to=14-8]
	\arrow["{ST \mu^S \mu^T_X}"', from=14-8, to=15-9]
	\arrow["{ST(\delta_{TX} \otimes S\delta_{TX})}"{description}, from=8-5, to=10-4]
	\arrow["{ST(\eta^{ST}_X \otimes SSTTX)}"{pos=0.9}, from=8-2, to=9-3]
	\arrow["{ST(S\eta^T_{TX} \otimes SSTTX)}"'{pos=0.1}, from=9-3, to=10-4]
	\arrow["{ST(S\eta^T_{TX} \otimes S\delta_{TX})}"', from=7-4, to=10-4]
	\arrow["{ST(STTX \otimes \mu^S_{TTX})}"', from=10-4, to=13-4]
	\arrow["{ST(\eta^{ST}_X \otimes STTX)}", from=12-2, to=13-3]
	\arrow["{ST(S\eta^T_{TX} \otimes STTX)}", from=13-3, to=13-4]
	\arrow["{ST(STTX \otimes S\eta^S_{TTX})}", from=13-4, to=14-4]
	\arrow["{S\eta^T(TX \otimes ST\eta^T_X)}", from=13-1, to=13-2]
	\arrow["{ST(\eta^T_X \otimes STTX)}"', from=12-2, to=13-2]
	\arrow["{ST(\eta^S_{TX} \otimes STTX)}"', from=13-2, to=13-3]
	\arrow["{ST(\eta_{\Gamma_S,STTX})}"', from=10-4, to=11-5]
	\arrow["{STT\delta_{TX}}"', from=9-5, to=11-5]
	\arrow["{ST(STT\eta^T_X \otimes SS\eta^T_{TTX})}"', from=14-4, to=15-4]
	\arrow["{ST\eta_{\Gamma_S,STTTX}}"', from=15-4, to=16-4]
	\arrow["{S\eta^T T \eta^S\eta^T T \eta^T_X }", curve={height=18pt}, from=14-1, to=16-4]
	\arrow[from=16-4, to=17-4]
	\arrow["{ST\delta_{TTX}}"', from=11-5, to=14-6]
	\arrow["{STS\mu^T_{TX}}", from=14-6, to=14-7]
	\arrow["{ST\eta^S_{STTX}}", from=14-7, to=14-8]
	\arrow["{ST in_1}", from=4-5, to=5-9]
	\arrow["{ST in_1}", from=5-5, to=6-9]
	\arrow["{ST in_1}", from=6-5, to=7-9]
	\arrow["{ST in_1}", from=8-5, to=8-9]
	\arrow["{ST in_1}", from=9-5, to=9-9]
	\arrow["{ST in_1}", from=10-6, to=10-9]
	\arrow["{in_2}", from=3-1, to=3-2]
	\arrow["{\eta_{\Sigma^+}\eta^T (X\otimes S(S\eta^T_X \otimes S\eta^T_X))}"'{pos=0}, from=3-2, to=4-3]
	\arrow["{in_2}"', from=5-1, to=5-2]
	\arrow["{\eta_{\Sigma^+}\eta^T (X\otimes STS\eta^T_X)}"'{pos=0.1}, from=5-2, to=6-3]
	\arrow["{in_2}"', from=7-1, to=7-2]
	\arrow["{\eta_{\Sigma^+}\eta^T (X\otimes SST\eta^T_X)}"', from=7-2, to=8-2]
	\arrow["{ST(\eta_{\Gamma_S,STTX})}"', from=14-4, to=14-5]
	\arrow["{ST\delta_{TTX}}"', from=14-5, to=14-6]
	\arrow["{\text{(Lemma~\ref{lem:mu:eta:delta})}}"{description, pos=0.6}, draw=none, from=14-4, to=11-5]
\end{tikzcd}
    }      
      \end{center}
There the bottom right subdiagram commutes as follows.
\begin{center}
\ajustedroit{
\begin{tikzcd}[ampersand replacement=\&]
	{ST(STTX \otimes SSTTX)} \&\& STTSTTX \&\& STSTTTX \\
	{ST(STTTX \otimes SSTTTX)} \&\&\&\& STSTTX \\
	STTSTTTX \&\&\&\& STSSTTX \\
	STTSTTX \& STTSTTX \& STSTTTX \&\& STSTTX \\
	\&\&\&\& STSTX \\
	STTSTX \&\& STSTTX \& STSTX \& STX
	\arrow["{\mu^{ST}_X}", from=5-5, to=6-5]
	\arrow["{ST(STT\eta^T_X \otimes SS\eta^T_{TTX})}"', from=1-1, to=2-1]
	\arrow["{ST\eta_{\Gamma_S,STTTX}}"', from=2-1, to=3-1]
	\arrow["{STS\mu^T_{TX}}", from=1-5, to=2-5]
	\arrow["{ST\eta^S_{STTX}}", from=2-5, to=3-5]
	\arrow["{ST(\eta_{\Gamma_S,STTX})}", from=1-1, to=1-3]
	\arrow["{ST\delta_{TTX}}", from=1-3, to=1-5]
	\arrow["{STTS\mu^T_{TX}}"', from=3-1, to=4-1]
	\arrow["{STTS\mu^T_X}"', from=4-1, to=6-1]
	\arrow["{ST\delta_{TX}}"', from=6-1, to=6-3]
	\arrow["{STS\mu^T_X}"', from=6-3, to=6-4]
	\arrow["{\mu^{ST}_X}"', from=6-4, to=6-5]
	\arrow["{STTSTT\eta^T_X}"{pos=0.6}, from=1-3, to=3-1]
	\arrow["{STTST\mu^T_X}", from=3-1, to=4-2]
	\arrow["{STTS\mu^T_X}", from=4-2, to=6-1]
	\arrow["{ST\delta_{TTX}}", from=4-2, to=4-3]
	\arrow["{STST\mu^T_X}"', from=4-3, to=6-3]
	\arrow["{ST \mu^S}", from=3-5, to=4-5]
	\arrow["{STS\mu^T_X}", from=4-5, to=5-5]
	\arrow["{STS\mu^T_{TX}}", from=4-3, to=4-5]
	\arrow[Rightarrow, no head, from=1-3, to=4-2]
	\arrow[curve={height=30pt}, Rightarrow, no head, from=2-5, to=4-5]
\end{tikzcd}}
\end{center}

\subsection{Coherence of $L₂$}

We then check coherence of $L₂$ in Figure~\ref{fig:coh:L2}, where, in
the middle triangle which unfolds $d_{SX,SX}$, we use
$η^S_{SX} ∘ η_{Σ⁺,X} ∘ in₁$ instead of $v_{SX}$ as a point. This is
justified because these morphisms are in fact equal, as the diagram
below shows.
\begin{center}
\begin{tikzcd}[ampersand replacement=\&]
	I \&\&\&\& {\Sigma^+X} \\
	\& {\Sigma^+\emptyset} \\
	{\Sigma^+SX} \&\& S\emptyset \&\& SX \\
	\&\& SSX
	\arrow["{in_1}"', from=1-1, to=3-1]
	\arrow["{\eta_{\Sigma^+,SX}}"', from=3-1, to=4-3]
	\arrow["{in_1}", from=1-1, to=1-5]
	\arrow["{\eta_{\Sigma^+,X}}", from=1-5, to=3-5]
	\arrow["{S\eta^S_X}", from=3-5, to=4-3]
	\arrow["{in_1}", from=1-1, to=2-2]
	\arrow["{\Sigma^+{!}}", from=2-2, to=1-5]
	\arrow["{\Sigma^+{!}}", from=2-2, to=3-1]
	\arrow["{\eta_{\Sigma^+,\emptyset}}", from=2-2, to=3-3]
	\arrow["{S{!}}", from=3-3, to=4-3]
	\arrow["{S{!}}", from=3-3, to=3-5]
\end{tikzcd}
\end{center}

\begin{figure}[htp]
  \ajustetourne{
\begin{tikzcd}[ampersand replacement=\&]
	{\Sigma^+X \otimes I} \& {I \otimes I + \Sigma X \otimes I} \&\& {I + \Sigma(X\otimes I)} \& {\Sigma^+(X\otimes I)} \& {S(X\otimes I)} \& {ST(X\otimes I)} \& {ST(X \otimes I + X + X)} \\
	{\red{S}X \otimes I} \& {\Sigma^+SX \otimes \Sigma^+X} \&\& {\Sigma^+X + \Sigma(SX \otimes \Sigma^+X)} \&\&\&\& {ST(STX \otimes I + X)} \\
	{S\red{S}X \otimes \red{\Sigma^+}SX} \&\&\&\& {SX + \Sigma(SX \otimes SX)} \&\&\& {ST(SSTX \otimes \Sigma^+STX + X)} \\
	{S\red{S}X \otimes \red{S}SX} \&\& {\Sigma^+SX \otimes SX} \& {\Sigma^+SX \otimes SSX} \& {I \otimes SSX + \Sigma SX \otimes SSX} \& {SSX + \Sigma(SX \otimes SSX)} \&\& {ST(SSTX \otimes SSTX + X)} \\
	{S\red{S}X \otimes S\red{S}SX} \&\&\&\&\&\&\& {ST(SSTX \otimes SSSTX + X)} \\
	{TS\red{S}X} \& TSSSX \& TSSX \&\&\& {S(SX \otimes SSX + SX + SX)} \&\& {ST(TSSTX + X)} \\
	{ST\red{S}X} \& STSSX \&\&\&\&\&\& {ST(STSTX + X)} \\
	{S\red{S}TX} \&\& SSTSX \& STSX \&\&\&\& STSTX \\
	\&\&\&\& SSTX \\
	\\
	\&\&\&\& STX
	\arrow[Rightarrow, no head, from=1-1, to=1-2]
	\arrow["{\lambda_I + st_{X,(I,id_I)}}", from=1-2, to=1-4]
	\arrow[Rightarrow, no head, from=1-4, to=1-5]
	\arrow["{\eta_{\Sigma^+}}", from=1-5, to=1-6]
	\arrow["{S\eta^T_{X\otimes I}}", from=1-6, to=1-7]
	\arrow["{ST in_1}", from=1-7, to=1-8]
	\arrow["{\eta_{\Sigma^+,X} \otimes I}"', from=1-1, to=2-1]
	\arrow["{\eta^S_{SX} \otimes in_1}"', from=2-1, to=3-1]
	\arrow["{SSX \otimes \eta_{\Sigma^+,SX}}"', from=3-1, to=4-1]
	\arrow["{SSX \otimes \eta^S_{SSX}}"', from=4-1, to=5-1]
	\arrow["{\eta_{\Gamma_S,SSX}}"', from=5-1, to=6-1]
	\arrow["{\delta_{SX}}"', from=6-1, to=7-1]
	\arrow["{S\delta_X}"', from=7-1, to=8-1]
	\arrow["{\mu^S_{TX}}"', from=8-1, to=11-5]
	\arrow["{ST(\eta^{ST}_X \otimes I + [X,X])}", from=1-8, to=2-8]
	\arrow["{ST(\eta^S_{STX} \otimes in_1 + X)}", from=2-8, to=3-8]
	\arrow["{ST(SSTX \otimes \eta_{\Sigma^+,STX} + X)}", from=3-8, to=4-8]
	\arrow["{ST(SSTX \otimes \eta^S_{SSTX} + X)}", from=4-8, to=5-8]
	\arrow["{ST(\eta_{\Gamma_S,SSTX} + X)}", from=5-8, to=6-8]
	\arrow["{ST(\delta_{STX} + X)}", from=6-8, to=7-8]
	\arrow["{ST[\mu^{ST}_X,\eta^{ST}_X]}", from=7-8, to=8-8]
	\arrow["{\mu^{ST}_X}", from=8-8, to=11-5]
	\arrow["{\Sigma^+\eta^S_X \otimes in_1}"{pos=1}, from=1-1, to=2-2]
	\arrow["{\Sigma^+SX \otimes \eta_{\Sigma^+,X}}"{pos=0.9}, from=2-2, to=4-3]
	\arrow["{\Sigma^+SX \otimes \eta^S_{SX}}", from=4-3, to=4-4]
	\arrow[Rightarrow, no head, from=4-4, to=4-5]
	\arrow["{\lambda_{SSX} + st_{SX,(SSX,\eta^S_{SX} \circ \eta_{\Sigma^+,X}\circ in_1)}}"', from=4-5, to=4-6]
	\arrow["{in_1 + \Sigma(\eta^S_X \otimes in_1)}"', from=1-4, to=2-4]
	\arrow[from=2-4, to=3-5]
	\arrow["{\eta^S_{SX} + \Sigma(SX \otimes \eta^S_{SX})}"{description}, from=3-5, to=4-6]
	\arrow["{[S in_3, S in_1 \circ \eta_{\Sigma^+,SX \otimes SSX} \circ in_2]}"{description}, from=4-6, to=6-6]
	\arrow[""{name=0, anchor=center, inner sep=0}, "{d_{SX,SX}}"{description}, from=4-3, to=6-6]
	\arrow["{STS\eta^S_X}"{description}, from=7-1, to=7-2]
	\arrow["{S\delta_{SX}}", from=7-2, to=8-3]
	\arrow["{SST\eta^S_X}", from=8-1, to=8-3]
	\arrow["{TSS\eta^S_X}", from=6-1, to=6-2]
	\arrow["{\delta_{SSX}}", from=6-2, to=7-2]
	\arrow["{T\mu^S_{SX}}", from=6-2, to=6-3]
	\arrow["{\delta_{SX}}", from=6-3, to=8-4]
	\arrow["{\mu^S_{TSX}}", from=8-3, to=8-4]
	\arrow[from=4-3, to=6-3]
	\arrow[from=6-6, to=8-4]
	\arrow["{S\delta_X}", from=8-4, to=9-5]
	\arrow["{\mu^S_{TX}}", from=9-5, to=11-5]
	\arrow["{\text{(?)}}", draw=none, from=6-2, to=2-2]
	\arrow["{\text{(??)}}", draw=none, from=6-6, to=4-8]
	\arrow["{\text{\eqref{eq:simple:d:delta}}}"{description}, Rightarrow, draw=none, from=8-4, to=0]
\end{tikzcd}
    }
  \caption{Coherence of $L₂$}
  \label{fig:coh:L2}
\end{figure}

The left-hand, question-marked subdiagram of Figure~\ref{fig:coh:L2} only commutes when postcomposed with
$δ_X ∘ Tμ^S_X = μ^S_{TX} ∘ Sδ_X ∘ δ_{SX}$, as follows.
\begin{center}
\ajustedroit{
\begin{tikzcd}[ampersand replacement=\&]
	{\Sigma^+X \otimes I} \&\&\&\& {\Sigma^+X \otimes \Sigma^+X} \&\&\&\&\& {\Sigma^+SX \otimes \Sigma^+X} \\
	{SX \otimes I} \&\& {\Sigma^+X \otimes \Sigma^+SX} \&\&\&\&\&\&\& {\Sigma^+SX \otimes SX} \\
	{SSX \otimes \Sigma^+SX} \&\& {SX \otimes \Sigma^+SX} \&\&\& {SX \otimes SX} \&\&\&\& {SSX \otimes SSX} \\
	{SSX \otimes SSX} \&\&\& {SX \otimes SSX} \\
	{SSX \otimes SSSX} \&\&\& {SX \otimes SX} \&\&\&\& {SSX \otimes SSX} \&\& {SSX \otimes SSSX} \\
	\\
	\&\&\&\&\&\&\&\&\& TSSX \\
	TSSX \&\&\&\&\& {SX \otimes SSX} \\
	\&\&\&\&\& TSX \&\&\&\& TSX \\
	TSSX \&\& TSSSX \&\&\& TSSX \& TSX \&\&\& STX
	\arrow["{\eta_{\Sigma^+,X} \otimes I}"', from=1-1, to=2-1]
	\arrow["{\eta^S_{SX} \otimes in_1}"', from=2-1, to=3-1]
	\arrow["{SSX \otimes \eta_{\Sigma^+,SX}}"', from=3-1, to=4-1]
	\arrow["{SSX \otimes \eta^S_{SSX}}"', from=4-1, to=5-1]
	\arrow["{\eta_{\Gamma_S,SSX}}"', from=5-1, to=8-1]
	\arrow["{\Sigma^+SX \otimes \eta_{\Sigma^+,X}}"{pos=0.6}, from=1-10, to=2-10]
	\arrow["{\eta_{\Sigma^+,SX} \otimes \eta^S_{SX}}", from=2-10, to=3-10]
	\arrow["{T\mu^S_{SX}}"', from=10-3, to=10-6]
	\arrow["{T\mu^S_X}"', from=10-6, to=10-7]
	\arrow["{TSS\eta^S_X}"', from=10-1, to=10-3]
	\arrow[Rightarrow, no head, from=8-1, to=10-1]
	\arrow["{\delta_X}"', from=10-7, to=10-10]
	\arrow["{\eta_{\Gamma_S,SSX}}", from=5-10, to=7-10]
	\arrow["{SSX \otimes \eta^S_{SSX}}", from=3-10, to=5-10]
	\arrow["{T\mu^S_X}"{description}, from=7-10, to=9-10]
	\arrow["{\delta_X}", from=9-10, to=10-10]
	\arrow["{\mu^S_X \otimes S\mu^S_X}"{description}, from=5-1, to=8-6]
	\arrow["{\mu^S_X \otimes S\mu^S_X}"{description}, from=5-10, to=8-6]
	\arrow["{\mu^S_X \otimes \mu^S_X}", from=4-1, to=5-4]
	\arrow["{SX \otimes \eta^S_{SX}}", from=5-4, to=8-6]
	\arrow["{SX \otimes in_1}", from=2-1, to=3-3]
	\arrow["{SX \otimes \eta_{\Sigma^+,SX}}"{pos=0.9}, from=3-3, to=4-4]
	\arrow["{SX \otimes \mu^S_X}"', from=4-4, to=5-4]
	\arrow["{\Sigma^+X \otimes in_1}", from=1-1, to=1-5]
	\arrow["{\Sigma^+\eta^S_X \otimes \Sigma^+X}", from=1-5, to=1-10]
	\arrow["{\eta_{\Sigma^+,X} \otimes \eta_{\Sigma^+,X}}", from=1-5, to=3-6]
	\arrow["{S\eta^S_X \otimes \eta^S_{SX}}", from=3-6, to=3-10]
	\arrow["{SX \otimes \eta^S_{SX}}", from=3-6, to=8-6]
	\arrow["{S\eta^S_X \otimes \eta^S_{SX}}"', from=3-6, to=5-8]
	\arrow["{SSX \otimes S\eta^S_{SX}}", from=5-8, to=5-10]
	\arrow["{\mu^S_X \otimes SSX}"', from=5-8, to=8-6]
	\arrow["{\Sigma^+X \otimes in_1}"{pos=0.7}, from=1-1, to=2-3]
	\arrow["{\eta_{\Sigma^+,X} \otimes \Sigma^+SX}", from=2-3, to=3-3]
	\arrow["{\Sigma^+X \otimes \Sigma^+\eta^S_X}", from=1-5, to=2-3]
	\arrow["{SX \otimes S\eta^S_X}"', from=3-6, to=4-4]
	\arrow[Rightarrow, no head, from=3-6, to=5-4]
	\arrow["{\eta_{\Gamma_S,SX}}", from=8-6, to=9-10]
	\arrow["{T\mu^S_X}", from=8-1, to=9-6]
	\arrow["{TS\eta^S_X}"', from=9-6, to=10-6]
	\arrow[Rightarrow, no head, from=9-6, to=10-7]
	\arrow[Rightarrow, no head, from=9-6, to=9-10]
\end{tikzcd}}
\end{center}

The right-hand, double question-marked subdiagram has a coproduct as its domain, so we proceed termwise.
The second term is chased as follows.
\begin{center}
  \ajustedroit{
\begin{tikzcd}[ampersand replacement=\&]
	{\Sigma(X\otimes I)} \& {\Sigma^+(X\otimes I)} \& {S(X\otimes I)} \&\& {ST(X\otimes I)} \& {ST(X \otimes I + X + X)} \\
	{\Sigma(SX \otimes \Sigma^+X)} \&\&\&\& {ST(STX \otimes I)} \& {ST(STX \otimes I + X)} \\
	{\Sigma(SX \otimes SX)} \&\& {S(SX \otimes \Sigma^+X)} \&\& {ST(SSTX \otimes \Sigma^+STX)} \& {ST(SSTX \otimes \Sigma^+STX + X)} \\
	{\Sigma(SX \otimes SSX)} \& {S(SX \otimes SX)} \&\&\& {ST(SSTX \otimes SSTX)} \& {ST(SSTX \otimes SSTX + X)} \\
	{\Sigma^+(SX \otimes SSX)} \&\&\&\&\& {ST(SSTX \otimes SSSTX + X)} \\
	{S(SX \otimes SSX)} \&\&\&\& {ST(SSTX \otimes SSSTX)} \\
	\&\&\&\& STTSSTX \& {ST(TSSTX + X)} \\
	{S(SX \otimes SSX + SX + SX)} \& STSX \&\&\& STSTSTX \& {ST(STSTX + X)} \\
	\& SSTX \&\&\& STSTX \\
	\&\& STX
	\arrow["{\eta_{\Sigma^+}}", color={rgb,255:red,214;green,92;blue,92}, from=1-2, to=1-3]
	\arrow["{S\eta^T_{X\otimes I}}", from=1-3, to=1-5]
	\arrow["{ST in_1}", from=1-5, to=1-6]
	\arrow["{ST(\eta^{ST}_X \otimes I + [X,X])}", from=1-6, to=2-6]
	\arrow["{ST(\eta^S_{STX} \otimes in_1 + X)}", from=2-6, to=3-6]
	\arrow["{ST(SSTX \otimes \eta_{\Sigma^+,STX} + X)}", from=3-6, to=4-6]
	\arrow["{ST(SSTX \otimes \eta^S_{SSTX} + X)}", from=4-6, to=5-6]
	\arrow["{ST(\eta_{\Gamma_S,SSTX} + X)}", from=5-6, to=7-6]
	\arrow["{ST(\delta_{STX} + X)}", from=7-6, to=8-6]
	\arrow["{ST[\mu^{ST}_X,\eta^{ST}_X]}", from=8-6, to=9-5]
	\arrow["{\mu^{ST}_X}", from=9-5, to=10-3]
	\arrow["{\Sigma(\eta^S_X \otimes in_1)}"', from=1-1, to=2-1]
	\arrow[from=2-1, to=3-1]
	\arrow["{\Sigma(SX \otimes \eta^S_{SX})}"', from=3-1, to=4-1]
	\arrow[from=8-1, to=8-2]
	\arrow["{S\delta_X}", from=8-2, to=9-2]
	\arrow["{\mu^S_{TX}}"', from=9-2, to=10-3]
	\arrow["{\eta_{\Sigma^+,SX \otimes SSX}}"', from=5-1, to=6-1]
	\arrow["{in_2}"', from=4-1, to=5-1]
	\arrow["{Sin_1}"', from=6-1, to=8-1]
	\arrow["{ST\mu^{ST}_X}"', from=8-5, to=9-5]
	\arrow["{ST(\eta^{ST}_X \otimes I)}", from=1-5, to=2-5]
	\arrow["{ST(\eta^S_{STX} \otimes in_1)}", from=2-5, to=3-5]
	\arrow["{ST(SSTX \otimes \eta_{\Sigma^+,STX})}", from=3-5, to=4-5]
	\arrow["{ST(SSTX \otimes \eta^S_{SSTX})}", from=4-5, to=6-5]
	\arrow["{ST(\eta_{\Gamma_S,SSTX})}", from=6-5, to=7-5]
	\arrow["{ST\delta_{STX}}", from=7-5, to=8-5]
	\arrow["{in_2}", from=1-1, to=1-2]
	\arrow["{S(SX \otimes \eta^S_{SX})}", from=4-2, to=6-1]
	\arrow[from=3-3, to=4-2]
	\arrow["{S(\eta^S_X \otimes in_1)}"', from=1-3, to=3-3]
	\arrow["{S \eta^T(S\eta^{ST}_X \otimes \Sigma^+ \eta^{ST}_X)}"', from=3-3, to=3-5]
	\arrow["{S \eta^T(S\eta^{ST}_X \otimes S\eta^{ST}_X)}"{description}, from=4-2, to=4-5]
	\arrow["{S \eta^T(S\eta^{ST}_X \otimes SS\eta^{ST}_X)}"{description}, from=6-1, to=6-5]
	\arrow["{S\eta_{\Gamma_S,SX}}"{description}, from=6-1, to=8-2]
	\arrow["{S\eta^T TS \eta^{ST}_X}"{description}, from=8-2, to=7-5]
	\arrow["{S\eta^T ST \eta^{ST}_X}"{description}, from=9-2, to=8-5]
	\arrow["{S\eta^T_{STX}}"{description}, from=9-2, to=9-5]
      \end{tikzcd}
      }
\end{center}
The first term is chased as follows.
\begin{center}
\ajustedroit{
  \begin{tikzcd}[ampersand replacement=\&]
    I \&\& {\Sigma^+(X\otimes I)} \& {S(X\otimes I)} \& {ST(X\otimes I)} \& {ST(X \otimes I + X + X)} \\
    \& {\Sigma^+\emptyset} \\
    {\Sigma^+X} \& S\emptyset \&\& {S(STX \otimes I)} \& {ST(STX \otimes I)} \& {ST(STX \otimes I + X)} \\
    \&\&\& {S(SSTX \otimes \Sigma^+STX)} \& {ST(SSTX \otimes \Sigma^+STX)} \& {ST(SSTX \otimes \Sigma^+STX + X)} \\
    \&\&\& {S(SSTX \otimes SSTX)} \& {ST(SSTX \otimes SSTX)} \& {ST(SSTX \otimes SSTX + X)} \\
    \&\& SS\emptyset \& {S(SSTX \otimes SSSTX)} \& {ST(SSTX \otimes SSSTX)} \& {ST(SSTX \otimes SSSTX + X)} \\
    SX \&\&\& STSSTX \& STTSSTX \& {ST(TSSTX + X)} \\
    \& STX \& S\emptyset \& SSTSTX \& STSTSTX \& {ST(STSTX + X)} \\
    SSX \&\&\& SSTX \&\& STSTX \\
    \\
    {S(SX \otimes SSX + SX + SX)} \& STSX \& SSTX \& STX
    \arrow["{in_1}", from=1-1, to=1-3]
    \arrow["{\eta_{\Sigma^+}}", from=1-3, to=1-4]
    \arrow["{S\eta^T_{X\otimes I}}", from=1-4, to=1-5]
    \arrow["{ST in_1}", from=1-5, to=1-6]
    \arrow["{ST(\eta^{ST}_X \otimes I + [X,X])}", from=1-6, to=3-6]
    \arrow["{ST(\eta^S_{STX} \otimes in_1 + X)}", from=3-6, to=4-6]
    \arrow["{ST(SSTX \otimes \eta_{\Sigma^+,STX} + X)}", from=4-6, to=5-6]
    \arrow["{ST(SSTX \otimes \eta^S_{SSTX} + X)}", from=5-6, to=6-6]
    \arrow["{ST(\eta_{\Gamma_S,SSTX} + X)}", from=6-6, to=7-6]
    \arrow["{ST(\delta_{STX} + X)}", from=7-6, to=8-6]
    \arrow["{ST[\mu^{ST}_X,\eta^{ST}_X]}", from=8-6, to=9-6]
    \arrow["{\mu^{ST}_X}", from=9-6, to=11-4]
    \arrow["{in_1}"', from=1-1, to=3-1]
    \arrow[from=3-1, to=7-1]
    \arrow["{\eta^S_{SX}}"', from=7-1, to=9-1]
    \arrow[from=11-1, to=11-2]
    \arrow["{S\delta_X}", from=11-2, to=11-3]
    \arrow["{\mu^S_{TX}}", from=11-3, to=11-4]
    \arrow["{ST\mu^{ST}_X}", from=8-5, to=9-6]
    \arrow["{ST\delta_{STX}}", from=7-5, to=8-5]
    \arrow["{ST\eta_{\Gamma_S,SSTX}}", from=6-5, to=7-5]
    \arrow["{ST(SSTX \otimes \eta^S_{SSTX})}", from=5-5, to=6-5]
    \arrow["{ST(SSTX \otimes \eta_{\Sigma^+,STX})}", from=4-5, to=5-5]
    \arrow["{ST(\eta^S_{STX} \otimes in_1)}", from=3-5, to=4-5]
    \arrow["{ST(\eta^{ST}_X \otimes I)}", from=1-5, to=3-5]
    \arrow["{S(\eta^{ST}_X \otimes I)}", from=1-4, to=3-4]
    \arrow["{S(\eta^S_{STX} \otimes in_1)}", from=3-4, to=4-4]
    \arrow["{S(SSTX \otimes \eta_{\Sigma^+,STX})}", from=4-4, to=5-4]
    \arrow["{S(SSTX \otimes \eta^S_{SSTX})}", from=5-4, to=6-4]
    \arrow["{S\eta_{\Gamma_S,SSTX}}", from=6-4, to=7-4]
    \arrow["{S\delta_{STX}}", from=7-4, to=8-4]
    \arrow["{S\eta^T_{STSTX}}", from=8-4, to=8-5]
    \arrow[from=2-2, to=3-2]
    \arrow["{in_1}", from=1-1, to=2-2]
    \arrow[from=2-2, to=1-3]
    \arrow[from=2-2, to=3-1]
    \arrow["{S{!}}"', from=3-2, to=7-1]
    \arrow["{S{!}}", from=3-2, to=1-4]
    \arrow["{\mu^S_{TX}}", from=9-4, to=11-4]
    \arrow["{S\eta^T_{STX}}", from=9-4, to=9-6]
    \arrow["{S\mu^{ST}_X}", from=8-4, to=9-4]
    \arrow["{S in_3}"', from=9-1, to=11-1]
    \arrow["{S\eta^T_X}"{description}, from=7-1, to=8-2]
    \arrow["{S\eta^T_{SX}}"{description}, from=9-1, to=11-2]
    \arrow["{SS\eta^T_X}"{description}, from=9-1, to=11-3]
    \arrow["{\eta^S_{STX}}"{description}, from=8-2, to=11-3]
    \arrow["{S{!}}"', from=3-2, to=8-2]
    \arrow[Rightarrow, no head, from=8-2, to=11-4]
    \arrow["{SS{!}}", from=6-3, to=9-4]
    \arrow["{\mu^S_\emptyset}", from=6-3, to=8-3]
    \arrow["{S{!}}", from=8-3, to=11-4]
    \arrow["{S\eta^S_\emptyset}", from=3-2, to=6-3]
    \arrow[Rightarrow, no head, from=3-2, to=8-3]
  \end{tikzcd}}
\end{center}

\subsection{Coherence of $R₂$}
We finally check the coherence of $R₂$, namely commutation of the following diagram.
\begin{center}
  \ajustedroit{
\begin{tikzcd}[ampersand replacement=\&]
	{\Sigma^+X \otimes I} \& {I \otimes I + \Sigma X \otimes I} \&\& {I + \Sigma(X\otimes I)} \& {\Sigma^+(X\otimes I)} \& {S(X\otimes I)} \& {ST(X\otimes I)} \& {ST(X \otimes I + X + X)} \\
	{\red{S}X \otimes I} \&\&\&\&\&\&\& {ST(STX \otimes I + X)} \\
	SX \&\&\&\&\&\&\& {ST(STX + X)} \\
	{ST\red{S}X} \&\&\&\&\&\&\& {ST(STSTX + X)} \\
	{S\red{S}TX} \&\&\&\&\&\&\& STSTX \\
	\&\&\&\& STX
	\arrow[Rightarrow, no head, from=1-1, to=1-2]
	\arrow["{\lambda_I + st_{X,(I,id_I)}}", from=1-2, to=1-4]
	\arrow[Rightarrow, no head, from=1-4, to=1-5]
	\arrow["{\eta_{\Sigma^+}}", from=1-5, to=1-6]
	\arrow["{S\eta^T_{X\otimes I}}", from=1-6, to=1-7]
	\arrow["{ST in_1}", from=1-7, to=1-8]
	\arrow["{\eta_{\Sigma^+,X} \otimes I}"', from=1-1, to=2-1]
	\arrow["{S\delta_X}"', from=4-1, to=5-1]
	\arrow["{\mu^S_{TX}}"', from=5-1, to=6-5]
	\arrow["{ST(\eta^{ST}_X \otimes I + [X,X])}", from=1-8, to=2-8]
	\arrow["{ST[\mu^{ST}_X,\eta^{ST}_X]}", from=4-8, to=5-8]
	\arrow["{\mu^{ST}_X}", from=5-8, to=6-5]
	\arrow["{\rho_{SX}}"', from=2-1, to=3-1]
	\arrow["{\eta^{ST}_{SX}}"', from=3-1, to=4-1]
	\arrow["{ST(\rho_{STX} + X)}", from=2-8, to=3-8]
	\arrow["{ST(\eta^{ST}_{STX} + X)}", from=3-8, to=4-8]
      \end{tikzcd}
    }
\end{center}
      Again, the domain is a coproduct, so we proceed termwise.

      The first term is chased as follows.
      \begin{center}
        \ajustedroit{
\begin{tikzcd}[ampersand replacement=\&]
	{I \otimes I} \&\&\& I \& {\Sigma^+(X\otimes I)} \& {S(X\otimes I)} \& {ST(X\otimes I)} \& {ST(X \otimes I + X + X)} \\
	{\Sigma^+X \otimes I} \& I \& {\Sigma^+ \emptyset} \&\&\& {S(STX \otimes I)} \&\& {ST(STX \otimes I + X)} \\
	{SX \otimes I} \& {\Sigma^+ X} \&\& {S \emptyset} \\
	\&\&\&\& SS\emptyset \& SSTX \&\& {ST(STX + X)} \\
	SX \&\&\& SSX \\
	\& SSX \&\& {} \&\& SSTSTX \\
	STSX \&\&\&\& S\emptyset \&\& STSTSTX \& {ST(STSTX + X)} \\
	\&\&\& SX \&\& STSTX \&\& STSTX \\
	SSTX \\
	\\
	\&\&\&\& STX
	\arrow["{in_1 \otimes I}"', from=1-1, to=2-1]
	\arrow["{\lambda_I}", from=1-1, to=1-4]
	\arrow["{in_1}", from=1-4, to=1-5]
	\arrow["{\eta_{\Sigma^+}}", from=1-5, to=1-6]
	\arrow["{S\eta^T_{X\otimes I}}", from=1-6, to=1-7]
	\arrow["{ST in_1}", from=1-7, to=1-8]
	\arrow["{\eta_{\Sigma^+,X} \otimes I}"', from=2-1, to=3-1]
	\arrow["{S\delta_X}"', from=7-1, to=9-1]
	\arrow["{\mu^S_{TX}}"', from=9-1, to=11-5]
	\arrow["{ST(\eta^{ST}_X \otimes I + [X,X])}", from=1-8, to=2-8]
	\arrow["{ST[\mu^{ST}_X,\eta^{ST}_X]}", from=7-8, to=8-8]
	\arrow["{\mu^{ST}_X}", from=8-8, to=11-5]
	\arrow["{\rho_{SX}}"', from=3-1, to=5-1]
	\arrow["{\eta^{ST}_{SX}}"', from=5-1, to=7-1]
	\arrow["{ST(\rho_{STX} + X)}", from=2-8, to=4-8]
	\arrow["{ST(\eta^{ST}_{STX} + X)}", from=4-8, to=7-8]
	\arrow["{S (\eta^{ST}_X \otimes I)}", from=1-6, to=2-6]
	\arrow["{S\rho_{STX}}", from=2-6, to=4-6]
	\arrow["{S\eta^{ST}_{STX}}", from=4-6, to=6-6]
	\arrow["{S\eta^T}", from=6-6, to=7-7]
	\arrow["{ST\mu^{ST}}", from=7-7, to=8-8]
	\arrow["{\mu^{ST}}", from=7-7, to=8-6]
	\arrow["{\mu^{ST}_X}", from=8-6, to=11-5]
	\arrow["{\mu^S_{TSTX}}"{description}, from=6-6, to=8-6]
	\arrow["{\mu^S_{TX}}"{description}, from=4-6, to=11-5]
	\arrow["{\rho_I}"{description}, from=1-1, to=2-2]
	\arrow["{in_1}", from=2-2, to=3-2]
	\arrow["{\eta_{\Sigma^+,X}}"{pos=0.1}, from=3-2, to=5-1]
	\arrow[Rightarrow, no head, from=2-2, to=1-4]
	\arrow["{in_1}"', from=2-2, to=2-3]
	\arrow["{\Sigma^+{!}}", from=2-3, to=3-2]
	\arrow["{\Sigma^+{!}}"', from=2-3, to=1-5]
	\arrow["{\eta_{\Sigma^+,\emptyset}}", from=2-3, to=3-4]
	\arrow["{S{!}}", from=3-4, to=5-1]
	\arrow["{S{!}}", from=3-4, to=1-6]
	\arrow["{SS{!}}", from=4-5, to=4-6]
	\arrow["{\mu^S_\emptyset}", from=4-5, to=7-5]
	\arrow["{S{!}}"', from=7-5, to=11-5]
	\arrow["{S\eta^S_\emptyset}", from=3-4, to=4-5]
	\arrow["{\eta^S_{SX}}"', from=5-1, to=6-2]
	\arrow["{S\eta^T_{SX}}"', from=6-2, to=7-1]
	\arrow["{SS\eta^T_X}", from=6-2, to=9-1]
	\arrow["{\mu^S_X}", from=6-2, to=8-4]
	\arrow["{S\eta^T_X}"', from=8-4, to=11-5]
	\arrow["{S{!}}", from=7-5, to=8-4]
	\arrow["{SS{!}}"', from=4-5, to=5-4]
	\arrow["{S\eta^S_X}"', from=5-1, to=5-4]
	\arrow["{\mu^S_X}", from=5-4, to=8-4]
\end{tikzcd}}
      \end{center}

      The second term is chased as follows.

      \begin{center}
        \ajustedroit{
          \begin{tikzcd}[ampersand replacement=\&]
            {\Sigma X \otimes I} \&\&\& {\Sigma(X\otimes I)} \& {\Sigma^+(X\otimes I)} \& {S(X\otimes I)} \& {ST(X\otimes I)} \& {ST(X \otimes I + X + X)} \\
            {\Sigma^+X \otimes I} \& {\Sigma X} \&\&\&\& {S(STX \otimes I)} \&\& {ST(STX \otimes I + X)} \\
            {\red{S}X \otimes I} \& {\Sigma^+ X} \\
            \&\&\&\&\&\&\& {ST(STX + X)} \\
            SX \&\&\&\&\& SSTX \\
            \&\&\&\&\& SSTSTX \& STSTSTX \\
            {ST\red{S}X} \&\&\&\&\& STSTX \&\& {ST(STSTX + X)} \\
            {S\red{S}TX} \&\&\&\&\&\&\& STSTX \\
            \\
            \\
            \&\&\&\& STX
            \arrow["{in_2}"', from=1-1, to=2-1]
            \arrow["{st_{X,(I,id_I)}}", from=1-1, to=1-4]
            \arrow["{\eta_{\Sigma^+}}", from=1-5, to=1-6]
            \arrow["{S\eta^T_{X\otimes I}}", from=1-6, to=1-7]
            \arrow["{ST in_1}", from=1-7, to=1-8]
            \arrow["{\eta_{\Sigma^+,X} \otimes I}"', from=2-1, to=3-1]
            \arrow["{S\delta_X}"', from=7-1, to=8-1]
            \arrow["{\mu^S_{TX}}"', from=8-1, to=11-5]
            \arrow["{ST(\eta^{ST}_X \otimes I + [X,X])}", from=1-8, to=2-8]
            \arrow["{ST[\mu^{ST}_X,\eta^{ST}_X]}", from=7-8, to=8-8]
            \arrow["{\mu^{ST}_X}", from=8-8, to=11-5]
            \arrow["{\rho_{SX}}"', from=3-1, to=5-1]
            \arrow["{\eta^{ST}_{SX}}"', from=5-1, to=7-1]
            \arrow["{ST(\rho_{STX} + X)}", from=2-8, to=4-8]
            \arrow["{ST(\eta^{ST}_{STX} + X)}", from=4-8, to=7-8]
            \arrow["{\rho_{\Sigma X}}", from=1-1, to=2-2]
            \arrow["{in_2}"', from=2-2, to=3-2]
            \arrow["{\eta_{\Sigma^+,X}}"{pos=0.3}, from=3-2, to=5-1]
            \arrow["{in_2}", from=1-4, to=1-5]
            \arrow["{\Sigma\rho_X}", from=1-4, to=2-2]
            \arrow["{\Sigma^+\rho_X}", from=1-5, to=3-2]
            \arrow["{S\rho_X}", from=1-6, to=5-1]
            \arrow["{S(\eta^{ST}_X \otimes I)}", from=1-6, to=2-6]
            \arrow["{S\rho_{STX}}", from=2-6, to=5-6]
            \arrow["{S\eta^{ST}_{STX}}"', from=5-6, to=6-6]
            \arrow["{\mu^{ST}_X}", from=7-6, to=11-5]
            \arrow["{S\eta^{ST}_X}", from=5-1, to=5-6]
            \arrow["{S\eta^T_{STSTX}}", from=6-6, to=6-7]
            \arrow["{ST\mu^{ST}_X}", from=6-7, to=8-8]
            \arrow["{\mu^{ST}_{STSTX}}", from=6-7, to=7-6]
            \arrow["{\mu^S_{TSTX}}"', from=6-6, to=7-6]
            \arrow["{S\eta^T_X}", from=5-1, to=11-5]
          \end{tikzcd}
          }
      \end{center}
\end{document}